\documentclass[a4paper,10pt,numbers=noenddot,headtopline,headsepline]{scrbook}
\pdfoutput=1

\usepackage[utf8]{inputenc}
\usepackage[T1]{fontenc}
\usepackage[english]{babel}
\usepackage[babel]{csquotes}
\usepackage{textcomp}
\usepackage{amssymb}
\usepackage{amsmath}
\usepackage{amsfonts}
\usepackage{amsopn}
\usepackage{amsbsy}
\usepackage{color}
\usepackage{xcolor}
\usepackage{enumerate}
\usepackage{graphics}
\usepackage{graphicx}
\numberwithin{equation}{section} 
\numberwithin{table}{section} 
\numberwithin{figure}{section} 
\colorlet{chapter}{black!75}
\addtokomafont{chapter}{\color{chapter}}

\makeatletter% siehe De-TeX-FAQ
\renewcommand*{\chapterformat}{%
\begingroup% damit \unitlength-Änderung lokal bleibt
\setlength{\unitlength}{1mm}%
\begin{picture}(20,40)(0,5)%
\setlength{\fboxsep}{0pt}%
\put(20,15){\line(1,0){\dimexpr
\textwidth-20\unitlength\relax\@gobble}}%
\put(0,0){\makebox(20,20)[r]{%
\fontsize{28\unitlength}{28\unitlength}\selectfont\thechapter
\kern-.04em% Ziffer in der Zeichenzelle nach rechts schieben
}}%
\put(20,15){\makebox(\dimexpr
\textwidth-20\unitlength\relax\@gobble,\ht\strutbox\@gobble)[l]{%
\ \normalsize\color{black}\chapapp~\thechapter\autodot
}}%
\end{picture} % <-- Leerzeichen ist hier beabsichtigt!
\endgroup
}
\usepackage[pdftex,unicode=true,naturalnames,plainpages=false,pdfpagelabels,breaklinks=true,pdfborderstyle={/S/U/W 1.2}]{hyperref}
\usepackage[amsmath,thmmarks,thref,hyperref]{ntheorem}

\theoremstyle{plain}
\newtheorem{theorem}{Theorem}[section]
\newtheorem{definition}[theorem]{Definition}
\newtheorem{lemma}[theorem]{Lemma}
\newtheorem{corollary}[theorem]{Corollary}
\newtheorem{proposition}[theorem]{Proposition}

\theorembodyfont{\upshape}
\newtheorem{remark}[theorem]{Remark}

\theoremstyle{nonumberplain}
\newtheorem{example}{Example}

\theoremsymbol{\ensuremath{_\Box}}
\newtheorem{proof}{Proof}
\usepackage[scaled]{beramono}
\usepackage{helvet}
\usepackage[charter]{mathdesign} % math font has ugly \mathcals
\SetMathAlphabet{\mathcal}{normal}{OMS}{cmsy}{m}{n} % fixes ugly \mathcals
\SetMathAlphabet{\mathcal}{bold}{OMS}{cmsy}{m}{n} % fixes ugly \mathcals
\usepackage{units}
\usepackage{tabularx}
\usepackage[style=alphabetic,
	citestyle=alphabetic,
	firstinits=true,
	backend=biber,
	hyperref=auto,
	sorting=nyt,
	maxcitenames=3,
	maxbibnames=99,
	minnames=3,
	arxiv=pdf,
	url=false,
	isbn=false,
	dateabbrev=true,
	datezeros=true,
	bibencoding=utf8]{biblatex}
\newbibmacro{string+doi}[1]{%
	\iffieldundef{doi}{#1}{\href{http://dx.doi.org/\thefield{doi}}{#1}}}
\DeclareFieldFormat
	{title}{\usebibmacro{string+doi}{\mkbibemph{#1}}\isdot}
\DeclareFieldFormat
	[article,inbook,incollection,inproceedings,patent,thesis,unpublished]
	{title}{\usebibmacro{string+doi}{\mkbibemph{#1}}\isdot}
\DeclareFieldFormat
	[article,inbook,incollection,inproceedings,patent,thesis,unpublished]
	{pages}{#1}
\DeclareFieldFormat
	[article]
	{journaltitle}{#1\isdot}
\DeclareFieldFormat
	[article]
	{volume}{\textbf{#1}}
\DefineBibliographyStrings{english}{pages={}}
\DeclareBibliographyDriver{article}{%
	\usebibmacro{bibindex}%
	\usebibmacro{begentry}%
	\usebibmacro{author/translator+others}%
	\setunit{\labelnamepunct}\newblock
	\usebibmacro{title}%
	\usebibmacro{journal+issuetitle}%
	\setunit*{\addcomma\space}
	\printfield{pages}
	\setunit*{\addcomma\space}
	\printfield{year}
	\usebibmacro{finentry}%
}

\newbibmacro*{journal+issuetitle}{%
	\usebibmacro{journal}%
	\setunit*{\addspace}%
	\iffieldundef{series}
	{}
	{\newunit
		\printfield{series}%
		\setunit{\addspace}}%
	\iffieldundef{volume}
		{}
		{\printfield{volume}%
		\setunit{\addspace}}%
	\setunit*{\addcolon\space}%
	\usebibmacro{issue}%
	\newunit}
\providecommand{\ie}{i.~e.~}
\providecommand{\eg}{e.~g.~}
\providecommand{\cf}{cf.~}

\providecommand{\R}{\mathbb{R}}
\providecommand{\Q}{\mathbb{Q}}
\providecommand{\C}{\mathbb{C}}
\renewcommand{\C}{\mathbb{C}}

\providecommand{\T}{\mathbb{T}}
\renewcommand{\T}{\mathbb{T}}
\providecommand{\N}{\mathbb{N}}
\providecommand{\Z}{\mathbb{Z}}
\providecommand{\ii}{\mathrm{i}}
\providecommand{\e}{\mathrm{e}}
\renewcommand{\Re}{\mathrm{Re} \,}
\renewcommand{\Im}{\mathrm{Im} \,}
\providecommand{\Hil}{\mathcal{H}}
\providecommand{\eps}{\varepsilon}
\providecommand{\Cont}{\mathcal{C}}
\providecommand{\image}{\mathrm{im} \, }
\providecommand{\ker}{\mathrm{ker} \, }
\providecommand{\ran}{\mathrm{ran} \, }
\providecommand{\im}{\mathrm{im} \, }
\providecommand{\span}{\mathrm{span} \,}
\providecommand{\supp}{\mathrm{supp} \,}
\providecommand{\ker}{\mathrm{ker} \,}
\providecommand{\det}{\mathrm{det} \,}
\providecommand{\trace}{\mathrm{Tr} \,}

\providecommand{\dd}{\mathrm{d}}
\providecommand{\id}{\mathrm{id}}

\providecommand{\order}{\mathcal{O}}
\providecommand{\Fourier}{\mathcal{F}}

\providecommand{\trace}{\mathrm{Tr}}

\providecommand{\abs}[1]{\left \lvert #1 \right \rvert}
\providecommand{\sabs}[1]{\lvert #1 \vert}
\providecommand{\babs}[1]{\bigl \lvert #1 \bigr \rvert}
\providecommand{\Babs}[1]{\Bigl \lvert #1 \Bigr \rvert}
\providecommand{\norm}[1]{\left \lVert #1 \right \rVert}
\providecommand{\snorm}[1]{\lVert #1 \rVert}
\providecommand{\bnorm}[1]{\bigl \lVert #1 \bigr \rVert}
\providecommand{\Bnorm}[1]{\Bigl \lVert #1 \Bigr \rVert}
\providecommand{\scpro}[2]{\left \langle #1 , #2 \right \rangle}
\providecommand{\sscpro}[2]{\langle #1 , #2 \rangle}
\providecommand{\bscpro}[2]{\bigl \langle #1 , #2 \bigr \rangle}
\providecommand{\Bscpro}[2]{\Bigl \langle #1 , #2 \Bigr \rangle}
\providecommand{\ket}[1]{\left \vert #1 \right \rangle}

\providecommand{\ipro}[2]{\left \langle #1 \vert #2 \right \rangle}

\providecommand{\opro}[2]{\left \vert #1 \right \rangle \left \langle #2 \right \vert}
\providecommand{\sopro}[2]{\vert #1 \rangle \langle #2 \vert}
\providecommand{\bopro}[2]{\bigl \vert #1 \bigr \rangle \bigl \langle #2 \bigr \vert}

\providecommand{\exp}{\mathrm{exp}}
\providecommand{\Schwartz}{\mathcal{S}}
\providecommand{\BCont}{\Cont_{\mathrm{b}}}
\providecommand{\wastlim}{{\mathrm{w}^{\ast}\mbox{-}\lim}}

\providecommand{\Pspace}{\R^{2n}}
\providecommand{\spec}{\mathrm{spec}}
\providecommand{\Rot}{\mathrm{Rot}}
\providecommand{\Sone}{\mathbb{S}^1}
\providecommand{\domain}{\mathcal{D}}
\bibliography{/Users/max/Library/texmf/tex/latex/max/bibliography.bib}
\title{Differential Equations in \\ Mathematical Physics}
\subtitle{Lecture Notes\\ APM351Y}

\author{\href{mailto:max.lein@utoronto.ca}{Max Lein}}
\date{\today} 

\publishers{\vspace{20mm}
Department of Mathematics \\
University of Toronto
} 

\begin{document}

\frontmatter
\maketitle

\tableofcontents

\mainmatter

\chapter{Introduction} % (fold)
\label{intro}
Physical theories are usually formulated as differential equations. Some such examples are Hamilton's equation of motion (describing the classical motion of a point particle) 
\begin{align}
	\frac{\dd}{\dd t} \left (
	\begin{matrix}
		q \\
		p \\
	\end{matrix}
	\right ) &= \left (
	\begin{matrix}
		+ p \\
		- \nabla_q V \\
	\end{matrix}
	\right )
	, 
	\label{intro:eqn:hamiltons_eom}
\end{align}
where $\nabla_q V = (\partial_{q_1} V , \ldots , \partial_{q_n} V)$ is the gradient of the potential $V$, the Schrödinger equation 
\begin{align}
	\ii \partial_t \Psi &= \bigl ( - \tfrac{1}{2m} \Delta_x + V \bigr ) \Psi 
	, 
	\label{intro:eqn:Schroedinger_eqn}
\end{align}
where $\Delta_x = \sum_{j = 1}^n \partial_{x_j}^2$ is the Laplace operator, the wave equation 
\begin{align}
	\partial_t^2 u - \Delta_x u = 0 
	, 
	\label{intro:eqn:wave_eqn}
\end{align}
the heat equation (with source) 
\begin{align}
	\partial_t u &= + \Delta_x u + f(t)
	\label{intro:eqn:heat_eqn}
\end{align}
and the Maxwell equations ($\nabla_x \times \mathbf{E}$ is the rotation of $\mathbf{E}$)
\begin{align}
	\frac{\dd}{\dd t} \left (
	\begin{matrix}
		\mathbf{E} \\
		\mathbf{H} \\
	\end{matrix}
	\right ) &= \left (
	\begin{matrix}
		+ \eps^{-1} \nabla_x \times \mathbf{H} \\
		- \mu^{-1} \nabla_x \times \mathbf{E} \\
	\end{matrix}
	\right ) - \left (
	\begin{matrix}
		\eps^{-1} j \\
		0 \\
	\end{matrix}
	\right )
	\label{intro:eqn:Maxwell_dynamical_eqns}
	\\
	\left (
	\begin{matrix}
		\nabla_x \cdot \eps \mathbf{E} \\
		\nabla_x \cdot \mu \mathbf{H} \\
	\end{matrix}
	\right ) &= \left (
	\begin{matrix}
		\rho \\
		0 \\
	\end{matrix}
	\right )
	\label{intro:eqn:Maxwell_conservation_laws}
	. 
\end{align}
Of course, there are many, many more examples, so it helps to classify them systematically. In the zoology of differential equations, the first distinction we need to make is between \emph{ordinary differential equations} (ODEs) and \emph{partial differential equations} (PDEs), \ie differential equations which involve derivatives of only one variable (\eg \eqref{intro:eqn:hamiltons_eom}) or several variables (all other examples). The order of a differential equation is determined by the highest derivative involved in it. 

On a second axis, we need to distinguish between non-linear (\eg \eqref{intro:eqn:hamiltons_eom} for $H(q,p) = \tfrac{1}{2m} p^2 + V(q)$ where $V$ is not a second-order polynomial) or linear equations (\eg \eqref{intro:eqn:hamiltons_eom} for $V(q) = \omega^2 \, q^2$ or \eqref{intro:eqn:Maxwell_conservation_laws}). 

We can write any differential equation in the form
\begin{align*}
	L(u) = f
\end{align*}
where the contribution $L(u)$ which depends on the solution $u$ is collected on the left-hand side and the $u$-\emph{in}dependent part $f$ is on the right-hand side. 

\emph{Linear} differential equations are the distinguished case where the \emph{operator} $L$ satisfies 
\begin{align}
	L \bigl ( \alpha_1 u_1 + \alpha_2 u_2 \bigr ) &= \alpha_1 \, L(u_1) + \alpha_2 \, L(u_2) 
	\label{intro:eqn:linearity} 
\end{align}
for all scalars $\alpha_1 , \alpha_2$ and suitable functions $u_1 , u_2$; otherwise, the differential equation is \emph{non-linear}. Among \emph{linear} differential equations we further distinguish between \emph{homogeneous} ($f = 0$) and \emph{inhomogeneous} ($f \neq 0$) linear differential equations. 

Solving linear differential equations is much, much easier, because linear combinations of solutions of the homogeneous equation $L(u) = 0$ are once again solutions of the homogeneous equation. In other words, \emph{the solutions form a vector space.} This makes it easier to find solutions which satisfy the correct initial conditions by, for instance, systematically finding all solutions to $L(u) = 0$ and then forming suitable linear combinations. 

However, there are cases when solving a non-linear problem may be more desirable. In case of many-particle quantum mechanics, a non-linear problem on a lower-dimensional space is often preferable to a high-dimensional linear problem. 

Secondly, one can often relate easier-to-solve \emph{ordinary} differential equations to \emph{partial} differential equations in a systematic fashion, \eg by means of semiclassical limits which relate \eqref{intro:eqn:hamiltons_eom} and \eqref{intro:eqn:Schroedinger_eqn}, by “diagonalizing” a PDE or considering an associated “eigenvalue problem”. 

A last important weapon in the arsenal of a mathematical physicist is to systematically exploit \emph{symmetries}.

\paragraph{Well-posedness} % (fold)
The most fundamental question one may ask about a differential equation on a domain of functions (which may include boundary conditions and such) is whether it is a \emph{well-posed problem}, \ie 
\begin{enumerate}[(1)]
	\item whether \emph{a} (non-trivial) solution \emph{exists}, 
	\item whether the solution is \emph{unique} and 
	\item whether the solution depends on the initial conditions in a continuous fashion (\emph{stability}). 
\end{enumerate}
Solving this sounds rather like an exercise, but it can be tremendously hard. (Proving the \href{http://www.claymath.org/millennium/Navier-Stokes_Equations/}{well-posedness of the Navier-Stokes equation} is one of the -- unsolved -- Millenium Prize Problems, for instance.) Indeed, there are countless examples where a differential equation either has no non-trivial solution\footnote{For linear-homogeneous differential equations, if the zero function is in the domain, it is automatically a solution. Hence, the zero function is often referred to as “trivial solution”.} or the solution is not unique; we will come across some of these cases in the exercises. 
% paragraph Well-posedness (end)

\paragraph{Course outline} % (fold)
This course is located at the intersection of mathematics and physics, so one of the tasks is to establish a dictionary between the mathematics and physics community. Both communities have benefitted from each other tremendously over the course of history: physicists would often generate new problems for mathematicians while mathematicians build and refine new tools to analyze problems from physics. 

Concretely, the course is structured so as to show the interplay between different fields of mathematics (\eg functional analysis, harmonic analysis and the theory of Schrödinger operators) as well as to consider different aspects of solving differential equations. 

However, the course is \emph{not} meant to be a comprehensive introduction to any of these fields in particular, but just give an overview, elucidate some of the connections and whet the appetite for more. 
% paragraph Course outline (end)
% chapter Introduction (end)
\chapter{Ordinary differential equations} % (fold)
\label{odes}
An ordinary differential equation -- or ODE for short -- is an equation of the form 
\begin{align}
	\tfrac{\dd}{\dd t} x &= \dot{x} 
	= F(x) 
	, 
	&&
	x(0) = x_0 \in U \subseteq \R^n 
	, 
	\label{odes:eqn:primordial_ode}
\end{align}
defined in terms of a \emph{vector field} $F = (F_1 , \ldots , F_n) : U \subseteq \R^n \longrightarrow \R^n$. Its solutions are curves $x = (x_1 , \ldots , x_n) : (-T,+T) \longrightarrow U$ in $\R^n$ for times up to $0 < T \leq +\infty$. You can also consider ODEs on other spaces, \eg $\C^d$. 

The solution to an ODE is a \emph{flow} $\Phi : (-T,+T) \times U \longrightarrow U$, \ie the map which satisfies 
\begin{align*}
	\Phi_t(x_0) = x(t) 
\end{align*}
where $x(t)$ is the solution of \eqref{odes:eqn:primordial_ode}. \emph{Locally}, $\Phi$ is a “group representation”\footnote{If $T = +\infty$, it really is a group representation of $(\R,+)$.} 
\begin{enumerate}[(i)]
	\item $\Phi_0 = \id_U$
	\item $\Phi_{t_1} \circ \Phi_{t_2} = \Phi_{t_1 + t_2}$ as long as $t_1 , t_2 , t_1 + t_2 \in (-T,+T)$
	\item $\Phi_t \circ \Phi_{-t} = \Phi_0 = \id_U$ for all $t \in (-T,+T)$
\end{enumerate}
The ODE \eqref{odes:eqn:primordial_ode} uniquely determines the flow $\Phi$ and vice versa: on the one hand, we can \emph{define} the flow map $\Phi_t (x_0) = x(t)$ for each initial condition $x_0$ and $t$ as the solution of \eqref{odes:eqn:primordial_ode}. Property~(i) follows from $x(0) = x_0$. The other two properties follow from the explicit construction of the solution later on in the proof of Theorem~\ref{odes:thm:Picard_Lindeloef}

On the other hand, assume we are given a flow $\Phi$ with properties (i)--(iii) above which is differentiable in $t$. Then we can recover the vector field $F$ at point $x_0$ by taking the time derivative of the flow, 
\begin{align*}
	\left . \frac{\dd}{\dd t} \Phi_t(x_0) \right \vert_{t = 0} &= \lim_{\delta \to 0} \tfrac{1}{\delta} \bigl ( \Phi_{\delta}(x_0) - \Phi_0(x_0) \bigr ) 
	\\
	&
	= \lim_{\delta \to 0} \tfrac{1}{\delta} \bigl ( \Phi_{\delta}(x_0) - x_0 \bigr ) 
	\\
	&
	=: F \bigl ( \Phi_0(x_0) \bigr ) = F(x_0)
	. 
\end{align*}
Now by property (ii), the above also holds at any later point in time: 
\begin{align*}
	\frac{\dd}{\dd t} \Phi_t(x_0) &= \lim_{\delta \to 0} \tfrac{1}{\delta} \bigl ( \Phi_{t + \delta}(x_0) - \Phi_t(x_0) \bigr ) 
	\\
	&
	\overset{(ii)}{=} \lim_{\delta \to 0} \frac{1}{\delta} \Bigl ( \Phi_{\delta} \bigl ( \underbrace{\Phi_t(x_0)}_{= y} \bigr ) - \Phi_t(x_0) \Bigr ) 
	\\
	&
	= \lim_{\delta \to 0} \tfrac{1}{\delta} \bigl ( \Phi_{\delta}(y) - y \bigr ) 
	\\
	&= F(y) 
	= F \bigl ( \Phi_t(x_0) \bigr ) 
\end{align*}
In other words, we have just shown 
\begin{align*}
	\frac{\dd}{\dd t} \Phi_t(x_0) &= F \bigl ( \Phi_t(x_0) \bigr ) 
	, 
	&&
	\Phi_0(x_0) = x_0 
	, 
	\\
	\Leftrightarrow \;
	\frac{\dd}{\dd t} \Phi_t &= F \circ \Phi_t 
	, 
	&&
	\Phi_0 = \id 
	. 
\end{align*}
Clearly, there are three immediate questions: \marginpar{2013.09.10}
\begin{enumerate}[(1)]
	\item When does the flow exist? 
	\item Is it unique? 
	\item How large can we make $T$? 
\end{enumerate}

\section{Linear ODEs} % (fold)
\label{odes:examples}
The simplest ODEs are \emph{linear} where the vector field 
\begin{align*}
	F(x) = H x
\end{align*}
is defined in terms of a $n \times n$ matrix $H = \bigl ( H_{jk} \bigr )_{1 \leq j , k \leq n} \in \mathrm{Mat}_{\C}(n)$. The flow $\Phi_t = \e^{t H}$ is given in terms of the \emph{matrix exponential}
\begin{align}
	\e^{t H} = \exp(t H) := \sum_{k = 0}^{\infty} \frac{t^k}{k!} \, H^k 
	. 
\end{align}
This series converges in the vector space $\mathrm{Mat}_{\C}(n) \cong \C^{n^2}$: if we choose 
\begin{align*}
	\norm{H} := \max_{1 \leq j , k \leq n} \abs{H_{jk}}
\end{align*}
as norm\footnote{We could have picked \emph{any} other norm on $\C^{n^2}$ since all norms on finite-dimensional vector spaces are equivalent.}, we can see that the sequence of partial sums $S_N := \sum_{k = 0}^N \frac{t^k}{k!} \, H^k$ is a Cauchy sequence, 
\begin{align*}
	\bnorm{S_N - S_K} &\leq \norm{\sum_{k = K+1}^N \frac{t^k}{k!} \, H^k} 
	\leq \sum_{k = K+1}^N \frac{t^k}{k!} \, \norm{H^k}
	\\
	&\leq \sum_{k = K+1}^N \frac{t^k}{k!} \, \norm{H}^k
	\leq \sum_{k = K+1}^{\infty} \frac{t^k}{k!} \, \norm{H}^k
	\\
	&\xrightarrow{K \to \infty} 0 
	. 
\end{align*}
Using the completeness of $\C^{n^2}$ with respect to the the norm $\norm{\cdot}$ (which follows from the completeness of $\C$ with respect to the absolute value), we now deduce that $S_N \rightarrow \e^{t H}$ as $N \to \infty$. Moreover, we also obtain the norm bound $\bnorm{\e^{t H}} \leq \e^{t \norm{H}}$. 
\medskip

\noindent
From linear algebra we know that \emph{any} matrix $H$ is similar to a matrix of the form 
\begin{align*}
	J = U^{-1} \, H \, U 
	= \left (
	\begin{matrix}
		J_{r_1}(\lambda_1) & 0 & 0 \\
		0 & \ddots & 0 \\
		0 & 0 & J_{r_N}(\lambda_N) \\
	\end{matrix}
	\right )
\end{align*}
where the $\{ \lambda_j \}_{j = 1}^N$ are the eigenvalues of $H$ and the 
\begin{align*}
	J_{r_j}(\lambda_j) = \left (
	\begin{matrix}
		\lambda_j & 1 & 0 & 0 \\
		0 & \ddots & \ddots & 0 \\
		\vdots & \ddots & \lambda_j & 1 \\
		0 & \cdots & 0 & \lambda_j \\
	\end{matrix}
	\right ) 
	=: \lambda \, \id_{\C^{r_j}} + N_{r_j}
	\in \mathrm{Mat}_{\C}(r_j)
\end{align*}
are the $r_j$-dimensional Jordan blocks associated to $\lambda_j$. Now the exponential 
\begin{align*}
	\e^{t H} &= \sum_{k = 0}^{\infty} \frac{t^k}{k!} \, \bigl ( U \, J \, U^{-1} \bigr )^k 
	= U \, \Biggl ( \sum_{k = 0}^{\infty} \frac{t^k}{k!} \, J^k \Biggr ) \, U^{-1} 
	\\
	&= U \, \e^{t J} \, U^{-1} 
\end{align*}
is also similar to the exponential of $t J$. Moreover, one can see that 
\begin{align*}
	\e^{t J} &= \left (
	\begin{matrix}
		\e^{t J_{r_1}(\lambda_1)} & 0 & 0 \\
		0 & \ddots & 0 \\
		0 & 0 & \e^{t J_{r_N}(\lambda_n)} \\
	\end{matrix}
	\right )
\end{align*}
and the matrix exponential of a Jordan block can be computed explicitly (\cf problem~6 on Sheet~02), 
\begin{align*}
	\e^{t J_r(\lambda)} &= \e^{t \lambda} 
	\, \left (
	\begin{matrix}
		1 & t & \frac{t^2}{2} & \cdots & \cdots & \frac{t^{r-1}}{(r-1)!} \\
		0 & 1 & t & \frac{t^2}{2} & \cdots & \frac{t^{r-2}}{(r-2)!} \\
		0 & 0 & \ddots & \ddots & \ddots & \vdots \\
		0 & 0 & 0 & 1 & t & \frac{t^2}{2} \\
		0 & 0 & 0 & 0 & 1 & t \\
		0 & 0 & 0 & 0 & 0 & 1 \\
	\end{matrix}
	\right )
	. 
\end{align*}
In case $H$ is similar to a diagonal matrix, 
\begin{align*}
	H = U \, \left (
	\begin{matrix}
		\lambda_1 & 0 & 0 \\
		0 & \ddots & 0 \\
		0 & 0 & \lambda_n \\
	\end{matrix}
	\right ) \, U^{-1} 
	, 
\end{align*}
this formula simplifies to 
\begin{align*}
	\e^{t H} = U \, \left (
	\begin{matrix}
		\e^{t \lambda_1} & 0 & 0 \\
		0 & \ddots & 0 \\
		0 & 0 & \e^{t \lambda_n} \\
	\end{matrix}
	\right ) \, U^{-1} 
	. 
\end{align*}
\begin{example}[Free Schrödinger equation]
	The solution to 
	\begin{align*}
		\ii \tfrac{\dd}{\dd t} \widehat{\psi}(t) &= \tfrac{1}{2} k^2 \, \widehat{\psi}(t) 
		, 
		&&
		\widehat{\psi}(0) = \widehat{\psi}_0 \in \C 
		, 
	\end{align*}
	is $\widehat{\psi}(t) = \e^{- \ii t \, \frac{1}{2} k^2} \widehat{\psi}_0$. 
\end{example}
\paragraph{Inhomogeneous linear ODEs} % (fold)
In case the linear ODE is inhomogeneous, 
\begin{align*}
	\dot{x}(t) = A x(t) + f(t)
	, 
	&&
	x(0) = x_0 
	, 
\end{align*}
it turns out we can still find a closed solution: 
\begin{align}
	x(t) &= \e^{t A} x_0 + \int_0^t \dd s \, \e^{(t-s) A} \, f(s) 
	\label{odes:eqn:inhomogeneous_linear_solution}
\end{align}
%
% paragraph Inhomogeneous linear ODEs (end)

\paragraph{Higher-order ODEs} % (fold)
\begin{align}
	\frac{\dd^n}{\dd t^n} x &= F(x) + f(t) 
	, 
	&&
	x^{(j)} = \alpha^{(j)} 
	, \; 
	j = 0 , \ldots , n-1
	, 
	\label{odes:eqn:higher_order_ode}
\end{align}
can be expressed as a system of \emph{first}-order ODEs. If $x \in \R$ for simplicity, we can write 
\begin{align*}
	y_1 &:= x 
	\\
	y_2 &:= \tfrac{\dd}{\dd t} x 
	\\
	y_j &:= \tfrac{\dd}{\dd t} y_{j-1} = \tfrac{\dd^j}{\dd t^j} x 
	\\
	y_n &:= \tfrac{\dd^{n-1}}{\dd t^{n-1}} x 
	= F(x) + f(t) 
	= F(y_1) + f(t) 
	. 
\end{align*}
Thus, assuming the vector field $F(x) = \lambda x$, $\lambda \in \R$, is linear and the dimension of the underlying space $1$, we obtain the first-order linear equation 
\begin{align}
	\frac{\dd}{\dd t} \left (
	\begin{matrix}
		y_1 \\
		\vdots \\
		\vdots \\
		\vdots \\
		y_n \\
	\end{matrix}
	\right ) 
	= \underbrace{\left (
	\begin{matrix}
		0 & 1 & 0 & \cdots & 0 \\
		\vdots & \ddots & \ddots & \ddots & \vdots \\
		\vdots &  & \ddots & \ddots & 0 \\
		0 &  &  & \ddots & 1 \\
		\lambda & 0 & \cdots & \cdots & 0 \\
	\end{matrix}
	\right )}_{=: H} 
	\left (
	\begin{matrix}
		y_1 \\
		\vdots \\
		\vdots \\
		\vdots \\
		y_n \\
	\end{matrix}
	\right ) + \left (
	\begin{matrix}
		0 \\
		\vdots \\
		\vdots \\
		0 \\
		f(t) \\
	\end{matrix}
	\right )
	\label{odes:eqn:higher_order_ode_reduced_first_order}
\end{align}
and the solution is given by \eqref{odes:eqn:inhomogeneous_linear_solution} with initial condition $y_0 = \bigl ( \alpha^{(0)} , \alpha^{(1)} , \ldots , \alpha^{(n-1)} \bigr )$. Then the solution to \eqref{odes:eqn:higher_order_ode} is then just the first component $y_1(t)$. 
\begin{example}[Newton's equation of motion for a free particle]
	The dynamics of a free particle of mass $m$ starting at $x_0$ with initial momentum $p_0$ is 
	\begin{align*}
		\ddot{x} = 0 
		, 
		&&
		x(0) = x_0 
		, \; 
		\dot{x}(0) = \tfrac{p_0}{m} 
		. 
	\end{align*}
	This leads to the first-order equation 
	\begin{align*}
		\frac{\dd}{\dd t} \left (
		\begin{matrix}
			y_1 \\
			y_2 \\
		\end{matrix}
		\right ) &= \left (
		\begin{matrix}
			0 & 1 \\
			0 & 0 \\
		\end{matrix}
		\right ) \left (
		\begin{matrix}
			y_1 \\
			y_2 \\
		\end{matrix}
		\right ) =: H y 
		. 
	\end{align*}
	The matrix $H$ is nilpotent, $H^2 = 0$, and thus the exponential series 
	\begin{align*}
		\e^{t H} = \sum_{k = 0}^{\infty} \frac{t^k}{k!} \, H^k 
		= \id + t \, H + 0 
		= \left (
		\begin{matrix}
			1 & t \\
			0 & 1 \\
		\end{matrix}
		\right )
	\end{align*}
	terminates after finitely many terms. Hence, the solution is 
	\begin{align*}
		x(t) = y_1(t) = x_0 + t \, \tfrac{p_0}{m} 
		. 
	\end{align*}
\end{example}
%
% paragraph Higher-order ODEs (end) 
% section Examples of ODEs (end)

\section{Existence and uniqueness of solutions} % (fold)
\label{odes:existence_uniqueness}
Now we turn back to the non-linear case. One important example of non-linear ODEs are Hamilton's equations of motion~\eqref{intro:eqn:hamiltons_eom} which describe the motion of a classical particle; this will be the content of Chapter~\ref{classical_mechanics}. 

The standard existence and uniqueness result is the Picard-Lindelöf theorem. The crucial idea is to approximate the solution to the ODE and to improve this approximation iteratively. To have a notion of convergence of solutions, we need to introduce the idea of  
\begin{definition}[Metric space]
	Let $\mathcal{X}$ be a set. A mapping $d : \mathcal{X} \times \mathcal{X} \longrightarrow [0,+\infty)$ with properties 
	\begin{enumerate}[(i)]
		\item $d(x,y) = 0$ exactly if $x = y$ (definiteness), 
		\item $d(x,y) = d(y,x)$ (symmetry), and 
		\item $d(x,z) \leq d(x,y) + d(y,z)$ (triangle inequality), 
	\end{enumerate}
	for all $x,y,z \in \mathcal{X}$ is called \emph{metric}. We refer to $(\mathcal{X},d)$ as \emph{metric space} (often only denoted as $\mathcal{X}$). A metric space $(\mathcal{X},d)$ is called \emph{complete} if all Cauchy sequences $(x_{(n)})$ (with respect to the metric) converge to some $x \in \mathcal{X}$. 
\end{definition}
A metric gives a notion of distance -- and thus a notion of convergence, continuity and open sets (a topology): quite naturally, one considers the topology generated by open balls defined in terms of $d$. There are more general ways to study convergence and alternative topologies (\eg Fréchet topologies or weak topologies) can be both useful and necessary. 
\begin{example}
	\begin{enumerate}[(i)]
		\item Let $\mathcal{X}$ be a set and define 
		\begin{align*}
			d(x,y) := 
			\begin{cases}
				1 \qquad x \neq y \\
				0 \qquad x = y \\
			\end{cases}
			. 
		\end{align*}
		It is easy to see $d$ satisfies the axioms of a metric and $\mathcal{X}$ is complete with respect to $d$. This particular choice leads to the discrete topology. 
		\item Let $\mathcal{X} = \Cont([a,b],U)$, $U \subseteq \R^n$ be the space of continuous functions on an interval. Then one naturally considers the metric  
		\begin{align*}
			d_{\infty}(f,g) := \sup_{x \in [a,b]} \babs{f(x) - g(x)} = \max_{x \in [a,b]} \babs{f(x) - g(x)}
		\end{align*}
		with respect to which $\Cont([a,b],U)$ is complete. 
	\end{enumerate}
\end{example}
One important result we need concerns the existence of so-called fixed points in complete metric spaces: 
\begin{theorem}[Banach's fixed point theorem]\label{odes:thm:Banach_fixed_point}
	Let $(\mathcal{X},d)$ be a complete metric space and $P$ a \emph{contraction}, \ie a map for which there exists $C \in [0,1)$ so that for all $x , y \in \mathcal{X}$ 
	\begin{align}
		d \bigl ( P(x) , P(y) \bigr ) \leq C \, d(x,y) 
		\label{odes:eqn:contraction_map}
	\end{align}
	holds. Then there exists a \emph{unique} fixed point $x_{\ast} = P(x_{\ast})$ so that for any $x_0 \in \mathcal{X}$, the sequence 
	\begin{align*}
		\bigl \{ P^n(x_0) \bigr \}_{n \in \N} = \bigl \{ \underbrace{P \circ \cdots \circ P}_{\mbox{$n$ times}} (x_0) \bigr \}_{n \in \N}
	\end{align*}
	converges to $x_{\ast} \in \mathcal{X}$. 
\end{theorem}
\begin{proof}
	Let us define $x_n := P^n(x_0)$ for brevity. To show existence of \emph{a} fixed point, we will prove that $\{ x_n \}_{n \in \N}$ is a Cauchy sequence. First of all, the distance between \emph{neighboring} sequence elements goes to $0$, 
	\begin{align}
		d(x_{n+1} , x_n) &= d \bigl ( P(x_n) , P(x_{n-1}) \bigr ) 
		\leq C \, d(x_n , x_{n-1}) 
		\notag \\
		&\leq C^n \, d(x_1,x_0) 
		. 
		\label{odes:eqn:contraction_estimate_neighbors}
	\end{align}
	Without loss of generality, we shall assume that $m < n$. Then, we use the triangle inequality to estimate the distance between $x_n$ and $x_m$ by distances to neighbors, 
	\begin{align*}
		d(x_n,x_m) &\leq 
		d(x_n , x_{n-1}) + d(x_{n-1},x_{n-2}) + \ldots + d(x_{m+1},x_m)
		= \sum_{j = m+1}^n d(x_j,x_{j-1})
		. 
	\end{align*}
	Hence, we can plug in the estimate on the distance between neighbors, \eqref{odes:eqn:contraction_estimate_neighbors}, and sum over $j$, 
	\begin{align*}
		d(x_n,x_m) &\leq \sum_{j = m+1}^n C^{j-1} \, d(x_1,x_0) 
		= C^m \, \sum_{j = 0}^{n-m} C^j \, d(x_1,x_0) 
		\\
		&\leq C^m \, d(x_1,x_0) \, \sum_{j = 0}^{\infty} C^j 
		= \frac{C^m}{1-C} \, d(x_1,x_0) 
		. 
	\end{align*}
	If we choose $m$ large enough, we can make the right-hand side as small as we want and thus, $\{ x_n \}_{n \in \N}$ is a Cauchy sequence. By assumption the space $(\mathcal{X},d)$ is complete, and thus, there exists $x_{\ast}$ so that 
	\begin{align*}
		\lim_{n \to \infty} x_n = \lim_{n \to \infty} P^n(x_0) = x_{\ast} 
		. 
	\end{align*}
	It remains to show that this fixed point is unique. Pick another initial point $x_0'$ with limit $x_{\ast}' = \lim_{n \to \infty} P^n(x_0')$; this limit exists by the arguments above. Both fixed points satisfy $P x_{\ast} = x_{\ast}$ and $P x_{\ast}' = x_{\ast}'$, because 
	\begin{align*}
		P(x_{\ast}) &= P \bigl ( \lim_{n \to \infty} P^n(x_0) \bigr ) 
		= \lim_{n \to \infty} P^{n+1}(x_0) 
		= x_{\ast} 
	\end{align*}
	holds (and analogously for $x_{\ast}'$). Using the contractivity property \eqref{odes:eqn:contraction_map}, we then estimate the distance between $x_{\ast}$ and $x_{\ast}'$ by 
	\begin{align*}
		d(x_{\ast},x_{\ast}') &= d \bigl ( P(x_{\ast}) , P(x_{\ast}') \bigr ) 
		\leq C \, d(x_{\ast},x_{\ast}') 
		. 
	\end{align*}
	Since $0 \leq C < 1$, the above equation can only hold if $d(x_{\ast},x_{\ast}') = 0$ which implies $x_{\ast} = x_{\ast}'$. Hence, the fixed point is also unique. \marginpar{2013.09.12}
\end{proof}
To ensure the existence of the flow, we need to impose conditions on the vector field. One common choice is to require $F$ to be \emph{Lipschitz}, meaning there exists a constant $L > 0$ so that 
\begin{align}
	\babs{F(x) - F(x')} \leq L \babs{x - x'} 
	\label{odes:eqn:Lipschitz_condition}
\end{align}
holds for all $x$ and $x'$ in some open neighborhood $U_{x_0}$ of the initial point $x_0 = x(0)$. The Lipschitz condition has two implications: first of all, it states that the vector field grows at most linearly. Secondly, if the vector field is continuously differentiable, $L$ is also a bound on the norm of the differential $\sup_{x \in U_{x_0}} \norm{DF(x)}$. However, not all vector fields which are Lipschitz need to be continuously differentiable, \eqref{odes:eqn:Lipschitz_condition} is in fact weaker than requiring that $\sup_{x \in U_{x_0}} \norm{DF(x)}$ is bounded. For instance, the vector field $F(x) = - \abs{x}$ in one dimension is Lipschitz on all of $\R$, but not differentiable at $x = 0$. 

So if the vector field is locally Lipschitz, the flow exists at least for some time interval: 
\begin{theorem}[Picard-Lindelöf]\label{odes:thm:Picard_Lindeloef}
	Let $F$ be a continuous vector field, $F \in \Cont(U,\R^n)$, $U \subseteq \R^n$ open, which defines a system of differential equations, 
	\begin{align}
		\dot{x} &= F(x) 
		\label{odes:eqn:ode_Picard_Lindeloef}
		. 
	\end{align}
	Assume for a certain initial condition $x_0 \in U$ there exists a closed ball 
	\begin{align*}
		B_{\rho}(x_0) := \bigl \{ x \in \R^n \; \vert \; \sabs{x - x_0} \leq \rho \bigr \} \subseteq U
		, 
		&&
		\rho > 0
		, 
	\end{align*}
	such that $F$ is Lipschitz on $B_{\rho}(x_0)$ with Lipschitz constant $L$ (\cf equation~\eqref{odes:eqn:Lipschitz_condition}). Then the initial value problem, equation~\eqref{odes:eqn:ode_Picard_Lindeloef} with $x(0) = x_0$, has a unique solution $t \mapsto x(t)$ for times $\abs{t} \leq T := \min \bigl ( \nicefrac{\rho}{v_{\max}} , \nicefrac{1}{2L} \bigr )$ where the maximal velocity is defined as $v_{\max} := \sup_{x \in B_{\rho}(x_0)} \sabs{F(x)}$. 
\end{theorem}
\begin{proof}
	The proof consists of three steps: 
	\medskip
	
	\noindent
	\textbf{Step 1:} We can rewrite the initial value problem equation~\eqref{odes:eqn:ode_Picard_Lindeloef} with $x(0) = x_0$ as 
	\begin{align*}
		x(t) = x_0 + \int_0^t \dd s \, F \bigl ( x(s) \bigr ) 
	\end{align*}
	This equation can be solved iteratively: we define $x_{(0)}(t) := x_0$ and the $n+1$th iteration 
	$x_{(n+1)}(t) := \bigl ( P (x_{(n)}) \bigr )(t)$ in terms of the so-called Picard map 
	\begin{align*}
		\bigl ( P (x_{(n)}) \bigr )(t) := x_0 + \int_0^t \dd s \, F \bigl ( x_{(n)}(t) \bigr ) 
		. 
	\end{align*}
	\textbf{Step 2:} We will determine $T > 0$ small enough so that $P : \mathcal{X} \longrightarrow \mathcal{X}$ is a contraction on the space of trajectories which start at $x_0$, 
	\begin{align*}
		\mathcal{X} := \Bigl \{ y \in \Cont \bigl ([-T,+T] , B_{\rho}(x_0) \bigr ) \; \big \vert \; y(0) = x_0 \Bigr \} 
		. 
	\end{align*}
	First, we note that $\mathcal{X}$ is a \emph{complete} metric space if we use 
	\begin{align*}
		\mathrm{d} (y,z) := \sup_{t \in [-T,+T]} \babs{y(t) - z(t)} 
		&& 
		y,z \in \mathcal{X} 
	\end{align*}
	to measure distances between trajectories. Hence, Banach's fixed point theorem~\ref{odes:thm:Banach_fixed_point} applies, and once we show that $P$ is a contraction, we know that $x_{(n)}$ converges to the (unique) fixed point $x = P(x) \in \mathcal{X}$; this fixed point, in turn, is the solution to the ODE \eqref{odes:eqn:ode_Picard_Lindeloef} with $x(0) = x_0$. 
	
	To ensure $P$ is a contraction with $C = \nicefrac{1}{2}$, we propose $T \leq \nicefrac{1}{2L}$. Then the Lipschitz property implies that for any $y,z \in \mathcal{X}$, we have 
	\begin{align*}
		\mathrm{d} \bigl ( P(y) , P(z) \bigr ) &= \sup_{t \in [-T,+T]} \abs{\int_0^t \dd s \, \bigl [ \bigl ( F(y) \bigr )(s) - \bigl ( F(z) \bigr )(s) \bigr ]} 
		\\
		&\leq T \, L \sup_{t \in [-T,+T]} \babs{y(t) - z(t)} 
		% \\
		% &
		\leq \tfrac{1}{2L} L \, \mathrm{d} (y,z) = \tfrac{1}{2} \mathrm{d} (y,z) 
		. 
	\end{align*}
	\textbf{Step 3:} We need to ensure that the trajectory does not leave the ball $B_{\rho}(x_0)$ for all $t \in [-T,+T]$: For any $y \in \mathcal{X}$, we have \marginpar{2013.09.17}
	\begin{align*}
		\babs{P(y) - x_0} = \abs{\int_0^t \dd s \, F \bigl ( y(s) \bigr )} \leq t \, v_{\max} \leq T \, v_{\max} < \rho 
		. 
	\end{align*}
	Hence, as long as $T \leq \min \bigl \{ \nicefrac{1}{2L} , \nicefrac{v_{\max}}{\rho} \bigr \}$ trajectories exist and do not leave $U$. This concludes the proof. 
\end{proof}
This existence result for a \emph{single} initial condition implies immediately that 
\begin{corollary}\label{odes:cor:Picard_Lindeloef_flow}
	Assume we are in the setting of Theorem~\ref{odes:thm:Picard_Lindeloef}. Then for any $x_V \in U$ there exists an open neighborhood $V$ of $x_V$ so that the flow exists as a map $\Phi : [-T,+T] \times V \longrightarrow U$. 
\end{corollary}
\begin{proof}
	Pick any $x_V \in U$. Then according to the Picard-Lindelöf Theorem~\ref{odes:thm:Picard_Lindeloef} there exists an open ball $B_{\rho}(x_V) \subseteq U$ around $x_V$ so that the trajectory $t \mapsto x(t)$ with $x(0) = x_V$ exists for all times $t \in [-T,+T]$ where $T = \min \bigl \{ \nicefrac{\rho}{v_{\max}} \nicefrac{1}{2L} \bigr \}$ and $v_{\max} = \sup_{x \in B_{\rho}(x_V)}$. 
	
	Now consider initial conditions $x_0 \in B_{\nicefrac{\rho}{2}}(x_V)$: since the vector field satisfies the same Lipschitz condition as before, the arguments from Step~3 of the proof of the Picard-Lindelöf Theorem tell us that for times $t \in [-\nicefrac{T}{2} , +\nicefrac{T}{2}]$, the particles will not leave $B_{\rho}(x_V)$. Put more concretely: the maximum velocity of the particle (which is the maximum of the vector field) dictates how far it can go. So even if we start at the border of the ball with radius $\nicefrac{\delta}{2}$, for short enough times (\eg $\nicefrac{T}{2}$) we can make sure it never reaches the boundary of $B_{\rho}(x_V)$. 
	
	This means the flow map $\Phi : [-\nicefrac{T}{2} , + \nicefrac{T}{2}] \times B_{\nicefrac{\rho}{2}}(x_V) \longrightarrow B_{\rho}(x_V)$ exists. Note that since $B_{\nicefrac{\rho}{2}}(x_V)$ is contained in $B_{\rho}(x_V) \subseteq U$, we know that the smaller ball is also a subset of $U$. 
\end{proof}
Another important fact is that the flow $\Phi$ inherits the smoothness of the vector field which generates it.%
\begin{theorem}\label{odes:thm:smoothness_flow}
	Assume the vector field $F$ is $k$ times continuously differentiable, $F \in \Cont^k ( U , \R^n )$, $U \subseteq \R^n$. Then the flow $\Phi$ associated to \eqref{odes:eqn:ode_Picard_Lindeloef} is also $k$ times continuously differentiable, \ie $\Phi \in \Cont^k \bigl ( [-T,+T] \times V , U \bigr )$ where $V \subset U$ is suitable. 
\end{theorem}
\begin{proof}
	% CHANGED add page reference and read proof (possibly `proof' in Arnold)
	We refer to Chapter 3, Section 7.3 in \cite{Arnold:ode:2006}. 
\end{proof}

\subsection{Interlude: the Grönwall lemma} % (fold)
\label{odes:existence_uniqueness:Groenwall}
To show \emph{global} uniqueness of the flow, we need to make use of the exceedingly useful Grönwall lemma. It is probably the simplest “differential inequality”. 
\begin{lemma}[Grönwall]\label{odes:lem:Groenwall}
	Let $u$ be differentiable on the interior of $I = [a,b]$ or $I = [a,+\infty)$, and satisfy the differential inequality 
	\begin{align*}
		\dot{u}(t) \leq \beta(t) \, u(t) 
	\end{align*}
	where $\beta$ is a real-valued, continuous function on $I$. Then 
	\begin{align}
		u(t) \leq u(a) \, \e^{\int_a^t \dd s \, \beta(s)} 
		\label{odes:eqn:Groenwall_estimate}
	\end{align}
	holds for all $t \in I$. 
\end{lemma}
\begin{proof}
	Define the function 
	\begin{align*}
		v(t) := \e^{\int_a^t \dd s \, \beta(s)} 
		. 
	\end{align*}
	Then $\dot{v}(t) = \beta(t) \, v(t)$ and $v(a) = 1$ hold, and we can use the assumption on $u$ to estimate 
	\begin{align*}
		\frac{\dd}{\dd t} \frac{u(t)}{v(t)} &= \frac{\dot{u}(t) \, v(t) - u(t) \, \dot{v}(t)}{v(t)^2} 
		\\
		&
		= \frac{\dot{u}(t) \, v(t) - u(t) \, \beta(t) \, v(t)}{v(t)^2} 
		\\
		&\leq \frac{\beta(t) \, u(t) \, v(t) - \beta(t) \, u(t) \, v(t)}{v(t)^2} 
		= 0 
		. 
	\end{align*}
	Hence, equation~\eqref{odes:eqn:Groenwall_estimate} follows from the mean value theorem, 
	\begin{align*}
		\frac{u(t)}{v(t)} \leq \frac{u(a)}{v(a)} = u(a) 
		. 
	\end{align*}
\end{proof}
One important application is to relate bootstrap information on the vector fields to information on the flow itself. For instance, if the vector fields of two ODEs are close, then also their flows remain close -- but only for logarithmic times at best. After that, the flows will usually diverge \emph{exponentially}. 

If one applies this reasoning to the \emph{same} ODE for different initial conditions, then one observes the same effect: no matter how close initial conditions are picked, usually, the trajectories will diverge exponentially. This fact gives rise to \emph{chaos}. 
\begin{proposition}\label{odes:prop:closeness_vector_fields_closeness_flows}
	Suppose the vector field $F_{\eps} = F_0 + \eps \, F_1$ satisfies the Lipschitz condition~\eqref{odes:eqn:Lipschitz_condition} for some open set $U \subseteq \R^n$ with Lipschitz constant $L > 0$, consisting of a leading-order term $F_0$ that is also Lipschitz with constant $L$, and a small, bounded perturbation $\eps \, F_1$, \ie $0 < \eps \ll 1$ and $C := \sup_{x \in U} \abs{F_1} < \infty$. 
	
	Then the flows $\Phi^{\eps}$ and $\Phi^0$ associated to $\dot{x}^{\eps} = F_{\eps}$ and $\dot{x}^{(0)} = F_0$ exist for the same times $t \in [-T,+T]$, and the two are $\order(\eps)$ close in the following sense: there exists an open neighborhood $V \subseteq U$ so that 
	\begin{align}
		\sup_{x \in V}\babs{\Phi^{\eps}_t(x) - \Phi^0_t(x)} \leq \eps \, \frac{C}{L} \, \bigl ( \e^{L \abs{t}} - 1 \bigr )
		\label{odes:eqn:closeness_vector_fields_closeness_flows}
	\end{align}
	holds. 
\end{proposition}
An important observation is that this result is in some sense optimal: while there are a few specific ODEs (in particular linear ones) for which estimates analogous to \eqref{odes:eqn:closeness_vector_fields_closeness_flows} hold for longer times, in general the Proposition really mirrors what happens: solutions will diverge exponentially. 

In physics, the time scale at which chaos sets in, $\order(\abs{\ln \eps})$, is known as \emph{Ehrenfest time scale}; classical chaos is the reason why semiclassics (approximating the \emph{linear} quantum evolution by \emph{non-linear} classical dynamics) is limited to the Ehrenfest time scale. 
\begin{proof}
	% CHANGED write proof
	First of all, due to Corollary~\ref{odes:cor:Picard_Lindeloef_flow}, there exist $T > 0$ and an open neighborhood $V$ so that for all initial conditions $x_0 \in V$, the trajectories do not leave $U$. 
	
	Let $x^{\eps}$ and $x^{(0)}$ be the trajectories which solve $\dot{x}^{\eps} = F_{\eps}$ and $\dot{x}^{(0)} = F_0$ for the initial condition $x^{\eps}(0) = x_0 = x^{(0)}(0)$. Moreover, let us define the difference vector $X(t) := x^{\eps}(t) - x^{(0)}(t)$ so that $u(t) = \sabs{X(t)}$. Using $\abs{x} - \abs{y} \leq \abs{x-y}$ for vectors $x , y \in \R^n$ (which follows from the triangle inequality), we obtain 
	\begin{align}
		\abs{\frac{\dd}{\dd t} X(t)} &= \abs{\lim_{\delta \to 0} \frac{X(t+\delta) - X(t)}{\delta}} 
		= \lim_{\delta \to 0} \frac{\babs{X(t+\delta) - X(t)}}{\sabs{\delta}} 
		\notag \\
		&\geq \lim_{\delta \to 0} \frac{\sabs{X(t+\delta)} - \sabs{X(t)}}{\sabs{\delta}} 
		= \frac{\dd}{\dd t} u(t) 
		. 
		\label{odes:eqn:abs_ddt_X_geq_ddt_abs_X}
	\end{align}
	Hence, we can estimate the derivative of $u(t)$ from above by 
	\begin{align*}
		\frac{\dd}{\dd t} u(t) &\leq \abs{\frac{\dd}{\dd t} \bigl ( x^{\eps}(t) - x^{(0)}(t) \bigr )} 
		= \Babs{F_{\eps} \bigl ( x^{\eps}(t) \bigr ) - F_0 \bigl ( x^{(0)}(t) \bigr )} 
		\\
		&\leq \Babs{F_0 \bigl ( x^{\eps}(t) \bigr ) - F_0 \bigl ( x^{(0)}(t) \bigr )} 
		+ \eps \, \babs{F_1 \bigl ( x^{\eps}(t) \bigr )} 
		\\
		&\leq L \, \babs{x^{\eps}(t) - x^{(0)}(t)} + \eps \, C 
		= L \, u(t) + \eps \, C 
		. 
	\end{align*}
	The above inequality is in fact equivalent to 
	\begin{align*}
		\frac{\dd}{\dd t} \bigl ( \e^{- L t} \, u(t) \bigr ) &= 
		\e^{- L t} \, \bigl ( \dot{u}(t) - L \, u(t) \bigr ) \leq \eps \, C \, \e^{- L t} 
		. 
	\end{align*}
	Now we integrate left- and right-hand side: since we assume both trajectories to start at the same point $x_0$, we have $u(0) = \babs{x^{\eps}(0) - x^{(0)}(0)} = \babs{x_0 - x_0} = 0$. Moreover, the integrands on both sides are non-negative functions, so that for $t > 0$ we obtain 
	\begin{align*}
		\int_0^t \dd s \, \frac{\dd}{\dd s} \bigl ( \e^{- L s} \, u(s) \bigr ) &= 
		\bigl [ \e^{- L s} \, u(s) \bigr ]_0^t 
		= \e^{- L t} \, u(t) 
		\\
		&\leq \int_0^t \dd s \, \eps \, C \, \e^{- L s}
		= - \eps \, \frac{C}{L} \, \bigl [ \e^{- L s} \bigr ]_0^t 
		= \eps \, \frac{C}{L} \, \bigl ( 1 - \e^{-L t} \bigr ) 
		. 
	\end{align*}
	A similar result is obtained for $t < 0$. Hence, we get 
	\begin{align*}
		u(t) \leq \eps \, \frac{C}{L} \, \bigl ( \e^{L \abs{t}} - 1 \bigr )
		. 
	\end{align*}
	Since the Lipschitz constant $L$ and $C$ were independent of the initial point $x_0$, this estimate holds for all $x_0 \in U$, and we have shown equation~\eqref{odes:eqn:closeness_vector_fields_closeness_flows}. 
\end{proof}
Note that the estimate~\eqref{odes:eqn:closeness_vector_fields_closeness_flows} also holds \emph{globally} (meaning for all times $t \in \R$ and initial conditions $x_0 \in \R^n$) if the vector fields are \emph{globally} Lipschitz, although we need the \emph{global} existence and uniqueness theorem, Corollary~\ref{frameworks:classical_mechanics:cor:existence_flow}, from the next subsection. 
% subsection Interlude: the Grönwall lemma (end)

\subsection{Conclusion: global existence and uniqueness} % (fold)
\label{odes:existence_uniqueness:global}
The Grönwall lemma is necessary to prove the uniqueness of the solutions in case the vector field is globally Lipschitz. 
\begin{corollary}\label{frameworks:classical_mechanics:cor:existence_flow}
	If the vector field $F$ satisfies the Lipschitz condition \emph{globally}, \ie there exists $L > 0$ such that 
	\begin{align*}
		\babs{F(x) - F(x')} \leq L \babs{x - x'} 
	\end{align*}
	holds for all $x , x' \in \R^n$, then $t \mapsto \Phi_t(x_0)$ exists globally for all $t \in \R$ and $x_0 \in \R^n$. 
\end{corollary}
%
% CHANGED revisit and improve the way the proof has been formulated! It's hard to understand! 
%
\begin{proof}
	% CHANGED add uniqueness --> I need the Grönwall lemma here! 
	% In this case only the Lipschitz condition enters the time until which we can solve the initial value problem, $T \leq \nicefrac{1}{2 L}$. 
	For every $x_0 \in \R^n$, we can solve the initial value problem at least for $\abs{t} \leq \nicefrac{1}{2 L}$. Since we do not require the particle to remain in a neighborhood of the initial point as in Theorem~\ref{odes:eqn:ode_Picard_Lindeloef}, the other condition on $T$ is void. 
	
	Using $\Phi_{t_1} \circ \Phi_{t_2} = \Phi_{t_1 + t_2}$ we obtain a \emph{global} trajectory $x : \R \longrightarrow \R^n$ for all times. However, we potentially lose uniqueness of the trajectory. So assume $\tilde{x}$ is another trajectory. Then we define $u(t) := \babs{x(t) - \tilde{x}(t)}$ and use the Lipschitz property~\eqref{odes:eqn:Lipschitz_condition} to deduce 
	\begin{align*}
		\frac{\dd}{\dd t} u(t) &= \frac{\dd}{\dd t} \babs{x(t) - \tilde{x}(t)} 
		\\
		&
		\overset{\eqref{odes:eqn:abs_ddt_X_geq_ddt_abs_X}}{\leq} \babs{\dot{x}(t) - \dot{\tilde{x}}(t)} 
		= \babs{F \bigl ( x(t) \bigr ) - F \bigl ( \tilde{x}(t) \bigr )} 
		\\
		&\leq L \, \babs{x(t) - \tilde{x}(t)} 
		. 
	\end{align*}
	Hence, the Grönwall lemma applies with $a = 0$ and 
	\begin{align*}
		u(0) = \babs{x(0) - \tilde{x}(0)} = \babs{x_0 - x_0} = 0 
		, 
	\end{align*}
	and we obtain the estimate 
	\begin{align*}
		0 \leq u(t) = \babs{x(t) - \tilde{x}(t)} \leq u(0) \, \e^{L t} = 0 
		, 
	\end{align*}
	\ie the two trajectories coincide and the solution is unique. \marginpar{2013.09.19}
\end{proof}
%
% subsection Conclusion: global existence and uniqueness (end)
% section Existence and uniqueness of solutions (end)

\section{Stability analysis} % (fold)
\label{odes:stability}
What if an ODE is not analytically solvable, beyond using numerics, what information can we extract? A simple way to learn something about the qualityative behavior is to linearize the vector field near fixed points, \ie those $x_0 \in \R^n$ for which 
\begin{align*}
	F(x_0) = 0 
	. 
\end{align*}
Writing the solution $x(t) = x_0 + \delta y(t)$ where $x(0) = x_0$ and Taylor expanding the vector field, we obtain another \emph{linear} ODE involving the differential $DF(x_0) = \bigl ( \partial_{x_j} F_k(x_0) \bigr )_{1 \leq j , k \leq d}$, 
\begin{align*}
	\frac{\dd}{\dd t} x(t) = \delta \, \frac{\dd}{\dd t} y(t) 
	&
	= F \bigl ( x_0 + \delta y(t) \bigr ) 
	\\
	&
	= F(x_0) + \delta \, DF(x_0) \, y(t) + \order(\delta^2) 
	= \delta \, DF(x_0) \, y(t) 
	\\
	\Longrightarrow \; \frac{\dd}{\dd t} y(t) &= DF(x_0) \, y(t) 
	. 
\end{align*}
The latter can be solved explicitly, namely $\e^{t DF(x_0)} y_0$. 

Now the stability of the solutions near fixed points is determined by the \emph{eigenvalues} $\{ \lambda_j \}_{j = 1}^N$ of $DF(x_0)$. 
\begin{definition}[Stability of fixed points]
	We call an ODE near a fixed point $x_0$
	\begin{enumerate}[(i)]
		\item \emph{stable} if $\Re \lambda_j < 0$ holds for all eigenvalues $\lambda_j$ of $DF(x_0)$, 
		\item \emph{marginally stable} (or Liapunov stable) $\Re \lambda_j \leq 0$ holds for all $j$ and $\Re \lambda_j = 0$ for at least one $j$, and 
		\item \emph{unstable} otherwise. 
	\end{enumerate}
\end{definition}
There is also another characterization of fixed points: 
\begin{definition}
	An ODE near a fixed point $x_0$ is called 
	\begin{enumerate}[(i)]
		\item \emph{elliptic} if $\Re \lambda_j = 0$ for all $j = 1 , \ldots , N$, and 
		\item \emph{hyperbolic} if $\Im \lambda_j = 0$ for all $j = 1 , \ldots , N$ and $\Re \lambda_j > 0$ for some $j$. 
	\end{enumerate}
\end{definition}
\begin{example}[Lorenz equation]
	The Rayleigh-Bernard equation describes the behavior of a liquid in between two warm plates of different temperatures: 
	\begin{align*}
		\left (
		\begin{matrix}
			\dot{x}_1 \\
			\dot{x}_2 \\
			\dot{x}_3 \\
		\end{matrix}
		\right ) &= \left (
		\begin{matrix}
			\sigma \, (x_2 - x_1) \\
			- x_1 \, x_3 + r x_1 - x_2 \\
			x_1 \, x_2 - b \, x_3 \\
		\end{matrix}
		\right ) =: F(x)
	\end{align*}
	Here, the variables $x_1$ and $x_2$ are temperatures, $x_3$ is a speed, the constant $\sigma$ is a volume ratio, $b$ is the so-called \emph{Prandtl number} and $r \geq 0$ is a control parameter. Typical values for the parameters are $\sigma \simeq 10$ and $b \simeq \nicefrac{8}{5}$. 
	
	This vector field may have several fixed points (depending on the choice of the other parameters), but $x = 0$ is always a fixed point independently of the values of $\sigma$, $b$ and $r$. 
	
	Let us analyze the stability of the ODE near the fixed point at $x = 0$: 
	\begin{align*}
		DF(0) = \left (
		\begin{matrix}
			- \sigma & + \sigma & 0 \\
			r & -1 & 0 \\
			0 & 0 & -b \\
		\end{matrix}
		\right )
	\end{align*}
	The block structure of the matrix yields that one of the eigenvalues is $-b$. 
	
	The other two eigenvalues can be inferred from finding the zeros of 
	\begin{align*}
		\chi(\lambda) &= \mathrm{det} \left (
		\begin{matrix}
			\lambda + \sigma & -\sigma \\
			-r & \lambda + 1 \\
		\end{matrix}
		\right ) 
		= (\lambda + \sigma)(\lambda + 1) - (-1)^2 \, r \sigma 
		\\
		&
		= \lambda^2 + (\sigma + 1) \, \lambda + (1-r) \, \sigma 
		, 
	\end{align*}
	namely 
	\begin{align*}
		\lambda_{\pm} = - \frac{\sigma + 1}{2} \pm \frac{1}{2} \sqrt{(\sigma + 1)^2 - (1 - r) \, \sigma} 
		. 
	\end{align*}
	If $r > 1$, then $(\sigma + 1)^2 - (1 - r) \, \sigma > (\sigma + 1)^2$, and thus $\lambda_- < 0 < \lambda_+$ as long as $\sigma \geq r$ is large enough. Hence, the fixed point at $x = 0$ is unstable and hyperbolic. 
	
	For the other case, $0 \leq r < 1$, both eigenvalues are negative, and thus $\lambda_{\pm} < 0$, the fixed point is stable (but neither elliptic nor hyperbolic). \marginpar{2013.09.24}
\end{example}
%
% section Stability analysis (end)
% chapter Ordinary differential equations (end)
\chapter{Classical mechanics} % (fold)
\label{classical_mechanics}
This section serves to give a short introduction to the hamiltonian point of view of classical mechanics and some of its beautiful mathematical structures. For simplicity, we only treat classical mechanics of a \emph{spinless point particle moving in $\R^n$}, for the more general theory, we refer to \cite{Marsden_Ratiu:intro_mechanics_symmetry:1999,Arnold:classical_mechanics:1997}. A more thorough introduction to this topic would easily take up a whole lecture, so we content ourselves with introducing the key precepts and apply our knowledge from Section~\ref{odes}. 

We will only treat \emph{hamiltonian} mechanics here: the dynamics is generated by the so-called hamilton function $H : \Pspace \longrightarrow \R$ which describes the energy of the system for a given configuration. Here, $\Pspace$ is also known as \emph{phase space}. Since only \emph{energy differences} are measurable, the hamiltonian $H' := H + E_0$, $E_0 \in \R$, generates the \emph{same dynamics} as $H$. This is obvious from the \emph{Hamilton's equations of motion}, 
\begin{subequations}\label{classical_mechanics:eqn:hamiltons_eom_simple}
	\begin{align}
		\dot{q}(t) &= + \nabla_p H \bigl ( q(t),p(t) \bigr ) 
		,
		\\ 
		\dot{p}(t) &= - \nabla_q H \bigl ( q(t),p(t) \bigr ) 
		, 
		% \notag
	\end{align}
\end{subequations}
which can be rewritten in matrix notation as 
\begin{align}
	J 
	\left (
	\begin{matrix}
		\dot{q}(t) \\
		\dot{p}(t) \\
	\end{matrix}
	\right ) 
	:& \negmedspace= 
	\left (
	\begin{matrix}
		0 & - \id_{\R^n} \\
		+ \id_{\R^n} & 0 \\
	\end{matrix}
	\right )
	\left (
	\begin{matrix}
		\dot{q}(t) \\
		\dot{p}(t) \\
	\end{matrix}
	\right ) 
	\label{classical_mechanics:eqn:hamiltons_eom} 
	\\
	&= \left (
	\begin{matrix}
		\nabla_q H \\
		\nabla_p H \\
	\end{matrix}
	\right ) \bigl ( q(t),p(t) \bigr ) 
	=: X_H \bigl ( q(t),p(t) \bigr ) 
	\notag 
	. 
\end{align}
The matrix $J$ appearing on the left-hand side is often called \emph{symplectic form} and leads to a geometric point of view of classical mechanics. For fixed initial condition $(q_0,p_0) \in \Pspace$ at time $t_0 = 0$, \ie initial position and momentum, the \emph{Hamilton's flow} 
\begin{align}
	\Phi : \R \times \Pspace \longrightarrow \Pspace 
\end{align}
maps $(q_0,p_0)$ onto the trajectory which solves the Hamilton's equations of motion, 
\begin{align*}
	\Phi_t(q_0,p_0) = \bigl ( q(t),p(t) \bigr )
	, 
	&&
	\bigl ( q(0),p(0) \bigr ) = (q_0,p_0) 
	. 
\end{align*}
If the flow exists for all $t \in \R$, it has the following nice properties: for all $t , t' \in \R$ and $(q_0,p_0) \in \Pspace$, we have 
\begin{enumerate}[(i)]
	\item $\Phi_t \bigl ( \Phi_{t'}(q_0,p_0) \bigr ) = \Phi_{t + t'}(q_0,p_0)$, 
	\item $\Phi_0(q_0,p_0) = (q_0,p_0)$, and 
	\item $\Phi_{t} \bigl ( \Phi_{-t}(q_0,p_0) \bigr ) = \Phi_{t - t}(q_0,p_0) = (q_0,p_0)$. 
\end{enumerate}
Mathematically, this means $\Phi$ is a \emph{group action} of $\R$ (with respect to time translations) on phase space $\Pspace$. This is a fancy way of saying: 
\begin{enumerate}[(i)]
	\item If we first evolve for time $t$ and then for time $t'$, this is the same as evolving for time $t + t'$. 
	\item If we do not evolve at all in time, nothing changes. 
	\item The system can be evolved forwards or backwards in time. 
\end{enumerate}
The above results immediately apply to the Hamilton's equations of motion: 
\begin{corollary}
	Let $H(q,p) = \tfrac{1}{2m} p^2 + V(q)$ be the hamiltonian which generates the dynamics according to equation~\eqref{classical_mechanics:eqn:hamiltons_eom} such that $\nabla_q V$ satisfies a global Lipschitz condition 
	\begin{align*}
		\babs{\nabla_q V(q) - \nabla_q V(q')} \leq L \, \babs{q - q'} 
		&& \forall q , q' \in \R^n 
		. 
	\end{align*}
	Then the hamiltonian flow $\Phi$ exists for all $t \in \R$. 
\end{corollary}
Obviously, if $\nabla_q V$ is only \emph{locally} Lipschitz, we have \emph{local} existence of the flow. 
\begin{remark}
	Note that if all second-order derivatives of $V$ are bounded, then the hamiltonian vector field $X_H$ is Lipschitz. 
	% TODO explain why 
\end{remark}

\section{The trinity of physical theories} % (fold)
\label{classical_mechanics:trinity}
It turns out to be useful to take a step back and analyze the \emph{generic structure} of most physical theories. They usually consist of \emph{three ingredients:} 
\begin{enumerate}[(i)]
	\item A notion of \emph{state} which encodes the configuration of the system, 
	\item a notion of \emph{observable} which predicts the outcome of measurements and 
	\item a \emph{dynamical equation} which governs how the physical system evolves. 
\end{enumerate}

\subsection{States} % (fold)
\label{classical_mechanics:trinity:states}
The classical particle always moves in \emph{phase space} $\R^{2n} \simeq \R^n_q \times \R^n_p$ of positions and momenta. Pure states in classical mechanics are simply points in phase space: a point particle's state at time $t$ is characterized by its position $q(t)$ and its momentum $p(t)$. More generally, one can consider \emph{distributions of initial conditions} which are relevant in statistical mechanics, for instance. 
\begin{definition}[Classical states]\label{classical_mechanics:defn:state}
	A classical state is a probability measure $\mu$ on phase space, that is a positive Borel measure\footnote{Unfortunately we do not have time to define Borel sets and Borel measures in this context. We refer the interested reader to chapter~1 of \cite{Lieb_Loss:analysis:2001}. Essentially, a Borel measure assigns a “volume” to “nice” sets, \ie Borel sets. } which is normed to $1$, 
	\begin{align*}
		\mu(U) &\geq 0 && \mbox{for all Borel sets $U \subseteq \Pspace$} \\
		\mu(\Pspace) &= \int_{\R^{2n}} \dd \mu 
		= 1 
		. 
	\end{align*}
	Pure states are \emph{point measures}, \ie if $(q_0,p_0) \in \Pspace$, then the associated pure state is given by $\mu_{(q_0,p_0)}(\cdot) := \delta_{(q_0,p_0)}(\cdot) = \delta \bigl ( \cdot - (q_0,p_0) \bigr )$.\footnote{Here, $\delta$ is the Dirac distribution which we will consider in detail in Chapter~\ref{S_and_Sprime}. } 
\end{definition}
%
% subsection States (end)

\subsection{Observables} % (fold)
\label{classical_mechanics:trinity:observables}
Observables $f$ such as position, momentum, angular momentum and energy describe the outcome of measurements. 
\begin{definition}[Classical observables]
	Classical observables $f \in \Cont^{\infty}(\R^{2n},\R)$ are smooth functions on $\R^{2n}$ with values in $\R$. 
\end{definition}
Of course, there are cases when observables are functions which are not smooth on all or $\R^{2n}$, \eg the Hamiltonian which describes Newtonian gravity in three dimensions, 
\begin{align*}
	H(q,p) = \frac{1}{2m} p^2 - \frac{g}{\abs{q}}
	, 
\end{align*}
has a singularity at $q = 0$, but $H \in \Cont^{\infty} \bigl ( \R^{2n} \setminus \{ 0 \} , \R \bigr )$. 

Intimately linked is the concept of 
\begin{definition}[Spectrum of an observable]
	The spectrum of a classical observables, \ie the set of possible outcomes of measurements, is given by 
	\begin{align*}
		\spec f := f(\Pspace) = \image f 
		. 
	\end{align*}
\end{definition}
If we are given a classical state $\mu$, then the \emph{expectation value $\mathbb{E}_{\mu}$} of an observable $f$ for the distribution of initial conditions $\mu$ is given by 
\begin{align*}
	\mathbb{E}_{\mu}(f) :=& \int_{\Pspace} \dd \mu(q,p) \, f(q,p)
	. 
\end{align*}
%
% subsection Observables (end)

\subsection{Dynamics: Schrödinger vs{.} Heisenberg picture} % (fold)
\label{classical_mechanics:trinity:dynamics}
There are two equivalent ways prescribe dynamics: either we evolve states in time and keep observables fixed or we keep states fixed in time and evolve observables. In \emph{quantum} mechanics, these two points of view are known as the \emph{Schrödinger} and \emph{Heisenberg picture} (after the famous physicists of the same name). Usually, there are situations when one point of view is more convenient than the other. 

In both cases the crucial ingredient in the dynamical equation is the \emph{energy observable} $H(q,p)$ also known as \emph{Hamilton function}, or \emph{Hamiltonian} for short. The prototypical form for a non-relativistic particle of mass $m > 0$ subjected to a potential $V$ is 
\begin{align*}
	H(q,p) = \frac{1}{2m} p^2 + V(q) 
	. 
\end{align*}
The first term, $\frac{1}{2m} p^2$, is also known as \emph{kinetic energy} while $V$ is the \emph{potential}. 

We have juxtaposed the Schrödinger and Heisenberg point of view in Table~\ref{classical_mechanics:table:Schroedinger_vs_Heisenberg}: in the Schrödinger picture, the states 
\begin{align}
	\mu(t) := \mu \circ \Phi_{-t}
\end{align}
are evolved \emph{backwards} in time while observables $f$ remain constant. (That may seem unintuitive at first, but we ask the skeptic to continue reading until the end of Section~\ref{classical_mechanics:equivalence_S_H}.) 

Conversely, in the Heisenberg picture, \emph{observables} 
\begin{align}
	f(t) := f \circ \Phi_t
\end{align}
evolve forwards in time whereas states $\mu$ remain fixed. In both cases, the dynamical equations can be written in terms of the so-called \emph{Poisson bracket}
\begin{align}
	\bigl \{ f , g \bigr \} := \sum_{j = 1}^n \bigl ( \partial_{p_j} f \, \partial_{q_j} g - \partial_{q_j} f \, \partial_{p_j} g \bigr ) 
	. 
	\label{classical_mechanics:eqn:Poisson_bracket} 
\end{align}
These equations turn out to be equivalent to proposing Hamilton's equations of motion~\eqref{classical_mechanics:eqn:hamiltons_eom} (\cf Proposition~\ref{classical_mechanics:prop:equations_of_motion_Poisson_bracket}). 

\begin{table}
	\hfil
	\begin{tabular}{l | c | c}
		 & \emph{Schrödinger picture} & \emph{Heisenberg picture} \\ \hline
		\emph{States} & $\mu(t) = \mu \circ \Phi_{-t}$ & $\mu$ \\
		\emph{Observables} & $f$ & $f(t) = f \circ \Phi_t$ \\
		\emph{Dynamical equation} & $\frac{\dd}{\dd t} \mu(t) = - \bigl \{ H , \mu(t) \bigr \}$ & $\frac{\dd}{\dd t} f(t) = \bigl \{ H , f(t) \bigr \}$
	\end{tabular}
	\hfil
	\caption{Comparison between Schrödinger and Heisenberg picture. }
	\label{classical_mechanics:table:Schroedinger_vs_Heisenberg}
\end{table}
\begin{example}
	For the special observables position $q_j$ and momentum $p_j$ equation~\eqref{classical_mechanics:eqn:eom_observables} reduces to the components of Hamilton's equations of motion~\eqref{classical_mechanics:eqn:hamiltons_eom}, 
	\begin{align*}
		\dot{q}_j &= \bigl \{ H , q_j \bigr \} 
		= + \partial_{p_j} H 
		\\
		\dot{p}_j &= \bigl \{ H , p_j \bigr \} 
		= - \partial_{q_j} H 
		. 
	\end{align*}
\end{example}
%
% subsection Dynamics: Schrödinger vs{.} Heisenberg picture (end)
% section The trinity of physical theories (end)

\section{Conservation of probability and the Liouville theorem} % (fold)
\label{classical_mechanics:Liouville}
In the Schrödinger picture, states are time-dependent while observables remain fixed in time. If $\mu$ is a state, we develop it \emph{backwards} in time, $\mu(t) = \mu \circ \Phi_{-t}$, using the hamiltonian flow $\Phi$ associated to \eqref{classical_mechanics:eqn:hamiltons_eom}. 

A priori, it is \emph{not} obvious that $\mu(t)$ is still a classical state in the sense of Definition~\ref{classical_mechanics:defn:state}. The first requirement, $\bigl ( \mu(t) \bigr )(U) \geq 0$, is still satisfied since $\Phi_{-t}(U)$ is again a subset of $\R^{2n}$\footnote{The fact that $\Phi_{-t}(U)$ is again Borel measurable follows from the continuity of $\Phi_t$ in $t$.} What is not obvious is whether $\mu(t)$ is still normed, \ie whether 
\begin{align*}
	\bigl ( \mu(t) \bigr )(\R^{2n}) &= \int_{\R^{2n}} \dd \mu(t) = 1 ? 
\end{align*}
\begin{proposition}\label{classical_mechanics:prop:states_stay_states}
	Let $\mu$ be a state on phase space $\Pspace$ and $\Phi_t$ the flow generated by a hamiltonian $H \in \Cont^{\infty}(\Pspace)$ which we assume to exist for $\abs{t} \leq T$ where $0 < T \leq \infty$ is suitable. Then $\mu(t)$ is again a state. 
\end{proposition}
The proof of this relies on a very deep result of classical mechanics, the so-called 
\begin{theorem}[Liouville]\label{classical_mechanics:thm:Liouville}
	The hamiltonian vector field is divergence free, \ie the hamiltonian flow preserves volume in phase space of bounded subsets $V$ of $\Pspace$ with smooth boundary $\partial V$. In particular, the functional determinant of the flow is constant and equal to 
	\begin{align*}
		\mathrm{det} \, \bigl ( D \Phi_t(q,p) \bigr ) = 1 
	\end{align*}
	for all $t \in \R$ and $(q,p) \in \Pspace$. 
\end{theorem}
\begin{figure}
	\hfil\includegraphics[height=4cm]{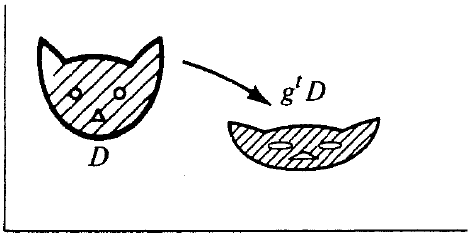}\hfil
	\caption{Phase space volume is preserved under the hamiltonian flow. }
\end{figure}
\begin{remark}\label{classical_mechanics:remark:Liouville}
	We will need a fact from the theory of dynamical systems: if $\Phi_t$ is the flow associated to a differential equation $\dot{x} = F(x)$ with $F \in \Cont^1(\R^n,R^d)$, then 
	\begin{align*}
		\frac{\dd }{\dd t} D \Phi_t(x) = DF \bigl ( \Phi_t(x) \bigr ) \, D \Phi_t(x) 
		, 
		&& D \Phi_t \big \vert_{t = 0} = \id_{\R^n} 
		, 
	\end{align*}
	holds for the differential of the flow. As a consequence, one can prove 
	\begin{align*}
		\frac{\dd }{\dd t} \bigl ( \mathrm{det} \, D \Phi_t(x) \bigr ) &= \mathrm{tr} \, \bigl ( DF \bigl (\Phi_t(x) \bigr ) \bigr ) \, \mathrm{\det} \, \bigl ( D \Phi_t(x) \bigr ) 
		\\
		&= \mathrm{div} \, F \bigl ( \Phi_t(x) \bigr ) \, \mathrm{\det} \, \bigl ( D \Phi_t(x) \bigr ) 
		, 
	\end{align*}
	and $\mathrm{det} \, \bigl ( D \Phi_t \bigr ) \big \vert_{t = 0} = 1$. There are more elegant, general and geometric proofs of this fact, but they are beyond the scope of our short introduction. 
\end{remark}
\begin{proof}[Theorem~\ref{classical_mechanics:thm:Liouville}]
	Let $H$ be the hamiltonian which generates the flow $\Phi_t$. Let us denote the hamiltonian vector field by 
	\begin{align*}
		X_H = \left (
		\begin{matrix}
			+ \nabla_p H \\
			- \nabla_q H \\
		\end{matrix}
		\right ) 
		. 
	\end{align*}
	Then a direct calculation yields 
	\begin{align*}
		\mathrm{div} \, X_H = \sum_{j = 1}^n \Bigl ( \partial_{q_j} \bigl ( + \partial_{p_j} H \bigr ) + \partial_{p_j} \bigl ( - \partial_{q_j} H \bigr ) \Bigr ) = 0 
	\end{align*}
	and the hamiltonian vector field is divergence free. This implies the hamiltonian flow $\Phi_t$ preserves volumes in phase space: let $V \subseteq \Pspace$ be a bounded region in phase space (a Borel subset) with smooth boundary. Then for all $- T \leq t \leq T$ for which the flow exists, we have 
	\begin{align*}
		\frac{\dd}{\dd t} \mathrm{Vol} \, \bigl ( \Phi_t(V) \bigr ) &= \frac{\dd}{\dd t} \int_{\Phi_t(V)} \dd q \, \dd p 
		\\
		&= \frac{\dd}{\dd t} \int_{V} \dd x' \, \dd p' \, \mathrm{det} \, \bigl ( D \Phi_t(q',p') \bigr )
		. 
	\end{align*}
	Since $V$ is bounded, we can bound $\mathrm{det} \, \bigl ( D \Phi_t \bigr )$ and its time derivative uniformly. Thus, we can interchange integration and differentiation and apply Remark~\ref{classical_mechanics:remark:Liouville}, 
	\begin{align*}
		\frac{\dd}{\dd t} \int_{V} \dd x' \, \dd p' \, \mathrm{det} \, \bigl ( D \Phi_t(q',p') \bigr ) &= \int_V \dd x' \, \dd p' \, \frac{\dd}{\dd t} \mathrm{det} \, \bigl ( D \Phi_t(q',p') \bigr ) 
		\\
		&
		= \int_V \dd x' \, \dd p' \, \underbrace{\mathrm{div} \, X_H \bigl ( \Phi_t(q',p') \bigr )}_{= 0} \, \mathrm{det} \, \bigl ( D \Phi_t(q',p') \bigr ) 
		\\
		&
		= 0 
		. 
	\end{align*}
	Hence $\frac{\dd}{\dd t} \mathrm{Vol} \, (V) = 0$ and the hamiltonian flow conserves phase space volume. The functional determinant of the flow is constant as the time derivative vanishes, 
	\begin{align*}
		\frac{\dd}{\dd t} \mathrm{det} \, \bigl ( D \Phi_t(q',p') \bigr ) = 0 
		, 
	\end{align*}
	and equal to $1$, 
	\begin{align*}
		\mathrm{det} \, \bigl ( D \Phi_t (q',p') \bigr ) \big \vert_{t = 0} = \mathrm{det} \, \id_{\Pspace} = 1 
		. 
	\end{align*}
	This concludes the proof. 
\end{proof}
With a different proof relying on alternating multilinear forms, the requirements on $V$ can be lifted, see \eg \cite[Theorem on pp.~204--207]{Arnold:classical_mechanics:1997}. 
\begin{proof}[Proposition~\ref{classical_mechanics:prop:states_stay_states}]
	Since $\Phi_t$ is continuous, it is also measurable. Thus $\mu(t) = \mu \circ \Phi_{-t}$ is also a Borel measure on $\Pspace$ ($\Phi_{-t}$ exists by assumption on $t$). In fact, $\Phi_t$ is a diffeomorphism on phase space. Liouville's theorem~\ref{classical_mechanics:thm:Liouville} not only ensures that the measure $\mu(t)$ remains positive, but also that it is normalized to 1: let $U \subseteq \Pspace$ be a Borel set. Then we conclude 
	\begin{align*}
		\bigl ( \mu(t) \bigr ) (U) &= \int_{U} \dd \bigl ( \mu(t) \bigr )(q,p) 
		= \int_{U} \dd \mu \bigl ( \Phi_{-t}(q,p) \bigr ) 
		\\ 
		&= \int_{\Phi_{-t}(U)} \dd \mu (q,p) \, \mathrm{det} \bigl ( D \Phi_t(q,p) \bigr ) 
		= \int_{\Phi_{-t}(U)} \dd \mu (q,p) \geq 0 
	\end{align*}
	where we have used the positivity of $\mu$ and the fact that $\Phi_{-t}(U)$ is again a Borel set by continuity of $\Phi_{-t}$. If we set $U = \Pspace$ and use the fact that the flow is a diffeomorphism, we see that $\Pspace$ is mapped onto itself, $\Phi_{-t}(\Pspace) = \Pspace$, and the normalization of $\mu$ leads to 
	\begin{align*}
		\bigl ( \mu(t) \bigr ) (\Pspace) &= \int_{\Phi_{-t}(\Pspace)} \dd \mu (q,p) = \int_{\Pspace} \dd \mu (q,p) = 1
	\end{align*}
	This concludes the proof. 
\end{proof}
%
% section Conservation of probability and the Liouville theorem (end)

\section{Equation of motion for observables and Poisson algebras} % (fold)
\label{classical_mechanics:Poisson}
Viewed in the Heisenberg picture, observables move in time, $f(t) = f \circ \Phi_t$, while states are remain fixed. Seeing as $\Phi_t$ is invertible, it maps $\R^{2n}$ onto itself, and thus the spectrum of the observable does not change in time, 
\begin{align*}
	\spec f(t) = \spec f 
	. 
\end{align*}
For many applications and arguments, it will be helpful to find a dynamical equation for $f(t)$ directly: 
\begin{proposition}\label{classical_mechanics:prop:equations_of_motion_Poisson_bracket}
	Let $f \in \Cont^{\infty}(\Pspace,\R)$ be an observable and $\Phi$ the hamiltonian flow which solves the equations of motion~\eqref{classical_mechanics:eqn:hamiltons_eom} associated to a hamiltonian $H \in \Cont^{\infty}(\Pspace,\R)$ which we assume to exist globally in time for all $(q_0,p_0) \in \Pspace$. Then 
	\begin{align}
		\frac{\dd }{\dd t} f(t) = \bigl \{ H , f(t) \bigr \} 
		\label{classical_mechanics:eqn:eom_observables}
	\end{align}
	holds where $\bigl \{ f , g \bigr \} := \sum_{j = 1}^n \bigl ( \partial_{p_j} f \, \partial_{q_j} g - \partial_{q_j} f \, \partial_{p_j} g \bigr )$ is the so-called \emph{Poisson bracket}. 
\end{proposition}
\begin{proof}
	Theorem~\ref{odes:thm:smoothness_flow} implies the smoothness of the flow from the smoothness of the hamiltonian. This means $f(t) \in \Cont^{\infty}(\Pspace,\R)$ is again a classical observable. By assumption, all initial conditions lead to trajectories that exist globally in time.\footnote{A slightly more sophisticated argument shows that the Proposition holds if the hamiltonian flow exists only locally in time. } For $(q_0,p_0)$, we compute the time derivative of $f(t)$ to be 
	\begin{align*}
		\left ( \frac{\dd }{\dd t} f(t) \right )(q_0,p_0) &= \frac{\dd }{\dd t} f \bigl ( q(t) , p(t) \bigr ) 
		\\
		&= \sum_{j = 1}^n \Bigl ( \partial_{q_j} f \circ \Phi_t(q_0,p_0) \, \dot{q}_j(t) + \partial_{p_j} f \circ \Phi_t(q_0,p_0) \, \dot{p}_j(t) \Bigr ) 
		\\
		&\overset{\ast}{=} \sum_{j = 1}^n \Bigl ( \partial_{q_j} f \circ \Phi_t(q_0,p_0) \, \partial_{p_j} H \circ \Phi_t(q_0,p_0) 
		+ \\
		&\qquad \qquad 
		+ \partial_{p_j} f \circ \Phi_t(q_0,p_0) \, \bigl ( - \partial_{q_j} H \circ \Phi_t(q_0,p_0) \bigr ) \Bigr ) 
		\\
		&= \bigl \{ H(t) , f(t) \bigr \}(q_0,p_0) 
		. 
	\end{align*}
	In the step marked with $\ast$, we have inserted the Hamilton's equations of motion. Compared to equation~\eqref{classical_mechanics:eqn:eom_observables}, we have $H$ instead of $H(t)$ as argument in the Poisson bracket. However, by setting $f(t) = H(t)$ in the above equation, we see that energy is a \emph{conserved quantity}, 
	\begin{align*}
		\frac{\dd }{\dd t} H(t) = \bigl \{ H(t) , H(t) \bigr \} = 0 
		.
	\end{align*}
	Hence, we can replace $H(t)$ by $H$ in the Poisson bracket with $f$ and obtain equation~\eqref{classical_mechanics:eqn:eom_observables}.
\end{proof}
\begin{remark}\label{classical_mechanics:remark:uniqueness_solution_eom_observables}
	One can prove that $f(t) = f \circ \Phi_t$ is the \emph{only} solution to equation~\eqref{classical_mechanics:eqn:eom_observables}, but that requires a little more knowledge about symplectic geometry (\cf Proposition~5.4.2 and Proposition~5.5.2 of \cite{Marsden_Ratiu:intro_mechanics_symmetry:1999}). \marginpar{2013.09.26}
\end{remark}
\paragraph{Conserved quantities} % (fold)
The proof immediately leads to the notion of conserved quantity: 
\begin{definition}[Conserved quantity/constant of motion]\label{classical:defn:conserved_quantity}
	An observable $f \in \Cont^{\infty}(\Pspace,\R)$ which is invariant under the flow $\Phi$ generated by the hamiltonian $H \in \Cont^{\infty}(\Pspace,\R)$, \ie 
	\begin{align*}
		f(t) = f(0) 
		, 
	\end{align*}
	or equivalently satisfies 
	\begin{align*}
		\frac{\dd }{\dd t} f(t) = \bigl \{ H , f(t) \bigr \} = 0 
		, 
	\end{align*}
	is called \emph{conserved quantity} or \emph{constant of motion}. 
\end{definition}
As is very often in physics and mathematics, we have completed the circle: starting from the Hamilton's equations of motion, we have proven that the time evolution of observables is given by the Poisson bracket. Alternatively, we could have \emph{started} by postulating 
\begin{align*}
	\frac{\dd}{\dd t} f(t) = \bigl \{ H , f(t) \bigr \} 
\end{align*}
for observables and we would have \emph{arrived} at the Hamilton's equations of motion by plugging in $q$ and $p$ as observables. 

Seemingly, to check whether an observable is a constant of motion requires one to \emph{solve} the equation of motion, but this is not so. A very important property of the Poisson bracket is the following: 
\begin{proposition}[Properties of the Poisson bracket]\label{classical_mechanics:prop:properties_Poisson_bracket}
	If $\Phi_t$ is the flow associated to $H \in \Cont^{\infty}(\R^{2n})$, then for any $f , g , h \in \Cont^{\infty}(\R^{2n})$, the following statements hold true: 
	\begin{enumerate}[(i)]
		\item $\{ \cdot \, , \cdot \} : \Cont^{\infty}(\R^{2n}) \times \Cont^{\infty}(\R^{2n}) \longrightarrow \Cont^{\infty}(\R^{2n})$
		\item $\{ f , g \} = - \{ g , f \}$ (antisymmetry)
		\item $\{ f , g \} \circ \Phi_t = \bigl \{ f \circ \Phi_t , g \circ \Phi_t \bigr \}$ 
		\item $\bigl \{ f , \{ g , h \} \bigr \} + \bigl \{ h , \{ f , g \} \bigr \} + \bigl \{ g , \{ h , f \} \bigr \} = 0$ (Jacobi identity)
		\item $\{ f g , h \} = f \, \{ g , h \} + g \, \{ f , h \}$ (derivation property)
	\end{enumerate}
\end{proposition}
\begin{proof}
	\begin{enumerate}[(i)]
		\item $\{ f , g \}$ consists of products of derivatives of $f$ and $g$, and thus their Poisson bracket is in the class $\Cont^{\infty}(\R^{2n})$. 
		\item The antisymmetry is obvious. 
		\item The fact that the time evolution and the Poisson bracket commute is a very deep result of Hamiltonian mechancs (see \eg \cite[Proposition~5.4.2]{Marsden_Ratiu:intro_mechanics_symmetry:1999}), but that goes far beyond our current capabilities. 
		\item This follows from either a straight-forward (and boring) calculation or one can use (iii): we compute the derivative of this equation and obtain 
		\begin{align*}
			\frac{\dd}{\dd t} &\{ f , g \} \circ \Phi_t - \frac{\dd}{\dd t} \bigl \{ f \circ \Phi_t , g \circ \Phi_t \bigr \} = 
			\\
			&= \bigl \{ H , \{ f , g \} \circ \Phi_t \bigr \} - \bigl \{ \{ H , f \circ \Phi_t \} , g \circ \Phi_t \bigr \} - \bigl \{ f \circ \Phi_t , \{ H , g \circ \Phi_t \} \bigr \}
			\\
			&= \bigl \{ H , \{ f , g \} \circ \Phi_t \bigr \} + \bigl \{ g \circ \Phi_t , \{ H , f \circ \Phi_t \} \bigr \} + \bigl \{ f \circ \Phi_t , \{ g \circ \Phi_t , H \} \bigr \}
			. 
		\end{align*}
		Setting $t = 0$ yields the Jacobi identity. 
		\item The derivation property follows directly from the product rule for partial derivatives and the definition of $\{ \cdot \, , \cdot \}$. 
	\end{enumerate}
\end{proof}
Even though we cannot prove (iii) with our current means, under the \emph{assumption} that the solution to equation~\eqref{classical_mechanics:eqn:eom_observables} for $f = 0$ is unique and given by $f(t) = 0$, we can deduce (iii) using the Jacobi identity (iv): $\{ f , g \} \circ \Phi_t = \bigl \{ f \circ \Phi_t , g \circ \Phi_t \bigr \}$ holds if this equality is true for $t = 0$ (which follows directly from $\Phi_0 = \id_{\R^{2n}}$) and if the time derivative of this equality is satisfied. Hence, we compare 
\begin{align*}
	\frac{\dd}{\dd t} \{ f , g \} \circ \Phi_t &= \bigl \{ H , \{ f , g \} \circ \Phi_t \bigr \} 
\end{align*}
which holds by Proposition~\ref{classical_mechanics:prop:equations_of_motion_Poisson_bracket} to 
\begin{align*}
	\frac{\dd}{\dd t} \bigl \{ f \circ \Phi_t , g \circ \Phi_t \bigr \} &= \bigl \{ \{ H , f \circ \Phi_t \} , g \circ \Phi_t \bigr \} + \bigl \{ f \circ \Phi_t , \{ H , g \circ \Phi_t \} \bigr \} 
	\\
	&= - \bigl \{ f \circ \Phi_t , \{ g \circ \Phi_t , H \} \bigr \} - \bigl \{ g \circ \Phi_t , \{ H , f \circ \Phi_t \} \bigr \}
	\\
	&= + \bigl \{ H , \{ f \circ \Phi_t , g \circ \Phi_t \} \bigr \} 
	. 
\end{align*}
Hence, the difference $\Delta(t) := \{ f , g \} \circ \Phi_t - \{ f \circ \Phi_t , g \circ \Phi_t \}$ satisfies equation~\eqref{classical_mechanics:eqn:eom_observables}, 
\begin{align*}
	\frac{\dd}{\dd t} \Delta(t) &= \bigl \{ H , \Delta(t) \bigr \} 
	, 
\end{align*}
with initial condition $\Delta(0) = 0$. We have assumed this equation has the unique solution $\Delta(t) = 0$ which yields (iii). 
\medskip

\noindent
Proposition~\ref{classical_mechanics:prop:properties_Poisson_bracket}~(iii) simplifies \emph{computing} the solution $f(t)$ and finding constants of motion: 
\begin{corollary}
	\begin{enumerate}[(i)]
		\item Equation~\eqref{classical_mechanics:eqn:eom_observables} is equivalent to 
		\begin{align*}
			\frac{\dd}{\dd t} f \circ \Phi_t &= \bigl \{ H , f \bigr \} \circ \Phi_t 
			, 
			&&
			f \circ \Phi_0 = f 
			. 
		\end{align*}
		\item $f$ is a constant of motion if and only if 
		\begin{align*}
			\{ H , f \} = 0 
			. 
		\end{align*}
	\end{enumerate}
\end{corollary}
\begin{proof}
	\begin{enumerate}[(i)]
		\item Since $H = H(t)$ is a constant of motion, the right-hand side of \eqref{classical_mechanics:eqn:eom_observables} can be written as 
		\begin{align*}
			\bigl \{ H , f \circ \Phi_t \bigr \} &= \bigl \{ H \circ \Phi_t , f \circ \Phi_t \bigr \} 
			\\
			&= \bigl \{ H , f \bigr \} \circ \Phi_t
		\end{align*}
		using Proposition~\ref{classical_mechanics:prop:properties_Poisson_bracket}~(iii). 
		\item This follows directly from (i) and the definition of constant of motion. 
	\end{enumerate}
\end{proof}
%
% paragraph Conserved quantities (end)
Moreover, it turn out to be \marginpar{2013.10.01}
\begin{definition}[Poisson algebra]
	Let $\mathcal{P} \subseteq \Cont^{\infty}(\R^{2n})$ be a subalgebra of the smooth functions (\ie $\mathcal{P}$ is closed under taking linear combinations and products). Moreover, assume that the Poisson bracket has the derivation property, $\{ f g , h \} = f \, \{ g , h \} +  g \, \{ f , h \}$, and 
	\begin{align*}
		\{ \cdot \, , \cdot \} : \mathcal{P} \times \mathcal{P} \longrightarrow \mathcal{P}
	\end{align*}
	maps $\mathcal{P} \times \mathcal{P}$ into $\mathcal{P}$. Then $\bigl ( \mathcal{P} , \{ \cdot \, , \cdot \} \bigr )$ is a Poisson algebra. 
\end{definition}
%
% section Equation of motion for observables and Poisson algebras (end)

\section{Equivalence of Schrödinger and Heisenberg picture} % (fold)
\label{classical_mechanics:equivalence_S_H}
We are still missing the proof that Heisenberg and Schrödinger picture equally describe the physics. The main observation is that taking expectation values in either lead to the same result, \ie that for any observable $f$ and state $\mu$ 
\begin{align*}
	\mathbb{E}_{\mu(t)}(f) &= \int_{\R^{2n}} \dd q \, \dd p \, \bigl ( \mu(t) \bigr )(q,p) \, f(q,p) 
	\\
	&
	= \int_{\R^{2n}} \dd q \, \dd p \, \mu \circ \Phi_{-t}(q,p) \, f(q,p) 
	\\
	&= \int_{\Phi_t(\R^{2n})} \dd q \, \dd p \, \mu(q,p) \, f \circ \Phi_t(q,p) \, \mathrm{det} \bigl ( D \Phi_t(q,p) \bigr ) 
	\\
	&
	\overset{\ast}{=} \int_{\R^{2n}} \dd q \, \dd p \, \mu(q,p) \, \bigl ( f(t) \bigr )(q,p)
	\\
	&
	= \mathbb{E}_{\mu} \bigl ( f(t) \bigr ) 
\end{align*}
holds. Note that the crucial ingredient in the step marked with $\ast$ is again the Liouville theorem~\ref{classical_mechanics:thm:Liouville}. Moreover, we see why states need to be evolved backwards in time. 
% section Equivalence of Schrödinger and Heisenberg picture (end)

\section{The inclusion of magnetic fields} % (fold)
\label{classical_mechanics:magnetic}
Magnetic fields can only be defined in dimension $2$ or higher; usually, one considers the case $n = 3$. There are two main ways to include magnetic fields into classical mechanics, \emph{minimal substitution}, and a more geometric way where the magnetic field enters into the geometric structure of phase space. Both descriptions are equivalent, though. 

The two-dimensional case can be obtained by restricting the three-dimensional case to the $q_1 \, q_2$-plane. 

The starting point from a physical point of view is the \emph{Lorentz force law}: an electric field $\mathbf{E}$ and a magnetic field $\mathbf{B}(q) = \bigl ( \mathbf{B}_1(q) , \mathbf{B}_2(q) , \mathbf{B}_3(q) \bigr )$ exert a force on a particle with charge $e$ at $q$ moving with velocity $v$ that is given by 
\begin{align}
	F_{\mathrm{L}} = e \, \mathbf{E} + v \times e \, \mathbf{B} 
	. 
	\label{classical_mechanics:eqn:Lorentz_force_law}
\end{align}
For simplicity, from now on, we set the charge $e = 1$. 
\medskip

\noindent
The goal of this section is to include magnetic fields in the framework of Hamiltonian mechanics. Electric fields $\mathbf{E} = - \nabla_q V$ appear as potentials in the Hamilton function $H(q,p) = \frac{1}{2m} p^2 + V(q)$, but magnetic fields are not gradients of a potential. Instead, one can express the magnetic field 
\begin{align*}
	\mathbf{B} = \nabla_q \times \mathbf{A} 
\end{align*}
as the curl of a vector potential $\mathbf{A} = (\mathbf{A}_1 , \mathbf{A}_2 , \mathbf{A}_3)$.

\subsection{Minimal substitution} % (fold)
\label{classical_mechanics:magnetic:minimal_substitution}
The standard way to add the interaction to a magnetic field is to consider the equations of motion for position $q$ and \emph{kinetic momentum} 
\begin{align}
	p^A(p,q) &= p - \mathbf{A}(q) 
	. 
	\label{classical_mechanics:eqn:kinetic_momentum}
\end{align}
Moreover, $p$ is replaced by $p^A = p - \mathbf{A}$ in the Hamiltonian, 
\begin{align*}
	H^A(q,p) := H \bigl ( q , p - \mathbf{A}(q) \bigr ) 
	. 
\end{align*}
One then proposes the usual equations of motion: 
\begin{align*}
	\left (
	\begin{matrix}
		0 & - \id_{\R^3} \\
		+ \id_{\R^3} & 0 \\
	\end{matrix}
	\right ) \left (
	\begin{matrix}
		\dot{q} \\
		\dot{p} \\
	\end{matrix}
	\right ) &= \left (
	\begin{matrix}
		\nabla_q H^A \\
		\nabla_p H^A \\
	\end{matrix}
	\right )
\end{align*}
Taking the time-derivative of kinetic momentum yields 
\begin{align}
	\frac{\dd}{\dd t} p^A_j &= \dot{p}_j - \sum_{k = 1}^3 \partial_{q_k} \mathbf{A}(q) \, \dot{q}_k 
	\notag \\
	&= - \tfrac{\dd}{\dd q_j} H^A - \sum_{k = 1}^3 \partial_{q_k} \mathbf{A}(q) \, \partial_{p_k} H^A 
	\notag \\
	&= - \partial_{q_j} H^A - \sum_{k = 1}^3 \partial_{p_k} H^A \, \partial_{q_j} \bigl ( p_k - \mathbf{A}_k \bigr ) - \sum_{k = 1}^3 \partial_{q_k} \mathbf{A}_j \, \partial_{p_k} H^A 
	\notag \\
	&= - \partial_{q_j} H^A + \sum_{k = 1}^3 \partial_{p_k} H^A \, \bigl ( \partial_{q_j} \mathbf{A}_k - \partial_{q_k} \mathbf{A}_j \bigr )
	. 
	\label{classical_mechanics:eqn:eom_kinetic_momentum}
\end{align}
If we set $B_{jk} := \partial_{q_j} \mathbf{A}_k - \partial_{q_k} \mathbf{A}_j$, then the magnetic field matrix 
\begin{align*}
	B := \bigl ( B_{jk} \bigr )_{1 \leq j , k \leq 3} = \left (
	\begin{matrix}
		0 & +\mathbf{B}_3 & -\mathbf{B}_2 \\
		-\mathbf{B}_3 & 0 & +\mathbf{B}_1 \\
		+\mathbf{B}_2 & -\mathbf{B}_1 & 0 \\
	\end{matrix}
	\right )
	, 
\end{align*}
and use 
\begin{align*}
	B p &= \left (
	\begin{matrix}
		0 & +\mathbf{B}_3 & -\mathbf{B}_2 \\
		-\mathbf{B}_3 & 0 & +\mathbf{B}_1 \\
		+\mathbf{B}_2 & -\mathbf{B}_1 & 0 \\
	\end{matrix}
	\right ) \left (
	\begin{matrix}
		p_1 \\
		p_2 \\
		p_3 \\
	\end{matrix}
	\right ) 
	\\
	&= \left (
	\begin{matrix}
		\mathbf{B}_3 \, p_2 - \mathbf{B}_2 \, p_3 \\
		- \mathbf{B}_3 \, p_1 + \mathbf{B}_1 \, p_3 \\
		\mathbf{B}_2 \, p_1 - \mathbf{B}_1 \, p_2 \\
	\end{matrix}
	\right )
	= p \times \mathbf{B} 
	, 
\end{align*}
we can simplify the equation for $p^A$ to 
\begin{align*}
	\frac{\dd}{\dd t} p^A &= - (\nabla_q H)^A + \nabla_p H^A \times \mathbf{B} 
	= - \nabla_q H + \dot{q} \times \mathbf{B} 
	. 
\end{align*}
For a non-relativistic particle with Hamiltonian $H(q,p) = \frac{1}{2m} p^2 + V(q)$, these equation reduce to the Lorentz force law~\eqref{classical_mechanics:eqn:Lorentz_force_law}: 
\begin{align*}
	\dot{p}^A &= - \nabla_q V + \frac{p^A}{m} \times \mathbf{B} 
	\\
	&= - (\nabla_q H)^A + B \, \nabla_p H^A 
	= - (\nabla_q H)^A + B \, \dot{q} 
\end{align*}
\paragraph{Changes of gauge} % (fold)
A choice of vector potential $\mathbf{A}$ is also called a \emph{choice of gauge}. If $\chi : \R^3 \longrightarrow \R$ is a scalar function then $\mathbf{A}$ and 
\begin{align*}
	\mathbf{A}' &= \mathbf{A} + \nabla_q \chi 
\end{align*}
are both vector potentials associated to $\mathbf{B}$, because $\nabla_q \times \nabla_q \chi = 0$, 
\begin{align*}
	\nabla_q \times \mathbf{A}' &= \nabla_q \times \bigl ( \mathbf{A} + \nabla_q \chi \bigr ) 
	\\
	&= \nabla_q \times \mathbf{A} + \nabla_q \times \nabla_q \chi 
	\\
	&= \nabla_q \times \mathbf{A} 
	= \mathbf{B} 
	. 
\end{align*}
Thus, we can either choose the gauge $\mathbf{A}$ or $\mathbf{A}'$, either one describes the physical situation: the equation of motion for $p^A$ only involves $\mathbf{B}$ rather than $\mathbf{A}$. Hence, gauges are usually chosen for convenience (\eg a particular “symmetry” or being divergence free, $\nabla_q \cdot \mathbf{A} = 0$). 
% paragraph Changes of gauge (end)
% subsection Minimal substitution (end)

\subsection{Magnetic symplectic form} % (fold)
\label{classical_mechanics:magnetic:magnetic_symplectic_form}
A second way to include the interaction to a magnetic field is to integrate it into the symplectic form: 
\begin{align}
	\left (
	\begin{matrix}
		B & - \id_{\R^3} \\
		+ \id_{\R^3} & 0 \\
	\end{matrix}
	\right ) \left (
	\begin{matrix}
		\dot{q} \\
		\dot{p} \\
	\end{matrix}
	\right ) &= \left (
	\begin{matrix}
		\nabla_q H \\
		\nabla_p H \\
	\end{matrix}
	\right )
	\label{classical_mechanics:eqn:magnetic_eom}
\end{align}
Note that neither the Hamiltonian nor the momentum are altered. Instead, in this variant, $q$ is position and $p$ is \emph{kinetic} momentum. Solving the above equation for $\dot{q}$ and $\dot{p}$ yields \eqref{classical_mechanics:eqn:eom_kinetic_momentum}, 
\begin{align*}
	\left (
	\begin{matrix}
		B & - \id_{\R^3} \\
		+ \id_{\R^3} & 0 \\
	\end{matrix}
	\right ) \left (
	\begin{matrix}
		\dot{q} \\
		\dot{p} \\
	\end{matrix}
	\right ) &= \left (
	\begin{matrix}
		B \, \dot{q} - \dot{p} \\
		\dot{q} \\
	\end{matrix}
	\right ) = \left (
	\begin{matrix}
		\nabla_q H \\
		\nabla_p H \\
	\end{matrix}
	\right ) 
	\\
	& \Leftrightarrow 
	\\
	\left (
	\begin{matrix}
		\dot{q} \\
		\dot{p} \\
	\end{matrix}
	\right ) &= \left (
	\begin{matrix}
		+ \nabla_p H \\
		- \nabla_q H + B \, \dot{q} \\
	\end{matrix}
	\right ) 
	= \left (
	\begin{matrix}
		+ \nabla_p H \\
		- \nabla_q H + \dot{q} \times \mathbf{B} \\
	\end{matrix}
	\right )
\end{align*}
In other words, we have again recovered the Lorentz force law~\eqref{classical_mechanics:eqn:Lorentz_force_law}. 

From a mathematical perspective, the advantage of this formulation is that magnetic fields are always “nicer” functions than associated vector potentials. 
% subsection Magnetic symplectic form (end)
% section The inclusion of magnetic fields (end)

\section{Stability analysis} % (fold)
\label{classical_mechanics:stability}
It turns out that hamiltonian systems are \emph{never} stable: we start by linearizing the hamiltonian vector field, 
\begin{align*}
	X_H(q,p) &= \left (
	\begin{matrix}
		+ \nabla_p H(q,p) \\
		- \nabla_q H(q,p) \\
	\end{matrix}
	\right )
	\; \Longrightarrow \; 
	D X_H(q,p) = \left (
	\begin{matrix}
		\nabla_q^T \nabla_p H(q,p) & \nabla_p^T \nabla_p H(q,p) \\
		- \nabla_q^T \nabla_q H(q,p) & - \nabla_p^T \nabla_q H(q,p) \\
	\end{matrix}
	\right ) 
	. 
\end{align*}
Thus, the linearized vector field is always trace-less, 
\begin{align*}
	\trace D X_H(q,p) &= \sum_{j = 1}^n \bigl ( \partial_{q_j} \partial_{p_j} H(q,p) - \partial_{p_j} \partial_{q_j} H(q,p) \bigr ) 
	= 0 
	, 
\end{align*}
and using that the sum of eigenvalues $\{ \lambda_j \}_{j = 1}^N$ of $D X_H(q,p)$ equals the trace of $D X_H(q,p)$, we know that the $\lambda_j$ (repeated according to their multiplicity) sum to $0$, 
\begin{align*}
	\trace D X_H(q,p) = 0 = \sum_{j = 1}^{2n} \lambda_j 
	. 
\end{align*}
Moreover, seeing as the entries of $D X_H(q,p)$ are real, the eigenvalues come in complex conjugate pairs $\{ \lambda_j , \overline{\lambda_j} \}$. Combined with the fact that $D X_H(q,p)$ is $2n \times 2n$-dimensional, we deduce that if $\lambda_j$ is an eigenvalue of $D X_H(q,p)$, then so is $- \lambda_j$. This suggests that hamiltonian systems tend to be either elliptic or hyperbolic.

\paragraph{Electric fields and other gradient fields} % (fold)
To better understand the influence interactions to electromagnetic fields (and other forces which can be expressed as the gradient of a potential) have to the dynamics of a particle, let us start with considering the purely electric case for a non-relativistic particle. Here, the interaction is given by the standard Hamiltonian 
\begin{align*}
	H(q,p) = \tfrac{1}{2m} p^2 + V(q) 
	. 
\end{align*}
which gives the interaction to the electric field $\mathbf{E} = - \nabla_q V$. Here, the vector field 
\begin{align*}
	X_H(q,p) = \left (
	\begin{matrix}
		\frac{p}{m} \\
		- \nabla_q V(q) \\
	\end{matrix}
	\right )
\end{align*}
vanishes at $(q_0,0)$ where $q_0$ is a critical point of $V$, \ie $\nabla_q V(q_0) = 0$. Its linearization 
\begin{align*}
	D X_H(q,p) &= \left (
	\begin{matrix}
		0 & \frac{1}{m} \, \id_{\R^n} \\
		- \mathrm{Hess} V(q) & 0 \\
	\end{matrix}
	\right )
\end{align*}
involves the Hessian
\begin{align*}
	\mathrm{Hess} V = \left (
	\begin{matrix}
		\partial_{q_1}^2 V & \cdots & \partial_{q_1} \partial_{q_n} V \\
		\vdots &  & \vdots \\
		\partial_{q_1} \partial_{q_n} V & \cdots & \partial_{q_n}^2 V \\
	\end{matrix}
	\right )
\end{align*}
of the potential. The block structure allows us to simplify the characteristic polynomial using 
\begin{align*}
	\mathrm{det} \left (
	\begin{matrix}
		A & B \\
		C & D \\
	\end{matrix}
	\right ) &= \mathrm{det} A \; \mathrm{det} \, \bigl ( D - B \, D^{-1} \, C \bigr ) 
	. 
\end{align*}
Then the zeros of the characteristic polynomial 
\begin{align*}
	\chi_{q_0}(\lambda) &= \mathrm{det} \, \bigl ( \lambda \, \id_{\R^{2n}} - D X_H(q_0,0) \bigr ) 
	\\
	&= \mathrm{det} \left (
	\begin{matrix}
		\lambda \, \id_{\R^n} & - \frac{1}{m} \, \id_{\R^n} \\
		+ \mathrm{Hess} \, V(q_0) & \lambda \, \id_{\R^n} \\
	\end{matrix}
	\right )
	\\
	&
	= \mathrm{det} \bigl ( \lambda \, \id_{\R^n} \bigr ) \; \mathrm{det} \, \bigl ( \lambda \, \id_{\R^n} + \lambda^{-1} \, m^{-1} \, \mathrm{Hess} \, V(q_0) \bigr )
	\\
	&= \mathrm{det} \, \bigl ( \lambda^2 \, \id_{\R^n} + m^{-1} \, \mathrm{Hess} \, V(q_0) \bigr )
\end{align*}
are the square roots of the eigenvalues of $- m^{-1} \, \mathrm{Hess} \, V(q_0)$. In case $q_0$ is a local \emph{maximum}, for instance, then 
\begin{align}
	\mathrm{Hess} \, V(q_0) > 0 
	\; \; \Leftrightarrow \; \; 
	x \cdot \mathrm{Hess} \, V(q_0) x > 0 \quad \forall x \in \R^3 
\end{align}
holds in the sense of matrices; this equation is equivalent to requiring all eigenvalues $\omega_j$ of the Hessian to be positive. Hence, the eigenvalues of the linearized vector field are 
\begin{align*}
	\lambda_{\pm j} &= \pm \sqrt{- \tfrac{\omega_j}{m}} 
	= \pm \ii \, \sqrt{\tfrac{\omega_j}{m}} 
	, 
\end{align*}
which means $(q_0,0)$ is a marginally stable, elliptic fixed point. 
% paragraph Stability of electric fields (end)

\paragraph{Magnetic fields} % (fold)
In case \emph{only} a magnetic field is present, \ie we consider the Hamilton function $H(q,p) = \frac{1}{2m} p^2$ and the \emph{magnetic} equations of motion~\eqref{classical_mechanics:eqn:magnetic_eom}. The corresponding vector field 
\begin{align*}
	X_H(q,p) &= \frac{1}{m} \, \left (
	\begin{matrix}
		p \\
		B(q) \, p \\
	\end{matrix}
	\right )
\end{align*}
linearizes to 
\begin{align*}
	D X_H(q,p) &= \frac{1}{m} \, \left (
	\begin{matrix}
		0 & \id_{\R^3} \\
		B'(q) \, p & B(q) \\
	\end{matrix}
	\right )
\end{align*}
where 
\begin{align*}
	B' \, p := \nabla_q^{\mathrm{T}} \bigl ( B \, p \bigr )
	. 
\end{align*}
Any of the fixed points are of the form $(q_0,0)$, so that we need to find the eigenvalues of 
\begin{align*}
	D X_H(q_0,0) &= \frac{1}{m} \, \left (
	\begin{matrix}
		0 & \id_{\R^3} \\
		0 & B(q_0) \\
	\end{matrix}
	\right ) 
	. 
\end{align*}
Using the block form of the matrix, we see right away that three eigenvalues are $0$ while the others are, up to a factor of $\nicefrac{1}{m}$, the eigenvalues of $B(q)$: 
\begin{align*}
	\chi_{q_0}(\lambda) :& \negmedspace= \mathrm{det} \, \bigl ( \lambda \, \id_{\R^6} - D X_H(q_0,0) \bigr ) 
	\\
	&= \mathrm{det} \, \left (
	\begin{matrix}
		\lambda \, \id_{\R^3} & - \tfrac{1}{m} \, \id_{\R^3} \\
		0 & \lambda \, \id_{\R^3} - \tfrac{1}{m} \, B(q_0) \\
	\end{matrix}
	\right )
	\\
	&= \lambda^3 \; \mathrm{det} \, \bigl ( \lambda \, \id_{\R^3} - \tfrac{1}{m} \, B(q_0) \bigr ) 
\end{align*}
The eigenvalues of $B(q_0)$ are $0$ and $\pm \ii \, \babs{\mathbf{B}(q_0)}$. This can be seen from $B = \overline{B} = - B^T$, $\mathrm{tr} B = 0$ and $\mathrm{det} \, B = 0$ which implies: (i) $\lambda_1 = 0$, (ii) $\lambda_2 = \overline{\lambda_3}$ and $\lambda_2 + \lambda_3 = 0$. \marginpar{2013.10.03}

Hence, magnetic field have \emph{metastable, elliptic} fixed points which are all of the form $(q_0,0)$. That means, we are confronted with two problems: first of all, $4$ of the eigenvalues are $0$, so there are many metastable directions. The second one is much more serious: linearization is a \emph{local} technique, and studying the stability via linearization hinges on the fact that you can separate fixed points by open neighborhoods. But here, none of the fixed points can be isolated from the others (in the sense that there does not exist an open neighborhood which contains only a single fixed point). So one needs more care: for instance, it is crucial to look at how the direction of $\mathbf{B}$ changes, looking at the linearization of the vector field is insufficient. 
% paragraph Stability of magnetic fields (end)
% section Stability analysis (end)
% chapter Classical mechanics (end)
\chapter{Banach \& Hilbert spaces} % (fold)
\label{spaces}
This section intends to introduce some fundamental notions on \emph{Banach spaces}; those are vector spaces of functions where the notion of distance is compatible with the linear structure. In addition, \emph{Hilbert spaces} also allow one to introduce a notion of angle via a \emph{scalar product}. 

Those notions are crucial to understand PDEs and ODEs, because these are defined on a \emph{domain} (similar to domains of functions). For instance, one may ask: 
\begin{enumerate}[(i)]
	\item How does the existence of solutions depend on the domain, \eg by imposing different boundary conditions? 
	\item How well can I approximate a solution with a given set of base vectors? This is important for numerics, because one needs to approximate elements of infinite-dimensional spaces by finite linear combinations. 
\end{enumerate}

\section{Banach spaces} % (fold)
\label{spaces:Banach}
Many vector spaces $\mathcal{X}$ can be equipped with a \emph{norm} $\norm{\cdot}$, and if they are complete with respect to that norm, the pair $\bigl ( \mathcal{X} , \norm{\cdot} \bigr )$ is a Banach space.

\subsection{Abstract Banach spaces} % (fold)
\label{spaces:Banach:generic}
The abstract definition of Banach spaces is quite helpful when we construct Banach spaces from other Banach spaces (\eg Banach spaces of integrable, vector-valued functions).\footnote{We will only consider vector spaces over $\C$, although much of what we do works just fine if the field of scalars is $\R$.} 
\begin{definition}[Normed space]
	Let $\mathcal{X}$ be a vector space. A mapping $\norm{\cdot} : \mathcal{X} \longrightarrow [0,+\infty)$ with properties 
	\begin{enumerate}[(i)]
		\item $\norm{x} = 0$ if and only if $x = 0$, 
		\item $\norm{\alpha x} = \abs{\alpha} \, \norm{x}$, and 
		\item $\norm{x + y} \leq \norm{x} + \norm{y}$, 
	\end{enumerate}
	for all $x,y \in \mathcal{X}$, $\alpha \in \C$, is called norm. The pair $(\mathcal{X},\norm{\cdot})$ is then referred to as normed space. 
\end{definition}
A norm on $\mathcal{X}$ quite naturally induces a metric by setting 
\begin{align*}
	d(x,y) := \norm{x-y}
\end{align*}
for all $x,y \in \mathcal{X}$. Unless specifically mentioned otherwise, one always works with the metric induced by the norm. 
\begin{definition}[Banach space]
	A complete normed space is a Banach space. 
\end{definition}
\begin{example}
	The space $\mathcal{X} = \Cont([a,b],\C)$ from the previous list of examples has a norm, the sup norm 
	\begin{align*}
		\norm{f}_{\infty} = \sup_{x \in [a,b]} \abs{f(x)} 
		. 
	\end{align*}
	Since $\Cont([a,b],\C)$ is complete, it is a Banach space.
\end{example}
%
% subsection Generic Banach spaces (end)

\subsection{Prototypical Banach spaces: $L^p(\Omega)$ spaces} % (fold)
\label{spaces:Banach:prototypical}
%
% FIXME change text 
The prototypical examples of Banach spaces are the so-called \emph{$L^p$ spaces} or $p$-integrable spaces; a rigorous definition requires a bit more care, so we refer the interested reader to \cite{Lieb_Loss:analysis:2001}. 

When we say integrable, we mean integrable with respect to the Lebesgue measure \cite[p.~6~ff.]{Lieb_Loss:analysis:2001}. For any open or closed set $\Omega \subseteq \R^n$, the space of \emph{$p$-integrable functions} $\mathcal{L}^p(\Omega)$ is a $\C$-vector space, but 
\begin{align*}
	\norm{\varphi}_{L^p(\Omega)} := \left ( \int_{\Omega} \dd x \, \abs{\varphi(x)}^p \right )^{\nicefrac{1}{p}}
\end{align*}
is not a norm: there are functions $\varphi \neq 0$ for which $\norm{\varphi} = 0$. Instead, $\norm{\varphi} = 0$ only ensures 
\begin{align*}
	\varphi(x) = 0 \mbox{ almost everywhere (with respect to the Lebesgue measure $\dd x$).} 
\end{align*}
Almost everywhere is sometimes abbreviated with a.~e. and the terms “almost surely” and “for almost all $x \in \Omega$” can be used synonymously. If we introduce the equivalence relation 
\begin{align*}
	\varphi \sim \psi :\Leftrightarrow \norm{\varphi - \psi} = 0 
	, 
\end{align*}
then we can define the vector space $L^p(\Omega)$: 
\begin{definition}[$L^p(\Omega)$]
	Let $1 \leq p < \infty$. Then we define 
	\begin{align*}
		\mathcal{L}^p(\Omega) := \Bigl \{ f : \Omega \longrightarrow \C \; \big \vert \; \mbox{$f$ measurable, } \int_{\Omega} \dd x \, \abs{f(x)}^p < \infty \Bigr \} 
	\end{align*}
	as the vector space of functions whose $p$th power is integrable. Then $L^p(\Omega)$ is the vector space 
	\begin{align*}
		L^p(\Omega) := \mathcal{L}^p(\Omega) / \sim 
	\end{align*}
	consisting of equivalence classes of functions that agree almost everywhere. With the $p$ norm 
	\begin{align*}
		\norm{f}_p := \biggl ( \int_{\Omega} \dd x \, \abs{f(x)}^p \biggr )^{\nicefrac{1}{p}} 
	\end{align*}
	it forms a normed space. 
\end{definition}
In practice, one usually does not distinguish between equivalence classes of functions $[f]$ (which make up $L^p(\Omega)$) and functions $f$. This abuse of notation is pervasive in the literature and it is perfectly acceptable to write $f \in L^p(\Omega)$ even though strictly speaking, one should write $[f] \in L^p(\Omega)$. Only when necessary, one takes into account that $f = 0$ actually means $f(x) = 0$ for almost all $x \in \Omega$. 

In case $p = \infty$, we have to modify the definition a little bit. 
\begin{definition}[$L^{\infty}(\Omega)$]
	We define 
	\begin{align*}
		\mathcal{L}^{\infty}(\Omega) := \Bigl \{ f : \Omega \longrightarrow \C \; \big \vert \; \mbox{$f$ measurable, } \exists 0 < K < \infty : \abs{f(x)} \leq K \mbox{ almost everywhere} \Bigr \} 
	\end{align*}
	to be the space of functions that are bounded almost everywhere and 
	\begin{align*}
		\norm{f}_{\infty} := \mathrm{ess} \sup_{x \in \Omega} \babs{f(x)} 
		:= \inf \bigl \{ K \geq 0 \; \big \vert \abs{f(x)} \leq K \mbox{ for almost all $x \in \Omega$} \bigr \} 
		. 
	\end{align*}
	Then the space $L^{\infty}(\Omega) := \mathcal{L}^{\infty}(\Omega) / \sim$ is defined as the vector space of equivalence classes where two functions are identified if they agree almost everywhere.\marginpar{2013.10.08} 
\end{definition}
\begin{theorem}[Riesz-Fischer]
	For any $1 \leq p \leq \infty$, $L^p(\Omega)$ is complete with respect to the $\norm{\cdot}_p$ norm and thus, a Banach space. 
\end{theorem}
\begin{definition}[Separable Banach space]
	A Banach space $\mathcal{X}$ is called separable if there exists a countable dense subset. 
\end{definition}
This condition is equivalent to asking that the space has a \emph{countable basis}. 
\begin{theorem}
	For any $1 \leq p < \infty$, the Banach space $L^p(\Omega)$ is separable. 
\end{theorem}
\begin{proof}
	We refer to \cite[Lemma~2.17]{Lieb_Loss:analysis:2001} for an explicit construction. The idea is to approximate arbitrary functions by functions which are constant on cubes and take only values in the rational complex numbers. 
\end{proof}
For future reference, we collect a few facts on $L^p(\Omega)$ spaces. In particular, we will make use of dominated convergence frequently. We will give them without proof, they can be found in standard text books on analysis, see \eg \cite{Lieb_Loss:analysis:2001}. 
\begin{theorem}[Monotone Convergence]
	Let $(f_k)_{k \in \N}$ be a sequence of non-decreasing functions in $L^1(\Omega)$ with pointwise limit $f$ defined almost everywhere. Define $I_k := \int_{\Omega} \dd x \, f_k(x)$; then the sequence $(I_k)$ is non-decreasing as well. If $I := \lim_{k \to \infty} I_k < \infty$, then $I = \int_{\Omega} \dd x \, f(x)$, \ie 
	\begin{align*}
		\lim_{k \to \infty} \int_{\Omega} \dd x \, f_k(x) = \int_{\Omega} \dd x \, \lim_{k \to \infty} f_k(x) 
		= \int_{\Omega} \dd x \, f(x) 
	\end{align*}
	holds. 
\end{theorem}
\begin{theorem}[Dominated Convergence]
	Let $(f_k)_{k \in \N}$ be a sequence of functions in $L^1(\Omega)$ that converges almost everywhere pointwise to some $f : \Omega \longrightarrow \C$. If there exists a non-negative $g \in L^1(\Omega)$ such that $\abs{f_k(x)} \leq g(x)$ holds almost everywhere for all $k \in \N$, then $g$ also bounds $\abs{f}$, \ie $\abs{f(x)} \leq g(x)$ almost everywhere, and $f \in L^1(\Omega)$. Furthermore, the limit $k \to \infty$ and integration with respect to $x$ commute and we have 
	\begin{align*}
		\lim_{k \to \infty} \int_{\Omega} \dd x \, f_k(x) = \int_{\Omega} \dd x \, \lim_{k \to \infty} f_k(x) = \int_{\Omega} \dd x \, f(x) 
		. 
	\end{align*}
\end{theorem}
\begin{example}[First half of Riemann-Lebesgue lemma]\label{spaces:example:Riemann-Lebesgue_half}
	We define the Fourier transform of $f \in L^1(\R^n)$ as 
	\begin{align*}
		(\Fourier f)(\xi) := \frac{1}{(2\pi)^{\nicefrac{n}{2}}} \int_{\R^n} \dd x \, \e^{- \ii \xi \cdot x} \, f(x) 
		. 
	\end{align*}
	The integrability of $f$ implies that $\Fourier f$ is uniformly bounded: 
	\begin{align*}
		\babs{(\Fourier f)(\xi)} &\leq \frac{1}{(2\pi)^{\nicefrac{n}{2}}} \int_{\R^n} \dd x \, \babs{\e^{- \ii \xi \cdot x} \, f(x)} 
		= \norm{f}_{L^1(\R^n)} 
	\end{align*}
	In fact, $\Fourier f$ is continuous, and the crucial tool in the proof is dominated convergence: to see that $\Fourier f$ is continuous in $\xi_0 \in \R^n$, let $(\xi_n)$ be any sequence converging to $\xi_0$. Since we can bound the integrand uniformly in $\xi$, 
	\begin{align*}
		\babs{\e^{- \ii \xi \cdot x} \, f(x)} \leq \abs{f(x)} 
		, 
	\end{align*}
	dominated convergence applies, and we may interchange integration and differentiation, 
	\begin{align*}
		\lim_{\xi \to \xi_0} (\Fourier f)(\xi) &= \lim_{\xi \to \xi_0} \frac{1}{(2\pi)^{\nicefrac{n}{2}}} \int_{\R^n} \dd x \, \e^{- \ii \xi \cdot x} \, f(x) 
		\\
		&= \frac{1}{(2\pi)^{\nicefrac{n}{2}}} \int_{\R^n} \dd x \, \lim_{\xi \to \xi_0} \bigl ( \e^{- \ii \xi \cdot x} \, f(x) \bigr ) 
		\\
		&= \frac{1}{(2\pi)^{\nicefrac{n}{2}}} \int_{\R^n} \dd x \, \e^{- \ii \xi_0 \cdot x} \, f(x) 
		= (\Fourier f)(\xi_0) 
		. 
	\end{align*}
	This means $\Fourier f$ is continuous in $\xi_0$. However, $\xi_0$ was chosen arbitrarily so that we have in fact $\Fourier f \in L^{\infty}(\R^n) \cap \Cont(\R^n)$. 
\end{example}
%
% subsection Prototypical Banach spaces: $L^p(\Omega)$ spaces (end)

\subsection{Boundary value problems} % (fold)
\label{spaces:boundary_value_problems}
Now let us consider the wave equation 
\begin{align}
	\partial_t^2 u - \partial_x^2 u &= 0 
	\label{spaces:eqn:wave_equation}
\end{align}
in one dimension for the initial conditions 
\begin{align*}
	u(x,0) &= \varphi(x) 
	, 
	\\
	u'(x,0) &= \psi(x) 
	. 
\end{align*}
This formal description defines only \emph{half} of the wave equation, the other half is to state clearly what space $u$ is taken from, and if it should satisfy additional conditions. A priori it is not clear whether $u$ is a function of $\R$ or a subset of $\R$, say, an interval $[a,b]$. The derivatives appearing in \eqref{spaces:eqn:wave_equation} need only exist in the interior of the spatial domain, \eg $(a,b)$. Moreover, we could impose integrability conditions on $u$, \eg $u \in L^1(\R)$. 

If $u$ is a function $[0,L]$ to $\C$, for instance, it turns out that we need to specify the behavior of $u$ or $u'$ at the \emph{boundary}, \eg 
\begin{enumerate}[(i)]
	\item Dirichlet boundary conditions: $u(t,0) = 0 = u(t,L)$ 
	\item Neumann boundary conditions: $\partial_x u(t,0) = 0 = \partial_x u(t,L)$
	\item Mixed or Robin boundary conditions: $\alpha_0 u(t,0) + \beta_0 \partial_x u(t,0) = 0$, $\alpha_L u(t,L) + \beta_L \partial_x u(t,L) = 0$
\end{enumerate}
If one of the boundaries is located at $\pm \infty$, then the corresponding boundary condition often becomes meaningless because “$u(+\infty,T)$” usually makes no sense. 
\medskip

\noindent
Let us start by solving \eqref{spaces:eqn:wave_equation} via a \emph{product ansatz}, \ie we assume $u$ is of the form 
\begin{align*}
	u(x,t) &= \tau(t) \, \xi(x) 
\end{align*}
for suitable functions $\xi$ and $\tau$. Plugging the product ansatz into the wave equation and assuming that $\tau(t)$ and $\xi(x)$ are non-zero yields 
\begin{align*}
	\ddot{\tau}(t) \, \xi(x) - \tau(t) \, \xi''(x) &= 0 
	\; \; \Longleftrightarrow \; \; 
	\frac{\ddot{\tau}(t)}{\tau(t)} = \frac{\xi''(x)}{\xi(x)} = \lambda \in \C 
	. 
\end{align*}
This means $\tau$ and $\xi$ each need to satisfy the harmonic oscillator equation, 
\begin{align*}
	\ddot{\tau} - \lambda \tau &= 0 
	\, , 
	\\
	\xi'' - \lambda \xi &= 0 
	. 
\end{align*}
Note that these two equations are \emph{coupled} via the constant $\lambda$ which has yet to be determined. The solutions to these equations are 
\begin{align*}
	\tau(t) &= 
	\begin{cases}
		a_1(0) + a_2(0) \, t & \lambda = 0 \\
		a_1(\lambda) \, \e^{+ t \sqrt{\lambda}} + a_2(\lambda) \, \e^{- t \sqrt{\lambda}} & \lambda \neq 0 \\
	\end{cases}
\end{align*}
and 
\begin{align*}
	\xi(x) &= 
	\begin{cases}
		b_1(0) + b_2(0) \, x & \lambda = 0 \\
		b_1(\lambda) \, \e^{+ x \sqrt{\lambda}} + b_2(\lambda) \, \e^{- x \sqrt{\lambda}} & \lambda \neq 0 \\
	\end{cases}
\end{align*}
for $a_1(\lambda) , a_2(\lambda) , b_1(\lambda) , b_2(\lambda) \in \C$ where we always choose the root whose imaginary part is positive. These solutions are smooth functions on $\R \times \R$.

\paragraph{The wave equation on $\R$} % (fold)
As mentioned in problem~4 on sheet~1, conditions on $u$ now restrict the admissible values for $\lambda$. For instance, if we ask that 
\begin{align*}
	u \in L^{\infty}(\R \times \R) 
	, 
\end{align*}
then this condition only allows $\lambda \leq 0$ and excludes the linear solutions, \ie $a_2(0) = 0 = b_2(0)$. 
% paragraph The wave equation on $\R$ (end)

\paragraph{The wave equation on $[0,L]$ with Dirichlet boundary conditions} % (fold)
Now assume we are interested in the case where the wave equation is considered on the interval $[0,L]$ with Dirichlet boundary conditions, $u(t,0) = 0 = u(t,L)$. It turns out that the boundary condition only allows for a discrete set of negative $\lambda$: one can see easily that for $\lambda = 0$, only the trivial solution satisfies the Dirichlet boundary conditions, \ie $a_1(0) = a_2(0) = 0$ and $b_1(0) = b_2(0) = 0$. 

For $\lambda \neq 0$, the first boundary condition 
\begin{align*}
	\xi_{\lambda}(0) &= a_1(\lambda) \, \e^{+ 0 \sqrt{\lambda}} + a_2(\lambda) \, \e^{- 0 \sqrt{\lambda}}
	\\
	&= a_1(\lambda) + a_2(\lambda) 
	\overset{!}{=} 0
\end{align*}
implies $a_2(\lambda) = - a_1(\lambda)$. Plugging that back into the second equation yields 
\begin{align*}
	\xi_{\lambda}(L) &= a_1(\lambda) \bigl (  \e^{+ L \sqrt{\lambda}} -  \e^{- L \sqrt{\lambda}} \bigr ) 
	\overset{!}{=} 0 
\end{align*}
which is equivalent to 
\begin{align*}
	\e^{2 L \sqrt{\lambda}} = 1 
	. 
\end{align*}
The solutions to this equation are of the form 
\begin{align*}
	\sqrt{\lambda} = \ii \frac{n \pi}{L}
\end{align*}
for some $n \in \N$, \ie $\lambda < 0$. Moreover, as discussed in problem~4, the only admissible solutions of the spatial equation are 
\begin{align*}
	\xi_n(x) &= \sin n \tfrac{\pi x}{L} 
	, 
	\qquad \qquad 
	n \in \N 
	. 
\end{align*}
That means, the solutions are indexed by an integer $n \in \N$, 
\begin{align*}
	u_n(t,x) = \Bigl ( a_1(n) \, \e^{+ \ii n \frac{\pi t}{L}} + a_2(n) \, \e^{- \ii n \frac{\pi t}{L}} \Bigr ) \, \sin n \tfrac{\pi x}{L}
	,  
\end{align*}
and a generic solution is of the form 
\begin{align}
	u(t,x) &= \sum_{n \in \N} \Bigl ( a_1(n) \, \e^{+ \ii n \frac{\pi t}{L}} + a_2(n) \, \e^{- \ii n \frac{\pi t}{L}} \Bigr ) \, \sin n \tfrac{\pi x}{L} 
	. 
	\label{spaces:eqn:solution_wave_equation_Fourier}
\end{align}
To obtain the coefficients, we need to solve 
\begin{align*}
	u(0,x) &= \sum_{n \in \N} \bigl ( a_1(n) + a_2(n) \bigr ) \, \sin n \tfrac{\pi x}{L} 
	\overset{!}{=} \varphi(x) 
	\\
	\partial_t u(0,x) &= \sum_{n \in \N} \frac{\ii \, n \, \pi}{L} \, \bigl ( a_1(n) - a_2(n) \bigr ) \, \sin n \tfrac{\pi x}{L} 
	\overset{!}{=} \psi(x) 
	. 
\end{align*}
Later on, we will see that all integrable functions on intervals can be expanded in terms of $\e^{\pm \ii n \lambda x}$ for suitable choices of $\lambda$ (\cf Section~\ref{Fourier}). Hence, we can expand 
\begin{align*}
	\varphi(x) &= \sum_{n \in \N} b_{\varphi}(n) \, \sin n \tfrac{\pi x}{L} 
	\\
	\psi(x) &= \sum_{n \in \N} b_{\psi}(n) \, \sin n \tfrac{\pi x}{L} 
\end{align*}
in the same fashion as $u$, and we obtain the following relation between the coefficients: 
\begin{align*}
	b_{\varphi}(n) &= a_1(n) + a_2(n) 
	\\
	b_{\psi}(n) &= \frac{\ii \, n \, \pi}{L} \, \bigl ( a_1(n) - a_2(n) \bigr )
\end{align*}
For each $n$, we obtain two linear equations with two unknowns, and we can solve them explicitly. Now the question is whether the sum in $u(t,\cdot)$ is also integrable: if we assume for now that the coefficients of $\varphi$ and $\psi$ are absolutely summable\footnote{In Chapter~\ref{Fourier} we will give conditions explicit conditions on $\varphi$ and $\psi$ which will ensure the absolute summability of the Fourier coefficients.},\marginpar{2013.10.10} 
\begin{align*}
	\sum_{n \in \N} \babs{a_1(n) + a_2(n)} &< \infty 
	\, , 
	\\
	\sum_{n \in \N} \sabs{n} \, \babs{a_1(n) - a_2(n)} &< \infty 
	\, , 
\end{align*}
then we deduce 
\begin{align*}
	2 \babs{a_1(n)} &= \babs{a_1(n) + a_1(n)} 
	\\
	&
	= \babs{a_1(n) - a_2(n) + a_2(n) + a_1(n)} 
	\\
	&\leq \babs{a_1(n) + a_2(n)} + \babs{a_1(n) - a_2(n)} 
	\\
	&
	\leq \babs{a_1(n) + a_2(n)} + \abs{n} \, \babs{a_1(n) - a_2(n)} 
	. 
\end{align*}
Thus, the expression on the right-hand side of \eqref{spaces:eqn:solution_wave_equation_Fourier} converges in $L^1([0,L])$, 
\begin{align*}
	\babs{u(x,t)} &\leq \biggl \lvert \sum_{n \in \N} \bigl ( a_1(n) \, \e^{+ \ii \frac{n \pi}{L} t} + a_2(n) \, \e^{- \ii \frac{n \pi}{L} t} \bigr ) \, \sin n \tfrac{\pi x}{L} \biggr \rvert
	\\
	&\leq \sum_{n \in \N} \bigl ( \sabs{a_1(n)} + \sabs{a_2(n)} \bigr ) 
	< \infty 
	. 
\end{align*}
Since the bound is independent of $t$ we deduce $u(t,\cdot)$ exists for all $t \in \R$. 

Overall, we have shown the following: if we place enough conditions on the initial values $u(0) = \varphi$ and $\partial_t u(0) = \psi$ (here: $\varphi, \psi \in L^1([0,L])$ and absolutely convergent Fourier series), then in fact $u(t) \in L^1([0,L])$ is integrable for all times (bounded functions on bounded intervals are integrable). 
% paragraph The wave equation on $[0,L]$ with Dirichlet boundary conditions (end)
% subsection Boundary value problems (end)
% section Banach spaces (end)

\section{Hilbert spaces} % (fold)
\label{spaces:Hilbert}
\emph{Hilbert spaces} $\Hil$ are Banach spaces with a \emph{scalar product} $\scpro{\cdot \,}{\cdot}$ which allows to measure the “angle” between two vectors. Most importantly, it yields a characterization of vectors which are orthogonal to one another, giving rise to the notion of \emph{orthonormal bases} (ONB). This type of basis is particularly efficient to work with and has some rather nice properties (\eg a Pythagorean theorem holds).

\subsection{Abstract Hilbert spaces} % (fold)
\label{spaces:Hilbert:generic}
First, let us define a Hilbert space in the abstract, starting with 
\begin{definition}[pre-Hilbert space and Hilbert space]
	A \emph{pre-Hilbert space} is a complex vector space $\Hil$ with scalar product 
	\begin{align*}
		\scpro{\cdot \, }{\cdot} : \Hil \times \Hil \longrightarrow \C 
		, 
	\end{align*}
	\ie a mapping with properties 
	\begin{enumerate}[(i)]
		\item $\scpro{\varphi}{\varphi} \geq 0$ and $\scpro{\varphi}{\varphi} = 0$ implies $\varphi = 0$ (positive definiteness), 
		\item $\scpro{\varphi}{\psi}^* = \scpro{\psi}{\varphi}$, and 
		\item $\scpro{\varphi}{\alpha \psi + \chi} = \alpha \scpro{\varphi}{\psi} + \scpro{\varphi}{\chi}$ 
	\end{enumerate}
	for all $\varphi , \psi , \chi \in \Hil$ and $\alpha \in \C$. This induces a natural norm $\norm{\varphi} := \sqrt{\sscpro{\varphi}{\varphi}}$ and metric $d(\varphi,\psi) := \norm{\varphi - \psi}$, $\varphi , \psi \in \Hil$. If $\Hil$ is complete with respect to the induced metric, it is a \emph{Hilbert space.} 
\end{definition}
The scalar product induces a norm 
\begin{align*}
	\norm{f} := \sqrt{\scpro{f}{f}} 
\end{align*}
and a metric $d(f,g) := \norm{f - g}$. 
\begin{example}
	\begin{enumerate}[(i)]
		\item $\C^n$ with scalar product
		\begin{align*}
			\scpro{z}{w} := \sum_{l = 1}^n z_j^* \, w_j 
		\end{align*}
		is a Hilbert space. 
		\item $\Cont([a,b],\C)$ with scalar product 
		\begin{align*}
			\scpro{f}{g} := \int_a^b \dd x \, f(x)^* \, g(x) 
		\end{align*}
		is just a pre-Hilbert space, since it is not complete. 
	\end{enumerate}
\end{example}
%
% subsection Generic Hilbert spaces (end)

\subsection{Orthonormal bases and orthogonal subspaces} % (fold)
\label{spaces:Hilbert:orthonormal}
Hilbert spaces have the important notion of orthonormal vectors and sequences which do not exist in Banach spaces. 
\begin{definition}[Orthonormal set]
	Let $\mathcal{I}$ be a countable index set. A family of vectors $\{ \varphi_k \}_{k \in \mathcal{I}}$ is called orthonormal set if for all $k,j \in \mathcal{I}$
	\begin{align*}
		\scpro{\varphi_k}{\varphi_j} = \delta_{kj} 
	\end{align*}
	holds. 
\end{definition}
As we will see, all vectors in a separable Hilbert spaces can be written in terms of a countable orthonormal basis. Especially when we want to approximate elements in a Hilbert space by elements in a proper closed subspace, the vector of best approximation can be written as a linear combination of basis vectors. 
\begin{definition}[Orthonormal basis]\label{spaces:Hilbert:defn:ONB}
	Let $\mathcal{I}$ be a countable index set. An orthonormal set of vectors $\{ \varphi_k \}_{k \in \mathcal{I}}$ is called orthonormal basis if and only if for all $\psi \in \Hil$, we have 
	\begin{align*}
		\psi = \sum_{k \in \mathcal{I}} \sscpro{\varphi_k}{\psi} \, \varphi_k 
		. 
	\end{align*}
	If $\mathcal{I}$ is countably infinite, $\mathcal{I} \cong \N$, then this means the sequence $\psi_n := \sum_{j = 1}^n \sscpro{\varphi_j}{\psi} \, \varphi_j$ of partial converges in norm to $\psi$, 
	\begin{align*}
		\lim_{n \rightarrow \infty} \Bnorm{\psi - \mbox{$\sum_{j = 1}^n$} \sscpro{\varphi_j}{\psi} \, \varphi_j} = 0
	\end{align*}
\end{definition}
\begin{example}
	Hermitian matrices $H = H^* \in \mathrm{Mat}_{\C}(n)$ give rise to a set of orthonormal vectors, namely the set of eigenvectors. To see that this is so, let $v_k$ be an eigenvector to $\lambda_n$ (the eigenvalues are repeated according to their multiplicity). Then for all eigenvectors $v_j$ and $v_k$, we compute 
	\begin{align*}
		0 = \scpro{v_j}{H v_k} - \scpro{v_j}{H v_k} &= \lambda_n \, \scpro{v_j}{v_k} - \scpro{H^* v_j}{v_k} 
		\\
		&= \lambda_n \, \scpro{v_j}{v_k} - \scpro{H v_j}{v_k} 
		= (\lambda_n - \lambda_j) \, \scpro{v_j}{v_k} 
	\end{align*}
	where we have used that the eigenvalues of hermitian matrices are real (this follows from repeating the above argument for $j = n$). Hence, either $\scpro{v_j}{v_k} = 0$ if $\lambda_j \neq \lambda_n$ or $\lambda_j = \lambda_n$. In the latter case, we obtain a higher-dimensional subspace of $\C^n$ associated to the eigenvalue $\lambda_n$ for which we can construct an orthonormal basis using the Gram-Schmidt procedure. 
	
	Thus, we obtain a basis of eigenvectors $\{ v_j \}_{j = 1}^n$. This basis is particularly convenient when working with the matrix $H$ since
	\begin{align*}
		H w &= H \, \sum_{j = 1}^n \scpro{v_j}{w} \, v_j 
		\\
		&= \sum_{j = 1}^n \lambda_j \, \scpro{v_j}{w} \, v_j 
		. 
	\end{align*}
	We will extend these arguments in the next chapter to operators on infinite-dimensional Hilbert spaces where one needs to take a little more care. 
\end{example}
With this general notion of orthogonality, we have a Pythagorean theorem: 
\begin{theorem}[Pythagoras]
	Given a finite orthonormal family $\{ \varphi_1 , \ldots , \varphi_n \}$ in a pre-Hilbert space $\Hil$ and $\varphi \in \Hil$, we have 
	\begin{align*}
		\bnorm{\varphi}^2 = \mbox{$\sum_{k = 1}^n$} \babs{\sscpro{\varphi_k}{\varphi}}^2 + \bnorm{\varphi - \mbox{$\sum_{k = 1}^n$} \sscpro{\varphi_k}{\varphi} \, \varphi_k}^2 
		. 
	\end{align*}
\end{theorem}
\begin{proof}
	It is easy to check that $\psi := \sum_{k = 1}^n \sscpro{\varphi_k}{\varphi} \, \varphi_k$ and $\psi^{\perp} := \varphi - \sum_{k = 1}^n \sscpro{\varphi_k}{\varphi} \, \varphi_k$ are orthogonal and $\varphi = \psi + \psi^{\perp}$. Hence, we obtain 
	\begin{align*}
		\norm{\varphi}^2 &= \sscpro{\varphi}{\varphi} 
		= \sscpro{\psi+\psi^{\perp}}{\psi+\psi^{\perp}} 
		= \sscpro{\psi}{\psi} + \sscpro{\psi^{\perp}}{\psi^{\perp}} 
		\\
		&= \bnorm{\mbox{$\sum_{k = 1}^n$} \sscpro{\varphi_k}{\varphi} \, \varphi_k}^2 + \bnorm{\varphi - \mbox{$\sum_{k = 1}^n$} \sscpro{\varphi_k}{\varphi} \, \varphi_k}^2 
		. 
	\end{align*}
	This concludes the proof. 
\end{proof}
A simple corollary are Bessel's inequality and the Cauchy-Schwarz inequality. 
\begin{theorem}
	Let $\Hil$ be a pre-Hilbert space. 
	\begin{enumerate}[(i)]
		\item Bessel's inequality holds: let $\bigl \{ \varphi_1 , \ldots \varphi_n \bigr \}$ be a finite orthonormal sequence. Then 
		\begin{align*}
			\snorm{\psi}^2 \geq \sum_{j = 1}^n \sabs{\sscpro{\varphi_j}{\psi}}^2 
			. 
		\end{align*}
		holds for all $\psi \in \Hil$. 
		\item The Cauchy-Schwarz inequality holds, \ie 
		\begin{align*}
			\sabs{\sscpro{\varphi}{\psi}} \leq \snorm{\varphi} \snorm{\psi} 
		\end{align*}
		is valid for all $\varphi , \psi \in \Hil$
	\end{enumerate}
\end{theorem}
\begin{proof}
	\begin{enumerate}[(i)]
		\item This follows trivially from the previous Theorem as $\snorm{\psi^{\perp}}^2 \geq 0$. 
		\item Pick $\varphi , \psi \in \Hil$. In case $\varphi = 0$, the inequality holds. So assume $\varphi \neq 0$ and define 
		\begin{align*}
			\varphi_1 := \frac{\varphi}{\norm{\varphi}} 
		\end{align*}
		which has norm $1$. We can apply (i) for $n = 1$ to conclude 
		\begin{align*}
			\norm{\psi}^2 \geq \abs{\sscpro{\varphi_1}{\psi}}^2 
			= \frac{1}{\norm{\varphi}^2} \abs{\sscpro{\varphi}{\psi}}^2 
			. 
		\end{align*}
		This is equivalent to the Cauchy-Schwarz inequality. 
	\end{enumerate}
\end{proof}
An important corollary says that the scalar product is continuous with respect to the norm topology. This is not at all surprising, after all the norm is induced by the scalar product! 
\begin{corollary}\label{hilbert_spaces:onb:cor:continuity_scalar_product}
	Let $\Hil$ be a Hilbert space. Then the scalar product is continuous with respect to the norm topology, \ie for two sequences $(\varphi_n)_{n \in \N}$ and $(\psi_m)_{m \in \N}$ that converge to $\varphi$ and $\psi$, respectively, we have 
	\begin{align*}
		\lim_{n,m \rightarrow \infty} \sscpro{\varphi_n}{\psi_m} &= \sscpro{\varphi}{\psi} 
		. 
	\end{align*}
\end{corollary}
\begin{proof}
	Let $(\varphi_n)_{n \in \N}$ and $(\psi_m)_{m \in \N}$ be two sequences in $\Hil$ that converge to $\varphi$ and $\psi$, respectively. Then by Cauchy-Schwarz, we have 
	\begin{align*}
		\lim_{n,m \rightarrow \infty} \babs{\sscpro{\varphi}{\psi} - \sscpro{\varphi_n}{\psi_m}} &= \lim_{n,m \rightarrow \infty} \babs{\sscpro{\varphi - \varphi_n}{\psi} - \sscpro{\varphi_n}{\psi_m - \psi}} 
		\\
		&\leq \lim_{n,m \rightarrow \infty} \babs{\sscpro{\varphi - \varphi_n}{\psi}} + \lim_{n,m \rightarrow \infty} \babs{\sscpro{\varphi_n}{\psi_m - \psi}} 
		\\
		&\leq \lim_{n,m \rightarrow \infty} \snorm{\varphi - \varphi_n} \, \snorm{\psi} + \lim_{n,m \rightarrow \infty} \snorm{\varphi_n} \, \snorm{\psi_m - \psi} 
		= 0 
	\end{align*}
	since there exists some $C > 0$ such that $\norm{\varphi_n} \leq C$ for all $n \in \N$. 
\end{proof}
Before we prove that a Hilbert space is separable exactly if it admits a \emph{countable} basis, we need to introduce the notion of orthogonal complement: if $A$ is a subset of a pre-Hilbert space $\Hil$, then we define 
\begin{align*}
	A^{\perp} := \bigl \{ \varphi \in \Hil \; \vert \; \sscpro{\varphi}{\psi} = 0 \; \forall \psi \in A \bigr \} 
	. 
\end{align*}
The following few properties of the orthogonal complement follow immediately from its definition: 
\begin{enumerate}[(i)]
	\item $\{ 0 \}^{\perp} = \Hil$ and $\Hil^{\perp} = \{ 0 \}$. 
	\item $A^{\perp}$ is a closed linear subspace of $\Hil$ for \emph{any} subset $A \subseteq \Hil$. 
	\item If $A \subseteq B$, then $B^{\perp} \subseteq A^{\perp}$. 
	\item If we denote the sub vector space spanned by the elements in $A$ by $\mathrm{span} \, A$, we have 
	\begin{align*}
		A^{\perp} = \bigl ( \mathrm{span} \, A \bigr )^{\perp} = \bigl ( \overline{\mathrm{span} \, A} \bigr )^{\perp}
	\end{align*}
	where $\overline{\mathrm{span} \, A}$ is the completion of $\mathrm{span} \, A$ with respect to the norm topology. \marginpar{2013.10.15}
\end{enumerate}
%
% subsection Orthonormal bases and orthogonal subspaces (end)

\subsection{Prototypical Hilbert spaces} % (fold)
\label{spaces:Hilbert:prototypical}
We have already introduced the \emph{Banach} space of square-integrable functions on $\R^n$, 
\begin{align*}
	\mathcal{L}^2(\Omega) := \Bigl \{ \varphi : \Omega \longrightarrow \C \; \big \vert \; \varphi \mbox{ measurable, } \int_{\Omega} \dd x \, \abs{\varphi(x)}^2 < \infty \Bigr \} 
	, 
\end{align*}
and this space can be naturally equipped with a scalar product, 
\begin{align}
	\scpro{f}{g} &= \int_{\Omega} \dd x \, f(x)^* \, g(x) 
	. 
	\label{spaces:eqn:L2_scalar_product}
\end{align}
When talking about wave functions in quantum mechanics, the Born rule states that $\abs{\psi(x)}^2$ is to be interpreted as a probability density on $\Omega$ for position (\ie $\psi$ is normalized, $\norm{\psi} = 1$). Hence, we are interested in solutions to the Schrödinger equation which are also square integrable. If $\psi_1 \sim \psi_2$ are two normalized functions in $\mathcal{L}^2(\Omega)$, then we get the same probabilities for both: if $\Lambda \subseteq \Omega \subseteq \R^n$ is a measurable set, then 
\begin{align*}
	\mathbb{P}_1(X \in \Lambda) = \int_{\Lambda} \dd x \, \abs{\psi_1(x)}^2 = \int_{\Lambda} \dd x \, \abs{\psi_2(x)}^2 = \mathbb{P}_2(X \in \Lambda) 
	. 
\end{align*}
This is proven via the triangle inequality and the Cauchy-Schwartz inequality: 
\begin{align*}
	0 &\leq \babs{\mathbb{P}_1(X \in \Lambda) - \mathbb{P}_2(X \in \Lambda)} 
	= \abs{\int_{\Lambda} \dd x \, \abs{\psi_1(x)}^2 - \int_{\Lambda} \dd x \, \abs{\psi_2(x)}^2} 
	\\
	&
	= \abs{\int_{\Lambda} \dd x \, \bigl ( \psi_1(x) - \psi_2(x) \bigr )^* \, \psi_1(x) + \int_{\Lambda} \dd x \, \psi_2(x)^* \, \bigl ( \psi_1(x) - \psi_2(x) \bigr )} 
	\\
	&\leq \int_{\Lambda} \dd x \, \babs{\psi_1(x) - \psi_2(x)} \, \babs{\psi_1(x)} + \int_{\Lambda} \dd x \, \babs{\psi_2(x)} \, \babs{\psi_1(x) - \psi_2(x)} 
	\\
	&\leq \norm{\psi_1 - \psi_2} \, \norm{\psi_1} + \norm{\psi_2} \, \norm{\psi_1 - \psi_2} 
	= 0
\end{align*}
Very often, another space is used in applications (\eg in tight-binding models): 
\begin{definition}[$\ell^2(S)$]\label{spaces:Hilbert:defn:ell2}
	Let $S$ be a countable set. Then 
	\begin{align*}
		\ell^2(S) := \Bigl \{ c : S \longrightarrow \C \; \big \vert \; \mbox{$\sum_{j \in S} c_j^* c_j < \infty$} \Bigr \} 
	\end{align*}
	is the Hilbert space of square-summable sequences with scalar product $\scpro{c}{c'} := \sum_{j \in S} c_j^* c'_j$. 
\end{definition}
On $\ell^2(S)$ the scalar product induces the norm $\norm{c} := \sqrt{\scpro{c}{c}}$. With respect to this norm, $\ell^2(S)$ is complete. 
% subsection Prototypical Hilbert spaces (end)

\subsection{Best approximation} % (fold)
\label{spaces:Hilbert:best_approx}
If $(\Hil,d)$ is a metric space, we can define the distance between a point $\varphi \in \Hil$ and a subset $A \subseteq \Hil$ as 
\begin{align*}
	d(\varphi,A) := \inf_{\psi \in A} d(\varphi,\psi) 
	. 
\end{align*}
If there exists $\varphi_0 \in A$ which minimizes the distance, \ie $d(\varphi,A) = d(\varphi,\varphi_0)$, then $\varphi_0$ is called \emph{element of best approximation} for $\varphi$ in $A$. This notion is helpful to understand why and how elements in an infinite-dimensional Hilbert space can be approximated by finite linear combinations -- something that is used in numerics all the time.  

If $A \subset \Hil$ is a convex subset of a Hilbert space $\Hil$, then one can show that there always exists an element of best approximation. In case $A$ is a linear subspace of $\Hil$, it is given by projecting an arbitrary $\psi \in \Hil$ down to the subspace $A$. 
\begin{theorem}\label{hilbert_spaces:onb:thm:best_approx}
	Let $A$ be a closed convex subset of a Hilbert space $\Hil$. Then there exists for each $\varphi \in \Hil$ exactly one $\varphi_0 \in A$ such that 
	\begin{align*}
		d(\varphi,A) = d(\varphi,\varphi_0) 
		. 
	\end{align*}
\end{theorem}
\begin{proof}
	We choose a sequence $(\psi_n)_{n \in \N}$ in $A$ with $d(\varphi,\psi_n) = \norm{\varphi - \psi_n} \rightarrow d(x,A)$. This sequence is also a Cauchy sequence: we add and subtract $\varphi$ to get 
	\begin{align*}
		\bnorm{\psi_n - \psi_m}^2 &= \bnorm{(\psi_n - \varphi) + (\varphi - \psi_m)}^2 
		. 
	\end{align*}
	If $\Hil$ were a normed space, we could have to use the triangle inequality to estimate the right-hand side from above. However, $\Hil$ is a Hilbert space and by using the parallelogram identity,\footnote{For all $\varphi , \psi \in \Hil$, the identity $2 \norm{\varphi}^2 + 2 \norm{\psi}^2 = \norm{\varphi + \psi}^2 + \norm{\varphi - \psi}^2$ holds. } we see that the right-hand side is actually \emph{equal} to 
	\begin{align*}
		\bnorm{\psi_n - \psi_m}^2 &= 2 \bnorm{\psi_n - \varphi}^2 + 2 \bnorm{\psi_m - \varphi}^2 - \bnorm{\psi_n + \psi_m - 2 \varphi}^2 
		\\
		&
		= 2 \bnorm{\psi_n - \varphi}^2 + 2 \bnorm{\psi_m - \varphi}^2 - 4 \bnorm{\tfrac{1}{2} (\psi_n + \psi_m) - \varphi}^2 
		\\
		&\leq 2 \bnorm{\psi_n - \varphi}^2 + 2 \bnorm{\psi_m - \varphi}^2 - 4 d(\varphi,A)
		\\
		&
		\xrightarrow{n \rightarrow \infty} 2 d(\varphi,A) + 2 d(\varphi,A) - 4 d(\varphi,A) 
		= 0
		. 
	\end{align*}
	By convexity, $\frac{1}{2} (\psi_n + \psi_m)$ is again an element of $A$. This is crucial once again for the uniqueness argument. Letting $n,m \rightarrow \infty$, we see that $(\psi_n)_{n \in \N}$ is a Cauchy sequence in $A$ which converges in $A$ as it is a closed subset of $\Hil$. Let us call the limit point $\varphi_0 := \lim_{n \rightarrow \infty} \psi_n$. Then $\varphi_0$ is \emph{an} element of best approximation, 
	\begin{align*}
		\bnorm{\varphi - \varphi_0} = \lim_{n \rightarrow \infty} \bnorm{\varphi - \psi_n} = d(\varphi,A)
		. 
	\end{align*}
	To show uniqueness, we assume that there exists another element of best approximation $\varphi_0' \in A$. Define the sequence $(\tilde{\psi}_n)_{n \in \N}$ by $\tilde{\psi}_{2n} := \varphi_0$ for even indices and $\tilde{\psi}_{2n + 1} := \varphi_0'$ for odd indices. By assumption, we have $\norm{\varphi - \varphi_0} = d(\varphi,A) = \norm{\varphi - \varphi_0'}$ and thus, by repeating the steps above, we conclude $(\tilde{\psi}_n)_{n \in \N}$ is a Cauchy sequence that converges to some element. However, since the sequence is alternating, the two elements $\varphi'_0 = \varphi_0$ are in fact identical. 
\end{proof}
As we have seen, the condition that the set is convex and closed is crucial in the proof. Otherwise the minimizer may not be unique or even contained in the set. 

This is all very abstract. For the case of a closed subvector space $E \subseteq \Hil$, we can express the element of best approximation in terms of the basis: not surprisingly, it is given by the projection of $\varphi$ onto $E$. 
\begin{theorem}\label{hilbert_spaces:onb:thm:best_approx_explicit}
	Let $E \subseteq \Hil$ be a closed subspace of a Hilbert space that is spanned by countably many orthonormal basis vectors $\{ \varphi_k \}_{k \in \mathcal{I}}$. Then for any $\varphi \in \Hil$, the element of best approximation $\varphi_0 \in E$ is given by 
	\begin{align*}
		\varphi_0 = \sum_{k \in \mathcal{I}} \sscpro{\varphi_k}{\varphi} \, \varphi_k 
		. 
	\end{align*}
\end{theorem}
\begin{proof}
	It is easy to show that $\varphi - \varphi_0$ is orthogonal to any $\psi = \sum_{k \in \mathcal{I}} \lambda_k \, \varphi_k \in E$: we focus on the more difficult case when $E$ is not finite-dimensional. Then, we have to approximate $\varphi_0$ and $\psi$ by finite linear combinations and take limits. We call $\varphi_0^{(n)} := \sum_{k = 1}^n \sscpro{\varphi_k}{\varphi} \, \varphi_k$ and $\psi^{(m)} := \sum_{l = 1}^m \lambda_l \, \varphi_l$. With that, we have 
	\begin{align*}
		\bscpro{\varphi - \varphi_0^{(n)}}{\psi^{(m)}} &= \Bscpro{\varphi - \mbox{$\sum_{k = 1}^n$} \sscpro{\varphi_k}{\varphi} \, \varphi_k}{\mbox{$\sum_{l = 1}^m$} \lambda_l \, \varphi_l} \
		\\
		&
		= \sum_{l = 1}^m \lambda_l \, \sscpro{\varphi}{\varphi_l} - \sum_{k = 1}^n \sum_{l = 1}^m \lambda_l \, \sscpro{\varphi_k}{\varphi}^* \, \sscpro{\varphi_k}{\varphi_l} 
		% \\
		% &= \sum_{l = 1}^m \lambda_l \, \sscpro{\varphi}{\varphi_l} - \sum_{k = 1}^n \sum_{l = 1}^m \lambda_l \, \sscpro{\varphi}{\varphi_k} \, \delta_{kl} 
		\\
		&= \sum_{l = 1}^m \lambda_l \, \sscpro{\varphi}{\varphi_l} \, \Bigl ( 1 - \mbox{$\sum_{k = 1}^n$} \, \delta_{kl} \Bigr ) 
		. 
	\end{align*}
	By continuity of the scalar product, Corollary~\ref{hilbert_spaces:onb:cor:continuity_scalar_product}, we can take the limit $n,m \rightarrow \infty$. The term in parentheses containing the sum is $0$ exactly when $l \in \{ 1 , \ldots , m \}$ and $1$ otherwise. Specifically, if $n \geq m$, the right-hand side vanishes identically. Hence, we have 
	\begin{align*}
		\bscpro{\varphi - \varphi_0}{\psi} &= \lim_{n,m \rightarrow \infty} \bscpro{\varphi - \varphi_0^{(n)}}{\psi^{(m)}} 
		= 0 
		, 
	\end{align*}
	in other words $\varphi - \varphi_0 \in E^{\perp}$. This, in turn, implies by the Pythagorean theorem that 
		\begin{align*}
			\norm{\varphi - \psi}^2 &= \norm{\varphi - \varphi_0}^2 + \norm{\varphi_0 - \psi}^2 
			\geq \norm{\varphi - \varphi_0}^2 
		\end{align*}
		and hence $\norm{\varphi - \varphi_0} = d(\varphi,E)$. Put another way, $\varphi_0$ is \emph{an} element of best approximation. Let us now show uniqueness. Assume, there exists another element of best approximation $\varphi'_0 = \sum_{k \in \mathcal{I}} \lambda'_k \, \varphi_k$. Then we know by repeating the previous calculation backwards that $\varphi - \varphi'_0 \in E^{\perp}$ and the scalar product with respect to any of the basis vectors $\varphi_k$ which span $E$ has to vanish, 
		\begin{align*}
			0 = \bscpro{\varphi_k}{\varphi - \varphi'_0} &= \sscpro{\varphi_k}{\varphi} - \sum_{l \in \mathcal{I}} {\lambda'_l} \, \sscpro{\varphi_k}{\varphi_l} 
			= \sscpro{\varphi_k}{\varphi} - \sum_{l \in \mathcal{I}} {\lambda'_l} \, \delta_{kl} 
			\\
			&= \sscpro{\varphi_k}{\varphi} - {\lambda'_k} 
			. 
		\end{align*}
		This means the coefficients with respect to the basis $\{ \varphi_k \}_{k \in \mathcal{I}}$ all agree with those of $\varphi_0$. Hence, the element of approximation is unique, $\varphi_0 = \varphi_0'$, and given by the projection of $\varphi$ onto $E$. 
\end{proof}
\begin{theorem}\label{hilbert_spaces:onb:thm:direct_sum_decomp}
	Let $E$ be a closed linear subspace of a Hilbert space $\Hil$. Then 
	\begin{enumerate}[(i)]
		\item $\Hil = E \oplus E^{\perp}$, \ie every vector $\varphi \in \Hil$ can be uniquely decomposed as $\varphi = \psi + \psi^{\perp}$ with $\psi \in E$, $\psi^{\perp} \in E^{\perp}$. 
		\item $E^{\perp \perp} = E$. 
	\end{enumerate}
\end{theorem}
\begin{proof}
	\begin{enumerate}[(i)]
		\item By Theorem~\ref{hilbert_spaces:onb:thm:best_approx}, for each $\varphi \in \Hil$, there exists $\varphi_0 \in E$ such that $d(\varphi,E) = d(\varphi,\varphi_0)$. From the proof of the previous theorem, we see that $\varphi_0^{\perp} := \varphi - \varphi_0 \in E^{\perp}$. Hence, $\varphi = \varphi_0 + \varphi_0^{\perp}$ is \emph{a} decomposition of $\varphi$. To show that it is unique, assume $\varphi'_0 + {\varphi'_0}^{\perp} = \varphi = \varphi_0 + \varphi_0^{\perp}$ is another decomposition. Then by subtracting, we are led to conclude that 
		\begin{align*}
			E \ni \varphi'_0 - \varphi_0 = \varphi_0^{\perp} - {\varphi'_0}^{\perp} \in E^{\perp} 
		\end{align*}
		holds. On the other hand, $E \cap E^{\perp} = \{ 0 \}$ and thus $\varphi_0 = {\varphi'_0}$ and $\varphi_0^{\perp} = {\varphi'_0}^{\perp}$, the decomposition is unique. 
		\item It is easy to see that $E \subseteq E^{\perp \perp}$. Let $\tilde{\varphi} \in E^{\perp \perp}$. By the same arguments as above, we can decompose $\tilde{\varphi} \in E^{\perp \perp} \subseteq \Hil$ into 
		\begin{align*}
			\tilde{\varphi} = \tilde{\varphi}_0 + {\tilde{\varphi}}_0^{\perp} 
		\end{align*}
		with $\tilde{\varphi}_0 \in E \subseteq E^{\perp \perp}$ and ${\tilde{\varphi}}_0^{\perp} \in E^{\perp}$. Hence, ${\tilde{\varphi}} - \tilde{\varphi}_0 \in E^{\perp \perp} \cap E^{\perp} = (E^{\perp})^{\perp} \cap E^{\perp} = \{ 0 \}$ and thus $\tilde{\varphi} = \tilde{\varphi}_0 \in E$. 
	\end{enumerate}
\end{proof}
Now we are in a position to prove the following important Proposition: 
\begin{proposition}
	A Hilbert space $\Hil$ is separable if and only if there exists a countable orthonormal basis. 
\end{proposition}
\begin{proof}
	$\Leftarrow$: 
	The set generated by the orthonormal basis $\{ \varphi_j \}_{j \in \mathcal{I}}$, $\mathcal{I}$ countable, and coefficients $z = q + i p$, $q , p \in \Q$, is dense in $\Hil$, 
	\begin{align*}
		\Bigl \{ \mbox{$\sum_{j = 1}^n$} z_j \varphi_j \in \Hil \; \big \vert \; \N \ni n \leq \abs{\mathcal{I}}, \; \varphi_j \in \{ \varphi_k \}_{k \in \N} , \; z_j = q_j + i p_j , \; q_j , p_j \in \Q \Bigr \} 
		. 
	\end{align*}
	\medskip
	
	\noindent
	$\Rightarrow$: Assume there exists a countable dense subset $\mathcal{D}$, \ie $\overline{\mathcal{D}} = \Hil$. If $\Hil$ is finite dimensional, the induction terminates after finitely many steps and the proof is simpler. Hence, we will assume $\Hil$ to be infinite dimensional. Pick a vector $\tilde{\varphi}_1 \in \mathcal{D} \setminus \{ 0 \}$ and normalize it. The normalized vector is then called $\varphi_1$. Note that $\varphi_1$ need not be in $\mathcal{D}$. By Theorem~\ref{hilbert_spaces:onb:thm:direct_sum_decomp}, we can split any $\psi \in \mathcal{D}$ into $\psi_1$ and $\psi_1^{\perp}$ such that $\psi_1 \in \mathrm{span} \, \{ \varphi_1 \} := E_1$, $\psi_1^{\perp} \in \mathrm{span} \, \{ \varphi_1 \}^{\perp} := E_1^{\perp}$ and 
	\begin{align*}
		\psi = \psi_1 + \psi_1^{\perp} 
		. 
	\end{align*}
	pick a second $\tilde{\varphi}_2 \in \mathcal{D} \setminus E_1$ (which is non-empty). Now we apply Theorem~\ref{hilbert_spaces:onb:thm:best_approx_explicit} (which is in essence Gram-Schmidt orthonormalization) to $\tilde{\varphi}_2$, \ie we pick the part which is orthogonal to $\varphi_1$, 
	\begin{align*}
		\varphi'_2 := \tilde{\varphi}_2 - \sscpro{\varphi_1}{\tilde{\varphi}_2} \, \varphi_2
	\end{align*}
	and normalize to $\varphi_2$, 
	\begin{align*}
		\varphi_2 := \frac{\varphi'_2}{\snorm{\varphi'_2}} 
		. 
	\end{align*}
	This defines $E_2 := \mathrm{span} \, \{ \varphi_1 , \varphi_2 \}$ and $\Hil = E_2 \oplus E_2^{\perp}$. 
	
	Now we proceed by induction: assume we are given $E_n = \mathrm{span} \, \{ \varphi_1 , \ldots , \varphi_n \}$. Take $\tilde{\varphi}_{n+1} \in  \mathcal{D} \setminus E_n$ and apply Gram-Schmidt once again to yield $\varphi_{n+1}$ which is the obtained from normalizing the vector 
	\begin{align*}
		\varphi'_{n+1} := \tilde{\varphi}_{n+1} - \sum_{k = 1}^n \sscpro{\varphi_k}{\tilde{\varphi}_{n+1}} \, \varphi_k
		. 
	\end{align*}
	This induction yields an orthonormal sequence $\{ \varphi_n \}_{n \in \N}$ which is by definition an orthonormal basis of $E_{\infty} := \overline{\mathrm{span} \, \{ \varphi_n \}_{n \in \N}}$ a closed subspace of $\Hil$. If $E_{\infty} \subsetneq \Hil$, we can split the Hilbert space into $\Hil = E_{\infty} \oplus E_{\infty}^{\perp}$. Then either $\mathcal{D} \cap (\Hil \setminus E_{\infty}) = \emptyset$ -- in which case $\mathcal{D}$ cannot be dense in $\Hil$ -- or $\mathcal{D} \cap (\Hil \setminus E_{\infty}) \neq \emptyset$. But then we have terminated the induction prematurely. 
\end{proof}
%
% subsection Orthonormal bases and orthogonal subspaces (end)

\subsection{Best approximation on $L^2([-\pi,+\pi])$ using $\{ \e^{+ \ii n x} \}_{n \in \Z}$} % (fold)
Let us consider the free Schrödinger equation 
\begin{align}
	\ii \, \partial_t \psi(t) &= - \partial_x^2 \psi(t) 
	\, , 
	&&
	\psi(0) = \psi_0 \in L^2([-\pi,+\pi]) 
	\, , 
	\label{spaces:eqn:free_Schroedinger_interval}
\end{align}
for a suitable initial condition $\psi_0$ (we will be more specific in a moment). Since $[-\pi,+\pi]$ has a boundary, we need to impose boundary conditions. We pick periodic boundary conditions, \ie we consider functions for which $\psi(-\pi) = \psi(+\pi)$. It is often convenient to think of functions with periodic boundary conditions as $2\pi$-periodic functions on $\R$, \ie $\psi(x + 2\pi) = \psi(x)$. 

Periodic functions are best analyzed using their \emph{Fourier series} (which we will discuss in detail in Chapter~\ref{Fourier}), \ie their expansion in terms of the exponentials $\{ \e^{+ \ii n x} \}_{n \in \Z}$. To simplify computations, let us choose a convenient prefactor for the scalar product, 
\begin{align*}
	\scpro{\varphi}{\psi} := \frac{1}{2\pi} \int_{-\pi}^{+\pi} \dd x \, \varphi(x)^* \, \psi(x) 
	. 
\end{align*}
Then a quick computation yields that $\{ \e^{+ \ii n x} \}_{n \in \Z}$ is indeed an orthonormal set: 
\begin{align*}
	\scpro{\e^{+ \ii j x}}{\e^{+ \ii n x}} &= \frac{1}{2\pi} \int_{-\pi}^{+\pi} \dd x \, \e^{+ \ii (n - j) x} 
	= 
	\begin{cases}
		1 & j = n \\
		0 & j \neq n \\
	\end{cases}
\end{align*}
As we will see later, this set of orthonormal vectors is also \emph{basis}, and we can expand \emph{any} vector $\psi$ in terms of its Fourier components, 
\begin{align*}
	\psi(x) &= \sum_{n \in \Z} \scpro{\e^{+ \ii n x}}{\psi} \, \e^{+ \ii n x} 
	. 
\end{align*}
Using the \emph{product ansatz} $\psi(t,x) = \tau(t) \, \varphi(x)$ (separation of variables) from Chapter~\ref{spaces:boundary_value_problems}, we obtain two coupled equations: 
\begin{align*}
	\ii \, \dot{\tau}(t) \, \varphi(x) &= - \tau(t) \varphi''(x) 
	\; \; \Longleftrightarrow \; \; 
	\ii \, \frac{\dot{\tau}(t)}{\tau(t)} = - \frac{\varphi''(x)}{\varphi(x)} = \lambda \in \C
\end{align*}
The periodic boundary conditions $\varphi(-\pi) = \varphi(+\pi)$ as well as the condition $\varphi \in L^2([-\pi,+\pi])$ eliminates many choices of $\lambda \in \C$. The equation for $\varphi$ is just a harmonic oscillator equation, and the periodicity requirement means that only $\lambda = n^2$ with $n \in \N_0 = \{ 0 , 1 , \ldots \}$ are admissible. The equation for $\tau$ can be obtained by elementary integration, 
\begin{align*}
	\ii \, \dot{\tau}_n(t) &= n^2 \tau_n(t) 
	\; \; \Rightarrow \; \; \tau_n(t) = \tau_n(0) \, \e^{- \ii n^2 t} 
	. 
\end{align*}
Hence, the solution to $\lambda = n^2$ is a scalar multiple of $\e^{- \ii n^2 t} \, \e^{+ \ii n x}$, and the solution to \eqref{spaces:eqn:free_Schroedinger_interval} can formally be written down as the sum 
\begin{align*}
	\psi(t,x) &= \sum_{n \in \Z} \e^{- \ii n^2 t} \, \scpro{\e^{+ \ii n x}}{\psi_0} \, \e^{+ \ii n x} 
	. 
\end{align*}
\emph{A priori} it is not clear in what sense this sum converges. It turns out the correct notion of convergence is to require the finiteness of 
\begin{align*}
	\norm{\psi(t)}^2 &= \sum_{n \in \Z} \babs{\e^{- \ii n^2 t} \, \scpro{\e^{+ \ii n x}}{\psi_0}}^2 
	\\
	&= \sum_{n \in \Z} \babs{\scpro{\e^{+ \ii n x}}{\psi_0}}^2 
	= \norm{\psi_0}^2 < \infty
\end{align*}
This condition, however, is automatically satisfied since we have assumed $\psi_0$ is an element of $L^2([-\pi,+\pi])$ from the start. In other words, the dynamics even \emph{preserve} the $L^2$-norm. In the next section, we will see that this is not at all accidental, but by design. 

Now to the part about best approximation: one immediate idea is to allow only $n^2 \in \{ 0 , 1 , \ldots , N \}$, because for $\sum_{n \in \Z} \babs{\scpro{\e^{+ \ii n x}}{\psi_0}}^2$ to converge, a necessary condition is that $\babs{\scpro{\e^{+ \ii n x}}{\psi_0}}^2 \to 0$ as $n \to \infty$. So for any upper bound for the error $\eps > 0$ (which we shall also call \emph{precision}), we can find $N(\eps)$ so that the initial condition can be approximated in norm up to an error $\eps$, 
\begin{align*}
	\norm{\psi_0}^2 - \sum_{n = -N(\eps)}^{+N(\eps)} \babs{\scpro{\e^{+ \ii n x}}{\psi_0}}^2 
	&= \sum_{\abs{n} > N(\eps)} \babs{\scpro{\e^{+ \ii n x}}{\psi_0}}^2
	< \eps^2 
	. 
\end{align*}
Since all the vectors $\{ \e^{+ \ii n x} \}_{n \in \Z}$ are orthonormal, the vector of best approximation is 
\begin{align*}
	\psi_{\mathrm{best}}(x) = \sum_{n = -N(\eps)}^{+N(\eps)} \scpro{\e^{+ \ii n x}}{\psi_0} \, \e^{+ \ii n x} 
	, 
\end{align*}
and we can repeat the above arguments to see that the time-evolved $\psi_{\mathrm{best}}(t)$ stays $\eps$-close to $\psi(t)$ \emph{for all times} in norm, 
\begin{align*}
	\norm{\psi(t) - \psi_{\mathrm{best}}(t)}^2 &= \sum_{\abs{n} > N(\eps)} \babs{\e^{- \ii n^2 t} \, \scpro{\e^{+ \ii n x}}{\psi_0}}^2 
	\\
	&= \sum_{\abs{n} > N(\eps)} \babs{\scpro{\e^{+ \ii n x}}{\psi_0}}^2 
	< \eps^2 
	. 
\end{align*}
In view of the Grönwall lemma~\ref{odes:lem:Groenwall}, one may ask why the two are close for all times. The crucial ingredient here is \emph{linearity} of \eqref{spaces:eqn:free_Schroedinger_interval}. 
% subsection Best approximation on $L^2([-\pi,+\pi])$ using $\{ \e^{+ \ii n x} \}_{n \in \Z}$ (end)
% section Hilbert spaces (end)

\section{Linear functionals, dual space and weak convergence} % (fold)
\label{hilbert_spaces:dual_space}

A very important notion is that of a functional. We have already gotten to know the free energy functional 
\begin{align*}
	E_{\mathrm{free}} : &\mathcal{D}(E_{\mathrm{free}}) \subset L^2(\R^n) \longrightarrow [0,+\infty) \subset \C 
	, 
	\\
	&\varphi \mapsto E_{\mathrm{free}}(\varphi) = \frac{1}{2m} \sum_{l = 1}^d \bscpro{(- i \hbar \partial_{x_l} \varphi)}{(- i \hbar \partial_{x_l} \varphi)} 
	. 
\end{align*}
This functional, however, is not linear, and it is not defined for all $\varphi \in L^2(\R^n)$. Let us restrict ourselves to a smaller class of functionals: 
\begin{definition}[Bounded linear functional]
	Let $\mathcal{X}$ be a normed space. Then a map 
	\begin{align*}
		 L : \mathcal{X} \longrightarrow \C 
	\end{align*}
	is a bounded linear functional if and only if 
	\begin{enumerate}[(i)]
		\item there exists $C>0$ such that $\abs{L(x)} \leq C \norm{x}$ and 
		\item $L(x + \mu y) = L(x) + \mu L(y)$ 
	\end{enumerate}
	hold for all $x,y \in \mathcal{X}$ and $\mu \in \C$. 
\end{definition}
A very basic fact is that boundedness of a linear functional is equivalent to its continuity. \marginpar{2013.10.17}
\begin{theorem}\label{hilbert_spaces:dual_space:thm:bounded_functionals_continuous}
	Let $L : \mathcal{X} \longrightarrow \C$ be a linear functional on the normed space $\mathcal{X}$. Then the following statements are equivalent: 
	\begin{enumerate}[(i)]
		\item $L$ is continuous at $x_0 \in \mathcal{X}$. 
		\item $L$ is continuous. 
		\item $L$ is bounded. 
	\end{enumerate}
\end{theorem}
\begin{proof}
	(i) $\Leftrightarrow$ (ii): This follows immediately from the linearity. 
	\medskip
	
	\noindent
	(ii) $\Rightarrow$ (iii): Assume $L$ to be continuous. Then it is continuous at $0$ and for $\eps = 1$, we can pick $\delta > 0$ such that %
	\begin{align*}
		\abs{L(x)} \leq \eps = 1 
	\end{align*}
	for all $x \in \mathcal{X}$ with $\norm{x} \leq \delta$. By linearity, this implies for any $y \in \mathcal{X} \setminus \{ 0 \}$ that 
	\begin{align*}
		\babs{L \bigl ( \tfrac{\delta}{\norm{y}} y \bigr )} = \tfrac{\delta}{\norm{y}} \, \babs{L(y)} \leq 1 
		. 
	\end{align*}
	Hence, $L$ is bounded with bound $\nicefrac{1}{\delta}$, 
	\begin{align*}
		\babs{L(y)} \leq \tfrac{1}{\delta} \norm{y} 
		. 
	\end{align*}
	(iii) $\Rightarrow$ (ii): Conversely, if $L$ is bounded by $C > 0$, 
	\begin{align*}
		\babs{L(x) - L(y)} \leq C \norm{x-y} 
		, 
	\end{align*}
	holds for all $x,y \in \mathcal{X}$. This means, $L$ is continuous: for $\eps > 0$ pick $\delta = \nicefrac{\eps}{C}$ so that 
	\begin{align*}
		\babs{L(x) - L(y)} \leq C \norm{x-y} \leq C \tfrac{\eps}{C} = \eps 
	\end{align*}
	holds for all $x,y \in \mathcal{X}$ such that $\norm{x - y} \leq \nicefrac{\eps}{C}$. 
\end{proof}
\begin{definition}[Dual space]\label{spaces:defn:dual_space}
	Let $\mathcal{X}$ be a normed space. The dual space $\mathcal{X}'$ is the vector space of bounded linear functionals endowed with the norm 
	\begin{align*}
		\norm{L}_* := \sup_{x \in \mathcal{X} \setminus \{ 0 \}} \frac{\abs{L(x)}}{\norm{x}} = \sup_{\substack{x \in \mathcal{X} \\ \norm{x} = 1}} \abs{L(x)} 
		. 
	\end{align*}
\end{definition}
Independently of whether $\mathcal{X}$ is complete, $\mathcal{X}'$ is a Banach space. 
\begin{proposition}\label{hilbert_spaces:dual_space:prop:completeness_Xstar}
	The dual space to a normed linear space $\mathcal{X}$ is a Banach space. 
\end{proposition}
\begin{proof}
	Let $( L_n )_{n \in \N}$ be a Cauchy sequence in $\mathcal{X}'$, \ie a sequence for which 
	\begin{align*}
		\norm{L_k - L_j}_* \xrightarrow{k,j \rightarrow \infty} 0
		. 
	\end{align*}
	We have to show that $( L_n )_{n \in \N}$ converges to some $L \in \mathcal{X}'$. For any $\eps > 0$, there exists $N(\eps) \in \N$ such that 
	\begin{align*}
		\norm{L_k - L_j}_* < \eps 
	\end{align*}
	for all $k,j \geq N(\eps)$. This also implies that for any $x \in \mathcal{X}$, $\bigl ( L_n(x) \bigr )_{n \in \N}$ converges as well, 
	\begin{align*}
		\babs{L_k(x) - L_j(x)} \leq \bnorm{L_k - L_j}_* \, \norm{x} < \eps \norm{x} 
		. 
	\end{align*}
	The field of complex numbers is complete and $\bigl ( L_n(x) \bigr )_{n \in \N}$ converges to some $L(x) \in \C$. We now \emph{define} 
	\begin{align*}
		L(x) := \lim_{n \rightarrow \infty} L_n(x) 
	\end{align*}
	for any $x \in \mathcal{X}$. Clearly, $L$ inherits the linearity of the $(L_n)_{n \in \N}$. The map $L$ is also bounded: for any $\eps > 0$, there exists $N(\eps) \in \N$ such that $\norm{L_j - L_n}_* < \eps$ for all $j , n \geq N(\eps)$. Then 
	\begin{align*}
		\babs{(L - L_n)(x)} &= \lim_{j \rightarrow \infty} \babs{(L_j - L_n)(x)} \leq \lim_{j \to \infty} \bnorm{L_j - L_n}_* \, \bnorm{x} 
		\\
		&< \eps \norm{x}
	\end{align*}
	holds for all $n \geq N(\eps)$. Since we can write $L$ as $L = L_n + (L - L_n)$, we can estimate the norm of the linear map $L$ by $\snorm{L}_* \leq \snorm{L_n}_* + \eps < \infty$. This means $L$ is a bounded linear functional on $\mathcal{X}$. 
\end{proof}
In case of Hilbert spaces, the dual $\Hil'$ can be canonically identified with $\Hil$ itself: 
\begin{theorem}[Riesz' Lemma]\label{hilbert_spaces:dual_space:thm:Riesz_Lemma}
	Let $\Hil$ be a Hilbert space. Then for all $L \in \Hil'$ there exist $\psi_L \in \Hil$ such that 
	\begin{align*}
		L(\varphi) = \sscpro{\psi_L}{\varphi} 
		. 
	\end{align*}
	In particular, we have $\norm{L}_* = \snorm{\psi_L}$. 
\end{theorem}
\begin{proof}
	Let $\ker L := \bigl \{ \varphi \in \Hil \; \vert \; L(\varphi) = 0 \bigr \}$ be the kernel of the functional $L$ and as such is a closed linear subspace of $\Hil$. If $\ker L = \Hil$, then $0 \in \Hil$ is the associated vector, 
	\begin{align*}
		L(\varphi) = 0 = \sscpro{0}{\varphi} 
		. 
	\end{align*}
	So assume $\ker L \subsetneq \Hil$ is a proper subspace. Then we can split $\Hil = \ker L \oplus (\ker L)^{\perp}$. Pick $\varphi_0 \in (\ker L)^{\perp}$, \ie $L(\varphi_0) \neq 0$. Then define 
	\begin{align*}
		\psi_L := \frac{L(\varphi_0)}{\snorm{\varphi_0}^2} \, \varphi_0 
		. 
	\end{align*}
	We will show that $L(\varphi) = \sscpro{\psi_L}{\varphi}$. If $\varphi \in \ker L$, then $L(\varphi) = 0 = \sscpro{\psi_L}{\varphi}$. One easily shows that for $\varphi = \alpha \, \varphi_0$, $\alpha \in \C$, 
	\begin{align*}
		L(\varphi) &= L(\alpha \, \varphi_0) = \alpha \, L(\varphi_0) 
		\\
		&= \sscpro{\psi_L}{\varphi} = \Bscpro{\tfrac{L(\varphi_0)^*}{\snorm{\varphi_0}^2} \varphi_0}{\alpha \, \varphi_0} 
		\\
		&= \alpha \, L(\varphi_0) \, \frac{\sscpro{\varphi_0}{\varphi_0}}{\snorm{\varphi_0}^2} 
		= \alpha \, L(\varphi_0) 
		. 
	\end{align*}
	Every $\varphi \in \Hil$ can be written as 
	\begin{align*}
		\varphi = \biggl ( \varphi - \frac{L(\varphi)}{L(\varphi_0)} \, \varphi_0 \biggr ) + \frac{L(\varphi)}{L(\varphi_0)} \, \varphi_0 
		. 
	\end{align*}
	Then the first term is in the kernel of $L$ while the second one is in the orthogonal complement of $\ker L$. Hence, $L(\varphi) = \sscpro{\psi_L}{\varphi}$ for all $\varphi \in \Hil$. If there exists a second $\psi_L' \in \Hil$, then for any $\varphi \in \Hil$
	\begin{align*}
		0 = L(\varphi) - L(\varphi) = \sscpro{\psi_L}{\varphi} - \sscpro{\psi_L'}{\varphi} 
		= \sscpro{\psi_L - \psi_L'}{\varphi} 
		. 
	\end{align*}
	This implies $\psi_L' = \psi_L$ and thus the element $\psi_L$ is unique. 
	
	To show $\snorm{L}_* = \snorm{\psi_L}$, assume $L \neq 0$. Then, we have 
	\begin{align*}
		\norm{L}_* &= \sup_{\norm{\varphi} = 1} \babs{L(\varphi)} \geq \babs{L \bigl ( \tfrac{\psi_L}{\snorm{\psi_L}} \bigr )} 
		\\
		&= \bscpro{\psi_L}{\tfrac{\psi_L}{\snorm{\psi_L}}} = \snorm{\psi_L} 
		. 
	\end{align*}
	On the other hand, the Cauchy-Schwarz inequality yields 
	\begin{align*}
		\norm{L}_* &= \sup_{\norm{\varphi} = 1} \babs{L(\varphi)} 
		= \sup_{\norm{\varphi} = 1} \babs{\sscpro{\psi_L}{\varphi}} 
		\\
		&\leq \sup_{\norm{\varphi} = 1} \snorm{\psi_L} \snorm{\varphi} 
		= \snorm{\psi_L} 
		. 
	\end{align*}
	Putting these two together, we conclude $\snorm{L}_* = \snorm{\psi_L}$. 
\end{proof}
\begin{definition}[Weak convergence]\label{spaces:defn:weak_convergence}
	Let $\mathcal{X}$ be a Banach space. Then a sequence $(x_n)_{n \in \N}$ in $\mathcal{X}$ is said to converge weakly to $x \in \mathcal{X}$ if for all $L \in \mathcal{X}'$ 
	\begin{align*}
		L(x_n) \xrightarrow{n \rightarrow \infty} L(x) 
	\end{align*}
	holds. In this case, one also writes $x_n \rightharpoonup x$. 
\end{definition}
Weak convergence, as the name suggests, is really weaker than convergence in norm. The reason why “more” sequences converge is that, a sense, uniformity is lost. If $\mathcal{X}$ is a Hilbert space, then applying a functional is the same as computing the inner product with respect to some vector $\psi_L$. If the “non-convergent part” lies in the orthogonal complement to $\{ \psi_L \}$, then this particular functional does not notice that the sequence has not converged yet. 
\begin{example}
	Let $\Hil$ be a separable infinite-dimensional Hilbert space and $\{ \varphi_n \}_{n \in \N}$ an ortho\-normal basis. Then the sequence $(\varphi_n)_{n \in \N}$ does not converge in norm, for as long as $n \neq k$ 
	\begin{align*}
		\norm{\varphi_n - \varphi_k} = \sqrt{2} 
		, 
	\end{align*}
	but it does converge weakly to $0$: for any functional $L = \sscpro{\psi_L}{\cdot}$, we see that $\bigl ( \abs{L(\varphi_n)} \bigr )_{n \in \N}$ is a sequence in $\R$ that converges to $0$. Since $\{ \varphi_n \}_{n \in \N}$ is a basis, we can write 
	\begin{align*}
		\psi_L = \sum_{n = 1}^{\infty} \sscpro{\varphi_n}{\psi_L} \, \varphi_n 
	\end{align*}
	and for the sequence of partial sums to converge to $\psi_L$, the sequence of coefficients 
	\begin{align*}
		\bigl ( \sscpro{\varphi_n}{\psi_L} \bigr )_{n \in \N} = \bigl ( L(\varphi_n)' \bigr )_{n \in \N} 
	\end{align*}
	must converge to $0$. Since this is true for any $L \in \Hil'$, we have proven that $\varphi_n \rightharpoonup 0$ (\ie $\varphi_n \rightarrow 0$ weakly). 
\end{example}
In case of $\mathcal{X} = L^p(\Omega)$, there are three basic mechanisms for when a sequence of functions $(f_k)$ does not converge in norm, but only weakly: 
\begin{enumerate}[(i)]
	\item \emph{$f_k$ oscillates to death:} take $f_k(x) = \sin (kx)$ for $0 \leq x \leq 1$ and zero otherwise. 
	\item \emph{$f_k$ goes up the spout:} pick $g \in L^p(\R)$ and define $f_k(x) := k^{\nicefrac{1}{p}} \, g(kx)$. This sequence explodes near $x = 0$ for large $k$. 
	\item \emph{$f_k$ wanders off to infinity:} this is the case when for some $g \in L^p(\R)$, we define $f_k(x) := g(x+k)$. 
\end{enumerate}
All of these sequences converge weakly to $0$, but do not converge in norm. \marginpar{2013.10.22}
% section linear_functionals_dual_space_and_weak_convergence (end)
% chapter hilbert_spaces (end)
\chapter{Linear operators} % (fold)
\label{operators}
Linear operators appear quite naturally in the analysis of \emph{linear} PDEs. Many PDEs share a common structure: for instance, the Schrödinger equation 
\begin{align}
	\ii \, \partial_t \psi(t) &= \bigl ( - \Delta + V \bigr ) \psi(t) 
	\, , 
	&&
	\psi(0) = \psi_0 
	\, , 
	\label{operators:eqn:Schroedinger_equation}
\end{align}
and the heat equation 
\begin{align}
	\partial_t \psi(t) &= - \bigl ( - \Delta + V \bigr ) \psi(t) 
	\, , 
	&&
	\psi(0) = \psi_0 
	\, , 
	\label{operators:eqn:heat_equation}
\end{align}
both involve the same operator 
\begin{align*}
	H = - \Delta + V
\end{align*}
on the right-hand side. \emph{Formally}, we can solve these equations in closed form, 
\begin{align*}
	\psi(t) &= \e^{- \ii t H} \psi_0 
\end{align*}
solves the Schrödinger equation~\eqref{operators:eqn:Schroedinger_equation} while $\psi(t) = \e^{- t H} \psi_0$ solves the heat equation, because we can \emph{formally} compute the derivative 
\begin{align*}
	\ii \frac{\dd}{\dd t} \psi(t) &= \ii \frac{\dd}{\dd t} \bigl ( \e^{- \ii t H} \psi_0 \bigr ) 
	= - \ii^2 \, H \, \e^{- \ii t H} \psi_0 
	\\
	&= H \psi(t) 
\end{align*}
and verify that also the initial condition is satisfied, $\psi(0) = \e^{0} \, \psi_0 = \psi_0$. A priori, these are just formal manipulations, though; If $H$ were a matrix, we know how to give rigorous meaning to these expressions, but in case of operators on infinite-dimensional spaces, this is much more involved. However, we see that the dynamics of both, the Schrödinger and the heat equation are generated by the \emph{same} operator $H$. 

As we will see in Chapter~\ref{operators:Maxwell}, also the Maxwell equations can be recast in the form~\eqref{operators:eqn:Schroedinger_equation}. This gives one access to all the powerful tools for the analysis of Schrödinger operators in order to gain understanding of the dynamics of electromagnetic waves (light). 

Moreover, one can see that the selfadjointness of $H = H^*$ leads to $U(t) := \e^{- \ii t H}$ being a \emph{unitary} operator, an operator which preserves the norm on the Hilbert space (“$\e^{- \ii t H}$ is just a phase”). Lastly, states will be closely related to (orthogonal) \emph{projections} $P$ which satisfy $P^2 = P$. 
\medskip

\noindent
This chapter will give a zoology of operators and expound on three particularly important classes of operator: selfadjoint operators, orthogonal projections and unitary operators as well as their relations. Given the brevity, much of what we do will not be rigorous. In fact, some of these results (\eg Stone's theorem) require extensive preparation until one can understand all these facets. For us, the important aspect is to elucidate the connections between these fundamental results and PDEs.

\section{Bounded operators} % (fold)
\label{operators:bounded}
The simplest operators are bounded operators. 
\begin{definition}[Bounded operator]
	Let $\mathcal{X}$ and $\mathcal{Y}$ be normed spaces. A linear operator $T : \mathcal{X} \longrightarrow \mathcal{Y}$ is called bounded if there exists $M \geq 0$ with $\norm{T x}_{\mathcal{Y}} \leq M \norm{x}_{\mathcal{X}}$. 
\end{definition}
Just as in the case of linear functionals, we have 
\begin{theorem}
	Let $T : \mathcal{X} \longrightarrow \mathcal{Y}$ be a linear operator between two normed spaces $\mathcal{X}$ and $\mathcal{Y}$. Then the following statements are equivalent: 
	\begin{enumerate}[(i)]
		\item $T$ is continuous at $x_0 \in \mathcal{X}$. 
		\item $T$ is continuous. 
		\item $T$ is bounded. 
	\end{enumerate}
\end{theorem}
\begin{proof}
	We leave it to the reader to modify the proof of Theorem~\ref{hilbert_spaces:dual_space:thm:bounded_functionals_continuous}. 
\end{proof}
We can introduce a norm on the operators which leads to a natural notion of convergence:%
\begin{definition}[Operator norm]
	Let $T : \mathcal{X} \longrightarrow \mathcal{Y}$ be a bounded linear operator between normed spaces. Then we define the operator norm of $T$ as 
	\begin{align*}
		\norm{T} := \sup_{\substack{x \in \mathcal{X} \\ \norm{x} = 1}} \norm{T x}_{\mathcal{Y}} 
		. 
	\end{align*}
	The space of all bounded linear operators between $\mathcal{X}$ and $\mathcal{Y}$ is denoted by $\mathcal{B}(\mathcal{X},\mathcal{Y})$. 
\end{definition}
One can show that $\norm{T}$ coincides with 
\begin{align*}
	\inf \bigl \{ M \geq 0 \; \vert \; \norm{T x}_{\mathcal{Y}} \leq M \norm{x}_{\mathcal{X}} \; \forall x \in \mathcal{X} \bigr \} = \norm{T} 
	. 
\end{align*}
The product of two bounded operators $T \in \mathcal{B}(\mathcal{Y},\mathcal{Z})$ and $S \in \mathcal{B}(\mathcal{X},\mathcal{Y})$ is again a bounded operator and its norm can be estimated from above by 
\begin{align*}
	\norm{T S} \leq \norm{T} \norm{S} 
	. 
\end{align*}
If $\mathcal{Y} = \mathcal{X} = \mathcal{Z}$, this implies that the product is jointly continuous with respect to the norm topology on $\mathcal{X}$. For Hilbert spaces, the following useful theorem holds: 
\begin{theorem}[Hellinger-Toeplitz]
	Let $A$ be a linear operator on a Hilbert space $\Hil$ with dense domain $\mathcal{D}(A)$ such that $\scpro{\psi}{A \varphi} = \scpro{A \psi}{\varphi}$ holds for all $\varphi , \psi \in \mathcal{D}(A)$. Then $\mathcal{D}(A) = \Hil$ if and only if $A$ is bounded. 
\end{theorem}
\begin{proof}
	$\Leftarrow$: If $A$ is bounded, then $\norm{A \varphi} \leq M \norm{\varphi}$ for some $M \geq 0$ and all $\varphi \in \Hil$ by definition of the norm. Hence, the domain of $A$ is all of $\Hil$. 
	\medskip
	
	\noindent
	$\Rightarrow$: This direction relies on a rather deep result of functional analysis, the so-called Open Mapping Theorem and its corollary, the Closed Graph Theorem. The interested reader may look it up in Chapter~III.5 of \cite{Reed_Simon:M_cap_Phi_1:1972}. 
\end{proof}
Let $T,S$ be bounded linear operators between the normed spaces $\mathcal{X}$ and $\mathcal{Y}$. If we define 
\begin{align*}
	(T + S) x := T x + S x 
\end{align*}
as addition and 
\begin{align*}
	\bigl ( \lambda \cdot T  \bigr ) x := \lambda T x 
\end{align*}
as scalar multiplication, the set of bounded linear operators forms a vector space. 
\begin{proposition}\label{operators:prop:B_X_Y_complete}
	The vector space $\mathcal{B}(\mathcal{X},\mathcal{Y})$ of bounded linear operators between normed spaces $\mathcal{X}$ and $\mathcal{Y}$ with operator norm forms a normed space. If $\mathcal{Y}$ is complete, $\mathcal{B}(\mathcal{X},\mathcal{Y})$ is a Banach space. 
\end{proposition}
\begin{proof}
	The fact $\mathcal{B}(\mathcal{X},\mathcal{Y})$ is a normed vector space follows directly from the definition. To show that $\mathcal{B}(\mathcal{X},\mathcal{Y})$ is a Banach space whenever $\mathcal{Y}$ is, one has to modify the proof of Theorem~\ref{hilbert_spaces:dual_space:prop:completeness_Xstar} to suit the current setting. This is left as an exercise. 
\end{proof}
Very often, it is easy to \emph{define} an operator $T$ on a ``nice'' dense subset $\mathcal{D} \subseteq \mathcal{X}$. Then the next theorem tells us that \emph{if}\ the operator is bounded, there is a unique bounded extension of the operator to the whole space $\mathcal{X}$. For instance, this allows us to instantly extend the Fourier transform from Schwartz functions to $L^2(\R^n)$ functions (see~Proposition~\ref{Fourier:R:thm:Parseval_Plancherel}). 
\begin{theorem}\label{operators:bounded:thm:extensions_bounded_operators}
	Let $\mathcal{D} \subseteq \mathcal{X}$ be a dense subset of a normed space and $\mathcal{Y}$ be a Banach space. Furthermore, let $T : \mathcal{D} \longrightarrow \mathcal{Y}$ be a bounded linear operator. Then there exists a unique bounded linear extension $\tilde{T} : \mathcal{X} \longrightarrow \mathcal{Y}$ and $\snorm{\tilde{T}} = \norm{T}_{\mathcal{D}}$. 
\end{theorem}
\begin{proof}
	We construct $\tilde{T}$ explicitly: let $x \in \mathcal{X}$ be arbitrary. Since $\mathcal{D}$ is dense in $\mathcal{X}$, there exists a sequence $(x_n)_{n \in \N}$ in $\mathcal{D}$ which converges to $x$. Then we set 
	\begin{align*}
		\tilde{T} x := \lim_{n \to \infty} T x_n 
		. 
	\end{align*}
	First of all, $\tilde{T}$ is linear. It is also well-defined: $(T x_n)_{n \in \N}$ is a Cauchy sequence in $\mathcal{Y}$, 
	\begin{align*}
		\bnorm{T x_n - T x_k}_{\mathcal{Y}} \leq \snorm{T}_{\mathcal{D}} \, \snorm{x_n - x_k}_{\mathcal{X}} \xrightarrow{n,k \to \infty} 0 
		, 
	\end{align*}
	where the norm of $T$ is defined as 
	\begin{align*}
		\snorm{T}_{\mathcal{D}} := \sup_{x \in \mathcal{D} \setminus \{ 0 \}} \frac{\norm{T x}_{\mathcal{Y}}}{\norm{x}_{\mathcal{X}}} 
		. 
	\end{align*}
	This Cauchy sequence in $\mathcal{Y}$ converges to some unique $y \in \mathcal{Y}$ as the target space is complete. Let $(x_n')_{n \in \N}$ be a second sequence in $\mathcal{D}$ that converges to $x$ and assume the sequence $(T x_n')_{n \in \N}$ converges to some $y' \in \mathcal{Y}$. We define a third sequence $(z_n)_{n \in \N}$ which alternates between elements of the first sequence $(x_n)_{n \in \N}$ and the second sequence $(x_n')_{n \in \N}$, \ie 
	\begin{align*}
		z_{2n - 1} &:= x_n 
		\\
		z_{2n} &:= x_n' 
	\end{align*}
	for all $n \in \N$. Then $(z_n)_{n \in \N}$ also converges to $x$ and $\bigl ( T z_n \bigr )$ forms a Cauchy sequence that converges to, say, $\zeta \in \mathcal{Y}$. Subsequences of convergent sequences are also convergent and they must converge to the same limit point. Hence, we conclude that 
	\begin{align*}
		\zeta &= \lim_{n \to \infty} T z_n = \lim_{n \to \infty} T z_{2n} = \lim_{n \to \infty} T x_n = y 
		\\
		&= \lim_{n \to \infty} T z_{2n - 1} = \lim_{n \to \infty} T x_n' = y' 
	\end{align*}
	holds and $\tilde{T} x$ does not depend on the particular choice of sequence which approximates $x$ in $\mathcal{D}$. It remains to show that $\snorm{\tilde{T}} = \snorm{T}_{\mathcal{D}}$: we can calculate the norm of $\tilde{T}$ on the dense subset $\mathcal{D}$ and use that $\tilde{T} \vert_{\mathcal{D}} = T$ to obtain 
	\begin{align*}
		\snorm{\tilde{T}} &= \sup_{\substack{x \in \mathcal{X} \\ \norm{x} = 1}} \snorm{\tilde{T} x} 
		= \sup_{x \in \mathcal{X} \setminus \{ 0 \}} \frac{\snorm{\tilde{T} x}}{\norm{x}} 
		= \sup_{x \in \mathcal{D} \setminus \{ 0 \}} \frac{\snorm{\tilde{T} x}}{\norm{x}} 
		\\
		&= \sup_{x \in \mathcal{D} \setminus \{ 0 \}} \frac{\norm{T x}}{\norm{x}} 
		= \norm{T}_{\mathcal{D}} 
		. 
	\end{align*}
	Hence, the norm of the extension $\tilde{T}$ is equal to the norm of the original operator $T$. 
\end{proof}
The spectrum of an operator is related to the set of possible outcomes of measurements in quantum mechanics. 
\begin{definition}[Spectrum]\label{operators:defn:spectrum}
	Let $T \in \mathcal{B}(\mathcal{X})$ be a bounded linear operator on a Banach space $\mathcal{X}$. We define: 
	\begin{enumerate}[(i)]
		\item The resolvent of $T$ is the set $\rho(T) := \bigl \{ z \in \C \; \vert \; T - z \, \id \mbox{ is bijective} \bigr \}$. 
		\item The spectrum $\sigma(T) := \C \setminus \rho(T)$ is the complement of $\rho(T)$ in $\C$. 
		\item The set of all eigenvalues is called \emph{point spectrum} 
		\begin{align*}
			\sigma_{\mathrm{p}}(T) := \bigl \{ z \in \C \; \vert \; T - z \, \id \mbox{ is not injective} \bigr \} 
			. 
		\end{align*}
		\item The \emph{continuous spectrum} is defined as 
		\begin{align*}
			\sigma_{\mathrm{c}}(T) := \bigl \{ z \in \C \; \vert \; T - z \, \id \mbox{ is injective, } \im (T - z \, \id) \subseteq \mathcal{X} \mbox{ dense} \bigr \} 
			. 
		\end{align*}
		\item The remainder of the spectrum is called \emph{residual spectrum}, 
		\begin{align*}
			\sigma_{\mathrm{r}}(T) := \bigl \{ z \in \C \; \vert \; T - z \, \id \mbox{ is injective, } \im (T - z \, \id) \subseteq \mathcal{X} \mbox{ not dense} \bigr \} 
		\end{align*}
	\end{enumerate}
\end{definition}
One can show that for all $\lambda \in \rho(T)$, the map $(T - z \, \id)^{-1}$ is a bounded operator and the spectrum is a closed subset of $\C$. One can show its $\sigma(T)$ is \emph{compact} and contained in $\bigl \{ \lambda \in \C \; \vert \; \abs{\lambda} \leq \norm{T} \bigr \} \subset \C$. 
\begin{example}[Spectrum of $H = -\partial_x^2$]
	\begin{enumerate}[(i)]
		\item The spectrum of $- \partial_x^2$ on $L^2(\R)$ is $\sigma(- \partial_x^2) = \sigma_{\mathrm{c}}(- \partial_x^2) = [0,\infty)$; it is purely continuous. 
		\item Restricting $- \partial_x^2$ to a bounded domain (\eg an interval $[a,b]$) turns out to be purely discrete, $\sigma(- \partial_x^2) = \sigma_{\mathrm{p}}(- \partial_x^2)$. 
	\end{enumerate}
\end{example}
Note that $- \partial_x^2$ need not have point spectrum. If an operator $H$ on a Banach space $\mathcal{X}$ has continuous spectrum, then there exist no eigenvectors \emph{in} $\mathcal{X}$. For instance, the eigenfunctions of $- \partial_x^2$ to $\lambda^2 \neq 0$ are of the form $\e^{\pm \ii \lambda x}$. However, none of these plane waves is square integrable on $\R$, $\norm{\e^{\pm \ii \lambda x}}_{L^2(\R)} = \infty$. 
% section bounded_operators (end)

\section{Adjoint operator} % (fold)
\label{operators:adjoint}
If $\mathcal{X}$ is a normed space, then we have defined $\mathcal{X}'$, the space of bounded linear functionals on $\mathcal{X}$. If $T : \mathcal{X} \longrightarrow \mathcal{Y}$ is a bounded linear operator between two normed spaces, it naturally defines the \emph{adjoint operator} $T' : \mathcal{Y}' \longrightarrow \mathcal{X}'$ via 
\begin{align}
	(T' L)(x) := L(Tx) 
	\label{operators:adjoint:eqn:adjoint_operator_functional}
\end{align}
for all $x \in \mathcal{X}$ and $L \in \mathcal{Y}'$. In case of Hilbert spaces, one can associate the \emph{Hilbert space adjoint}. We will almost exclusively work with the latter and thus drop “Hilbert space” most of the time. 
\begin{definition}[Adjoint and selfadjoint operator]
	Let $\Hil$ be a Hilbert space and $A \in \mathcal{B}(\Hil)$ be a bounded linear operator on the Hilbert space $\Hil$. Then for any $\varphi \in \Hil$, the equation 
	\begin{align*}
		\scpro{A \psi}{\varphi} = \scpro{\psi}{\phi} 
		&& 
		\forall \psi \in \mathcal{D}(A) 
	\end{align*}
	defines a vector $\phi$. For each $\varphi \in \Hil$, we set $A^* \varphi := \phi$ and $A^*$ is called the (Hilbert space) \emph{adjoint} of $A$. In case $A^* = A$, the operator is called \emph{selfadjoint}. 
\end{definition}
Hilbert and Banach space adjoint are related through the map $C \psi := \sscpro{\psi}{\cdot \,} = L_{\psi}$, because then the Hilbert space adjoint is defined as 
\begin{align*}
	A^* := C^{-1} A' C 
	. 
\end{align*}
\begin{example}[Adjoint of the time-evolution group]
	$\bigl ( \e^{- \ii t H} \bigr )^* = \e^{+ \ii t H^*} = \e^{+ \ii t H}$ 
\end{example}
\begin{proposition}
	Let $A , B \in \mathcal{B}(\Hil)$ be two bounded linear operators on a Hilbert space $\Hil$ and $\alpha \in \C$. Then, we have: 
	\begin{enumerate}[(i)]
		\item $(A + B)^* = A^* + B^*$ 
		\item $(\alpha A)^* = \alpha^* \, A^*$
		\item $(A B)^* = B^* A^*$
		\item $\norm{A^*} = \norm{A}$
		\item $A^{**} = A$
		\item $\norm{A^* A} = \norm{A A^*} = \norm{A}^2$
		\item $\ker A = (\im A^*)^{\perp}$, $\ker A^* = (\im A)^{\perp}$
	\end{enumerate}
\end{proposition}
\begin{proof}
	Properties (i)-(iii) follow directly from the defintion. 
	
	To show (iv), we note that $\norm{A} \leq \norm{A^*}$ follows from 
	\begin{align*}
		\norm{A \varphi} &= \abs{\Bscpro{\tfrac{A \varphi}{\snorm{A \varphi}}}{A \varphi}} 
		\overset{*}{=} \sup_{\norm{L}_* = 1} \abs{L(A \varphi)} 
		\\
		&= \sup_{\norm{\psi_L} = 1} \abs{\scpro{A^* \psi_L}{\varphi}} 
		\leq \norm{A^*} \norm{\varphi} 
	\end{align*}
	where in the step marked with $\ast$, we have used that we can calculate the norm from picking the functional associated to $\tfrac{A \varphi}{\snorm{A \varphi}}$: for a functional with norm 1, $\norm{L}_* =  1$, the norm of $L(A \varphi)$ cannot exceed that of $A \varphi$
	\begin{align*}
		\abs{L(A \varphi)} &= \sabs{\sscpro{\psi_L}{A \varphi}} 
		\leq \snorm{\psi_L} \snorm{A \varphi} 
		= \snorm{A \varphi} 
		. 
	\end{align*}
	Here, $\psi_L$ is the vector such that $L = \sscpro{\psi_L}{\cdot \,}$ which exists by Theorem~\ref{hilbert_spaces:dual_space:thm:Riesz_Lemma}. This theorem also ensures $\norm{L}_* = \snorm{\psi_L}$. On the other hand, from 
	\begin{align*}
		\bnorm{A^* \psi_L} &= \bnorm{L_{A^* \psi_L}}_* 
		% = \bnorm{A' L}_* 
		% = \sup_{\norm{\varphi} = 1} \bnorm{A' L(\varphi)} 
		= \sup_{\norm{\varphi} = 1} \babs{\bscpro{A^* \psi_L}{\varphi}} 
		\\
		&\leq \sup_{\norm{\varphi} = 1} \norm{\psi_L} \norm{A \varphi} 
		= \norm{A} \norm{L}_* 
		= \norm{A} \snorm{\psi_L} 
	\end{align*}
	we conclude $\norm{A^*} \leq \norm{A}$. Hence, $\norm{A^*} = \norm{A}$. 
	
	(v) is clear. For (vi), we remark 
	\begin{align*}
		\norm{A}^2 &= \sup_{\norm{\varphi} = 1} \norm{A \varphi}^2 
		= \sup_{\norm{\varphi} = 1} \bscpro{\varphi}{A^* A \varphi} 
		\\
		&\leq \sup_{\norm{\varphi} = 1} \norm{A^* A \varphi} 
		= \norm{A^* A} 
		. 
	\end{align*}
	This means 
	\begin{align*}
		\norm{A}^2 \leq \norm{A^* A} \leq \snorm{A^*} \norm{A} = \norm{A}^2 
		. 
	\end{align*}
	which combined with (iv), 
	\begin{align*}
		\norm{A}^2 = \norm{A^*}^2 \leq \norm{A A^*} \leq \snorm{A} \norm{A^*} = \norm{A}^2 
	\end{align*}
	implies $\norm{A^* A}= \norm{A}^2 = \norm{A A^*}$. (vii) is left as an exercise. 
\end{proof}
\begin{definition}
	Let $\Hil$ be a Hilbert space and $A \in \mathcal{B}(\Hil)$. Then $A$ is called 
	\begin{enumerate}[(i)]
		% \item normal if $A^* A = A A^*$. 
		\item selfadjoint (or hermitian) if $A^* = A$, 
		\item unitary if $A^* = A^{-1}$, 
		% \item a projection if $A^2 = A$. 
		\item an orthogonal projection if $A^2 = A$ and $A^* = A$, and 
		\item positive semidefinite (or non-negative) iff $\bscpro{\varphi}{A \varphi} \geq 0$ for all $\varphi \in \Hil$ and positive (definite) if the inequality is strict. 
	\end{enumerate}
\end{definition}
This leads to a particular characterization of the spectrum as a set \cite[Theorem~VII.12]{Reed_Simon:M_cap_Phi_1:1972}: 
\begin{theorem}[Weyl's criterion]\label{operators:thm:Weyl_criterion}
	Let $H$ be a bounded selfadjoint operator on a Hilbert space $\Hil$. Then $\lambda \in \sigma(H)$ holds if and only if there exists a sequence $\{ \psi_n \}_{n \in \N}$ so that $\snorm{\psi_n} = 1$ and 
	\begin{align*}
		\lim_{n \to \infty} \bnorm{H \psi_n - \lambda \, \psi_n}_{\Hil} = 0 
		. 
	\end{align*}
	%
	% commented out bits on essential spectrum 
	% We have $\lambda \in \sigma_{\mathrm{ess}}(H)$ if and only if we can choose the sequence $\{ \psi_n \}_{n \in \N}$ to be orthonormal. 
\end{theorem}
\begin{example}[Weyl's criterion for $H = -\partial_x^2$ on $L^2(\R)$]
	For any $\lambda \in \R \setminus \{ 0 \}$, one can choose a sequence $\{ \psi_n \}_{n \in \N}$ of normalized and cut off plane waves $\e^{\pm \ii \lambda x}$. To make sure they are normalized, we know that pointwise $\psi_n(x) \to 0$ as $n \to \infty$. \marginpar{2013.10.24}
\end{example}
%
% section adjoint_operator (end)

\section{Unitary operators} % (fold)
\label{operators:unitary}
Unitary operators $U$ have the nice property that 
\begin{align*}
	\scpro{U \varphi}{U \psi} = \scpro{\varphi}{U^* U \psi} 
	= \scpro{\varphi}{U^{-1} U \psi} 
	= \scpro{\varphi}{\psi} 
\end{align*}
for all $\varphi , \psi \in \Hil$. In case of quantum mechanics, we are interested in solutions to the Schrödinger equation 
\begin{align*}
	\ii \frac{\dd }{\dd t} \psi(t) = H \psi(t) 
	, 
	&& 
	\psi(t) = \psi_0 
	, 
\end{align*}
for a hamilton operator which satisfies $H^* = H$. Assume that $H$ is bounded (this is really the case for many simple quantum systems). Then the unitary group generated by $H$, 
\begin{align*}
	U(t) = \e^{- \ii t H} 
	, 
\end{align*}
can be written as a power series, 
\begin{align*}
	\e^{- \ii t H} = \sum_{n = 0}^{\infty} \frac{1}{n!} (-\ii t)^n \, H^n 
	\, , 
\end{align*}
where $H^0 := \id$ by convention. The sequence of partial sums converges in the operator norm to $\e^{- \ii t H}$, 
\begin{align*}
	\sum_{n = 0}^N \frac{1}{n!} (-i t)^n \, H^n \xrightarrow{N \to \infty} \e^{- \ii t H} 
	, 
\end{align*}
since we can make the simple estimate 
\begin{align*}
	\norm{\sum_{n = 0}^{\infty} \frac{1}{n!} (- \ii t)^n \, H^n \psi} &\leq \sum_{n = 0}^{\infty} \frac{1}{n!} \abs{t}^n \norm{H^n \psi} 
	\leq \sum_{n = 0}^{\infty} \frac{1}{n!} \abs{t}^n \norm{H}^n \, \norm{\psi} 
	\\ 
	&= \e^{\abs{t} \norm{H}} \norm{\psi} < \infty 
	. 
\end{align*}
This shows that the power series of the exponential converges in the operator norm independently of the choice of $\psi$ to a bounded operator. Given a unitary evolution group, it is suggestive to obtain the hamiltonian which generates it by deriving $U(t) \psi$ with respect to time. This is indeed the correct idea. The left-hand side of the Schrödinger equation (modulo a factor of $\ii$) can be expressed as a limit 
\begin{align*}
	\frac{\dd}{\dd t} \psi(t) = \lim_{\delta \to 0} \tfrac{1}{\delta} \bigl ( \psi(t+\delta) - \psi(t) \bigr ) 
	. 
\end{align*}
This limit really exists, but before we compute it, we note that since 
\begin{align*}
	\psi(t+\delta) - \psi(t) = \e^{- \ii (t+\delta) H} \psi_0 - \e^{- \ii t H} \psi_0 
	= \e^{- \ii t H} \bigl ( \e^{- \ii \delta H} - 1 \bigr ) \psi_0 
	\, , 
\end{align*}
it suffices to consider differentiability at $t = 0$: taking limits in norm of $\Hil$, we get 
\begin{align*}
	\frac{\dd}{\dd t} \psi(0) &= \lim_{\delta \to 0} \, \tfrac{1}{\delta} \bigl ( \psi(\delta) - \psi_0 \bigr ) = \lim_{\delta \to 0} \, \frac{1}{\delta} \left ( \sum_{n = 0}^{\infty} \frac{(- \ii)^n}{n!} \delta^n H^n \psi_0 - \psi_0 \right ) 
	\\
	&
	= \lim_{\delta \to 0} \sum_{n = 1}^{\infty} \frac{(- \ii)^n}{n!} \delta^{n-1} H^n \psi_0 
	= - \ii H \psi_0 
	.  
\end{align*}
Hence, we have established that $\e^{- \ii t H} \psi_0$ solves the Schrödinger condition with $\psi(0) = \psi_0$, 
\begin{align*}
	\ii \frac{\dd}{\dd t} \psi(t) = H \psi(t) 
	. 
\end{align*}
However, this procedure \emph{does not work} if $H$ is unbounded (\ie the generic case)! Before we proceed, we need to introduce several different notions of convergence of sequences of operators which are necessary to define derivatives of $U(t)$. 
\begin{definition}[Convergence of operators]
	Let $\{ A_n \}_{n \in \N} \subset \mathcal{B}(\Hil)$ be a sequence of bounded operators. We say that the sequence converges to $A \in \mathcal{B}(\Hil)$
	\begin{enumerate}[(i)]
		\item uniformly/in norm if $\lim_{n \to \infty} \bnorm{A_n - A} = 0$. 
		\item strongly if $\lim_{n \to \infty} \bnorm{A_n \psi - A \psi} = 0$ for all $\psi \in \Hil$. 
		\item weakly if $\lim_{n \to \infty} \bscpro{\varphi}{A_n \psi - A \psi} = 0$ for all $\varphi , \psi \in \Hil$. 
	\end{enumerate}
\end{definition}
Convergence of a sequence of operators in norm implies strong and weak convergence, but not the other way around. In the tutorials, we will also show explicitly that weak convergence does not necessarily imply strong convergence. 
\begin{example}
	With the arguments above, we have shown that if $H = H^*$ is selfadjoint and bounded, then $t \mapsto \e^{- \ii t H}$ is \emph{uniformly} continuous. 
\end{example}
If $\norm{H} = \infty$ on the other hand, uniform continuity is too strong a requirement. If $H = - \frac{1}{2} \Delta_x$ is the free Schrödinger operator on $L^2(\R^n)$, then the Fourier transform $\Fourier$ links the position representation on $L^2(\R^n)$ to the momentum representation on $L^2(\R^n)$. In this representation, the free Schrödinger operator $H$ simplifies to the multiplication operator 
\begin{align*}
	\hat{H} = \tfrac{1}{2} \hat{k}^2 
\end{align*}
acting on $L^2(\R^n)$. This operator is not bounded since $\sup_{k \in \R^n} \tfrac{1}{2} k^2 = \infty$ (\cf problem~24). More elaborate mathematical arguments show that for any $t \in \R$, the norm of the difference between $\hat{U}(t) = \e^{- \ii t \frac{1}{2} \hat{k}^2}$ and $\hat{U}(0) = \id$
\begin{align*}
	\bnorm{\hat{U}(t) - \id} = \sup_{k \in \R^n} \babs{\e^{- \ii t \frac{1}{2} k^2} - 1} = 2 
\end{align*}
is exactly $2$ and $\hat{U}(t)$ \emph{cannot} be uniformly continuous in $t$. However, if $\widehat{\psi} \in L^2(\R^n)$ is a wave function, the estimate 
\begin{align*}
	\bnorm{\hat{U}(t) \widehat{\psi} - \widehat{\psi}}^2 &= \int_{\R^n} \dd k \, \babs{\e^{- \ii t \frac{1}{2} k^2} - 1}^2 \, \babs{\widehat{\psi}(k)}^2 
	\\
	&\leq 2^2 \int_{\R^n} \dd k \, \babs{\widehat{\psi}(k)}^2 
	= 4 \bnorm{\widehat{\psi}}^2 
\end{align*}
shows we can invoke the Theorem of Dominated Convergence to conclude $\hat{U}(t)$ is \emph{strongly continuous} in $t \in \R$. 
\begin{definition}[Strongly continuous one-parameter unitary evolution group]
	A family of unitary operators $\{ U(t) \}_{t \in \R}$ on a Hilbert space $\Hil$ is called a strongly continuous one-parameter unitary group -- or unitary evolution group for short -- if 
	\begin{enumerate}[(i)]
		\item $t \mapsto U(t)$ is strongly continuous and 
		\item $U(t) U(t') = U(t+t')$ as well as $U(0) = \id_{\Hil}$ 
	\end{enumerate}
	hold for all $t,t' \in \R$. 
\end{definition}
This is again a \emph{group representation of $\R$} just as in the case of the classical flow $\Phi$. The form of the Schrödinger equation, 
\begin{align*}
	\ii \frac{\dd }{\dd t} \psi(t) = H \psi(t) 
	, 
\end{align*}
also suggests that strong continuity/differentiability is the correct notion. Let us once more consider the free hamiltonian $H = - \frac{1}{2} \Delta_x$ on $L^2(\R^n)$. We will show that its domain is 
\begin{align*}
	\mathcal{D}(H) = \bigl \{ \varphi \in L^2(\R^n) \; \vert \; - \Delta_x \varphi \in L^2(\R^n) \bigr \} 
	. 
\end{align*}
In Chapter~\ref{S_and_Sprime}, we will learn that $\mathcal{D}(H)$ is mapped by the Fourier transform onto 
\begin{align*}
	\mathcal{D}(\hat{H}) = \bigl \{ \widehat{\psi} \in L^2(\R^n) \; \vert \; \hat{k}^2 \widehat{\psi} \in L^2(\R^n) \bigr \} 
	. 
\end{align*}
Dominated Convergence can once more be used to make the following claims rigorous: for any $\widehat{\psi} \in \mathcal{D}(\hat{H})$, we have 
\begin{align}
	\lim_{t \to 0} &\Bnorm{\tfrac{\ii}{t} \bigl ( \hat{U}(t) - \id \bigr ) \widehat{\psi} - \tfrac{1}{2} \hat{k}^2 \widehat{\psi}} 
	\leq \lim_{t \to 0} \Bnorm{\tfrac{\ii}{t} \bigl ( \hat{U}(t) - \id \bigr ) \widehat{\psi}} +  \bnorm{\tfrac{1}{2} \hat{k}^2 \widehat{\psi}}
	\label{operators:unitary:eqn:free_evolution_strongly_continuous}
	. 
\end{align}
The second term is finite since $\widehat{\psi} \in \mathcal{D}(\hat{H})$ and we have to focus on the first term. On the level of functions, 
\begin{align*}
	\lim_{t \to 0} \tfrac{\ii}{t} \bigl ( \e^{- \ii t \frac{1}{2} k^2} - 1 \bigr ) 
	= \ii \frac{\dd}{\dd t} \e^{- \ii t \frac{1}{2} k^2} \Big \vert_{t = 0} 
	= \tfrac{1}{2} k^2 
\end{align*}
holds pointwise. Furthermore, by the mean value theorem, for any finite $t \in \R$ with $\abs{t} \leq 1$, for instance, then there exists $0 \leq t_0 \leq t$ such that \marginpar{2013.10.29}
\begin{align*}
	\tfrac{1}{t} \bigl ( \e^{- \ii t \frac{1}{2} k^2} - 1 \bigr ) = \partial_t \e^{- \ii t \frac{1}{2} k^2} \big \vert_{t = t_0} 
	= - \ii \, \tfrac{1}{2} k^2 \, \e^{- \ii t_0 \frac{1}{2} k^2} 
	. 
\end{align*}
This can be bounded uniformly in $t$ by $\frac{1}{2} k^2$. Thus, also the first term can be bounded by $\bnorm{\frac{1}{2} \hat{k}^2 \widehat{\psi}}$ uniformly. By Dominated Convergence, we can interchange the limit $t \to 0$ and integration with respect to $k$ on the left-hand side of equation~\eqref{operators:unitary:eqn:free_evolution_strongly_continuous}. But then the integrand is zero and thus the domain where the free evolution group is differentiable coincides with the domain of the Fourier transformed hamiltonian, 
\begin{align*}
	\lim_{t \to 0} \norm{\tfrac{\ii}{t} \bigl ( \hat{U}(t) - \id) \widehat{\psi} - \tfrac{1}{2} \hat{k}^2 \widehat{\psi}} = 0 
	. 
\end{align*}
This suggests to use the following definition: 
\begin{definition}[Generator of a unitary group]
	A densely defined linear operator on a Hilbert space $\Hil$ with domain $\mathcal{D}(H) \subseteq \Hil$ is called generator of a \emph{unitary evolution group} $U(t)$, $t \in \R$, if 
	\begin{enumerate}[(i)]
		\item the domain coincides with 
		\begin{align*}
			\widetilde{\mathcal{D}(H)} = \Bigl \{ \varphi \in \Hil \; \big \vert \;  t \mapsto U(t) \varphi \mbox{ differentiable} \Bigr \} = \mathcal{D}(H)
		\end{align*}
		\item and for all $\psi \in \mathcal{D}(H)$, the Schrödinger equation holds, 
		\begin{align*}
			\ii \frac{\dd}{\dd t} U(t) \psi = H U(t) \psi 
			. 
		\end{align*}
	\end{enumerate}
\end{definition}
This is only one of the two implications: usually we are given a hamiltonian $H$ and we would like to know under which circumstances this operator generates a unitary evolution group. We will answer this question conclusively in the next section with Stone's Theorem. 
\begin{theorem}
	Let $H$ be the generator of a strongly continuous evolution group $U(t)$, $t \in \R$. Then the following holds: 
	\begin{enumerate}[(i)]
		\item $\mathcal{D}(H)$ is invariant under the action of $U(t)$, \ie $U(t) \mathcal{D}(H) = \mathcal{D}(H)$ for all $t \in \R$. 
		\item $H$ commutes with $U(t)$, \ie $[U(t) , H] \psi := U(t) \, H \psi - H \, U(t) \psi = 0$ for all $t \in \R$ and $\psi \in \mathcal{D}(H)$. 
		\item $H$ is \emph{symmetric}, \ie $\scpro{H \varphi}{\psi} = \scpro{\varphi}{H \psi}$ holds for all $\varphi , \psi \in \mathcal{D}(H)$. 
		\item $U(t)$ is uniquely determined by $H$. 
		\item $H$ is uniquely determined by $U(t)$. 
	\end{enumerate}
\end{theorem}
\begin{proof}
	\begin{enumerate}[(i)]
		\item Let $\psi \in \mathcal{D}(H)$. To show that $U(t) \psi$ is still in the domain, we have to show that the norm of $H U(t) \psi$ is finite. Since $H$ is the generator of $U(t)$, it is equal to 
		\begin{align*}
			H \psi = \ii \frac{\dd}{\dd s} U(s) \psi \bigg \vert_{s = 0} = \lim_{s \to 0} \tfrac{\ii}{s} \bigl ( U(s) - \id \bigr ) \psi 
			. 
		\end{align*}
		Let us start with $s > 0$ and omit the limit. Then 
		\begin{align*}
			\Bnorm{\tfrac{\ii}{s} \bigl ( U(s) - \id \bigr ) U(t) \psi} &= \Bnorm{U(t) \tfrac{\ii}{s} \bigl ( U(s) - \id \bigr ) \psi} 
			= \Bnorm{\tfrac{\ii}{s} \bigl ( U(s) - \id \bigr ) \psi} < \infty 
		\end{align*}
		holds for all $s > 0$. Taking the limit on left and right-hand side yields that we can estimate the norm of $H U(t) \psi$ by the norm of $H \psi$ -- which is finite since $\psi$ is in the domain. This means $U(t) \mathcal{D}(H) \subseteq \mathcal{D}(H)$. To show the converse, we repeat the proof for $U(-t) = U(t)^{-1} = U(t)^*$ to obtain 
		\begin{align*}
			\mathcal{D}(H) = U(-t) U(t) \mathcal{D}(H) \subseteq U(t) \mathcal{D}(H) 
			. 
		\end{align*}
		Hence, $U(t) \mathcal{D}(H) = \mathcal{D}(H)$. 
		\item This follows from an extension of the proof of (i): since the domain $\mathcal{D}(H)$ coincides with the set of vectors on which $U(t)$ is strongly differentiable and is left invariant by $U(t)$, taking limits on left- and right-hand side of 
		\begin{align*}
			\Bnorm{\tfrac{\ii}{s} \bigl ( U(s) - \id \bigr ) U(t) \psi - U(t) \tfrac{\ii}{s} \bigl ( U(s) - \id \bigr ) \psi} = 0 
		\end{align*}
		leads to $[H,U(t)] \psi = 0$. 
		\item This follows from differentiating $\scpro{U(t) \varphi}{U(t) \psi}$ for arbitrary $\varphi , \psi \in \mathcal{D}(H)$ and using $\bigl [ U(t) , H \bigr ] = 0$ as well as the unitarity of $U(t)$ for all $t \in \R$. 
		\item Assume that both unitary evolution groups, $U(t)$ and $\tilde{U}(t)$, have $H$ as their generator. For any $\psi \in \mathcal{D}(H)$, we can calculate the time derivative of $\bnorm{(U(t) - \tilde{U}(t)) \psi}^2$, 
		\begin{align*}
			\frac{\dd }{\dd t} \Bnorm{\bigl ( U(t) - \tilde{U}(t) \bigr ) \psi}^2 &= 2 \frac{\dd }{\dd t} \bigl ( \norm{\psi}^2 - \Re \bscpro{U(t) \psi}{\tilde{U}(t) \psi} \bigr ) 
			\\
			&= - 2 \Re \Bigl ( \bscpro{- \ii \, H U(t) \psi}{\tilde{U}(t) \psi} + \bscpro{U(t) \psi}{- \ii \, H \tilde{U}(t) \psi} \Bigr ) 
			\\
			&= 0 
			. 
		\end{align*}
		Since $U(0) = \id = \tilde{U}(0)$, this means $U(t)$ and $\tilde{U}(t)$ agree at least on $\mathcal{D}(H)$. Using the fact that there is only bounded extension of a bounded operator to all of $\Hil$, Theorem~\ref{operators:bounded:thm:extensions_bounded_operators}, we conclude they must be equal on all of $\Hil$. 
		\item This follows from the definition of the generator and the density of the domain. 
	\end{enumerate}
\end{proof}
Now that we have collected a few facts on unitary evolution groups, one could think that \emph{symmetric} operators generate evolution groups, but \emph{this is false!} The standard example to showcase this fact is the group of translations on $L^2([0,1])$. Since we would like $T(t)$ to conserve ``mass'' -- or more accurately, probability, we define for $\varphi \in L^2([0,1])$ and $0 \leq t < 1$
\begin{align*}
	\bigl ( T(t) \varphi \bigr )(x) := \left \{
	\begin{matrix}
		\varphi(x - t) & x - t \in [0,1] \\
		\varphi(x - t + 1) & x - t + 1 \in [0,1] \\
	\end{matrix}
	\right . 
	. 
\end{align*}
For all other $t \in \R$, we extend this operator periodically, \ie we plug in the fractional part of $t$. Clearly, $\bscpro{T(t) \varphi}{T(t) \psi} = \bscpro{\varphi}{\psi}$ holds for all $\varphi , \psi \in L^2([0,1])$. Locally, the infinitesimal generator is $- \ii \partial_x$ as a simple calculation shows: 
\begin{align*}
	\biggl ( \ii \frac{\dd}{\dd t} \bigl ( T(t) \varphi \bigr ) \biggr )(x) \bigg \vert_{t = 0} &= \ii \frac{\dd}{\dd t} \varphi(x - t) \bigg \vert_{t = 0} 
	= - \ii \partial_x \varphi(x) 
\end{align*}
However, $T(t)$ does not preserve the maximal domain of $- \ii \partial_x$, 
\begin{align*}
	\mathcal{D}_{\max}(- \ii \partial_x) = \bigl \{ \varphi \in L^2([0,1]) \; \vert \; -\ii \partial_x \varphi \in L^2([0,1]) \bigr \} 
	. 
\end{align*}
Any element of the maximal domain has a continuous representative, but if $\varphi(0) \neq \varphi(1)$, then for $t > 0$, $T(t) \varphi$ will have a discontinuity at $t$. We will denote the operator $- \ii \partial_x$ on $\mathcal{D}_{\max}(- \ii \partial_x)$ with $\mathsf{P}_{\max}$. Let us check whether $\mathsf{P}_{\max}$ is symmetric: for any $\varphi , \psi \in \mathcal{D}_{\max}(- \ii \partial_x)$, we compute 
\begin{align}
	\bscpro{\varphi}{-\ii \partial_x \psi} &= \int_0^1 \dd x \, \varphi^*(x) \, (- \ii \partial_x \psi)(x) 
	= \Bigl [ - \ii \varphi^*(x) \, \psi(x) \Bigr ]_0^1 - \int_0^1 \dd x \, (-\ii) \partial_x \varphi^*(x) \, \psi(x) 
	\notag \\
	&
	= \ii \bigl ( \varphi^*(0) \, \psi(0) - \varphi^*(1) \, \psi(1) \bigr ) + \int_0^1 \dd x \, (- \ii \partial_x \varphi)^*(x) \, \psi(x) 
	\notag \\
	&= \ii \bigl ( \varphi^*(0) \, \psi(0) - \varphi^*(1) \, \psi(1) \bigr ) + \bscpro{- \ii \partial_x \varphi}{\psi} 
	\label{operators:unitary:eqn:symmetry_translations_interval}
	. 
\end{align}
In general, the boundary terms do not disappear and the maximal domain is “too large” for $- \ii \partial_x$ to be symmetric. Thus, it is not at all surprising, $T(t)$ does not leave $\mathcal{D}_{\max}(- \ii \partial_x)$ invariant. Let us try another domain: one way to make the boundary terms disappear is to choose 
\begin{align*}
	\mathcal{D}_{\min}(- \ii \partial_x) := \Bigl \{ \varphi \in L^2([0,1]) \; \big \vert \; - \ii \partial_x \varphi \in L^2([0,1]) , \; \varphi(0) = 0 = \varphi(1) \Bigr \} 
	. 
\end{align*}
We denote $- \ii \partial_x$ on this “minimal” domain with $\mathsf{P}_{\min}$. In this case, the boundary terms in equation~\eqref{operators:unitary:eqn:symmetry_translations_interval} vanish which tells us that $\mathsf{P}_{\min}$ is symmetric. Alas, the domain is still not invariant under translations $T(t)$, even though $\mathsf{P}_{\min}$ is symmetric. This is an example of a symmetric operator which \emph{does not} generate a unitary group. 

There is another thing we have missed so far: the translations allow for an additional phase factor, \ie for $\varphi , \psi \in L^2([0,1])$ and $\vartheta \in [0,2\pi)$, we define for $0 \leq t < 1$
\begin{align*}
	\bigl ( T_{\vartheta}(t) \varphi \bigr )(x) := \left \{
	\begin{matrix}
		\varphi(x - t) & x - t \in [0,1] \\
		\e^{\ii \vartheta} \varphi(x - t + 1) & x - t + 1 \in [0,1] \\
	\end{matrix}
	\right . 
	. 
\end{align*}
while for all other $t$, we plug in the fractional part of $t$. The additional phase factor cancels in the inner product, $\bscpro{T_{\vartheta}(t) \varphi}{T_{\vartheta}(t) \psi} = \bscpro{\varphi}{\psi}$ still holds true for all $\varphi , \psi \in L^2([0,1])$. In general $T_{\vartheta}(t) \neq T_{\vartheta'}(t)$ if $\vartheta \neq \vartheta'$ and the unitary groups are genuinely different. Repeating the simple calculation from before, we see that the local generator still is $- \ii \partial_x$ and it would seem we can generate a family of unitary evolutions from a \emph{single} generator. The confusion is resolved if we focus on \emph{invariant domains}: choosing $\vartheta \in [0,2\pi)$, we define $\mathsf{P}_{\vartheta}$ to be the operator $- \ii \partial_x$ on the domain 
\begin{align*}
	\mathcal{D}_{\vartheta}(- \ii \partial_x) := \Bigl \{ \varphi \in L^2([0,1]) \; \big \vert \; - \ii \partial_x \varphi \in L^2([0,1]) , \; \varphi(0) = \e^{- \ii \vartheta} \varphi(1) \Bigr \} 
	. 
\end{align*}
A quick look at equation~\eqref{operators:unitary:eqn:symmetry_translations_interval} reassures us that $\mathsf{P}_{\vartheta}$ is symmetric and a quick calculation shows it is also \emph{invariant} under the action of $T_{\vartheta}(t)$. Hence, $\mathsf{P}_{\vartheta}$ is the generator of $T_{\vartheta}$, and the \emph{definition of an unbounded operator is incomplete without spelling out its domain}. 
\begin{example}[The wave equation with boundary conditions]
	% TODO expand on this example
	Another example where the domain is crucial in the properties is the wave equation on $[0,L]$, 
	\begin{align*}
		\partial_t^2 u(x,t) - \partial_x^2 u(x,t) = 0 
		, 
		&&
		u \in \Cont^2([0,L] \times \R)
		. 
	\end{align*}
	Here, $u$ is the amplitude of the vibration, \ie the lateral deflection. If we choose Dirichlet boundary conditions at both ends, \ie $u(0) = 0 = u(L)$, we model a closed pipe, if we choose Dirichlet boundary conditions on one end, $u(0) = 0$, and von Neumann boundary conditions on the other, $u'(L) = 0$, we model a half-closed pipe. Choosing domains is a question of physics! 
\end{example}
%
% section unitary_operators (end)

\section{Selfadjoint operators} % (fold)
\label{operators:selfadjoint_operators}
Although we do not have time to explore this very far, the crucial difference between $\mathsf{P}_{\min}$ and $\mathsf{P}_{\vartheta}$ is that the former is only symmetric while the latter is also selfadjoint. We first recall the definition of the adjoint of a possibly unbounded operator: 
\begin{definition}[Adjoint operator]
	Let $A$ be a densely defined linear operator on a Hilbert space $\Hil$ with domain $\mathcal{D}(A)$. Let $\mathcal{D}(A^*)$ be the set of $\varphi \in \Hil$ for which there exists $\phi \in \Hil$ with 
	\begin{align*}
		\scpro{A \psi}{\varphi} = \scpro{\psi}{\phi} 
		&& 
		\forall \psi \in \mathcal{D}(A) 
		. 
	\end{align*}
	For each $\varphi \in \mathcal{D}(A^*)$, we define $A^* \varphi := \phi$ and $A^*$ is called the adjoint of $A$. 
\end{definition}
\begin{remark}
	By Riesz Lemma, $\varphi$ belongs to $\mathcal{D}(A^*)$ if and only if 
	\begin{align*}
		\babs{\scpro{A \psi}{\varphi}} \leq C \norm{\psi} 
		&&
		\forall \psi \in \mathcal{D}(A) 
		. 
	\end{align*}
	This is equivalent to saying $\varphi \in \mathcal{D}(A^*)$ if and only if $\psi \mapsto \sscpro{A \psi}{\varphi}$ is continuous on $\mathcal{D}(A)$. As a matter of fact, we could have used to latter to \emph{define} the adjoint operator. 
\end{remark}
One word of caution: even if $A$ is densely defined, $A^*$ need not be. 
\begin{example}
	Let $f \in L^{\infty}(\R)$, but $f \not\in L^2(\R)$, and pick $\psi_0 \in L^2(\R)$. Define 
	\begin{align*}
		\mathcal{D}(T_f) := \Bigl \{ \psi \in L^2(\R) \; \vert \; \int_{\R} \dd x \, \abs{f(x) \, \psi(x)} < \infty \Bigr \} 
		. 
	\end{align*}
	Then the adjoint of the operator 
	\begin{align*}
		T_f \psi := \sscpro{f}{\psi} \, \psi_0 
		, 
		&&
		\psi \in \mathcal{D}(T_f) 
		, 
	\end{align*}
	has domain $\mathcal{D}(T_f^*) = \{ 0 \}$. Let $\psi \in \mathcal{D}(T_f)$. Then for any $\varphi \in \mathcal{D}(T_f^*)$ 
	\begin{align*}
		\bscpro{T_f \psi}{\varphi} &= \bscpro{\sscpro{f}{\psi} \, \psi_0}{\varphi} = \bscpro{\psi}{f} \, \bscpro{\psi_0}{\varphi} 
		\\
		&= \bscpro{\psi}{\sscpro{\psi_0}{\varphi} f} 
		. 
	\end{align*}
	Hence $T_f^* \varphi = \sscpro{\psi_0}{\varphi} f$. However $f \not\in L^2(\R)$ and thus $\varphi = 0$ is the only possible choice for which $T_f^* \varphi$ is well defined. 
\end{example}
Symmetric operators, however, are special: since $\scpro{H \varphi}{\psi} = \scpro{\varphi}{H \psi}$ holds by definition for all $\varphi , \psi \in \Hil$, the domain of $H^*$ is contained in that of $H$, $\mathcal{D}(H^*) \supseteq \mathcal{D}(H)$. In particular, $\mathcal{D}(H^*) \subseteq \Hil$ is also dense. Thus, $H^*$ is an \emph{extension of $H$}. 
\begin{definition}[Selfadjoint operator]
	Let $H$ be a symmetric operator on a Hilbert space $\Hil$ with domain $\mathcal{D}(H)$. $H$ is called selfadjoint, $H^* = H$, iff $\mathcal{D}(H^*) = \mathcal{D}(H)$. 
\end{definition}
One word regarding notation: if we write $H^* = H$, we do not just imply that the “operator prescription” of $H$ and $H^*$ is the same, but that the \emph{domains} of both coincide. 
\begin{example}
	In this sense, $\mathsf{P}_{\min}^* \neq \mathsf{P}_{\min}$. 
\end{example}
The central theorem of this section is Stone's Theorem: 
\begin{theorem}[Stone]
	 To every strongly continuous one-parameter unitary group $U$ on a Hilbert space $\Hil$, there exists a selfadjoint operator $H = H^*$ which generates $U(t) = \e^{- \ii t H}$. Conversely, every selfadjoint operator $H$ generates the unitary evolution group $U(t) = \e^{- \ii t H}$. 
\end{theorem}
A complete proof \cite[Chapter~VIII.3]{Reed_Simon:M_cap_Phi_1:1972} is beyond our capabilities. \marginpar{2013.10.31}
% section selfadjoint_operators (end)

\section{Recasting the Maxwell equations as a Schrödinger equation} % (fold)
\label{operators:Maxwell}
The Maxwell equations in a medium with electric permittivity $\eps$ and magnetic permeability $\mu$ are given by the two \emph{dynamical} 
\begin{subequations}
	\label{operators:eqn:dynamical_Maxwell}
	\begin{align}
		\partial_t \mathbf{E}(t) &= + \eps^{-1} \, \nabla_x \times \mathbf{H}(t)
		\\
		\partial_t \mathbf{H}(t) &= - \mu^{-1} \, \nabla_x \times \mathbf{E}(t)
		% \notag 
	\end{align}
\end{subequations}
and the two \emph{kinetic Maxwell equations}
\begin{subequations}
	\label{operators:eqn:source_Maxwell_eqns}
	\begin{align}
		\nabla_x \cdot \eps \mathbf{E}(t) &= \rho 
		\\
		\nabla_x \cdot \mu \mathbf{H}(t) &= j 
		. 
		% \notag 
	\end{align}
\end{subequations}
Here, the source terms in the kinetic equations are the \emph{charge density} $\rho$ and the \emph{current density} $j$; in the absence of sources, $\rho = 0$ and $j = 0$, the Maxwell equations are homogeneous. 

We can rewrite the dynamical equations~\eqref{operators:eqn:dynamical_Maxwell} as a Schrödinger-type equation, 
\begin{align}
	\ii \frac{\dd}{\dd t} \left (
	\begin{matrix}
		\mathbf{E}(t) \\
		\mathbf{H}(t) \\
	\end{matrix}
	\right ) &= M(\eps,\mu) \left (
	\begin{matrix}
		\mathbf{E}(t) \\
		\mathbf{H}(t) \\
	\end{matrix}
	\right )
	\, , 
	&&
	\left (
	\begin{matrix}
		\mathbf{E}(0) \\
		\mathbf{H}(0) \\
	\end{matrix}
	\right ) = \left (
	\begin{matrix}
		\mathbf{E}^{(0)} \\
		\mathbf{H}^{(0)} \\
	\end{matrix}
	\right ) \in L^2(\R^3,\C^6)
	\, , 
\end{align}
where the \emph{Maxwell operator}
\begin{align}
	M(\eps,\mu) := W \, \Rot 
	:= \left (
	\begin{matrix}
		\eps^{-1} & 0 \\
		0 & \mu^{-1} \\
	\end{matrix}
	\right ) \left (
	\begin{matrix}
		0 & + \ii \nabla_x^{\times} \\
		- \ii \nabla_x^{\times} & 0 \\
	\end{matrix}
	\right )
\end{align}
takes the role of the Schrödinger operator $H = - \Delta_x + V$. It can be conveniently written as the product of the multiplication operator $W$ which contains the \emph{material weights} $\eps$ and $\mu$, and the \emph{free Maxwell operator} $\Rot$. Here, $\nabla_x^{\times}$ is just a short-hand for the curl, $\nabla_x^{\times} \mathbf{E} := \nabla_x \times \mathbf{E}$. The solution can now be expressed just like in the Schrödinger case, 
\begin{align*}
	\left (
	\begin{matrix}
		\mathbf{E}(t) \\
		\mathbf{H}(t) \\
	\end{matrix}
	\right ) = \e^{- \ii t M(\eps,\mu)} \left (
	\begin{matrix}
		\mathbf{E}^{(0)} \\
		\mathbf{H}^{(0)} \\
	\end{matrix}
	\right ) 
	\, , 
\end{align*}
where the initial conditions must satisfy the no sources condition (equations~\eqref{operators:eqn:source_Maxwell_eqns} for $\rho = 0$ and $j = 0$); one can show that this is enough to ensure that also the time-evolved fields $\mathbf{E}(t)$ and $\mathbf{H}(t)$ satisfy the no sources conditions for all times. 

Physically, the condition that $\mathbf{E}$ and $\mathbf{H}$ be square-integrable stems from the requirement that the \emph{field energy} 
\begin{align}
	\mathcal{E}(\mathbf{E},\mathbf{H}) := \frac{1}{2} \int_{\R^3} \dd x \Bigl ( \eps(x) \, \babs{\mathbf{E}(x)}^2 + \mu(x) \, \babs{\mathbf{H}(x)}^2 \Bigr ) 
\end{align}
be finite; Moreover, the field energy is a \emph{conserved quantity}, 
\begin{align*}
	\mathcal{E} \bigl ( \mathbf{E}(t) , \mathbf{H}(t) \bigr ) = \mathcal{E} \bigl ( \mathbf{E}^{(0)} , \mathbf{H}^{(0)} \bigr )
	. 
\end{align*}
It is not coincidental that the expression for $\mathcal{E}$ looks like the square of a weighted $L^2$-norm: if we assume that $\eps , \mu \in L^{\infty}(\R^3)$ are bounded away from $0$ and $+\infty$, \ie there exist $c , C > 0$ for which 
\begin{align*}
	0 < c \leq \eps(x) , \mu(x) \leq C < + \infty 
\end{align*}
holds almost everywhere in $x \in \R^3$, then $\eps^{-1}$ and $\mu^{-1}$ are also bounded away from $0$ and $+\infty$ in the above sense. Hence, we deduce 
\begin{align*}
	\Psi = (\psi^E , \psi^H) \in L^2(\R^3,\C^6) 
	\; \; \Longleftrightarrow \; \; 
	\Psi \in \Hil(\eps,\mu) := L^2_{\eps}(\R^3,\C^3) \oplus L^2_{\mu}(\R^3,\C^3)
\end{align*}
where $L^2_{\eps}(\R^3,\C^3)$ and $L^2_{\mu}(\R^3,\C^3)$ are defined analogously to problem~22. By definition, $\Psi$ is an element of $\Hil(\eps,\mu)$ if and only if the norm $\norm{\Psi}_{\Hil(\eps,\mu)}$ induced by the weighted scalar product 
\begin{align}
	\scpro{\Psi}{\Phi}_{\Hil(\eps,\mu)} :& \negmedspace= \bscpro{(\psi^E,\psi^H)}{(\phi^E,\phi^H)}_{\Hil(\eps,\mu)}
	= \scpro{\psi^E}{\phi^E}_{\eps} + \scpro{\psi^H}{\phi^H}_{\mu}
	\\
	:& \negmedspace= \int_{\R^3} \dd x \, \eps(x) \, \psi^E(x) \cdot \phi^E(x) + \int_{\R^3} \dd x \, \mu(x) \, \psi^H(x) \cdot \phi^H(x) 
	\notag 
\end{align}
is finite. Adapting the arguments from problem~22, we conclude that $L^2(\R^3,\C^6)$ and $\Hil(\eps,\mu)$ can be canonically identified as \emph{Banach} spaces. Now the field energy can be expressed as 
\begin{align*}
	\mathcal{E}(\mathbf{E},\mathbf{H}) = \tfrac{1}{2} \, \bnorm{(\mathbf{E},\mathbf{H})}_{\Hil(\eps,\mu)}^2 
	:= \tfrac{1}{2} \bscpro{(\mathbf{E},\mathbf{H})}{(\mathbf{E},\mathbf{H})}_{\Hil(\eps,\mu)} 
	\, , 
\end{align*}
and the conservation of field energy suggests that $\e^{- \ii t M(\eps,\mu)}$ is unitary with respect to the \emph{weighted} scalar product $\scpro{\cdot \,}{\cdot}_{\Hil(\eps,\mu)}$. 

Indeed, this is the case: the scalar product can alternatively be expressed in terms of $W^{-1}$ and the unweighted scalar product on $L^2(\R^3,\C^6)$, 
\begin{align}
	\scpro{\Psi}{\Phi}_{\Hil(\eps,\mu)} &= \bscpro{\Psi}{W^{-1} \Phi}_{L^2(\R^3,\C^6)} 
	= \bscpro{W^{-1} \Psi}{\Phi}_{L^2(\R^3,\C^6)} 
	. 
	\label{operators:eqn:weighted_scalar_product_W_unweighted}
\end{align}
The last equality holds true, because $W^{\pm 1}$ are multiplication operators with scalar real-valued functions in the electric and magnetic component, and thus 
\begin{align*}
	\scpro{\psi^E}{\eps \phi^E}_{L^2(\R^3,\C^3)} = \scpro{\eps \psi^E}{\phi^E}_{L^2(\R^3,\C^3)}
\end{align*}
holds for all $\psi^E , \phi^E \in L^2(\R^3,\C^3)$, for instance. Under the assumption that the free Maxwell $\Rot$ is selfadjoint on $L^2(\R^3,\C^6)$, then one can also show the selfadjointness of the Maxwell operator $M(\eps,\mu) = W \, \Rot$ by using its product structure and equation~\eqref{operators:eqn:weighted_scalar_product_W_unweighted}: 
\begin{align*}
	\bscpro{\Psi}{M(\eps,\mu) \, \Phi}_{\Hil(\eps,\mu)} &= \bscpro{\Psi}{W^{-1} \, W \, \Rot \, \Phi}_{L^2(\R^3,\C^6)} 
	= \bscpro{\Psi}{\Rot \, \Phi}_{L^2(\R^3,\C^6)} 
	\\
	&= \bscpro{\Rot \, \Psi}{\Phi}_{L^2(\R^3,\C^6)} 
	= \bscpro{W^{-1} \, W \, \Rot \, \Psi}{\Phi}_{L^2(\R^3,\C^6)} 
	\\
	&
	= \bscpro{M(\eps,\mu) \, \Psi}{W^{-1} \, \Phi}_{L^2(\R^3,\C^6)} 
	= \bscpro{M(\eps,\mu) \, \Psi}{\Phi}_{\Hil(\eps,\mu)}
\end{align*}
These arguments \emph{imply} that $\e^{- \ii t M(\eps,\mu)}$ is unitary with respect to $\scpro{\cdot \,}{\cdot}_{\Hil(\eps,\mu)}$, and thus, we deduce that the dynamics conserve energy, 
\begin{align*}
	\mathcal{E} \bigl ( \mathbf{E}(t) , \mathbf{H}(t) \bigr ) &= \frac{1}{2} \, \Bnorm{\e^{- \ii t M(\eps,\mu)} \bigl ( \mathbf{E}^{(0)} , \mathbf{H}^{(0)} \bigr )}_{\Hil(\eps,\mu)}^2 
	\\
	&
	= \frac{1}{2} \, \Bnorm{\bigl ( \mathbf{E}^{(0)} , \mathbf{H}^{(0)} \bigr )}_{\Hil(\eps,\mu)}^2 
	= \mathcal{E} \bigl ( \mathbf{E}^{(0)} , \mathbf{H}^{(0)} \bigr ) 
	. 
\end{align*}
Moreover, the formulation of the Maxwell equations as a Schrödinger equation also allows us to prove that the dynamics $\e^{- \ii t M(\eps,\mu)}$ map real-valued fields onto real-valued fields: define complex conjugation $(C \Psi)(x) := \overline{\Psi(x)}$ on $\Hil(\eps,\mu)$ component-wise. Then the fact that $\eps$ and $\mu$ are real-valued implies $C$ commutes with $W$, 
\begin{align*}
	\bigl ( C \, W \, \Psi \bigr )(x) &= \overline{\bigl ( W \, (\psi^E,\psi^H) \bigr )(x)} 
	= \overline{\bigl ( \eps^{-1}(x) \, \psi^E(x) \, , \, \mu^{-1}(x) \, \phi^H(x) \bigr )}
	\\
	&= \Bigl ( \eps^{-1}(x) \, \overline{\psi^E(x)} \, , \, \mu^{-1}(x) \, \overline{\phi^H(x)} \Bigr )
	= \bigl ( W \, C \, \Psi \bigr )(x) 
\end{align*}
In problem~23, we have shown that $C \, \Rot \, C = - \Rot$, and thus 
\begin{align*}
	C \, M(\eps,\mu) \, C &= C \, W \, \Rot \, C 
	= W \, C \, \Rot \, C 
	\\
	&= - W \, \Rot 
	= - M(\eps,\mu) 
\end{align*}
holds just as in the case of the free Maxwell equations. This means the unitary evolution operator and complex conjugation \emph{commute,} 
\begin{align*}
	C \, \e^{- \ii t M(\eps,\mu)} \, C &= \e^{+ \ii t \, C \, M(\eps,\mu) \, C} 
	= \e^{- \ii t M(\eps,\mu)} 
	\, , 
\end{align*}
as does the real part operator $\Re := \tfrac{1}{2} \bigl ( \id_{\Hil(\eps,\mu)} + C \bigr )$, 
\begin{align*}
	\Bigl [ \e^{- \ii t M(\eps,\mu)} , \Re \Bigr ] &= 0 
	\, , 
\end{align*}
Now if the initial state $\bigl ( \mathbf{E}^{(0)} , \mathbf{H}^{(0)} \bigr ) = \Re \, \bigl ( \mathbf{E}^{(0)} , \mathbf{H}^{(0)} \bigr )$ is real-valued, then so is the time-evolved state, 
\begin{align*}
	\bigl ( \mathbf{E}(t) , \mathbf{H}(t) \bigr ) &= \e^{- \ii t M(\eps,\mu)} \, \bigl ( \mathbf{E}^{(0)} , \mathbf{H}^{(0)} \bigr ) 
	\\
	&
	= \e^{- \ii t M(\eps,\mu)} \, \Re \, \bigl ( \mathbf{E}^{(0)} , \mathbf{H}^{(0)} \bigr ) 
	\\
	&
	= \Re \, \e^{- \ii t M(\eps,\mu)} \, \bigl ( \mathbf{E}^{(0)} , \mathbf{H}^{(0)} \bigr ) 
	\\
	&
	= \Re \, \bigl ( \mathbf{E}(t) , \mathbf{H}(t) \bigr )
	. 
\end{align*}
The reformulation of the Maxwell equations as a Schrödinger-type equation was first made rigorous by \cite{Birman_Solomyak:L2_theory_Maxwell_operator:1987}; it allows to adapt and apply many of the techniques first developed for the analysis Schrödinger operators to Maxwell operators, \eg under suitable conditions one can derive ray optics as a “semiclassical limit” of the above equations. \marginpar{2013.11.05}
% section Recasting the Maxwell equations as a Schrödinger equation (end)
% chapter application_magnetic_space_adiabatic_perturbation_theory (end)
\chapter{The Fourier transform} % (fold)
\label{Fourier}
The wave equation, the free Schrödinger equation and the heat equation all admit the same class of “fundamental solutions”, namely exponential functions $\e^{- \ii \xi \cdot x}$. In some cases, the boundary conditions impose restrictions on the admissible values of $\xi$. 

This is because these equations share a common \emph{symmetry}, namely \emph{invariance under translations} (we will be more precise in Chapters~\ref{Fourier:T:periodic_operators} and \ref{Fourier:R:heat}--\ref{Fourier:R:Schroedinger}), and the Fourier transform converts a \textbf{P}DE into an \textbf{O}DE. Moreover, certain properties of the function are tied to certain properties of the Fourier transform, the most famous being that the regularity of $f$ is linked to the decay rate of $\Fourier f$.

\section{The Fourier transform on $\T^n$} % (fold)
\label{Fourier:T}
Let us consider the Fourier transform on the torus $\T^n := (\Sone)^n$ which will be identified with $[-\pi,+\pi]^n$. Moreover, we will view functions on $\T^n$ with $2\pi \Z^n$-periodic functions on $\R^n$ whenever convenient. Now we proceed to define the central notion of this section: 
\begin{definition}[Fourier transform on $\T^n$]\label{Fourier:T:defn:Fourier_transform}
	For all $f \in L^1(\T^n)$, we set 
	\begin{align*}
		(\Fourier f)(k) := \hat{f}(k) := \frac{1}{(2\pi)^n} \int_{\T^n} \dd x \, \e^{- \ii k \cdot x} \, f(x) 
	\end{align*}
	for $k \in \Z^n$. The \emph{Fourier series} is the formal sum 
	\begin{align}
		\sum_{k \in \Z^n} \hat{f}(k) \, \e^{+ \ii k \cdot x} 
		. 
		\label{Fourier:eqn:T_Fourier_series}
	\end{align}
\end{definition}
If all we know is that $f$ is integrable, then the question on whether the Fourier series exists turns out to be surprisingly hard. In fact, Kolmogorov has shown that there exist integrable functions for which the Fourier series diverges for almost all $x \in \T^n$. However, if we have additional information on $f$, \eg if $f \in \Cont^r(\T^n)$ for $r \geq n+1$, then we can show that \eqref{Fourier:eqn:T_Fourier_series} exists as an absolutely convergent sum. 
\begin{example}
	To compute the Fourier coefficients for $f(x) = x \in L^1([-\pi,+\pi])$, we need to distinguish the cases $k = 0$, 
	\begin{align*}
		(\Fourier x)(0) &= \frac{1}{2\pi} \int_{-\pi}^{+\pi} \dd x \, x 
		= \left [ \frac{1}{4\pi} x^2 \right ]_{-\pi}^{+\pi} 
		= 0
		, 
	\end{align*}
	and $k \neq 0$, 
	\begin{align*}
		(\Fourier x)(k) &= \frac{1}{2\pi} \int_{-\pi}^{+\pi} \dd x \, \e^{- \ii k x} \, x 
		\\
		&= \left [ \frac{\ii}{2\pi k} \, x \, \e^{- \ii k x} \right ]_{-\pi}^{+\pi} - \frac{\ii}{k} \int_{-\pi}^{+\pi} \dd x \, \e^{- \ii k x} \cdot 1 
		\\
		&= (-1)^k \, \frac{\ii}{k} 
		. 
	\end{align*}
	Thus, the Fourier coefficients decay like $\nicefrac{1}{\abs{k}}$ for large $\abs{k}$, 
	\begin{align*}
		(\Fourier x)(k) &= 
		\begin{cases}
			0 & k = 0 \\
			(-1)^k \, \frac{\ii}{k} & k \in \Z \setminus \{ 0 \} \\
		\end{cases}
		. 
	\end{align*}
	We will see later on that this is because $f(x) = x$ has a discontinuity at $x = +\pi$ (which is identified with the point $x = -\pi$). 
\end{example}
Before we continue, it is useful to introduce \emph{multiindex notation}: for any $\alpha \in \N_0^n$, we set 
\begin{align*}
	\partial_x^{\alpha} f := \partial_{x_1}^{\alpha_1} \cdots \partial_{x_n}^{\alpha_n} f 
\end{align*}
and similarly $x^{\alpha} := x_1^{\alpha_1} \cdots x_n^{\alpha_n}$. The integer $\abs{\alpha} := \sum_{j = 1}^n \alpha_j$ is the degree of $\alpha$.

\subsection{Fundamental properties} % (fold)
\label{Fourier:T:fundamentals}
First, we will enumerate various fundamental properties of the Fourier transform on $\T^n$: 
\begin{proposition}[Fundamental properties of $\Fourier$]\label{Fourier:T:prop:fundamentals}
	Let $f \in L^1(\T^n)$. 
	\begin{enumerate}[(i)]
		\item $\Fourier : L^1(\R^n) \longrightarrow \ell^{\infty}(\Z^n)$ 
		\item $\Fourier$ is linear. 
		\item $\Fourier \bar{f}(k) = \overline{(\Fourier f)(-k)}$ 
		\item $\bigl ( \Fourier f(- \, \cdot \,) \bigr )(k) = (\Fourier f)(-k)$
		\item $\bigl ( \Fourier (T_y f) \bigr ) = \e^{- \ii k \cdot y} \, (\Fourier f)(k)$ where $(T_y f)(x) := f(x-y)$ for $y \in \T^n$ 
		\item $(\Fourier f)(k - j) = \bigl ( \Fourier (\e^{+ \ii j \cdot x} f) \bigr )(k)$, $j \in \Z^n$
		\item For all $f \in \Cont^r(\T^n)$, we have $\bigl ( \Fourier (\partial_x^{\alpha} f) \bigr )(k) = \ii^{\abs{\alpha}} \, k^{\alpha} \, (\Fourier f)(k)$ for all $\abs{\alpha} \leq r$. 
	\end{enumerate}
\end{proposition}
\begin{proof}
	(i) can be deduced from the estimate 
	\begin{align*}
		\bnorm{\Fourier f}_{\ell^{\infty}(\Z^n)} &= \sup_{k \in \Z^n} \babs{\Fourier f(k)} 
		\\
		&
		\leq \sup_{k \in \Z^n} \frac{1}{(2\pi)^n} \int_{\T^n} \dd x \, \babs{\e^{- \ii k \cdot x} \, f(x)} 
		= (2\pi)^{-n} \, \snorm{f}_{L^1(\T^n)} 
		. 
	\end{align*}
	(ii)--(vi) follow from direct computation. 
	
	For (vii), we note that continuous functions on $\T^n$ are also integrable, and thus we see that if $f \in \Cont^r(\T^n)$, then also $\partial_x^{\alpha} f \in L^1(\T^n)$ holds for any $\abs{\alpha} \leq r$. This means $\Fourier \bigl ( \partial_x^{\alpha} f \bigr )$ exists, and we obtain by partial integration 
	\begin{align*}
		\bigl ( \Fourier ( \partial_x^{\alpha} f ) \bigr )(k) &= \frac{1}{(2\pi)^n} \int_{\T^n} \dd x \, \e^{- \ii k \cdot x} \, \partial_x^{\alpha} f(x) 
		\\
		&
		= \frac{(-1)^{\abs{\alpha}}}{(2\pi)^n} \int_{\T^n} \dd x \, \bigl ( \partial_x^{\alpha} \e^{- \ii k \cdot x} \bigr ) \, f(x) 
		\\
		&= \ii^{\abs{\alpha}} \, k^{\alpha} \, (\Fourier f)(k) 
		.
	\end{align*}
	Note that the periodicity of $f$ and its derivatives implies the boundary terms vanish. This finishes the proof.
\end{proof}
\begin{example}[$\Fourier$ representation of heat equation]
	Let us consider the heat equation 
	\begin{align*}
		\partial_t u = \Delta_x u
	\end{align*}
	on $[-\pi,+\pi]^n$. We will see that we can write 
	\begin{align}
		u(t,x) &= \sum_{k \in \Z^n} \hat{u}(t,k) \, \e^{+ \ii k \cdot x}
		\label{Fourier:T:eqn:Fourier_series_heat_equation}
	\end{align}
	in terms of Fourier coefficients $\hat{u}(t,k)$. If we \emph{assume} we can interchange taking derivatives and the sum, then this induces an equation involving the coefficients, 
	\begin{align*}
		\Fourier \bigl ( \partial_t u \bigr ) = \Fourier \bigl ( \Delta_x u \bigr )
		\; \; \Longrightarrow \; \; 
		\partial_t \hat{u}(t,k) &= - k^2 \hat{u}(t,k) 
		. 
	\end{align*}
	And \emph{if} the sum \eqref{Fourier:T:eqn:Fourier_series_heat_equation} converges to an integrable function $u(t)$, then clearly 
	\begin{align*}
		(\Fourier u)(t,k) &= \hat{u}(t,k) 
	\end{align*}
	holds. 
\end{example}
We have indicated before that it is not at all clear in what sense the Fourier series~\eqref{Fourier:eqn:T_Fourier_series} exists. The simplest type of convergence is \emph{absolute} convergence, and here Dominated Convergence gives a sufficient condition under which we can interchange taking limits and summation. This helps to prove that the Fourier series yields a continuous or $\Cont^r$ function if the Fourier coefficients decay fast enough.  
\begin{lemma}[Dominated convergence for sums]\label{Fourier:T:lem:dominated_convergence}
	Let $a^{(j)} \in \ell^1(\Z^n)$ be a sequence of absolutely summable sequences so that the pointwise limits $\lim_{j \to \infty} a^{(j)}(k) = a(k)$ exist. Moreover, assume there exists a non-negative sequence $b = \bigl ( b(k) \bigr )_{k \in \Z^n} \in \ell^1(\Z^n)$ so that 
	\begin{align*}
		\babs{a^{(j)}(k)} \leq b(k) 
	\end{align*}
	holds for all $k \in \Z^n$ and $j \in \N$. Then summing over $k \in \Z^n$ and taking the limit $\lim_{j \to \infty}$ commute, \ie 
	\begin{align*}
		\lim_{j \to \infty} \sum_{k \in \Z^n} a^{(j)}(k) = \sum_{k \in \Z^n} \lim_{j \to \infty} a^{(j)}(k) 
		\, , 
	\end{align*}
	and $a = \bigl ( a(k) \bigr )_{k \in \Z^n} \in \ell^1(\Z^n)$. 
\end{lemma}
\begin{proof}
	Let $\eps > 0$. Then we deduce from the triangle inequality and $\babs{a^{(j)}(k)} \leq b_k$, $\babs{a(k)} \leq b(k)$ that 
	\begin{align*}
		\abs{\sum_{k \in \Z^n} a^{(j)}(k) - \sum_{k \in \Z^n} a(k)} &\leq \sum_{\abs{k} \leq N} \babs{a^{(j)}(k) - a(k)} + \sum_{\abs{k} > N} \babs{a^{(j)}(k) - a(k)}
		\\
		&\leq \sum_{\abs{k} \leq N} \babs{a^{(j)}(k) - a(k)} + 2 \, \sum_{\abs{k} > N} b(k)  
		. 
	\end{align*}
	If we choose $N \in \N_0$ large enough, we can estimate the second term \emph{independently of $j$} and make it less than $\nicefrac{\eps}{2}$. The first term is a \emph{finite} sum converging to $0$, and hence, we can make it $< \nicefrac{\eps}{2}$ if we choose $j \geq K$ large enough. Hence, 
	\begin{align*}
		\abs{\sum_{k \in \Z^n} a^{(j)}(k) - \sum_{k \in \Z^n} a(k)} < \nicefrac{\eps}{2} + \nicefrac{\eps}{2} 
		= \eps
	\end{align*}
	holds for $j \geq K$. Moreover, $a = \bigl ( a(k) \bigr )_{k \in \Z^n}$ is absolutely summable since $\babs{a(k)} \leq b(k)$ and $b = \bigl ( b(k) \bigr )_{k \in \Z^n}$ is in $\ell^1(\Z^n)$. 
\end{proof}
\begin{corollary}[Continuity, smoothness and decay properties]\label{Fourier:T:prop:continuity_smoothness_decay}
	\begin{enumerate}[(i)]
		\item Assume the Fourier coefficients $\Fourier f \in \ell^1(\Z^n)$ of $f \in L^1(\T^n)$ are absolutely summable. Then 
		\begin{align*}
			f(x) &= \sum_{k \in \Z^n} (\Fourier f)(k) \, \e^{+ \ii k \cdot x} 
		\end{align*}
		holds almost everywhere and $f$ has a continuous representative. 
		\item Assume the Fourier coefficients of $f \in L^1(\T^n)$ are such that $\bigl ( \abs{k}^s \hat{f}(k) \bigr )_{k \in \Z^n}$ is absolutely summable for some $s \in \N$. Then 
		\begin{align*}
			\partial_x^{\alpha} f(x) &= \sum_{k \in \Z^n} \ii^{\abs{\alpha}} \, k^{\alpha} \, (\Fourier f)(k) \, \e^{+ \ii k \cdot x} 
		\end{align*}
		holds almost everywhere for all $\abs{\alpha} \leq s$ and $f$ has a $\Cont^s(\T^n)$ representative. Moreover, the Fourier series of $\partial_x^{\alpha} f$, $\abs{\alpha} \leq s$, exist as absolutely convergent sums. 
	\end{enumerate}
\end{corollary}
\begin{proof}
	\begin{enumerate}[(i)]
		\item This follows from Dominated Convergence, observing that 
		\begin{align*}
			b(k) := \babs{\hat{f}(k) \, \e^{+ \ii k \cdot x}} = \babs{\hat{f}(k)} 
		\end{align*}
		is a summable sequence which dominates each term of the sum on the right-hand side of \eqref{Fourier:eqn:T_Fourier_series} independently of $x \in \T^n$. 
		\item For any multiindex $\alpha \in \N_0^n$ with $\abs{\alpha} \leq s$, we estimate each term in the sum from above by $\hat{f}(k)$ times 
		\begin{align*}
			\Babs{\partial_x^{\alpha} \bigl ( \e^{+ \ii k \cdot x} \bigr )} &= \babs{k^{\alpha} \, \e^{+ \ii k \cdot x}} 
			= \babs{k^{\alpha}} 
			\\
			&\leq C \, \sabs{k}^{\abs{\alpha}} 
			\leq C \, \sabs{k}^s 
			. 
		\end{align*}
		By assumption, $\bigl ( \sabs{k}^s \hat{f}(k) \bigr )_{k \in \Z^n}$ and thus also $\bigl ( \sabs{k}^{\abs{\alpha}} \hat{f}(k) \bigr )_{k \in \Z^n}$, $\abs{\alpha} \leq s$, are elements of $\ell^1(\Z^n)$, and hence, we have found the sum which dominates the right-hand side of 
		\begin{align*}
			\partial_x^{\alpha} f(x) &= \sum_{k \in \Z^n} \ii^{\abs{\alpha}} \, k^{\alpha} \, \hat{f}(k) \, \e^{+ \ii k \cdot x} 
		\end{align*}
		for all $x \in \T^n$. Thus, a Dominated Convergence argument implies we can interchange differentiation with summation and that the sum depends on $x$ in a continuous fashion. 
	\end{enumerate}
\end{proof}
In what follows, we will need a multiplication, the convolution, defined on $L^1(\T^n)$ and $\ell^1(\Z^n)$. The convolution $\ast$ is intrinsically linked to the Fourier transform: similar to the case of $\R^n$, we define 
\begin{align*}
	(f \ast g)(x) := \int_{\T^n} \dd y \, f(x-y) \, g(y) 
	. 
\end{align*}
where we have used the identification between periodic functions on $\R^n$ and functions on $\T^n$. Moreover, a straightforward modification to the arguments in problem~16 show that 
\begin{align}
	\norm{f \ast g}_{L^1(\T^n)} \leq \norm{f}_{L^1(\T^n)} \, \norm{g}_{L^1(\T^n)} 
	. 
	\label{Fourier:eqn:convolution_Tn}
\end{align}
There is also a convolution on $\ell^1(\Z^n)$: for any two sequences $a , b \in \ell^1(\Z^n)$, we set 
\begin{align}
	(a \ast b)(k) := \sum_{j \in \Z^n} a(k-j) \, b(j) 
	. 
	\label{Fourier:eqn:convolution_Zn}
\end{align}
More careful arguments allow one to generalize the convolution as a map between different spaces, \eg $\ast : L^1(\T^n) \times L^p(\T^n) \longrightarrow L^p(\T^n)$. 

The Fourier transform intertwines pointwise multiplication of functions or Fourier coefficients with the convolution, a fact that will be eminently useful in applications. 
\begin{proposition}
	\begin{enumerate}[(i)]
		\item $f , g \in L^1(\T^n)$ $\; \Longrightarrow \;$ $\Fourier (f \ast g) = (2\pi)^n \, \Fourier f \, \Fourier g$
		\item $f , g \in L^1(\T^n)$, $\hat{f} , \hat{g} \in \ell^1(\Z^n)$ $\; \Longrightarrow \;$ $\displaystyle \sum_{k \in \Z^n} \bigl ( \Fourier f \ast \Fourier g \bigr ) \, \e^{+ \ii k \cdot x} = f(x) \, g(x)$
	\end{enumerate}
\end{proposition}
\begin{proof}
	\begin{enumerate}[(i)]
		\item The convolution of two $L^1(\T^n)$ functions is integrable, and thus, the Fourier transform of $f \ast g$ exists. A quick computation yields the claim: 
		\begin{align*}
			\bigl ( \Fourier (f \ast g) \bigr )(k) &= \frac{1}{(2\pi)^n} \int_{\T^n} \dd x \, \e^{- \ii k \cdot x} \, (f \ast g)(x)
			\\
			&= \frac{1}{(2\pi)^n} \int_{\T^n} \dd x \, \int_{\T^n} \dd y \, \e^{- \ii k \cdot x} \, f(x-y) \, g(y) 
			\\
			&= \frac{1}{(2\pi)^n} \int_{\T^n} \dd x \, \int_{\T^n} \dd y \, \e^{- \ii k \cdot (x-y)} \, f(x-y) \, \e^{- \ii k \cdot y} \, g(y) 
			\\
			&= (2\pi)^n \, (\Fourier f)(k) \, (\Fourier g)(k) 
		\end{align*}
		\item By assumption on the Fourier series of $f$ and $g$, the sequence $\hat{f} \ast \hat{g}$ is absolutely summable, and hence 
		\begin{align*}
			\sum_{k \in \Z^n} \bigl ( \hat{f} \ast \hat{g} \bigr )(k) \, \e^{+ \ii k \cdot x} &= \sum_{k , j \in \Z^n} \hat{f}(k-j) \, \hat{g}(j) \, \e^{+ \ii (k-j) \cdot x} \, \e^{+ \ii j \cdot x} 
			\\
			&= \Biggl ( \sum_{k \in \Z^n} \hat{f}(k) \e^{+ \ii k \cdot x} \Biggr ) \, \Biggl ( \sum_{j \in \Z^n} \hat{g}(j) \e^{+ \ii j \cdot x} \Biggr ) 
		\end{align*}
		exists for all $x \in \T^n$. We will show later in Theorem~\ref{Fourier:thm:T_Fourier_inversion} that for almost all $x \in \T^n$, the sum $\sum_{k \in \Z^n} \hat{f}(k) \, \e^{+ \ii k \cdot x}$ equals $f(x)$ and similarly for $g$. 
	\end{enumerate}
\end{proof}
One helpful fact in applications is that the $L^p(\T^n)$ spaces are nested (\cf Section~\ref{Fourier:R:Lp}): 
\begin{lemma}\label{Fourier:T:lem:nesting_Lp_spaces}
	$L^q(\T^n) \subseteq L^p(\T^n)$ for $1 \leq p \leq q \leq +\infty$
\end{lemma}
That means it suffices to define $\Fourier$ on $L^1(\T^n)$; Life on $\R^n$ is not so simple, though. 
\begin{proof}
	We content ourselves with a sketch: the main idea is to split the integral of $f$ into a region where $\abs{f} \leq 1$ and $\abs{f} > 1$, and then use the compactness of $\T^n$. \marginpar{2013.11.07}
\end{proof}
%
% subsection fundamental_properties (end)

\subsection{Approximating Fourier series by trigonometric polynomials} % (fold)
\label{Fourier:T:trig_pol}
The idea of the Fourier series is to approximate $L^1(\T^n)$ functions by 
\begin{definition}[Trigonometric polynomials]
	A trigonometric polynomial on $\T^n$ is a function of the form 
	\begin{align*}
		P(x) = \sum_{k \in \Z^n} a(k) \, \e^{+ \ii k \cdot x} 
	\end{align*}
	where $\{ a(k) \}_{k \in \Z}$ is a finitely supported sequence in $\Z^n$. The \emph{degree of $P$} is the largest number $\sum_{j = 1}^n \abs{k_j}$ so that $a(k) \neq 0$ where $k = (k_1, \ldots , k_n)$. We denote the set of trigonometric polynomials by $\mathrm{Pol}(\T^n)$. 
\end{definition}
Writing a function in terms of its Fourier series is initially just an ansatz, \ie we do not know whether the formal sum $\sum_{k \in \Z^n} \hat{f}(k) \, \e^{+ \ii k \cdot x}$ converges in any meaningful way. In some sense, this is akin to approximating functions using the Taylor series: only \emph{analytic} functions can be locally expressed in terms of a power series in $x - x_0$, but a smooth function need not have a convergent or useful Taylor series at a point. 

The situation is similar for Fourier series: we need \emph{additional conditions on $f$} to ensure that its Fourier series converges in a strong sense (\eg absolute convergence). However, suitable resummations \emph{do} converge, and one particularly convenient way to approximate $f \in L^1(\T^n)$ by trigonometric polynomials is to convolve it with an 
\begin{definition}[Approximate identity]\label{Fourier:T:defn:approximate_id}
	An approximate identity or \emph{Dirac sequence} is a family of non-negative functions $(\delta_{\eps})_{\eps \in (0,\eps_0)} \subset L^1(\T^n)$, $\eps_0 > 0$, with the following two properties: 
	\begin{enumerate}[(i)]
		\item $\norm{\delta_{\eps}}_{L^1(\T^n)} = 1$ holds for all $\eps \in (0,\eps_0)$. 
		\item For any $R > 0$ we have $\displaystyle \lim_{\eps \to 0} \int_{\abs{x} \leq R} \dd x \, \delta_{\eps}(x) = 1$. 
	\end{enumerate}
\end{definition}
Sometimes the assumption that the $\delta_{\eps}$ are non-negative is dropped. One can show that Dirac sequences are also named approximate identities, because 
\begin{theorem}\label{Fourier:T:thm:approximate_id_convolution}
	Let $(\delta_{\eps})_{n \in \N}$ be an approximate identity. Then for all $f \in L^1(\T^n)$ we have 
	\begin{align*}
		\lim_{\eps \to 0} \bnorm{\delta_{\eps} \ast f - f}_{L^1(\T^n)} = 0 
		. 
	\end{align*}
\end{theorem}
The proof is somewhat tedious: one needs to localize $\delta_{\eps}$ close to $0$ and away from $0$, approximate $f$ by a step function near $x = 0$ and use property~(ii) of approximate identities away from $x = 0$. The interested reader may look it up in \cite[Theorem~1.2.19~(i)]{Grafakos:Fourier_analysis:2008}. 
\medskip

\noindent
One standard example of an approximate identity in this context is the the \emph{Féjer kernel} which is constructed from the \emph{Dirichlet kernel}. In one dimension, the Dirichlet kernel is 
\begin{align}
	d_N(x_1) := \frac{1}{2\pi} \, \sum_{\abs{k} \leq N} \e^{+ \ii k_1 x_1} 
	= \frac{\sin (2N+1) x_1}{\sin x_1} 
	. 
\end{align}
Now the one-dimensional Féjer kernel is the Césaro mean of $d_N$, 
\begin{align}
	f_N(x_1) := \frac{1}{2\pi} \, \sum_{k = -N}^{+N} \left ( 1 - \frac{\abs{k}}{N+1} \right ) \, \e^{+ \ii k x_1}
	= \frac{1}{N+1} \left ( \frac{\sin (N+1) x_1}{\sin x_1} \right )^2 
	. 
\end{align}
The higher-dimensional Dirichlet and Féjer kernels are then defined as 
\begin{align}
	D_N(x) := \prod_{j = 1}^n d_N(x_j)
\end{align}
and 
\begin{align}
	F_N(x) :& \negmedspace= \frac{1}{(N+1)^2} \sum_{k_1 = 0}^N \cdots \sum_{k_n = 0}^N d_{k_1}(x_1) \cdots d_{k_n}(x_n) 
	\\
	&
	= \frac{1}{(2\pi)^n} \, \sum_{\substack{k \in \Z^n \\ \abs{k_j} \leq N}} \left ( \prod_{j = 1}^N \left ( 1 - \frac{\abs{k_j}}{N+1} \right ) \right ) \, \e^{+ \ii k \cdot x}
	\notag 
	\\
	&= \frac{1}{(2\pi)^n} \frac{1}{(N+1)^n} \prod_{j = 1}^n \left ( \frac{\sin (N+1) x_j}{\sin x_j} \right )^2 
	. 
\end{align}
One can see that 
\begin{lemma}
	$(F_N)_{N \in \N}$ is an approximate identity. 
\end{lemma}
\begin{proof}
	Since $F_N$ is the product of one-dimensional Féjer kernels and all the integrals factor, it suffices to consider the one-dimensional case: we note that $f_N$ is non-negative. Then the fact that $\int_{\T^n} \dd x \, \e^{+ \ii k \cdot x} = 0$ for $k \in \Z^n \setminus \{ 0 \}$ and $\int_{\T^n} \dd x \, \e^{0} = 2 \pi$ implies (i). 
	
	(ii) is equivalent to proving $\lim_{N \to \infty} \int_{\abs{x} > R} \dd x \, f_N(x)$. But this follows from writing $f_N$ in terms of sines. 
\end{proof}
Moreover, if we convolve $f \in L^1(\T^n)$ with $F_N$, then 
\begin{align*}
	(F_N \ast f)(x) &= \sum_{\substack{k \in \Z^n \\ \abs{k_j} \leq N}} \left ( 1 - \frac{\abs{k_1}}{N+1} \right ) \cdots  \left ( 1 - \frac{\abs{k_n}}{N+1} \right ) \, \hat{f}(k) \, \e^{+ \ii k \cdot x} 
\end{align*}
is a trigonometric polynomial of degree $N$, and thus 
\begin{proposition}\label{Fourier:T:prop:Pol_dense_Lp}
	$\mathrm{Pol}(\T^n)$ is dense in $L^p(\T^n)$ for any $1 \leq p < \infty$. 
\end{proposition}
\begin{proof}
	Since $F_N$ is an approximate identity, $\lim_{N \to \infty} \bnorm{F_N \ast f - f}_{L^p(\T^n)} = 0$, and thus the trigonometric polynomial $F_N \ast f$ approximates $f$ arbitrarily well in norm. 
\end{proof}
%
% subsection approximating_fourier_series_by_trigonometric_polynomials (end)

\subsection{Decay properties of Fourier coefficients} % (fold)
\label{Fourier:T:decay_properties}
A fundamental question is to ask about the behavior of the Fourier coefficients for large $\abs{k}$. This is important in many applications, because it may mean that certain bases are more efficient than others. The simplest of these criteria is the 
\begin{lemma}[Riemann-Lebesgue lemma]\label{Fourier:T:cor:Riemann_Lebesgue}
	$f \in L^1(\T^n)$ $\Rightarrow$ $\displaystyle \lim_{\abs{k} \to \infty} \hat{f}(k) = 0$ 
\end{lemma}
\begin{proof}
	% p.~176 Proposition~3.2.1 
	% 
	By Proposition~\ref{Fourier:T:prop:Pol_dense_Lp}, we can approximate any $f \in L^1(\T^n)$ arbitrarily well by a trigonometric polynomial $P$, \ie for any $\eps > 0$, there exists a $P \in \mathrm{Pol}(\T^n)$ so that $\norm{f - P}_{L^1(\T^n)} < \eps$. Since the Fourier coefficients of $P$ satisfy $\lim_{\abs{k} \to \infty} \hat{P}(k) = 0$ (only finitely many are non-zero), this also implies that the Fourier coefficients of $f$ satisfy 
	\begin{align*}
		0 \leq \lim_{\abs{k} \to \infty} \babs{\hat{f}(k)} \leq \lim_{\abs{k} \to \infty} \Bigl ( \babs{\hat{f}(k) - \hat{P}(k)} + \babs{\hat{P}(k)} \Bigr ) < \eps + 0 = \eps 
		. 
	\end{align*}
	Given as $\eps$ can be chosen arbitrarily small, the above in fact implies $\displaystyle \lim_{\abs{k} \to \infty} \hat{f}(k) = 0$. 
\end{proof}
\begin{proposition}\label{Fourier:T:prop:F_injective}
	$\Fourier : L^1(\T^n) \longrightarrow \ell^{\infty}(\Z^n)$ is injective, \ie
	\begin{align*}
		\Fourier f = \Fourier g
		\; \; \Longleftrightarrow \; \; 
		f = g \in L^1(\T^n)
		. 
	\end{align*}
\end{proposition}
\begin{proof}
	% p.~169 Proposition 3.1.13
	% 
	Given that $f$ is linear, it suffices to consider the case $f = 0$: 
	
	“$\Rightarrow$:” Assume $\Fourier f = 0$. Then $F_N \ast f = 0$. Since $\{ F_N \}_{N \in \N}$ is an approximate identity, $0 = F_N \ast f \to f$ as $N \to \infty$, \ie $f = 0$. 
	
	“$\Leftarrow$:” In case $f = 0$, also all of the Fourier coefficients vanish, $\Fourier f = 0$. 
\end{proof}
\begin{proposition}[Fourier inversion]\label{Fourier:thm:T_Fourier_inversion}
	Suppose $f \in L^1(\T^n)$ has an absolutely convergent Fourier series, \ie $\Fourier f \in \ell^1(\Z^n)$, 
	\begin{align*}
		\bnorm{\Fourier f}_{\ell^1(\Z^n)} = \sum_{k \in \Z^n} \babs{\hat{f}(k)} < \infty 
		. 
	\end{align*}
	Then 
	\begin{align}
		f(x) = \sum_{k \in \Z^n} \hat{f}(k) \, \e^{+ \ii k \cdot x} 
		\label{Fourier:T:eqn:f_eq_Fourier_series}
	\end{align}
	holds almost everywhere, and $f$ is almost-everywhere equal to a continuous function. 
\end{proposition}
\begin{proof}
	Clearly, left- and right-hand side of equation~\eqref{Fourier:T:eqn:f_eq_Fourier_series} have the same Fourier coefficients, and thus, by Proposition~\ref{Fourier:T:prop:F_injective} they are equal as elements of $L^1(\T^n)$. The continuity of the right-hand side follows from Corollary~\ref{Fourier:T:prop:continuity_smoothness_decay}~(i). 
\end{proof}
\begin{theorem}[Regularity $f$ $\leftrightarrow$ decay $\Fourier f$]\label{Fourier:T:thm:decay_properties}
	\begin{enumerate}[(i)]
		\item Let $s \in \N_0$, $\delta \in (0,1)$ and assume that the Fourier coefficients of $f \in L^1(\T^n)$ decay as 
		\begin{align}
			\babs{\hat{f}(k)} \leq C (1 + \abs{k})^{-n-s-\delta} 
			. 
			\label{Fourier:T:eqn:decay_assumption_Fourier_coefficients}
		\end{align}
		Then $f \in \Cont^s(\T^n)$. 
		\item The Fourier coefficients of $f \in \Cont^s(\T^n)$ satisfy $\displaystyle \lim_{\abs{k} \to \infty} \bigl ( \sabs{k}^r \, \hat{f}(k) \bigr ) = 0$ for $r \leq s$. 
		\item $f \in \Cont^{\infty}(\T^n)$ holds if and only if for all $r \geq 0$ there exists $C_r > 0$ such that 
		\begin{align*}
			\babs{\hat{f}(k)} \leq C_r \, \bigl ( 1 + \abs{k} \bigr )^{-r} 
			. 
		\end{align*}
	\end{enumerate}
\end{theorem}
\begin{proof}
	\begin{enumerate}[(i)]
		\item The decay assumption \eqref{Fourier:T:eqn:decay_assumption_Fourier_coefficients} ensures that $\bigl ( \ii^{\abs{\alpha}} \, k^{\alpha} \hat{f}(k) \bigr )_{k \in \Z^n}$ is absolutely summable if $\abs{\alpha} \leq s$, and thus, by Proposition~\ref{Fourier:T:prop:fundamentals}~(vii) and \ref{Fourier:thm:T_Fourier_inversion} left- and right-hand side of 
		\begin{align*}
			\partial_x^{\alpha} f(x) &= \sum_{k \in \Z^n} \ii^{\abs{\alpha}} \, k^{\alpha} \hat{f}(k) \, \e^{+ \ii k \cdot x} 
		\end{align*}
		are equal and continuous in $x$. 
		\item Clearly, $\abs{k}^r \leq \abs{k}^s$ holds for $r \leq s$ and $k \in \Z^n$. Moreover, $\partial^{\alpha}_x f \in L^1(\T^n)$, $\abs{\alpha} \leq s$, and thus, the Riemann-Lebesgue lemma~\ref{Fourier:T:cor:Riemann_Lebesgue} implies $\lim_{\abs{k} \to \infty} \bigl ( \sabs{k}^r \, \hat{f}(k) \bigr ) = 0$. 
		\item “$\Rightarrow$:” $f \in \Cont^{\infty}(\T^n)$, then (ii) implies that $\lim_{\abs{k} \to \infty} \bigl ( \sabs{k}^r \, \hat{f}(k) \bigr ) = 0$ holds for all $r \geq 0$. 
		
		“$\Leftarrow$:” Conversely, if $\hat{f}(k)$ decays faster than any polynomial, then (i) implies for each $s \geq 0$ the function $f$ is an element of $\Cont^{s-1-n}(\T^n)$, and thus $f \in \Cont^{\infty}(\T^n)$. \marginpar{2013.11.14}
	\end{enumerate}
\end{proof}
%
% subsection decay_properties_of_fourier_coefficients (end)

\subsection{The Fourier transform on $L^2(\T^n)$} % (fold)
\label{Fourier:T:L2}
We will dedicate a little more time to the example of the free Schrödinger equation on $\T^n \cong [-\pi,+\pi]^n$, 
\begin{align*}
	\ii \partial_t \psi = - \Delta_x \psi 
	\, , 
\end{align*}
where we equip $\Delta_x := \partial_{x_1}^2 + \ldots + \partial_{x_n}^2$ with periodic boundary conditions. If we can show that $\{ \e^{+ \ii k \cdot x} \}_{k \in \Z^n}$ is an orthonormal \emph{basis}, then any $\psi \in L^2(\T^n)$ has a \emph{unique} Fourier expansion 
\begin{align}
	\psi(x) &= \sum_{k \in \Z^n} \widehat{\psi}(k) \, \e^{+ \ii k \cdot x}
	\label{Fourier:T:eqn:Fourier_expansion}
\end{align}
where the sum converges in the $L^2$-sense and 
\begin{align*}
	\widehat{\psi}(k) = \bscpro{\e^{+ \ii k \cdot x}}{\psi}_{L^2(\T^n)} 
	= \frac{1}{(2\pi)^n} \int_{\T^n} \dd x \, \e^{- \ii k \cdot x} \, \psi(x) 
\end{align*}
is the $k$th Fourier coefficient. Note that we have normalized the scalar product
\begin{align*}
	\scpro{f}{g}_{L^2(\T^n)} := \frac{1}{(2\pi)^n} \int_{\T^n} \dd x \, \overline{f(x)} \, g(x) 
\end{align*}
so that the $\e^{+ \ii k \cdot x}$ have $L^2$-norm $1$. Lemma~\ref{Fourier:T:lem:nesting_Lp_spaces} tells us that $\widehat{\psi} = \Fourier \psi$ is well-defined, because $\psi \in L^2(\T^n) \subset L^1(\T^n)$. Hence, if $\{ \e^{+ \ii k \cdot x} \}_{k \in \Z^n}$ is a basis, any $L^2(\T^n)$ function can be expressed as a Fourier series. 
\begin{lemma}
	$\{ \e^{+ \ii k \cdot x} \}_{k \in \Z^n}$ is an orthonormal basis of $L^2(\T^n)$. 
\end{lemma}
\begin{proof}
	The orthonormality of $\{ \e^{+ \ii k \cdot x} \}_{k \in \Z^n}$ follows from a straight-forward calculation analogous to the one-dimensional case. The injectivity of $\Fourier : L^2(\T^n) \subset L^1(\T^n) \longrightarrow \ell^{\infty}(\Z^n)$ (Proposition~\ref{Fourier:T:prop:F_injective}) means $\psi = 0 \in L^2(\T^n)$ if and only if $\Fourier \psi = 0$. Hence, $\{ \e^{+ \ii k \cdot x} \}_{k \in \Z^n}$ is a basis. 
\end{proof}
\begin{proposition}\label{Fourier:R:prop:Parseval_Plancherel}
	Let $f , g \in L^2(\T^n)$. Then the following holds true: 
	\begin{enumerate}[(i)]
		\item \emph{Parseval's identity:} $\displaystyle \scpro{f}{g}_{L^2(\T^n)} = \scpro{\Fourier f}{\Fourier g}_{\ell^2(\Z^n)}$ 
		\item \emph{Plancherel's identity:} $\displaystyle \snorm{f}_{L^2(\T^n)} = \snorm{\Fourier f}_{\ell^2(\Z^n)}$ 
		\item $\Fourier : L^2(\T^n) \longrightarrow \ell^2(\Z^n)$ is a unitary. 
	\end{enumerate}
\end{proposition}
\begin{proof}
	(i) follows from the fact that $\{ \e^{+ \ii k \cdot x} \}_{k \in \Z^n}$ is an orthonormal basis and that the coefficients of the basis expansion 
	\begin{align*}
		\bscpro{\e^{+ \ii k \cdot x}}{\psi}_{L^2(\T^n)} = (\Fourier \psi)(k)
	\end{align*}
	coincide with the Fourier coefficients.
	
	(ii) and (iii) are immediate consequences of (i). 
\end{proof}
%
% subsection the_free_schrödinger_equation (end)

\subsection{Periodic operators} % (fold)
\label{Fourier:T:periodic_operators}
An important class of operators are those which commute with translations, 
\begin{align*}
	\bigl [ H , T_y \bigr ] = H \, T_y - T_y \, H = 0 
	, 
	\qquad \qquad 
	\forall y \in \T^n 
	, 
\end{align*}
because this symmetry implies the exponential functions $\e^{+ \ii k \cdot x}$ are eigenfunctions of $H$. 
\begin{theorem}
	Suppose that $H : L^p(\T^n) \longrightarrow L^q(\T^n)$, $1 \leq p,q \leq \infty$, is a bounded linear operator which commutes with translations. Then there exists $\{ h(k) \}_{k \in \Z} \in \ell^{\infty}(\Z^n)$ so that 
	\begin{align}
		(H f)(x) = \sum_{k \in \Z^n} h(k) \, \hat{f}(k) \, \e^{+ \ii k \cdot x} 
		\label{Fourier:T:thm:periodic_op}
	\end{align}
	holds for all $f \in \Cont^{\infty}(\T^n)$. Moreover, we have $\bnorm{\bigl ( h(k) \bigr )_{k \in \Z^n}}_{\ell^{\infty}(\Z^n)} \leq \norm{H}_{\mathcal{B}(L^p(\T^n),L^q(\T^n))}$. 
\end{theorem}
An important example class of examples are \emph{differential operators}, 
\begin{align*}
	H := \sum_{\substack{a \in \N_0^d \\ \abs{a} \leq N}} \beta(a) \, \partial_x^a 
	, 
	&&
	\beta(a) \in \C 
	, 
\end{align*}
whose eigenvalues are 
\begin{align*}
	H \e^{+ \ii k \cdot x} &= \biggl ( \sum_{\substack{a \in \N_0^d \\ \abs{a} \leq N}} f_a(k) \, \ii^{\abs{a}} \, k^a \biggr ) \, \e^{+ \ii k \cdot x} 
	= h(k) \, \e^{+ \ii k \cdot x} 
	. 
\end{align*}
The most famous example so far was $- \Delta_x$ whose eigenvalues are $k^2$. 
\begin{proof}
	We already know that the Fourier series of $f \in \Cont^{\infty}(\T^n)$ converges absolutely, and its Fourier coefficients decay faster than any power of $\abs{k}$ (Theorem~\ref{Fourier:T:thm:decay_properties}). So consider the functions $\varphi_k(x) := \e^{+ \ii k \cdot x}$, $k \in \Z$. The exponential functions are eigenfunctions of the translation operator,
	\begin{align*}
		(T_y \varphi_k)(x) &= \varphi_k(x-y) 
		= \e^{+ \ii k \cdot (x-y)} 
		\\
		&
		= \e^{- \ii k \cdot y} \, \varphi_k(x) 
		, 
	\end{align*}
	and thus, the fact that $T$ commutes with translations implies 
	\begin{align*}
		\bigl ( T_y H \varphi_k \bigr )(x) &= (H \varphi_k)(x - y) 
		= \bigl ( H (T_y \varphi_k) \bigr )(x) 
		\\
		&
		= \e^{- \ii k \cdot y} \, (H \varphi_k)(x) 
		= \varphi_k(-y) \, (H \varphi_k)(x)
		. 
	\end{align*}
	Now writing $x = y - (y - x)$ and interchanging the roles of $x$ and $y$ yields that $\varphi_k(x) = \e^{+ \ii k \cdot x}$ is an eigenfunction of $H$, 
	\begin{align*}
		(H \varphi_k)(x) &= (H \varphi_k)(x - y + y) 
		= \e^{- \ii k \cdot (y-x)} \, (H \varphi_k)(y) 
		\\
		&
		= \e^{+ \ii k \cdot x} \, \bigl ( \e^{- \ii k \cdot y} \, (H \varphi_k)(y) \bigr ) 
		% \\
		% &
		= \bigl ( \e^{- \ii k \cdot y} \, (H \varphi_k)(y) \bigr ) \, \varphi_k(x) 
		\\
		&
		=: h(k) \, \varphi_k(x) 
		. 
	\end{align*}
	It is easy to see that $h(k)$ is in fact independent of the choice of $y \in \T^n$. The above also means $\abs{h(k)} \leq \norm{H}_{\mathcal{B}(L^p(\T^n),L^q(\T^n))}$ holds for all $k \in \Z$, and taking the supremum over $k$ yields $\norm{h}_{\ell^{\infty}(\Z)} \leq \norm{H}_{\mathcal{B}(L^p(\T^n),L^q(\T^n))}$. 
	
	Hence, equation~\eqref{Fourier:T:thm:periodic_op} holds for all $f \in \Cont^{\infty}(\T^n)$, \eg for all trigonometric polynomials. Since those are dense in $L^p(\T^n)$ and $T$ restricted to $\Cont^{\infty}(\T^n)$ is bounded, there exists a unique bounded extension on all of $L^p(\T^n)$ (Theorem~\ref{operators:bounded:thm:extensions_bounded_operators}). \marginpar{2013.11.19}
\end{proof}
%
% subsection periodic_operators (end)

\subsection{Solving quantum problems using $\Fourier$} % (fold)
\label{sub:solving_quantum}
The Fourier transform helps to simplify solving quantum problems. The idea is to convert the Hamiltonian to a multiplication operator. We start with a very simple example:

\subsubsection{The free Schrödinger equation} % (fold)
\label{Fourier:T:periodic_operators:free_schroedinger}
Let us return to the example of the Schrödinger equation: if we denote the free Schödinger operator in the position representation with $H := - \Delta_x$ acting on $L^2(\T^n)$ and impose periodic boundary conditions (\cf the discussion about the shift operator on the interval in Chapter~\ref{operators:unitary}), then the Fourier transform $\Fourier : L^2(\T^n) \longrightarrow \ell^2(\Z^n)$ connects position and momentum representation, \ie 
\begin{align}
	H^{\Fourier} := \Fourier \, H \, \Fourier^{-1} = \hat{k}^2 
	\label{Fourier:T:eqn:free_Schroedinger_Fourier}
\end{align}
acts on suitable vectors from $\ell^2(\Z^n)$ by multiplication with $k^2$. The names position and momentum representation originate from physics, because here, the variable $x \in \T^n$ is interpreted as a position while $k \in \Z^n$ is a momentum. 

To arrive at \eqref{Fourier:T:eqn:free_Schroedinger_Fourier}, we use that any $\psi \in L^2(\T^n)$ has a Fourier expansion since the orthonormal set $\{ \e^{+ \ii k \cdot x} \}_{k \in \Z^n}$ is also a basis and that the $k$ Fourier coefficient of 
\begin{align*}
	- \Delta_x \psi(x) &= - \Delta_x \sum_{k \in \Z^n} \widehat{\psi}(k) \, \e^{+ \ii k \cdot x} 
	= \sum_{k \in \Z^n} k^2 \, \widehat{\psi}(k) \, \e^{+ \ii k \cdot x} 
\end{align*}
is just $k^2 \, \widehat{\psi}(k)$, provided that the sum on the right-hand side is square summable. The latter condition just means that $- \Delta_x \psi$ must exist in $L^2(\T^n)$. 

Clearly, the solution to the free Schrödinger equation in momentum representation is the multiplication operator 
\begin{align*}
	U^{\Fourier}(t) = \e^{- \ii t \hat{k}^2} 
	, 
\end{align*}
and thus we obtain the solution in position representation as well, 
\begin{align*}
	U(t) = \Fourier^{-1} \, U^{\Fourier}(t) \, \Fourier 
\end{align*}
Applied to a vector, this yields 
\begin{align*}
	\psi(t) &= \sum_{k \in \Z^n} \e^{- \ii t k^2} \, \widehat{\psi}_0(k) \, \e^{+ \ii k \cdot x} 
	. 
\end{align*}
Note that this sum exists in $L^2(\T^n)$ if and only if the initial state $\psi$ is square-integrable. 
% subsubsection the_free_schrödinger_equation (end)

\subsubsection{Tight-binding models in solid state physics} % (fold)
\label{Fourier:T:periodic_operators:tight_binding}
Quantum mechanical models are encoded via choosing a hamiltonian and a Hilbert space on which it acts. In the previous section, we have started in a situation where the position variable took values in $\T^n$, \ie wave functions were elements of $L^2(\T^n)$. Tight-binding models, for instance, are but one example where the position variable is a lattice vector and the wave function $\psi \in \ell^2(\Z^n)$ a square summable sequence. 

Tight-binding models describe conduction in semiconductors and insulators: here, the electron may jump from its current position to neighboring atoms with a certain amplitude. One usually restricts oneself to the case where only the hopping amplitudes to \emph{nearest} and sometimes \emph{next-nearest neighbors} are included: in many cases, one can prove that these hopping amplitudes decay exponentially with the distance, and hence, one only needs to include the leading-order terms.

\paragraph{Single-band model} % (fold)
Let us consider a two-dimensional lattice. First of all, we note that the number of nearest neighbors actually depends on the crystal structure. For a simple square lattice, the number of nearest neighbors is $4$ while for a hexagonal lattice, there are only $3$. Let us start with the square lattice: then the hamiltonian 
\begin{align}
	H = \id_{\ell^2(\Z^2)} + q_1 \, \mathfrak{s}_1 + q_2 \, \mathfrak{s}_2 + \overline{q_1} \, \mathfrak{s}_1^* + \overline{q_2} \, \mathfrak{s}_2^* 
	= H^* 
	\label{Fourier:T:tight_binding_square} 
\end{align}
which includes only nearest-neighbor hopping with amplitudes $q_1 , q_2 \in \C$ is defined in terms of the \emph{shift operators}
\begin{align*}
	(\mathfrak{s}_j \psi)(\gamma) := \psi(\gamma - e_j)
	, 
	\qquad \qquad 
	\psi \in \ell^2(\Z^2)
	. 
\end{align*}
Here, $e_j$ stands for either $e_1 = (1,0)$ or $e_2 = (0,1)$ and $\gamma \in \Z^2$. If one sees $\psi(\gamma)$ as the Fourier coefficients for 
\begin{align*}
	(\Fourier^{-1} \psi)(k) = \sum_{\gamma \in \Z^2} \psi(\gamma) \, \e^{+ \ii \gamma \cdot k} 
	, 
\end{align*}
then one can see that the shift operator in momentum representation is just the multiplication operator $\mathfrak{s}_j^{\Fourier} := \Fourier^{-1} \mathfrak{s}_j \Fourier = \e^{+ \ii \hat{k}_j}$: 
\begin{align*}
	\bigl ( \Fourier^{-1} \mathfrak{s}_j \psi \bigr )(k) &= \sum_{\gamma \in \Z^2} \psi(\gamma - e_j) \, \e^{+ \ii \gamma \cdot k}
	= \sum_{\gamma \in \Z^2} \psi(\gamma) \, \e^{+ \ii (\gamma + e_j) \cdot k}
	\\
	&= \e^{+ \ii k_j} \, (\Fourier^{-1} \psi)(k) 
\end{align*}
Note that $\mathfrak{s}_j^{\Fourier} = \Fourier^{-1} \mathfrak{s}_j \Fourier$ makes sense as a composition of bounded linear operators and that $\mathfrak{s}_j^{\Fourier} : L^2(\T^2) \longrightarrow L^2(\T^2)$ is again unitary. 

Hence, the Hamilton operator \eqref{Fourier:T:tight_binding_square} in momentum representation transforms to 
\begin{align*}
	H^{\Fourier} &= 1 + q_1 \, \e^{+ \ii \hat{k}_1} + q_2 \, \e^{+ \ii \hat{k}_2} + \overline{q_1} \, \e^{- \ii \hat{k}_1} + \overline{q_2} \, \e^{- \ii \hat{k}_2} 
	\\
	&= 1 + 2 \Re \bigl ( q_1 \, \e^{+ \ii \hat{k}_1} \bigr ) + 2 \Re \bigl ( q_2 \, \e^{+ \ii \hat{k}_2} \bigr ) 
	. 
\end{align*}
It turns out that in the absence of magnetic fields, the hopping amplitudes can be chosen to be real, and then $H^{\Fourier}$ becomes the multiplication operator associated to the \emph{band function} 
\begin{align*}
	E(k) = 1 + 2 q_1 \, \cos k_1 + 2 q_2 \, \cos k_2 
	. 
\end{align*}
In other words, the Fourier transform converts an operator of shifts into a multiplication operator. That means we can solve the Schrödinger equation in momentum representation, 
\begin{align*}
	\ii \partial_t \widehat{\psi}(t) &= H^{\Fourier} \widehat{\psi}(t) 
	, 
	\qquad \qquad 
	\widehat{\psi}(0) = \psi_0 \in L^2(\T^2) 
	, 
\end{align*}
because also the unitary evolution group is just a multiplication operator, 
\begin{align*}
	U^{\Fourier}(t) &= \e^{- \ii t E(\hat{k})} 
	. 
\end{align*}
Moreover, the unitary evolution group associated to the Schrödinger equation in position representation 
\begin{align*}
	\ii \partial_t \psi(t) = H \psi(t) 
	, 
	\qquad \qquad 
	\psi(0) = \psi_0 \in \ell^2(\Z^2) 
	, 
\end{align*}
is obtained by changing back to the position representation with $\Fourier$, 
\begin{align*}
	U(t) &= \Fourier \, U^{\Fourier}(t) \, \Fourier^{-1} 
	\\
	&= \Fourier \, \e^{- \ii t E(\hat{k})} \, \Fourier^{-1} 
	. 
\end{align*}
%
% paragraph single_band_model (end)

\paragraph{Two-band model} % (fold)
The situation is more interesting and more complicated for hexagonal lattices: here, there are three nearest neighbors and two atoms per unit cell. The following operators have been studied as simplified operators for graphene and boron-nitride (see \eg \cite{Hasegawa_Konno_Nakano_Kohomoto:zero_modes_honeycomb:2006,DeNittis_Lein:piezo_graphene:2013}). Here, the relevant Hilbert space is $\ell^2(\Z^2,\C^2)$ where the “internal” $\C^2$ degree of freedom corresponds to the two atoms in the unit cell (black and white atoms in Figure~\ref{Fourier:R:figure:honeycomb}). Here, nearest neighbors are atoms of a “different color”, and the relevant operator takes the form 
\begin{align*}
	H = \left (
	\begin{matrix}
		0 & 1_{\ell^2(\Z^2)} + q_1 \, \mathfrak{s}_1 + q_2 \, \mathfrak{s}_2 \\
		1_{\ell^2(\Z^2)} + \overline{q_1} \, \mathfrak{s}_1^* + \overline{q_2} \, \mathfrak{s}_2^* & 0 \\
	\end{matrix}
	\right )
	= H^* 
	, 
\end{align*}
\begin{figure}
	\hfil\resizebox{80mm}{!}{\includegraphics{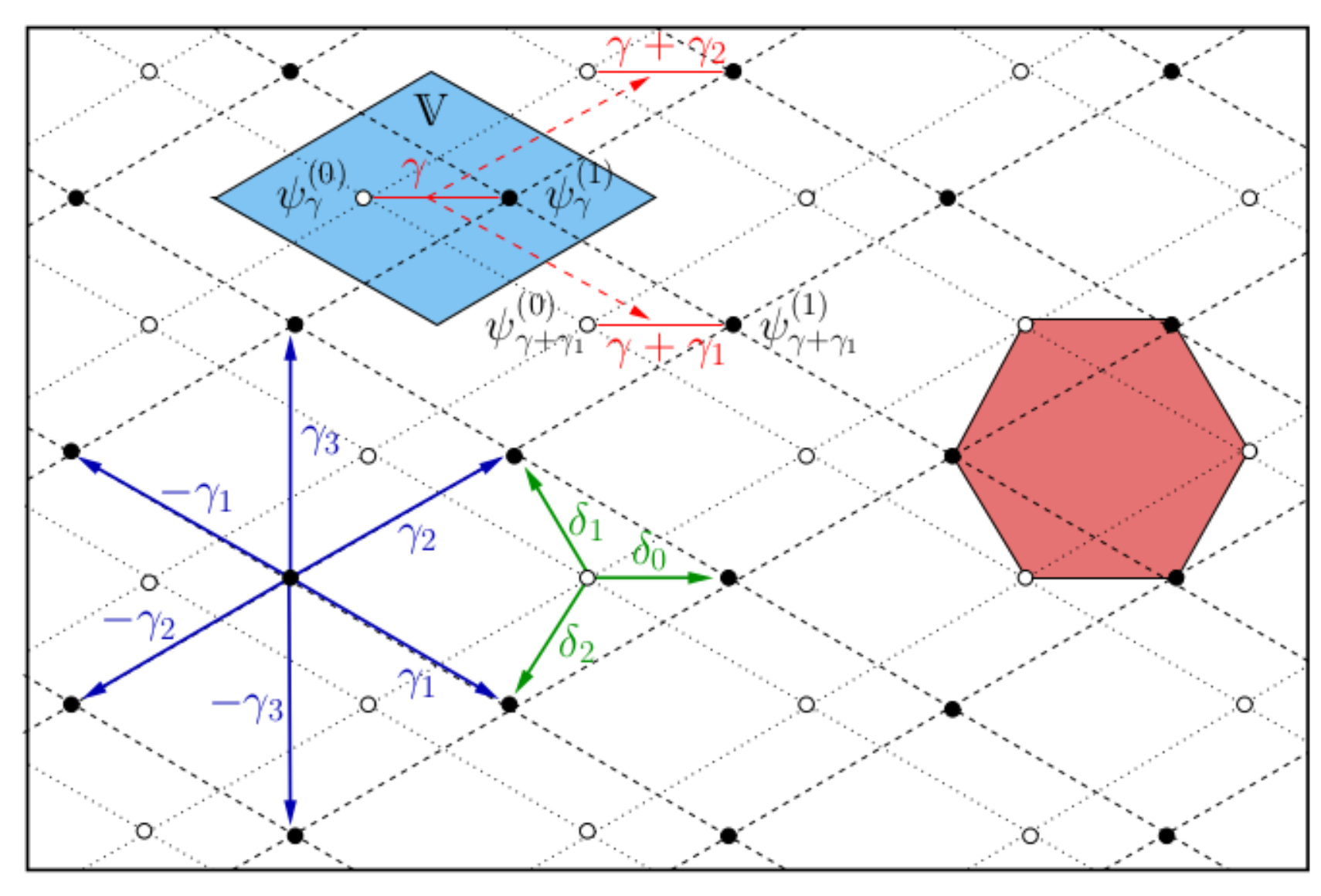}}\hfil
	\caption{The honeycomb lattice is the superposition of two triangular lattices where the fundamental cell contains two atoms, one black and one white.}
	\label{Fourier:R:figure:honeycomb}
\end{figure}
and acts on $\psi = \bigl ( \psi^{(0)} , \psi^{(1)} \bigr ) \in \ell^2(\Z^2,\C^2)$ as 
\begin{align*}
	(H \psi)(\gamma) = \left (
	\begin{matrix}
		\psi^{(1)}(\gamma) + q_1 \, \psi^{(1)}(\gamma - e_1) + q_2 \, \psi^{(1)}(\gamma - e_2) \\
		\psi^{(0)}(\gamma) + \overline{q_1} \, \psi^{(0)}(\gamma + e_1) + \overline{q_2} \, \psi^{(0)}(\gamma + e_2) \\
	\end{matrix}
	\right )
	. 
\end{align*}
One can again use the Fourier transform to relate $H$ to a matrix-valued multiplication operator on $L^2(\T^2,\C^2)$, namely 
\begin{align*}
	H^{\Fourier} = \left (
	\begin{matrix}
		0 & \overline{\varpi(\hat{k})} \\
		\varpi(\hat{k}) & 0 \\
	\end{matrix}
	\right ) 
	= T(\hat{k}) 
	. 
\end{align*}
where we have defined 
\begin{align*}
	\varpi(k) = 1 + q_1 \, \e^{- \ii k_1} + q_2 \, \e^{- \ii k_2}
	, 
\end{align*}
and one can conveniently write $T$ in terms of the Pauli matrices as 
\begin{align*}
	T(k) &= \Re \bigl ( \varpi(k) \bigr ) \, \sigma_1 + \Im \bigl ( \varpi(k) \bigr ) \, \sigma_2 
	. 
\end{align*}
The advantage is that there exist closed formulas for the eigenvalues 
\begin{align*}
	E_{\pm}(k) &= \pm \babs{\varpi(k)}
\end{align*}
of $T(k)$ which are interpreted as the \emph{upper} and \emph{lower band functions} and the two eigenprojections 
\begin{align*}
	P_{\pm}(k) &= \frac{1}{2} \left ( \id_{\C^2} \pm \frac{\Re \bigl ( \varpi(k) \bigr ) \, \sigma_1 + \Im \bigl ( \varpi(k) \bigr ) \, \sigma_2}{\abs{\varpi(k)}} \right ) 
\end{align*}
associated to $E_{\pm}(k)$. Now if one wants to solve the associated Schrödinger equation in momentum representation, 
\begin{align*}
	\ii \partial_t \widehat{\psi}(t) &= H^{\Fourier} \widehat{\psi}(t) 
	, 
	\qquad \qquad 
	\widehat{\psi}(0) = \widehat{\psi}_0 \in L^2(\T^2,\C^2)
	, 
\end{align*}
we can express the unitary evolution group in terms of the eigenvalues and projections as 
\begin{align*}
	U^{\Fourier}(t) = \e^{- \ii t \abs{\varpi(k)}} \, P_+(\hat{k}) + \e^{+ \ii t \abs{\varpi(k)}} \, P_-(\hat{k}) 
	. 
\end{align*}
Also here, the Fourier transform connects the evolution group in momentum and position representation, $U(t) = \Fourier \, U^{\Fourier}(t) \, \Fourier^{-1}$. \marginpar{2013.11.21}
% paragraph two_band_model (end)
% subsubsection tight_binding_models_in_solid_state_physics (end)
% subsection solving_quantum_problems_using_fourier_ (end)
% section the_fourier_transform_on_t_n_ (end)

\section{The Fourier transform on $\R^n$} % (fold)
\label{Fourier:R}
There also exists a Fourier transform on $\R^n$ which is defined analogously to Definition~\ref{Fourier:T:defn:Fourier_transform}. In spirit, the $L^1(\R^n)$ theory is very similar to that for the discrete Fourier transform.

\subsection{The Fourier transform on $L^1(\R^n)$} % (fold)
\label{Fourier:R:L1}
The Fourier transform on $\R^n$ is first defined on $L^1(\R^n)$ as follows: 
\begin{definition}[Fourier transform]\label{Fourier:R:defn:Fourier_transform}
	For any $f \in L^1(\R^n)$, we define its Fourier transform 
	\begin{align*}
		(\Fourier f)(\xi) := \hat{f}(\xi) := \frac{1}{(2\pi)^{\nicefrac{n}{2}}} \int_{\R^n} \dd x \, \e^{- \ii \xi \cdot x} \, f(x) 
		. 
	\end{align*}
\end{definition}
The prefactor $(2\pi)^{-\nicefrac{n}{2}}$ is a matter of convention. Our choice is motivated by the fact that $\Fourier$ will define a \emph{unitary}  map $L^2(\R^n) \longrightarrow L^2(\R^n)$.

\subsubsection{Fundamental properties} % (fold)
\label{Fourier:R:L1:fundamental_properties}
Let us begin by investigating some of the fundamental properties. First, just like in the case of the Fourier transform on $\T^n$, the Fourier transform decays at $\infty$. 
\begin{lemma}[Riemann-Lebesgue lemma]\label{Fourier:R:lem:Riemann_Lebesgue}
	The Fourier transform of any $f \in L^1(\R^n)$ is an element of $\Cont_{\infty}(\R^n)$, \ie $\Fourier f$ is continuous, bounded and decays at infinity,  
	\begin{align*}
		\lim_{\abs{\xi} \to \infty} \Fourier f(\xi) = 0
		. 
	\end{align*}
\end{lemma}
\begin{proof}
	The first part, $\Fourier f \in L^{\infty}(\R^n) \cap \Cont(\R^n)$, has already been shown on page~\pageref{spaces:example:Riemann-Lebesgue_half}. It remains to show that $\Fourier f$ decays at infinity. But that follows from the fact that any integrable function can be approximated arbitrarily well by a finite linear combination of step functions
	\begin{align*}
		1_A(x) := 
		\begin{cases}
			1 & x \in A \\
			0 & x \not\in A \\
		\end{cases}
		, 
	\end{align*}
	and $\lim_{\abs{\xi} \to \infty} (\Fourier 1_A)(\xi) = 0$. Thus, $\lim_{\abs{\xi} \to \infty} (\Fourier f)(\xi) = 0$ follows. 
\end{proof}
We begin to enumerate a few important properties of the Fourier transform. These are tremendously helpful in computations. 
\begin{proposition}[Fundamental properties of $\Fourier$]\label{Fourier:R:prop:fundamentals}
	Let $f \in L^1(\R^n)$. 
	\begin{enumerate}[(i)]
		\item $\Fourier \bar{f}(\xi) = \overline{(\Fourier f)(-\xi)}$ 
		\item $\bigl ( \Fourier f(- \, \cdot \,) \bigr )(\xi) = (\Fourier f)(-\xi)$
		\item $\bigl ( \Fourier (T_y f) \bigr )(\xi) = \e^{- \ii \xi \cdot y} \, (\Fourier f)(\xi)$ where $(T_y f)(x) := f(x-y)$ for $y \in \R^n$ 
		\item $(\Fourier f)(\xi - \eta) = \bigl ( \Fourier (\e^{+ \ii \eta \cdot x} f) \bigr )(\xi)$
		\item $\bigl ( \Fourier (S_{\lambda} f) \bigr )(\xi) = \lambda^n \, (\Fourier f)(\lambda \xi)$ where $(S_{\lambda} f)(x) := f(\nicefrac{x}{\lambda})$, $\lambda > 0$ 
		\item For all $f \in \Cont^r(\R^n)$ with $\partial_x^{\alpha} f \in L^1(\R^n)$, $\abs{\alpha} \leq r$, we have $\bigl ( \Fourier (\partial_x^{\alpha} f) \bigr )(\xi) = \ii^{\abs{\alpha}} \, \xi^{\alpha} \, (\Fourier f)(\xi)$ for all $\abs{\alpha} \leq r$. 
	\end{enumerate}
\end{proposition}
\begin{lemma}[Fourier transform of a Gaußian]\label{Fourier:R:lem:Fourier_Gaussian}
	The Fourier transform of the Gaußian function $g_{\lambda}(x) := \e^{- \frac{\lambda}{2} x^2}$ is 
	\begin{align*}
		(\Fourier g_{\lambda})(\xi) = \lambda^{- \nicefrac{n}{2}} \, \e^{- \frac{1}{2 \lambda} \xi^2} 
		= \lambda^{- \nicefrac{n}{2}} \, g_{\nicefrac{1}{\lambda}}(\xi) 
		. 
	\end{align*}
\end{lemma}
\begin{proof}
	By the scaling relation in Proposition~\ref{Fourier:R:prop:fundamentals}, it suffices to prove the Lemma for $\lambda = 1$. Moreover, we may set $n = 1$, because 
	\begin{align*}
		g_1(x) = \prod_{j = 1}^n \frac{1}{\sqrt{2\pi}} \e^{- \frac{1}{2} x_j^2} 
	\end{align*}
	is just the product of one-dimensional Gaußians. Completing the square, we can express the Fourier transform 
	\begin{align*}
		(\Fourier g_1)(\xi) &= \frac{1}{\sqrt{2\pi}} \int_{\R} \dd x \, \e^{- \ii \xi \cdot x} \, \e^{- \frac{1}{2} x^2} 
		\\
		&= \frac{1}{\sqrt{2\pi}} \int_{\R} \dd x \, \e^{- \frac{1}{2} \xi^2} \, \e^{- \frac{1}{2} (x + \ii \xi)^2} 
		\\
		&= g_1(\xi) \, f(\xi)
	\end{align*}
	as the product of a Gaußian $g_1(\xi)$ with 
	\begin{align*}
		f(\xi) &= \frac{1}{\sqrt{2\pi}} \int_{\R} \dd x \, \e^{- \frac{1}{2} (x + \ii \xi)^2} 
		. 
	\end{align*}
	A simple limiting argument shows that we can differentiate $f$ under the integral sign as often as we would like, \ie $f \in \Cont^{\infty}(\R)$, and that its first derivative vanishes, 
	\begin{align*}
		\frac{\dd}{\dd \xi} f(\xi) &= \frac{1}{\sqrt{2\pi}} \int_{\R} \dd x \, \frac{\dd}{\dd \xi} \Bigl ( \e^{- \frac{1}{2} (x + \ii \xi)^2} \Bigr )
		\\
		&
		= \frac{1}{\sqrt{2\pi}} \int_{\R} \dd x \, \bigl ( - \ii \, (x + \ii \xi) \bigr ) \, \e^{- \frac{1}{2} (x + \ii \xi)^2}
		\\
		&= \frac{1}{\sqrt{2\pi}} \int_{\R} \dd x \, \ii \frac{\dd}{\dd x} \Bigl ( \e^{- \frac{1}{2} (x + \ii \xi)^2} \Bigr )
		\\
		&= \frac{\ii}{\sqrt{2\pi}} \Bigl [ \e^{- \frac{1}{2} (x + \ii \xi)^2} \Bigr ]_{-\infty}^{+\infty}
		= 0 
		. 
	\end{align*}
	But a smooth function whose first derivative vanishes everywhere is constant, and its value is 
	\begin{align*}
		f(0) &= \frac{1}{\sqrt{2\pi}} \int_{\R} \dd x \, \e^{- \frac{1}{2} x^2} 
		= 1 
		. 
	\end{align*}
\end{proof}
The Fourier transform also has an inverse: \marginpar{2013.11.26}
\begin{proposition}
	Assume $f \in L^1(\R^n)$ is such that also its Fourier transform $\hat{f} = \Fourier f$ is integrable. Then for this function, the inverse Fourier transform 
	\begin{align*}
		(\Fourier^{-1} \hat{f})(x) &= \frac{1}{(2\pi)^{\nicefrac{n}{2}}} \int_{\R^n} \dd \xi \, \e^{+ \ii \xi \cdot x} \, \hat{f}(\xi) 
		= (\Fourier \hat{f})(-x)
		= f
	\end{align*}
	agrees with $f \in L^1(\R^n)$. 
\end{proposition}
We will postpone the proof until the next subsection, but the idea is that in the sense of distributions (\cf Section~\ref{S_and_Sprime}) one has 
\begin{align*}
	\int_{\R^n} \dd x \, \e^{+ \ii \xi \cdot x} = (2\pi)^n \, \delta(\xi)
\end{align*}
where $\delta$ is the Dirac distribution. A rigorous argument is more involved, though. 
% subsubsection fundamental_properties (end)

\subsubsection{The convolution} % (fold)
\label{Fourier:R:L1:convolution}
The convolution arises naturally from the group structure (a discussion for another time) and it appears naturally in the discussion, because the Fourier transform intertwines the pointwise product of functions and the convolution (\cf Proposition~\ref{Fourier:R:prop:Fourier_convolution}). 
\begin{definition}[Convolution]
	We define the convolution of $f , g \in L^1(\R^n)$ to be 
	\begin{align*}
		(f \ast g)(x) := \int_{\R^n} \dd y \, f(x-y) \, g(y) 
		. 
	\end{align*}
\end{definition}
We have seen in the exercises that $\ast : L^1(\R^n) \times L^1(\R^n) \longrightarrow L^1(\R^n)$, \ie the convolution of two $L^1$ functions is again integrable. The Fourier transform intertwines the convolution and the pointwise product of functions: 
\begin{proposition}\label{Fourier:R:prop:Fourier_convolution}
	$\Fourier (f \ast g) = (2\pi)^{\nicefrac{n}{2}} \, \Fourier f \; \Fourier g$ holds for all $f , g \in L^1(\R^n)$. 
\end{proposition}
\begin{proof}
	For any $f , g \in L^1(\R^n)$ also their convolution is integrable. Then we obtain the claim from direct computation: 
	\begin{align*}
		\bigl ( \Fourier ( f \ast g ) \bigr )(\xi) &= \frac{1}{(2\pi)^{\nicefrac{n}{2}}} \int_{\R^n} \dd x \, \e^{- \ii \xi \cdot x} \, (f \ast g)(x) 
		\\
		&= \frac{1}{(2\pi)^{\nicefrac{n}{2}}} \int_{\R^n} \dd x \int_{\R^n} \dd y \, \e^{- \ii \xi \cdot (x-y)} \, \e^{- \ii \xi \cdot y} \, f(x-y) \, g(y)
		\\
		&= (2\pi)^{\nicefrac{n}{2}} \, \Fourier f(\xi) \; \Fourier g(\xi)
	\end{align*}
\end{proof}
Also convolutions on $L^1(\R^n)$ have approximate identities. 
\begin{definition}[Approximate identity]\label{Fourier:R:defn:approximate_id}
	An approximate identity or \emph{Dirac sequence} is a family of non-negative functions $(\delta_{\eps})_{\eps \in (0,\eps_0)} \subset L^1(\R^n)$, $\eps_0 > 0$, with the following two properties: 
	\begin{enumerate}[(i)]
		\item $\norm{\delta_{\eps}}_{L^1(\R^n)} = 1$ holds for all $\eps \in (0,\eps_0)$. 
		\item For any $R > 0$ we have $\displaystyle \lim_{\eps \to 0} \int_{\abs{x} \leq R} \dd x \, \delta_{\eps}(x) = 1$. 
	\end{enumerate}
\end{definition}
The name “approximate identity” again derives from the following 
\begin{theorem}\label{Fourier:R:thm:approximate_id_convolution}
	Let $(\delta_{\eps})_{n \in \N}$ be an approximate identity. Then for all $f \in L^1(\R^n)$ we have 
	\begin{align*}
		\lim_{\eps \to 0} \bnorm{\delta_{\eps} \ast f - f}_{L^1(\R^n)} = 0 
		. 
	\end{align*}
\end{theorem}
The interested reader may look up the proof in \cite[Chapter~2.16]{Lieb_Loss:analysis:2001}. 
\begin{example}
	Let $\chi \in L^1(\R^n)$ be a non-negative function normalized to $1$, $\norm{\chi}_{L^1(\R^n)} = 1$. Then one can show that 
	\begin{align*}
		\delta_k(x) := k^n \, \chi(k x)
	\end{align*}
	is an approximate identity. 
\end{example}
With this in mind, we can show that 
\begin{lemma}\label{Fourier:R:lem:Cinfty_c_dense_L1}
	$\Cont^{\infty}_{\mathrm{c}}(\R^n)$ is dense in $L^1(\R^n)$. 
\end{lemma}
\begin{proof}
	We will only sketch the proof: Using the linearity of the convolution and the decomposition 
	\begin{align*}
		f = \bigl ( f_{\Re +} - f_{\Re -} \bigr ) + \ii \, \bigl ( f_{\Im +} - f_{\Im -} \bigr ) \in L^1(\R^n)
	\end{align*}
	implies we may assume $f \geq 0$ without loss of generality. 
	
	First, we smoothen $f$ by convolving it with a smooth approximate identity, because 
	\begin{align*}
		\partial_x^a \bigl ( \delta_k \ast f \bigr ) = (\partial_x^a \delta_k) \ast f
	\end{align*}
	holds as shown in the homework assignments. Clearly, the convolution of two non-negative functions is non-negative. One may start with $\chi \in \Cont^{\infty}_{\mathrm{c}}(\R^n)$ and then scale it like in the example. 
	
	To make the support compact, we multiply with a cutoff function, \eg pick a second function $\mu \in \Cont^{\infty}_{\mathrm{c}}(\R^n)$ taking values between $0$ and $1$ which is normalized to $1 = \norm{\mu}_{L^1(\R^n)}$ and satisfies 
	\begin{align*}
		\mu(x) = 
		\begin{cases}
			1 & \abs{x} \leq 1 \\
			0 & \abs{x} \geq 2 \\
		\end{cases}
		. 
	\end{align*}
	Clearly, $\mu_j(x) := \mu(\nicefrac{x}{j}) \xrightarrow{j \to \infty} 1$ almost everywhere in $x$, and thus $\mu_j \, (\delta_k \ast f)$ converges to $\delta_k \ast f$ as $j \to \infty$ by the Monotone Convergence theorem. Thus, $f_k := \mu_k \, (\delta_k \ast f) \in \Cont^{\infty}_c(\R^n)$ converges to $f$ in $L^1(\R^n)$. 
\end{proof}
%
% subsubsection the_convolution (end)
% subsection the_fourier_transform_on_l_1_r_d_ (end)

\subsection{Decay of Fourier transforms} % (fold)
\label{Fourier:R:decay}
We can prove an analog of Theorem~\ref{Fourier:T:thm:decay_properties}. 
\begin{theorem}[Regularity $f$ $\leftrightarrow$ decay $\Fourier f$]\label{Fourier:R:thm:decay_properties}
	\begin{enumerate}[(i)]
		\item Let $s \in \N_0$, $\delta \in (0,1)$ and assume that the Fourier transform of $f \in L^1(\R^n)$ decays as 
		\begin{align}
			\babs{\hat{f}(\xi)} \leq C (1 + \abs{\xi})^{-n-s-\delta} 
			. 
			\label{Fourier:R:eqn:decay_assumption_Fourier_coefficients}
		\end{align}
		Then $f \in \Cont^s(\R^n)$ and $\partial_x^a f \in L^{\infty}(\R^n)$ holds for all $\abs{a} \leq s$. 
		\item Assume $f \in \Cont^s(\R^n)$ is such that all derivatives up to order $s$ are integrable. Then the Fourier transform $\Fourier f$ satisfies $\displaystyle \lim_{\abs{\xi} \to \infty} \bigl ( \sabs{\xi}^r \, \hat{f}(\xi) \bigr ) = 0$ for $r \leq s$. 
		\item $f \in \Cont^{\infty}(\R^n)$, $\partial_x^{\alpha} f \in L^1(\R^n)$ for all $\alpha \in \N_0^n$ holds if and only if for all $r \geq 0$ there exists $C_r > 0$ such that 
		\begin{align*}
			\babs{\hat{f}(\xi)} \leq C_r \, \bigl ( 1 + \abs{\xi} \bigr )^{-r} 
			. 
		\end{align*}
	\end{enumerate}
\end{theorem}
\begin{proof}
	\begin{enumerate}[(i)]
		\item The decay~\eqref{Fourier:R:eqn:decay_assumption_Fourier_coefficients} implies $\hat{f}$ and $\xi^{\alpha} \, \hat{f}$, $\abs{\alpha} \leq s$, are in fact integrable. Thus, the inverse 
		\begin{align*}
			\bigl ( \Fourier^{-1} (\ii^{\abs{\alpha}} \, \xi^{\alpha} \, \hat{f}) \bigr )(x) = \partial_x^{\alpha} (\Fourier \hat{f})(-x) 
			= \partial_x^{\alpha} f
		\end{align*}
		exists as long as $\abs{\alpha} \leq s$ and is an element of $L^{\infty}(\R^n)$ by the Riemann-Lebesgue lemma~\ref{Fourier:R:lem:Riemann_Lebesgue}. 
		\item This is a consequence of Proposition~\ref{Fourier:R:prop:fundamentals} and the Riemann-Lebesgue lemma~\ref{Fourier:R:lem:Riemann_Lebesgue}. 
		\item Just like in the discrete case, this follows immediately from (i) and (ii). 
	\end{enumerate}
\end{proof}
%
% subsection decay_of_fourier_transforms (end)

\subsection{The Fourier transform on $L^2(\R^n)$} % (fold)
\label{Fourier:R:Lp}
The difficulty of \emph{defining} the Fourier transform on $L^p(\R^n)$ spaces is that they do not nest akin to Lemma~\ref{Fourier:T:lem:nesting_Lp_spaces}. Hence, we will need some preparation. First of all, it is easy to see that the convolution can also be seen as a continuous map 
\begin{align*}
	\ast : L^1(\R^n) \times L^p(\R^n) \longrightarrow L^p(\R^n)
	, 
\end{align*}
Moreover, convolving with a suitable approximate identity is a standard method to regularize $L^p(\R^n)$ functions: 
\begin{lemma}\label{Fourier:R:lem:approximate_identity_Lp}
	If $(\delta_{\eps})$ is an approximate identity and $f \in L^p(\R^n)$, $1 \leq p < \infty$, then $\delta_{\eps} \ast f$ converges to $f$ in $L^p(\R^n)$. 
\end{lemma}
The interested reader may look up the proof in \cite[Chapter~2.16]{Lieb_Loss:analysis:2001}. Hence, we can generalize Lemma~\ref{Fourier:R:lem:Cinfty_c_dense_L1} to 
\begin{lemma}
	$\Cont^{\infty}_{\mathrm{c}}(\R^n)$ is dense in $L^p(\R^n)$ for $1 \leq p < \infty$. 
\end{lemma}
An immediate consequence of the Lemma is that $L^1(\R^n) \cap L^p(\R^n)$ is dense in $L^p(\R^n)$, because $\Cont^{\infty}_{\mathrm{c}}(\R^n) \subset L^1(\R^n) \cap L^p(\R^n) \subset L^p(\R^n)$ lies densely in $L^p(\R^n)$. 

The Gaußian can be conveniently used to define the Fourier transform on $L^2(\R^n)$. 
\begin{theorem}[Plancherel's theorem]\label{Fourier:R:thm:Parseval_Plancherel}
	If $f \in L^1(\R^n) \cap L^2(\R^n)$, then $\hat{f}$ is in $L^2(\R^n)$ and has the same $L^2(\R^n)$-norm as $f$,
	\begin{align*}
		\norm{f} = \bnorm{\hat{f}} 
		. 
	\end{align*}
	Hence, $f \mapsto \Fourier f$ has a unique continuous extension to a \emph{unitary} map $\Fourier : L^2(\R^n) \longrightarrow L^2(\R^n)$. Moreover, Parseval's formula holds, \ie 
	\begin{align*}
		\scpro{f}{g} = \bscpro{\Fourier f}{\Fourier g}
	\end{align*}
	holds for all $f , g \in L^2(\R^n)$. 
\end{theorem}
\begin{proof}
	Initially, pick any $f \in L^1(\R^n) \cap L^2(\R^n)$. Then according to the Riemann-Lebesgue lemma~\ref{Fourier:R:lem:Riemann_Lebesgue}, $\Fourier f$ is essentially bounded, and hence 
	\begin{align}
		\int_{\R} \dd \xi \, \babs{\hat{f}(\xi)}^2 \, \e^{- \frac{\eps}{2} \xi^2}
		\label{Fourier:R:eqn:Plancherel_proof_lhs}
	\end{align}
	is finite. Since $f \in L^1(\R^n)$ by assumption, the function $f(x) \, \overline{f(y)} \, \e^{- \frac{\eps}{2} \xi^2}$ depending on the three variables $(x,y,\xi)$ is an element of $L^1(\R^{3n})$. Then writing out the Fourier transforms in the above integral and integrating over $\xi$, we obtain 
	\begin{align*}
		\eqref{Fourier:R:eqn:Plancherel_proof_lhs} &= \frac{1}{(2\pi)^n} \int_{\R^n} \dd \xi \int_{\R^n} \dd x \, \int_{\R^n} \dd y \, \e^{- \ii \xi \cdot (x-y)} \, f(x) \, \overline{f(y)} \, \e^{- \frac{\eps}{2} \xi^2} 
		\\
		&= \frac{1}{(2\pi)^{\nicefrac{n}{2}}} \int_{\R^n} \dd x \, \int_{\R^n} \dd y \, f(x) \, \overline{f(y)} \, (2\pi)^{\nicefrac{n}{2}} \, \eps^{-\nicefrac{n}{2}} \, \e^{- \frac{1}{2\eps} (x-y)^2} 
		\\
		&= \int_{\R^n} \dd y \, \overline{f(y)} \, \bigl ( \eps^{-\nicefrac{n}{2}} \e^{- \frac{1}{2\eps} x^2} \ast f \bigr )(y) 
		. 
	\end{align*}
	Then by Lemma~\ref{Fourier:R:lem:approximate_identity_Lp} the function $\eps^{-\nicefrac{n}{2}} \e^{- \frac{1}{2\eps} x^2} \ast f$ converges to $f$ in $L^2$ as $\eps \to 0$, and by Dominated convergence the above expression approaches $\norm{f}^2$. On the other hand, the above is equal to \eqref{Fourier:R:eqn:Plancherel_proof_lhs}; Moreover, we may interchange integration and limit $\eps \to 0$ in \eqref{Fourier:R:eqn:Plancherel_proof_lhs} by the Monotone Convergence theorem, and thus we have proven $\norm{f} = \bnorm{\hat{f}}$ if $f \in L^1(\R^n) \cap L^2(\R^n)$. 
	
	By density of $L^1(\R^n) \cap L^2(\R^n)$, this equality extends to all $f \in L^2(\R^n)$ and with the help of Theorem~\ref{operators:bounded:thm:extensions_bounded_operators}, we deduce that the Fourier transform extends to a continuous map $\Fourier : L^2(\R^n) \longrightarrow L^2(\R^n)$. 
	
	Parseval's formula $\scpro{f}{g} = \bscpro{\Fourier f}{\Fourier g}$ follows from the polarization identity 
	\begin{align*}
		\scpro{f}{g} = \frac{1}{2} \Bigl ( \bnorm{f + g}^2 - \ii \, \bnorm{f + \ii \, g}^2 - (1 - \ii) \norm{f}^2 - (1 - \ii) \norm{g}^2 \Bigr ) 
		. 
	\end{align*}
	Parseval's formula also implies the unitarity of $\Fourier$. 
\end{proof}
The above Theorem also gives a \emph{definition} of the Fourier transform on $L^2(\R^n)$. \marginpar{2013.11.28}
% subsection the_fourier_transform_on_l_p_r_n_ (end)

\subsection{The solution of the heat equation on $\R^n$} % (fold)
\label{Fourier:R:heat}
Just like in the discrete case, the Fourier transform converts the heat equation, a linear \textbf{P}DE, into a linear \textbf{O}DE. This connection explains how to arrive at the solution to the heat equation given in problem~18: we first Fourier transform left- and right-hand side of 
\begin{align}
	\partial_t u(t) &= + D \, \Delta_x u(t) 
	, 
	\qquad \qquad 
	u(0) = u_0 \in L^1(\R^n)
	, 
	\label{Fourier:R:eqn:heat_equation_position_rep}
\end{align}
and obtain the heat equation in Fourier representation, 
\begin{align}
	\partial_t \hat{u}(t) &= - D \, \hat{\xi}^2 \, \hat{u}(t) 
	, 
	\qquad \qquad 
	\hat{u}(0) = \hat{u}_0 \in \Cont_{\infty}(\R^n)
	. 
	\label{Fourier:R:eqn:heat_equation_Fourier_rep}
\end{align}
Here, $\hat{u}$ and $\hat{u}_0$ are the Fourier transforms of $u$ and $u_0$, respectively, and $D > 0$ is the diffusion constant. $\hat{\xi}^2$ stands for the multiplication operator associated to the function $\xi \mapsto \xi^2$, \ie we set 
\begin{align*}
	\bigl ( \hat{\xi}^2 \hat{u} \bigr )(\xi) = \xi^2 \, \hat{u}(\xi) 
	. 
\end{align*}
Since the Laplacian in Fourier representation is just a multiplication operator, the solution to \eqref{Fourier:R:eqn:heat_equation_Fourier_rep} is 
\begin{align*}
	\hat{u}(t) &= \e^{- t D \hat{\xi}^2} \hat{u}_0 
	\in \Cont_{\infty}(\R^n) 
	. 
\end{align*}
The first factor $\e^{- t D \xi^2}$ is just a Gaußian, hence $\hat{u}(t)$ is integrable for $t > 0$ and its inverse Fourier exists. The solution in position representation is just the inverse Fourier transform applied to $\hat{u}(t)$: Proposition~\ref{Fourier:R:prop:Fourier_convolution} tells us that $u(t)$ can be seen as the convolution of the inverse Fourier transform of $\e^{- t D \xi^2}$ convolved with $u_0$, and using Lemma~\ref{Fourier:R:lem:Fourier_Gaussian} we compute 
\begin{align}
	u(t) &= \Fourier^{-1} \bigl ( \e^{- t D \hat{\xi}^2} \hat{u}_0 \bigr ) 
	\notag \\
	&= (2\pi)^{- \nicefrac{n}{2}} \, \Bigl ( \Fourier^{-1} \bigl ( \e^{- t D \hat{\xi}^2} \bigr ) \Bigr ) \ast u_0 
	\notag \\
	% &= \frac{1}{(2\pi)^{\nicefrac{n}{2}}} \, \frac{1}{(2 D t)^{\nicefrac{n}{2}}} \, \e^{- \frac{x^2}{4 D t}} \ast u_0 
	% \\
	&= \frac{1}{(4 \pi D t)^{\nicefrac{n}{2}}} \, \e^{- \frac{x^2}{4 D t}} \ast u_0 
	=: G(t) \ast u_0 
	. 
	\label{Fourier:R:eqn:fundamental_solution_heat_equation}
\end{align}
Note that the right-hand side exists in $L^1(\R^n)$ as the convolution of two $L^1(\R^n)$ functions. Moreover, one can show that $\frac{1}{(4 \pi D t)^{\nicefrac{n}{2}}} \, \e^{- \frac{x^2}{4 D t}}$ is a Dirac sequence as $t \searrow 0$ so that $\lim_{t \searrow 0} u(t) = u_0$ holds where the limit is understood in the $L^1(\R^n)$-sense.

\paragraph{Uniqueness of solutions to the heat equation} % (fold)
Let us revisit a topic that has not seen much attention up to now, namely whether solutions exist for all times and whether they are unique. 
\begin{theorem}\label{Fourier:R:thm:uniqueness_solution_heat_equation_L1}
	The initial value problem \eqref{Fourier:R:eqn:heat_equation_position_rep} has \eqref{Fourier:R:eqn:fundamental_solution_heat_equation} as its \emph{unique} solution if we require $u(t) , \partial_t u(t) \in L^1(\R^n)$ to hold for $t > 0$. 
\end{theorem}
\begin{proof}
	The arguments preceding this theorem show that $u(t)$ as given by \eqref{Fourier:R:eqn:fundamental_solution_heat_equation} defines \emph{a} solution which is integrable for $t > 0$; moreover, computing $\Delta_x G(t)$ shows that it is also integrable, and hence, $\partial_t u(t) \in L^1(\R^n)$ is also established. 
	
	So assume $\tilde{u}(t)$ is another integrable solution to \eqref{Fourier:R:eqn:heat_equation_position_rep} with $\partial_t \tilde{u}(t) \in L^1(\R^n)$ and define the difference 
	\begin{align*}
		g(t) := u(t) - \tilde{u}(t) 
		. 
	\end{align*}
	Clearly, this difference vanishes as $t = 0$, \ie $g(0) = 0$. Since $u(t)$ and $\tilde{u}(t)$ as well as their time-derivatives are integrable, also $g(t) , \partial_t g(t) \in L^1(\R^n)$. 
	
	The Riemann-Lebesgue Lemma~\ref{Fourier:R:lem:Riemann_Lebesgue} implies that the Fourier transform of the difference $\hat{g}(t) := \Fourier \Delta(t)$ is an element of $\Cont_{\infty}(\R^n)$. Hence, equations~\eqref{odes:eqn:abs_ddt_X_geq_ddt_abs_X} and \eqref{Fourier:R:eqn:heat_equation_Fourier_rep} yield 
	\begin{align*}
		\frac{\dd}{\dd t} \babs{\hat{g}(t,\xi)} \leq \abs{\frac{\dd}{\dd t} \hat{g}(t,\xi)} 
		= D \, \xi^2 \, \babs{\hat{g}(t,\xi)} 
		, 
		&&
		\forall \xi \in \R^n 
		, 
	\end{align*}
	which is then the initial estimate for the Grönwall lemma~\ref{odes:lem:Groenwall}, 
	\begin{align*}
		0 \leq \babs{\hat{g}(t,\xi)} \leq \babs{\hat{g}(0,\xi)} \, \e^{D \int_0^t \dd s \, \xi^2}
		= \babs{\hat{g}(0,\xi)} \, \e^{D t \xi^2} 
		= 0 
		. 
	\end{align*}
	Since $\hat{g}(t,\xi)$ is continuous in $\xi$, this shows that \eqref{Fourier:R:eqn:fundamental_solution_heat_equation} is the \emph{unique} solution in $L^1(\R^n)$. \marginpar{2014.01.07}
\end{proof}
To show that the condition $u(t) \in L^1(\R^n)$ is crucial for the uniqueness, we give a concrete counterexample first found by Tychonoff \cite{Tychonoff:nonuniqueness_solution_heat_equation:1935}. For simplicity, we reduce to the one-dimensional case and set the diffusion constant $D$ to $1$. Define the function 
\begin{align*}
	u(t,x) := \frac{1}{\sqrt{1-t}} \, \e^{+ \frac{x^2}{4(1-t)}} 
\end{align*}
for $t \in [0,1)$. A simple computation yields that $u(t)$ satisfies the heat equation 
\begin{align*}
	\partial_t u(t) = + \partial_x^2 u(t) 
\end{align*}
to the initial condition $u(0,x) = \e^{\frac{x^2}{4}}$: 
\begin{align*}
	\partial_t u(t,x) &= + \tfrac{1}{2} \, (1-t)^{- \nicefrac{3}{2}} \, \e^{\frac{x^2}{4(1-t)}} + (1 - t)^{- \nicefrac{1}{2}} \cdot (+1) \cdot \frac{x^2}{4(1-t)^2} \, \e^{\frac{x^2}{4(1-t)}}
	\\
	&= \frac{2 (1-t) + x^2}{4 (1-t)^{\nicefrac{5}{2}}} \, \e^{\frac{x^2}{4(1-t)}} 
	\\
	\partial_x u(t,x) &= \frac{x}{2 (1-t)^{\nicefrac{3}{2}}} \, \e^{\frac{x^2}{4(1-t)}}
	\\
	\Rightarrow \; \partial_x^2 u(t,x) &= \frac{1}{2 (1-t)^{\nicefrac{3}{2}}} \, \e^{\frac{x^2}{4(1-t)}} + \frac{x^2}{4 (1-t)^{\nicefrac{5}{2}}} \, \e^{\frac{x^2}{4(1-t)}}
	\\
	&= \frac{2 (1-t) + x^2}{4 (1-t)^{\nicefrac{5}{2}}} \, \e^{\frac{x^2}{4(1-t)}}
\end{align*}
Clearly, the solution explodes as $t \nearrow 1$ and is not integrable for $t \in [0,1)$. 

More careful study shows that asking for $u(t) \in L^1(\R^n)$ is a lot stronger than necessary. In fact, we will see in the next chapter that requiring the solution $u(t)$ to remain a tempered distribution suffices to make the solution unique. 
% paragraph uniqueness_of_solutions_to_the_heat_equation (end)
% subsection the_solution_of_the_heat_equation_on_r_ (end)

\subsection{The solution of the free Schrödinger equation on $\R^n$} % (fold)
\label{Fourier:R:Schroedinger}
Now we will apply this to the free Schrödinger operator $H = - \tfrac{1}{2} \Delta_x$ where the natural space of solutions is $L^2(\R^n)$. Rewriting the free Schrödinger equation in the \emph{momentum representation} yields 
\begin{align*}
	\Fourier \bigl ( \ii \partial_t \psi(t) \bigr ) &= \ii \partial_t \Fourier \psi(t) = \Fourier \bigl ( - \tfrac{1}{2} \Delta_x \psi(t) \bigr ) 
	\\
	&
	= \Fourier \bigl ( - \tfrac{1}{2} \Delta_x \bigr ) \Fourier^{-1} \Fourier \psi(t) 
	. 
\end{align*}
Parseval's theorem~\ref{Fourier:R:thm:Parseval_Plancherel} tells us that $\Fourier : L^2(\R^n) \longrightarrow L^2(\R^n)$ is a unitary, and thus we can equivalently look for $\widehat{\psi}(t) = \Fourier \bigl ( \psi(t) \bigr ) \in L^2(\R^n)$ which solves the free Schrödinger equation in the momentum representation, 
\begin{align*}
	\ii \partial_t \widehat{\psi}(t) &= \Fourier \, \bigl ( - \tfrac{1}{2} \Delta_x \bigr ) \, \Fourier^{-1} \widehat{\psi}(t)
	, 
	&&
	\widehat{\psi}(0) = \widehat{\psi}_0 
	. 
\end{align*}
If we compute the right-hand side at $t = 0$ and assume $\psi_0 \in L^1(\R^n) \cap L^2(\R^n)$, by Proposition~\ref{Fourier:R:prop:fundamentals}~(vi) this leads to 
\begin{align*}
	\Fourier \bigl ( - \tfrac{1}{2} \Delta_x \psi_0 \bigr ) (\xi) = \tfrac{1}{2} \xi^2 \, (\Fourier \psi_0)(\xi) =: \bigl ( H^{\Fourier} \Fourier \psi_0 \bigr )(\xi)
	. 
\end{align*}
We will revisit this point in Chapter~\ref{S_and_Sprime:schwartz_functions}. Note that $- \Delta_x \psi_0 \in L^2(\R^n)$ since this is precisely the domain of definition of $H$, and thus $\mathcal{D}(H^{\Fourier})$ consists of those $L^2(\R^n)$-functions $\widehat{\psi}$ for which $\hat{\xi}^2 \, \widehat{\psi}$ is also in $L^2(\R^n)$. 

Again, the Fourier transform converts the PDE into the linear ODE 
\begin{align}
	\ii \partial_t \widehat{\psi}(t) &= \tfrac{1}{2} \hat{\xi}^2 \, \widehat{\psi}(t) 
	, 
	&&
	\widehat{\psi}(0) = \Fourier \psi_0 = \widehat{\psi}_0 
	, 
\end{align}
which can be solved explicitly by 
\begin{align}
	\widehat{\psi}(t) = U^{\Fourier}(t) \widehat{\psi}_0 
	:= \e^{- \ii t \, \frac{1}{2} \hat{\xi}^2} \, \widehat{\psi}_0 
	. 
\end{align}
The unitary evolution group associated to $H^{\Fourier} = \frac{1}{2} \hat{\xi}^2$ is the unitary multiplication operator $U^{\Fourier}(t) = \e^{- \ii t \, \frac{1}{2} \hat{\xi}^2}$, and hence, the evolution group generated by $H = - \frac{1}{2} \Delta_x$ is 
\begin{align*}
	U(t) = \Fourier^{-1} \, \e^{- \ii t \, \frac{1}{2} \hat{\xi}^2} \Fourier 
	. 
\end{align*}
$U(t)$ is also unitary: $U^{\Fourier}(t)$ and $\Fourier$ are unitary, and thus 
\begin{align*}
	U(t) \, U(t)^* &= \Bigl ( \Fourier^{-1} \, \e^{- \ii t \, \frac{1}{2} \hat{\xi}^2} \Fourier \Bigr ) \, \Bigl ( \Fourier^{-1} \, \e^{- \ii t \, \frac{1}{2} \hat{\xi}^2} \Fourier \Bigr )^* 
	\\
	&= \Fourier^* \, \e^{- \ii t \, \frac{1}{2} \hat{\xi}^2} \Fourier \, \Fourier^* \, \e^{+ \ii t \, \frac{1}{2} \hat{\xi}^2} \, \Fourier 
	= \id_{L^2(\R^n)} 
	. 
\end{align*}
Similarly, one deduces $U(t)^* \, U(t) = \id_{L^2(\R^n)}$. One may be tempted to follow the computation leading up to \eqref{Fourier:R:eqn:fundamental_solution_heat_equation}, replace $t$ by $- \ii t$ and write the solution 
\begin{align}
	\psi(t) = p(t) \ast \psi_0 
	\label{Fourier:R:eqn:solution_free_Schroedinger_convolution}
\end{align}
as a convolution of the initial condition $\psi_0$ with the function 
\begin{align*}
	p(t,x) := \frac{\e^{+ \ii \frac{x^2}{2 t}}}{(2 \pi \ii t)^{\nicefrac{n}{2}}}  
	. 
\end{align*}
From a technical point of view, the derivation of \eqref{Fourier:R:eqn:solution_free_Schroedinger_convolution} is more delicate and will be postponed to Chapter~\ref{S_and_Sprime}.

\paragraph{The uncertainty principle} % (fold)
One of the fundamentals of quantum mechanics is \emph{Heisenberg's uncertainty principle}, namely that one cannot arbitrarily localize wave functions in position \emph{and} momentum space simultaneously. This is a particular case of a much more general fact about non-commuting (quantum) observables: 
\begin{theorem}[Heisenberg's uncertainty principle]
	Let $A , B : \Hil \longrightarrow \Hil$ be two bounded selfadjoint operators on the Hilbert space $\Hil$. We define the expectation value 
	\begin{align*}
		\mathbb{E}_{\psi}(A) := \bscpro{\psi}{A \psi} 
	\end{align*}
	with respect to $\psi \in \Hil$ with $\norm{\psi} = 1$ and the variance 
	\begin{align*}
		\sigma_{\psi}(A)^2 := \mathbb{E}_{\psi} \Bigl ( \bigl (A - \mathbb{E}_{\psi}(A) \bigr )^2 \Bigr ) 
	\end{align*}
	Then Heisenberg's uncertainty relation holds: 
	\begin{align}
		\tfrac{1}{2} \babs{\mathbb{E}_{\psi} \bigl ( \ii [A,B] \bigr )} \leq \sigma_{\psi}(A) \, \sigma_{\psi}(B) 
		\label{Fourier:R:eqn:uncertainty_principle}
	\end{align}
\end{theorem}
\begin{proof}
	Let $\psi \in \Hil$ be an arbitrary normalized vector. Due to the selfadjointness of $A$ and $B$, the expectation values are real, 
	\begin{align*}
		\mathbb{E}_{\psi}(A) &= \bscpro{\psi}{A \psi} 
		= \bscpro{A^* \psi}{\psi} 
		= \bscpro{A \psi}{\psi} 
		\\
		&= \overline{\bscpro{\psi}{A \psi}} 
		= \overline{\mathbb{E}_{\psi}(A)} 
		. 
	\end{align*}
	In general $A$ and $B$ will not have mean $0$, but 
	\begin{align*}
		\tilde{A} := A - \mathbb{E}_{\psi}(A) 
	\end{align*}
	and $\tilde{B} := B - \mathbb{E}_{\psi}(B)$ do. Hence, we can express the variance of $A$ as an expectation value: 
	\begin{align*}
		\sigma_{\psi}(A)^2 &= \mathbb{E}_{\psi} \Bigl ( \bigl ( A - \mathbb{E}_{\psi}(A) \bigr )^2 \Bigr ) 
		= \mathbb{E}_{\psi} \bigl ( \tilde{A}^2 \bigr )
	\end{align*}
	Moreover, the commutator of $\tilde{A}$ and $\tilde{B}$ coincides with that of $A$ and $B$, 
	\begin{align*}
		\bigl [ \tilde{A} , \tilde{B} \bigr ] &= 
		\bigl [ A , B \bigr ] - \bigl [ \mathbb{E}_{\psi}(A) , B \bigr ] - \bigl [ A , \mathbb{E}_{\psi}(B) \bigr ] + \bigl [ \mathbb{E}_{\psi}(A) , \mathbb{E}_{\psi}(B) \bigr ]
		% \bigl ( A - \mathbb{E}_{\psi}(A) \bigr ) \, \bigl ( B - \mathbb{E}_{\psi}(B) \bigr ) - \bigl ( B - \mathbb{E}_{\psi}(B) \bigr ) \, \bigl ( A - \mathbb{E}_{\psi}(A) \bigr ) 
		% \\
		% &= \Bigl ( A B - \mathbb{E}_{\psi}(A) \, B - \mathbb{E}_{\psi}(B) \, A + \mathbb{E}_{\psi}(A) \, \mathbb{E}_{\psi}(B) \Bigr ) 
		% + \\
		% &\qquad 
		% - \Bigl ( B A - \mathbb{E}_{\psi}(A) \, B - \mathbb{E}_{\psi}(B) \, A + \mathbb{E}_{\psi}(A) \, \mathbb{E}_{\psi}(B) \Bigr ) 
		\\
		&= [A , B]
		. 
	\end{align*}
	Then expressing the left-hand side of \eqref{Fourier:R:eqn:uncertainty_principle} in terms of the shifted observables $\tilde{A}$ and $\tilde{B}$, and using the Cauchy-Schwarz inequality as well as the selfadjointness yields Heisenberg's inequality, 
	\begin{align*}
		\babs{\mathbb{E}_{\psi} \bigl ( \ii [ A , B ] \bigr )} 
		&= \babs{\mathbb{E}_{\psi} \bigl ( [ \tilde{A} , \tilde{B} ] \bigr )} 
		= \Babs{\bscpro{\psi}{\tilde{A} \tilde{B} \psi} - \bscpro{\psi}{\tilde{B} \tilde{A} \psi}} 
		\\
		&
		\leq \babs{\bscpro{\tilde{A} \psi}{\tilde{B} \psi}} + \babs{\bscpro{\tilde{B} \psi}{\tilde{A} \psi}} 
		\leq 2 \, \bnorm{\tilde{A} \psi} \, \bnorm{\tilde{B} \psi} 
		\\
		&
		= 2 \, \sqrt{\bscpro{\tilde{A} \psi}{\tilde{A} \psi}} \, \sqrt{\bscpro{\tilde{B} \psi}{\tilde{B} \psi}} 
		= 2 \, \sqrt{\bscpro{\psi}{\tilde{A}^2 \psi}} \, \sqrt{\bscpro{\psi}{\tilde{B}^2 \psi}} 
		\\
		&
		= 2 \, \sigma_{\psi}(A) \, \sigma_{\psi}(B) 
		. 
	\end{align*}
\end{proof}
Often Heisenberg's inequality is just stated for the \emph{position observable} $x_j$ (multiplication by $x_j$) and the \emph{momentum observable} $- \ii \hbar \partial_{x_k}$: even though these are unbounded selfadjoint operators (\cf the discussion in Chapters~\ref{operators:unitary}--\ref{operators:selfadjoint_operators}), this introduces only technical complications. For instance, the above arguments hold verbatim if we require in addition $\psi \in \Cont^{\infty}_{\mathrm{c}}(\R^n) \subset L^2(\R^n)$, and vectors of this type lie dense in $L^2(\R^n)$. Then the left-hand side of Heisenberg's inequality reduces to $\nicefrac{\hbar}{2}$ because 
\begin{align*}
	\bigl [ x_j , (- \ii \hbar \partial_{x_k}) \bigr ] \psi &= x_j \, (- \ii \hbar \partial_{x_k} \psi) - (- \ii \hbar) \partial_{x_k} \bigl ( x_j \, \psi \bigr )
	\\
	&= \ii \hbar \, \delta_{kj} \, \psi 
\end{align*}
and $\psi$ is assumed to have norm $1$, 
\begin{align}
	\sigma_{\psi}(x_j) \, \sigma_{\psi} \bigl ( - \ii \hbar \partial_{x_k} \bigr ) \geq \tfrac{\hbar}{2} 
	. 
	\label{Fourier:R:eqn:uncertainty_principle_x_p}
\end{align}
Skipping over some of the details (there are technical difficulties defining the commutator of two unbounded operators), we see that one cannot do better than $\nicefrac{\hbar}{2}$ but there are cases when the right-hand side of \eqref{Fourier:R:eqn:uncertainty_principle_x_p} is not even finite. 

The physical interpretation of \eqref{Fourier:R:eqn:uncertainty_principle} is that one cannot measure non-commuting observables simultaneously with arbitrary precision. In his original book on quantum mechanics \cite{Heisenberg:prinzipien_quantentheorie:1930}, Heisenberg spends a lot of care to explain why in specific experiments position and momentum along the same direction cannot be measured simultaneously with arbitrary precision, \ie why increasing the resolution of the position measurement increases the error of the momentum measurement and vice versa. 
% paragraph the_uncertainty_principle (end)
% subsection the_solution_of_the_free_schrödinger_equation_on_r_ (end)
% section the_fourier_transform_on_r_n_ (end)
% chapter The Fourier transform (end)
\chapter{Schwartz functions and tempered distributions} % (fold)
\label{S_and_Sprime}
Often one wants to find more general solutions to PDEs, \eg one may ask whether the heat equation makes sense in case the initial condition is an element of $L^{\infty}(\R^n)$? A very fruitful ansatz which we will explore in Chapter~\ref{Greens_functions} is to ask whether “weak solutions” to a PDE exist. Weak means that the solution may be a so-called \emph{distribution} which is a linear functional from a space of \emph{test functions}. 

\emph{Schwartz functions} $\Schwartz(\R^d)$ are such a space of test functions, \ie a space of “very nicely behaved functions”. The dual of this space of test functions, the tempered distributions, allow us to extend common operations such as Fourier transforms and derivatives to objects which may not even be functions.

\section{Schwartz functions} % (fold)
\label{S_and_Sprime:schwartz_functions}
The motivation to define Schwartz functions on $\R^d$ comes from dealing with Fourier transforms: our class of test functions $\Schwartz(\R^d)$ has three defining properties: 
\begin{enumerate}[(i)]
	\item $\Schwartz(\R^d)$ forms a vector space. 
	\item \emph{Stability under derivation}, $\partial_x^{\alpha} \Schwartz(\R^d) \subset \Schwartz(\R^d)$: for all multiindices $\alpha \in \N_0^d$ and $f \in \Schwartz(\R^d)$, we have $\partial_x^{\alpha} f \in \Schwartz(\R^d)$. 
	\item \emph{Stability under Fourier transform}, $\Fourier \Schwartz(\R^d) \subseteq \Schwartz(\R^d)$: for all $f \in \Schwartz(\R^d)$, the Fourier transform 
	\begin{align}
		\Fourier^{\pm 1} f : \xi \mapsto \frac{1}{(2\pi)^{\nicefrac{d}{2}}} \int_{\R^d} \dd x \, \e^{\mp \ii x \cdot \xi} \, f(x) \in \Schwartz(\R^d)
		\label{S_and_Sprime:schwartz_functions:eqn:Fourier_transform}
	\end{align}
	is also a test function. 
\end{enumerate}
These relatively simple requirements have surprisingly rich implications: 
\begin{enumerate}[(i)]
	\item $\Schwartz(\R^d) \subset L^1(\R^d)$, \ie any $f \in \Schwartz(\R^d)$ and all of its derivatives are integrable. 
	\item $\Fourier : \Schwartz(\R^d) \longrightarrow \Schwartz(\R^d)$ acts bijectively: if $f \in \Schwartz(\R^d) \subset L^1(\R^d)$, then $\Fourier f \in \Schwartz(\R^d) \subset L^1(\R^d)$. 
	\item For all $\alpha \in \N_0^d$, we have $\Fourier \bigl ( (i \partial_x)^{\alpha} f \bigr ) = x^{\alpha} \, \Fourier f \in \Schwartz(\R^d)$. This holds as all derivatives are integrable. 
	\item Hence, for all $a , \alpha \in \N_0^d$, we have $x^a \partial_x^{\alpha} f \in \Schwartz(\R^d)$. 
	\item Translations of Schwartz functions are again Schwartz functions, $f( \cdot - x_0) \in \Schwartz(\R^d)$; this follows from $\Fourier f(\cdot - x_0) = \e^{- \ii \xi \cdot x_0} \, \Fourier f \in \Schwartz(\R^d)$ for all $x_0 \in \R^d$. 
\end{enumerate}
This leads to the following definition: 
\begin{definition}[Schwartz functions]
	The space of Schwartz functions 
	\begin{align*}
		\Schwartz(\R^d) := \Bigl \{ f \in \Cont^{\infty}(\R^d) \; \big \vert \; \forall a , \alpha \in \N_0^d : \norm{f}_{a \alpha} < \infty \Bigr \} 
	\end{align*}
	is defined in terms of the family of seminorms\footnote{A seminorm has all properties of a norm except that $\norm{f} = 0$ does not necessarily imply $f = 0$. } indexed by $a , \alpha \in \N_0^d$ 
	\begin{align*}
		\norm{f}_{a \alpha} := \sup_{x \in \R^d} \babs{x^a \partial_x^{\alpha} f(x)} 
		, 
		&& f \in \Cont^{\infty}(\R^d) 
		. 
	\end{align*}
\end{definition}
The family of seminorms defines a so-called \emph{Fréchet topology:} put in simple terms, to make sure that sequences in $\Schwartz(\R^d)$ converge to rapidly decreasing smooth functions, we need to control all derivatives as well as the decay. This is also the reason why there is \emph{no norm on $\Schwartz(\R^d)$} which generates the same topology as the family of seminorms. However, $\norm{f}_{a \alpha} = 0$ for all $a , \alpha \in \N_0^d$ ensures $f = 0$, all seminorms put together can distinguish points. 
\begin{example}
	Two simple examples of Schwartz functions are 
	\begin{align*}
		f(x) = \e^{- a x^2} 
		, 
		&& 
		a > 0 
		, 
	\end{align*}
	and 
	\begin{align*}
		g(x) = \left \{
		\begin{matrix}
			\e^{- \frac{1}{1 - x^2} + 1} & \abs{x} < 1 \\
			0 & \abs{x} \geq 1 \\
		\end{matrix}
		\right . 
		. 
	\end{align*}
	The second one even has compact support. 
\end{example}
The first major fact we will establish is completeness. 
\begin{theorem}
	The space of Schwartz functions endowed with 
	\begin{align*}
		\mathrm{d}(f,g) := \sum_{n = 0}^{\infty} 2^{-n} \sup_{\abs{a} + \abs{\alpha} = n} \frac{\norm{f - g}_{a \alpha}}{1 + \norm{f - g}_{a \alpha}} 
	\end{align*}
	is a complete metric space. 
\end{theorem}
\begin{proof}
	$\mathrm{d}$ is positive and symmetric. It also satisfies the triangle inequality as $x \mapsto \frac{x}{1 + x}$ is concave on $\R^+_0$ and all of the seminorms satisfy the triangle inequality. Hence, $\bigl ( \Schwartz(\R^d) , \mathrm{d} \bigr )$ is a metric space. \marginpar{2014.01.09}
	
	To show completeness, take a Cauchy sequence $(f_n)$ with respect to $\mathrm{d}$. By definition and positivity, this means $(f_n)$ is also a Cauchy sequence with respect to all of the seminorms $\norm{\cdot}_{a \alpha}$. Each of the $\bigl ( x^a \partial_x^{\alpha} f_n \bigr )$ converge to some $g_{a \alpha}$ as the space of bounded continuous functions $\BCont(\R^d)$ with $\sup$ norm is complete. It remains to show that $g_{a \alpha} = x^a \partial_x^{\alpha} g_{00}$. Clearly, only taking derivatives is problematic: we will prove this for $\abs{\alpha} = 1$, the general result follows from a simple induction. Assume we are interested in the sequence $(\partial_{x_k} f_n)$, $k \in \{ 1 , \ldots , d \}$. With $\alpha_k := (0 , \ldots , 0 , 1 , 0, \ldots)$ as the multiindex that has a $1$ in the $k$th entry and $e_k := (0 , \ldots , 0 , 1 , 0, \ldots) \in \R^d$ as the $k$th canonical base vector, we know that 
	\begin{align*}
		f_n(x) = f_n(x - x_k e_k) + \int_0^{x_k} \dd s \, \partial_{x_k} f_n \bigl ( x + (s - x_k) e_k \bigr ) 
	\end{align*}
	as well as 
	\begin{align*}
		g_{00}(x) = g_{00}(x - x_k e_k) + \int_0^{x_k} \dd s \, \partial_{x_k} g_{00} \bigl ( x + (s - x_k) e_k \bigr ) 
	\end{align*}
	hold since $f_n \to g_{00}$ and $\partial_{x_k} f_n \to g_{0\alpha_k}$ uniformly. Hence, $g_{00} \in \Cont^1(\R^d)$ and the derivative of $g_{00}$ coincides with $g_{0\alpha_k}$, $\partial_{x_k} g_{00} = g_{0 \alpha_k}$. We then proceed by induction to show $g_{00} \in \Cont^{\infty}(\R^d)$. This means $\mathrm{d}(f_n,g_{00}) \to 0$ as $n \to \infty$ and $\Schwartz(\R^d)$ is complete. 
\end{proof}
The $L^p$ norm of each element in $\Schwartz(\R^d)$ can be dominated by two seminorms: 
\begin{lemma}\label{S_and_Sprime:schwartz_functions:lem:Lp_estimate}
	Let $f \in \Schwartz(\R^d)$. Then for each $1 \leq p < \infty$, the $L^p$ norm of $f$ can be dominated by a finite number of seminorms, 
	\begin{align*}
		\norm{f}_{L^p(\R^d)} \leq C_1(d) \norm{f}_{00} + C_2(d) \max_{\abs{a} = 2n(d)} \norm{f}_{a 0} 
		, 
	\end{align*}
	where $C_1(d) , C_2(d) \in \R^+$ and $n(d) \in \N_0$ only depend on the dimension of $\R^d$. Hence, $f \in L^p(\R^d)$. 
\end{lemma}
\begin{proof}
	We split the integral on $\R^d$ into an integral over the unit ball centered at the origin and its complement: let $B_n := \max_{\abs{a} = 2n} \norm{f}_{a0}$, then 
	\begin{align*}
		\norm{f}_{L^p(\R^d)} &= \biggl ( \int_{\R^d} \dd x \, \abs{f(x)}^p \biggr )^{\nicefrac{1}{p}} 
		\leq \biggl ( \int_{\abs{x} \leq 1} \dd x \, \abs{f(x)}^p \biggr )^{\nicefrac{1}{p}} + \biggl ( \int_{\abs{x} > 1} \dd x \, \abs{f(x)}^p \biggr )^{\nicefrac{1}{p}} 
		\\
		&\leq \norm{f}_{00} \, \biggl ( \int_{\abs{x} \leq 1} \dd x \, 1 \biggr )^{\nicefrac{1}{p}} +  \biggl ( \int_{\abs{x} > 1} \dd x \, \abs{f(x)}^p \frac{\abs{x}^{2np}}{\abs{x}^{2np}} \biggr )^{\nicefrac{1}{p}} 
		\\
		&\leq \mathrm{Vol}(B_1(0))^{\nicefrac{1}{p}} \, \norm{f}_{00} + B_n \, \biggl ( \int_{\abs{x} > 1} \dd x \, \frac{1}{\abs{x}^{2np}} \biggr )^{\nicefrac{1}{p}} 
		. 
	\end{align*}
	If we choose $n$ large enough, $\abs{x}^{-2np}$	is integrable and can be computed explicitly, and we get 
	\begin{align*}
		\norm{f}_{L^p(\R^d)} \leq C_1(d) \, \norm{f}_{00} + C_2(d) \, \max_{\abs{a} = 2n} \norm{f}_{a0} 
		. 
	\end{align*}
	This concludes the proof. 
\end{proof}
\begin{lemma}\label{S_and_Sprime:schwartz_functions:lem:density_of_Cinfty_compact}
	The smooth functions with compact support $\Cont^{\infty}_c(\R^d)$ are dense in $\Schwartz(\R^d)$. 
\end{lemma}
\begin{proof}
	Take any $f \in \Schwartz(\R^d)$ and choose  
	\begin{align*}
		g(x) = \left \{
		\begin{matrix}
			\e^{- \frac{1}{1 - x^2} + 1} & \abs{x} \leq 1 \\
			0 & \abs{x} > 1 \\
		\end{matrix}
		\right . 
		. 
	\end{align*}
	Then $f_n := g(\nicefrac{\cdot}{n}) \, f$ converges to $f$ in $\Schwartz(\R^d)$, \ie 
	\begin{align*}
		\lim_{n \to \infty} \bnorm{f_n - f}_{a \alpha} = 0 
	\end{align*}
	holds for all $a , \alpha \in \N_0^d$. 
\end{proof}
Next, we will show that $\Fourier : \mathcal{S}(\R^d) \longrightarrow \mathcal{S}(\R^d)$ is a continuous and bijective map from $\mathcal{S}(\R^d)$ onto itself. 
\begin{theorem}\label{S_and_Sprime:schwartz_functions:thm:Fourier_is_bijection}
	The Fourier transform $\Fourier$ as defined by equation~\eqref{S_and_Sprime:schwartz_functions:eqn:Fourier_transform} maps $\Schwartz(\R^d)$ continuously and bijectively onto itself. The inverse $\Fourier^{-1}$ is continuous as well.  Furthermore, for all $f \in \Schwartz(\R^d)$ and $a , \alpha \in \N_0^d$, we have 
	\begin{align}
		\Fourier \bigl ( x^a (- \ii \partial_x)^{\alpha} f \bigr ) = (+ \ii \partial_{\xi})^a \xi^{\alpha} \Fourier f 
		. 
		\label{S_and_Sprime:schwartz_functions:eqn:Fourier_is_bijection}
	\end{align}
\end{theorem}
\begin{proof}
	We need to prove $\Fourier \bigl ( x^a (- \ii \partial_x)^{\alpha} f \bigr ) = (+ \ii \partial_{\xi})^a \xi^{\alpha} \Fourier f$ first: since $x^{\alpha} \partial_x^a f$ is integrable, its Fourier transform exists and is continuous by Dominated Convergence. For any $a , \alpha \in \N_0^d$, we have 
	\begin{align*}
		\Bigl ( \Fourier \bigl ( x^a (- \ii \partial_x)^{\alpha} f \bigr ) \Bigr )(\xi) &= 
		\frac{1}{(2\pi)^{\nicefrac{d}{2}}} \int_{\R^d} \dd x \, \e^{- \ii x \cdot \xi} \, x^a \, (- \ii \partial_x)^{\alpha} f(x) 
		\\
		&= \frac{1}{(2\pi)^{\nicefrac{d}{2}}} \int_{\R^d} \dd x \, (+ \ii \partial_{\xi})^a \e^{- \ii x \cdot \xi} \, (- \ii \partial_x)^{\alpha} f(x) 
		\\
		&
		\overset{\ast}{=} \frac{1}{(2\pi)^{\nicefrac{d}{2}}} (+ \ii \partial_{\xi})^a \int_{\R^d} \dd x \, \e^{- \ii x \cdot \xi} \, (- \ii \partial_x)^{\alpha} f(x) 
		. 
	\end{align*}
	In the step marked with $\ast$, we have used Dominated Convergence to interchange integration and differentiation. Now we integrate partially $\abs{\alpha}$ times and use that the boundary terms vanish, 
	\begin{align*}
		\Bigl ( \Fourier \bigl ( x^a (- \ii \partial_x)^{\alpha} f \bigr ) \Bigr )(\xi) &= \frac{1}{(2\pi)^{\nicefrac{d}{2}}} (+ \ii \partial_{\xi})^a \int_{\R^d} \dd x \, (+ \ii \partial_x)^{\alpha} \e^{- \ii x \cdot \xi} \, f(x) 
		\\
		&
		= \frac{1}{(2\pi)^{\nicefrac{d}{2}}} (+ \ii \partial_{\xi})^a \int_{\R^d} \dd x \, \xi^{\alpha} \e^{- \ii x \cdot \xi} \, f(x) 
		\\
		&= \bigl ( (+ \ii \partial_{\xi})^a \xi^{\alpha} \Fourier f \bigr )(\xi) 
		. 
	\end{align*}
	To show $\Fourier$ is continuous, we need to estimate the seminorms of $\Fourier f$ by those of $f$: for any $a , \alpha \in \N_0^d$, it holds 
	\begin{align*}
		\bnorm{\Fourier f}_{a \alpha} &= \sup_{\xi \in \R^d} \babs{\bigl ( \xi^a \partial_{\xi}^{\alpha} \Fourier f \bigr )(\xi)} = \sup_{\xi \in \R^d} \Babs{\Bigl ( \Fourier \bigl ( \partial_x^a x^{\alpha} f \bigr ) \Bigr )(x)} 
		\\
		&\leq \frac{1}{(2\pi)^{\nicefrac{d}{2}}} \bnorm{\partial_x^a x^{\alpha} f}_{L^1(\R^d)} 
		. 
	\end{align*}
	In particular, this implies $\Fourier f \in \Schwartz(\R^d)$. Since $\partial_x^a x^{\alpha} f \in \Schwartz(\R^d)$, we can apply Lemma~\ref{S_and_Sprime:schwartz_functions:lem:Lp_estimate} and estaimte the right-hand side by a finite number of seminorms of $f$. Hence, $\Fourier$ is continuous: if $f_n$ is a Cauchy sequence in $\Schwartz(\R^d)$ that converges to $f$, then $\Fourier f_n$ has to converge to $\Fourier f \in \Schwartz(\R^d)$. \marginpar{2014.01.14}
	
	To show that $\Fourier$ is a bijection with continuous inverse, we note that it suffices to prove $\Fourier^{-1} \Fourier f = f$ for functions $f$ in a dense subset, namely $\Cont^{\infty}_c(\R^d)$ (see Lemma~\ref{S_and_Sprime:schwartz_functions:lem:density_of_Cinfty_compact}). Pick $f$ so that the support of is contained in a cube $W_n = [-n,+n]^d$ with sides of length $2n$. We can write $f$ on $W_n$ as a uniformly convergent Fourier series, 
	\begin{align*}
		f_n(x) = \sum_{\xi \in \frac{\pi}{n} \Z^d} \hat{f}_n(\xi) \, \e^{\ii x \cdot \xi} 
		, 
	\end{align*}
	with 
	\begin{align*}
		\hat{f}_n(\xi) &= \frac{1}{\mathrm{Vol}(W_n)} \int_{W_n} \dd x \, \e^{- \ii x \cdot \xi} \, f(x) 
		% \\
		% &
		= \frac{(2 \pi)^{\nicefrac{d}{2}}}{(2n)^d} \frac{1}{(2\pi)^{\nicefrac{d}{2}}} \int_{\R^d} \dd x \, \e^{- \ii x \cdot \xi} \, f(x) 
		. 
	\end{align*}
	The second equality holds if $n$ is large enough so that $\supp \, f$ fits into the cube $[-n,+n]^d$. Hence, $f_n$ can be expressed as 
	\begin{align*}
		f_n(x) = \sum_{\xi \in \frac{\pi}{n} \Z^d} \frac{1}{(2\pi)^{\nicefrac{d}{2}}} \frac{\pi^d}{n^d} \, (\Fourier f)(\xi) \, \e^{\ii x \cdot \xi} 
	\end{align*}
	which is a Riemann sum that converges to 
	\begin{align*}
		f(x) &= \frac{1}{(2\pi)^{\nicefrac{d}{2}}} \int_{\R^d} \dd x \, \e^{\ii x \cdot \xi} \, (\Fourier f)(\xi) 
		= \bigl ( \Fourier^{-1} \Fourier f \bigr )(x)
	\end{align*}
	as $\Fourier f \in \Schwartz$. This concludes the proof. 
\end{proof}
Hence, we have shown that $\Schwartz(\R^d)$ has the defining properties that suggested its motivation in the first place. The Schwartz functions also have other nice properties whose proofs are left as an exercise. 
\begin{proposition}
	The Schwartz functions have the following properties: 
	\begin{enumerate}[(i)]
		\item With pointwise multiplication $\cdot : \Schwartz(\R^d) \times \Schwartz(\R^d) \longrightarrow \Schwartz(\R^d)$, the space of $\Schwartz(\R^d)$ forms a Fréchet algebra (\ie the multiplication is continuous in both arguments). 
		\item For all $a , \alpha$, the map $f \mapsto x^a \partial_{\xi}^{\alpha} f$ is continuous on $\Schwartz(\R^d)$. 
		\item For any $x_0 \in \R^d$, the map $\tau_{x_0} : f \mapsto f(\cdot - x_0)$ continuous on $\Schwartz(\R^d)$. 
		\item For any $f \in \Schwartz(\R^d)$, $\frac{1}{h} \bigl ( \tau_{h e_k} f - f)$ converges to $\partial_{x_k} f$ as $h \to 0$ where $e_k$ is the $k$th canonical base vector of $\R^d$. 
		% \item $\Schwartz(\R^d)$ is separable. 
	\end{enumerate}
\end{proposition}
The next important fact will be mentioned without proof: 
\begin{theorem}\label{S_and_Sprime:schwartz_functions:thm:S_dense_in_Lp}
	$\Schwartz(\R^d)$ is dense in $L^p(\R^d)$, $1 \leq p < \infty$. 
\end{theorem}
This means, we can approximate any $L^p(\R^d)$ function by a test function. We will use this and the next theorem to extend the Fourier transform to $L^2(\R^d)$. 
\begin{theorem}\label{S_and_Sprime:schwartz_functions:thm:unitarity_Fourier_on_S}
	For all $f , g \in \Schwartz(\R^d)$, we have 
	\begin{align*}
		\int_{\R^d} \dd x \, (\Fourier f)(x) \, g(x) = \int_{\R^d} \dd x \, f(x) \, (\Fourier g)(x) 
		. 
	\end{align*}
	This implies $\sscpro{\Fourier f}{g} = \sscpro{f}{\Fourier^{-1} g}$ and $\sscpro{\Fourier f}{\Fourier g} = \sscpro{f}{g}$ where $\scpro{\cdot}{\cdot}$ is the usual scalar product on $L^2(\R^d)$. 
\end{theorem}
\begin{proof}
	Using Fubini's theorem, we conclude we can first integrate with respect to $\xi$ instead of $x$, 
	\begin{align*}
		\int_{\R^d} \dd x \, (\Fourier f)(x) \, g(x) &= \int_{\R^d} \dd x \, \frac{1}{(2\pi)^{\nicefrac{d}{2}}} \int_{\R^d} \dd \xi \, \e^{- \ii x \cdot \xi} \, f(\xi) \, g(x) 
		\\
		&
		= \int_{\R^d} \dd \xi \, f(\xi) \, \frac{1}{(2\pi)^{\nicefrac{d}{2}}} \int_{\R^d} \dd x \, \e^{- \ii x \cdot \xi} g(x) 
		% \\
		% &
		= \int_{\R^d} \dd \xi \, f(\xi) \, (\Fourier g)(\xi) 
		. 
	\end{align*}
	To prove the second part, we remark that compared to the scalar product on $L^2(\R^d)$, we are missing a complex conjugation of the first function. Furthermore, $(\Fourier f)^* = \Fourier^{-1} f^*$ holds. From this, it follows that $\sscpro{\Fourier f}{g} = \sscpro{f}{\Fourier^{-1} g}$ and upon replacing $g$ with $\Fourier g$, that $\sscpro{\Fourier f}{\Fourier g} = \sscpro{f}{\Fourier^{-1} \Fourier g} = \sscpro{f}{g}$. 
\end{proof}
Consequently, the convolution defines a multiplication on $\Schwartz(\R^d)$: 
\begin{corollary}
	$\Schwartz(\R^d) \ast \Schwartz(\R^d) \subseteq \Schwartz(\R^d)$ 
\end{corollary}
\begin{proof}
	Let $f , g \in \Schwartz(\R^d)$. Because Schwartz functions are also integrable, $f \ast g$ exists in $L^1(\R^d)$ and satisfies $\Fourier (f \ast g) = (2\pi)^{\nicefrac{d}{2}} \Fourier f \, \Fourier g$ (Proposition~\ref{Fourier:R:prop:Fourier_convolution}). This means we can rewrite $f \ast g = (2 \pi)^{\nicefrac{d}{2}} \Fourier^{-1} \bigl ( \Fourier f \, \Fourier g \bigr )$ as the Fourier transform of a product of Schwartz functions, and thus $f \ast g \in \Schwartz(\R^d)$. 
\end{proof}
Now we will apply this to the free Schrödinger operator $H = - \tfrac{1}{2} \Delta_x$. First of all, we conclude from Theorem~\ref{S_and_Sprime:schwartz_functions:thm:S_dense_in_Lp} that the domain of $H$, 
\begin{align*}
	\Schwartz(\R^d) \subset \mathcal{D}(H) = \bigl \{ \varphi \in L^2(\R^d) \; \vert \; - \Delta_x \varphi \in L^2(\R^d) \bigr \} \subset L^2(\R^d)
	, 
\end{align*}
is dense. Since derivatives of Schwartz functions are Schwartz functions, $H$ maps $\Schwartz(\R^d)$ to itself, and we deduce that the solution 
\begin{align*}
	\psi(t) = U(t) \psi_0 = \Fourier \, \e^{- \ii t \frac{1}{2} \hat{\xi}^2} \, \Fourier^{-1} \psi_0
\end{align*}
to initial conditions $\psi_0 \in \Schwartz(\R^d) \subset L^2(\R^d)$ remains a Schwartz function: $\Fourier^{\pm 1}$ leaves $\Schwartz(\R^d)$ invariant (Theorem~\ref{S_and_Sprime:schwartz_functions:thm:Fourier_is_bijection}) as does multiplication by $\e^{- \ii t \, \frac{1}{2} \xi^2}$, because derivatives of that function are of the form polynomial times $\e^{- \ii t \, \frac{1}{2} \xi^2}$. 

For these initial conditions, we can also rigorously prove equation~\eqref{Fourier:R:eqn:solution_free_Schroedinger_convolution}: 
\begin{proposition}\label{S_and_Sprime:schwartz_functions:prop:free_Schroedinger}
	Let $\psi_0 \in \Schwartz(\R^d) \subset L^2(\R^d)$. Then for $t \neq 0$ the global solution of the free Schrödinger equation with initial condition $\psi_0$ is given by 
	\begin{align}
		\psi(t,x) &= \frac{1}{(2 \pi \ii t)^{\nicefrac{d}{2}}} \int_{\R^d} \dd y \, \e^{\ii \frac{(x-y)^2}{2 t}} \, \psi_0(y) 
		=: \int_{\R^d} \dd y \, p(t,x-y) \, \psi_0(y) 
		. 
		\label{S_and_Sprime:schwartz_functions:eqn:free_propagator}
	\end{align}
	This expression converges in the $L^2$ norm to $\psi_0$ as $t \to 0$. 
\end{proposition}
\begin{proof}
	We denote the Fourier transform of $\e^{- \ii t \frac{1}{2} \hat{\xi}^2}$ by 
	\begin{align*}
		U(t) := \Fourier^{-1} \e^{- \ii t \frac{1}{2} \hat{\xi}^2} \Fourier 
		. 
	\end{align*}
	If $t = 0$, then the bijectivity of the Fourier transform on $\Schwartz(\R^d)$, Theorem~\ref{S_and_Sprime:schwartz_functions:thm:Fourier_is_bijection}, yields $U(0) = \id_{\Schwartz}$. 
	
	So let $t \neq 0$. If $\widehat{\psi}_0$ is a Schwartz function, so is $\e^{- \ii t \frac{1}{2} \xi^2} \widehat{\psi}_0$. As the Fourier transform is a unitary map on $L^2(\R^d)$ (Proposition~\ref{Fourier:R:prop:Parseval_Plancherel}) and maps Schwartz functions onto Schwartz functions (Theorem~\ref{S_and_Sprime:schwartz_functions:thm:Fourier_is_bijection}), $U(t)$ also maps Schwartz functions onto Schwartz functions. Plugging in the definition of the Fourier transform, for any $\psi_0 \in \Schwartz(\R^d)$ and $t \neq 0$ we can write out $U(t) \psi_0$ as 
	\begin{align}
		\bigl ( \Fourier^{-1} &\e^{- \ii t \frac{1}{2} \hat{\xi}^2} \Fourier \psi_0 \bigr )(x) = \frac{1}{(2\pi)^d} \int_{\R^d} \dd \xi \, \e^{+ \ii x \cdot \xi} \, \e^{- \ii t \frac{1}{2} \xi^2} \, \int_{\R^d} \dd y \, \e^{- \ii y \cdot \xi} \, \psi_0(y) 
		\notag \\
		% &= \frac{1}{(2 \pi)^d} \int_{\R^d} \dd y \, \int_{\R^d} \dd \xi \, \e^{- \ii (y - x) \cdot \xi} \e^{- \ii t \frac{1}{2} \xi^2} \, \psi_0(y) 
		% \\
		&= \frac{1}{(2 \pi)^{\nicefrac{d}{2}}} \int_{\R^d} \dd \xi \int_{\R^d} \dd y \, \e^{\ii \frac{(x - y)^2}{2 t}} \, \left ( \frac{1}{(2 \pi)^{\nicefrac{d}{2}}} \, \e^{- \ii \frac{t}{2} ( \xi - \nicefrac{(x - y)}{t})^2} \right ) \, \psi_0(y) 
		. 
		\label{S_and_Sprime:schwartz_functions:eqn:proof_free_propagator}
	\end{align}
	We need to regularize the integral: if we write the right-hand side of the above as 
	\begin{align*}
		\mbox{r. h. s.} &= \lim_{\eps \searrow 0} \frac{1}{(2 \pi)^{\nicefrac{d}{2}}} \int_{\R^d} \dd \xi \int_{\R^d} \dd y \, \e^{\ii \frac{(x - y)^2}{2 t}} \, \left ( \frac{1}{(2 \pi)^{\nicefrac{d}{2}}} \, \e^{- (\eps + \ii) \frac{t}{2} ( \xi - \nicefrac{(x - y)}{t})^2} \right ) \, \psi_0(y) 
		\\
		&= \lim_{\eps \searrow 0} \frac{1}{(2 \pi)^{\nicefrac{d}{2}}} \int_{\R^d} \dd y \, \e^{\ii \frac{(x - y)^2}{2 t}} \, \left ( \frac{1}{(2 \pi)^{\nicefrac{d}{2}}} \int_{\R^d} \dd \xi \, \e^{- (\eps + \ii) \frac{t}{2} ( \xi - \nicefrac{(x - y)}{t})^2} \right ) \, \psi_0(y) 
		, 
	\end{align*}
	we can use Fubini to change the order of integration. The inner integral can be computed by interpreting it as an integral in the complex plane, 
	\begin{align*}
		\frac{1}{(2 \pi)^{\nicefrac{d}{2}}} \int_{\R^d} \dd \xi \, \e^{- (\eps + \ii) \frac{t}{2} ( \xi - \nicefrac{(x - y)}{t})^2} = \frac{1}{\bigl ( (\eps + \ii) t \bigr )^{\nicefrac{d}{2}}} 
		. 
	\end{align*}
	Plugged back into equation~\eqref{S_and_Sprime:schwartz_functions:eqn:proof_free_propagator} and combined with the Dominated Convergence Theorem, this yields equation~\eqref{S_and_Sprime:schwartz_functions:eqn:free_propagator}. 
\end{proof}
%
% section schwartz_functions (end)

\section{Tempered distributions} % (fold)
\label{S_and_Sprime:Sprime}

Tempered distributions are linear functionals on $\Schwartz(\R^d)$. 
\begin{definition}[Tempered distributions]
	The tempered distributions $\Schwartz'(\R^d)$ are the continuous linear functions on the Schwartz functions $\Schwartz(\R^d)$. If $L \in \Schwartz'(\R^d)$ is a linear functional, we will often write 
	\begin{align*}
		\bigl ( L , f \bigr ) := L(f) 
		&&
		\forall f \in \Schwartz(\R^d) 
		. 
	\end{align*}
\end{definition}
\begin{example}
	\begin{enumerate}[(i)]
		\item The $\delta$ distribution defined via 
		\begin{align*}
			\delta(f) := f(0) 
		\end{align*}
		is a linear continuous functional on $\Schwartz(\R^d)$. (See exercise sheet~12.) 
		\item Let $g \in L^p(\R^d)$, $1 \leq p < \infty$, then for $f \in \Schwartz(\R^d)$, we define 
		\begin{align}
			L_g(f) = \int_{\R^d} \dd x \, g(x) \, f(x) =: (g , f) 
			. 
			\label{S_and_Sprime:Sprime:eqn:distributions_of_functions}
		\end{align}
		As $f \in \Schwartz(\R^d) \subset L^q(\R^d)$, $\frac{1}{p} + \frac{1}{q} = 1$, by Hölder's inequality, we have 
		\begin{align*}
			\babs{\bigl ( g , f \bigr )} \leq \norm{g}_p \, \norm{f}_q 
			. 
		\end{align*}
		Since $\norm{f}_q$ can be bounded by a finite linear combination of Fréchet seminorms of $f$, $L_g$ is continuous, and the inclusion map $\imath : L^p(\R^d) \longrightarrow \Schwartz'(\R^d)$ is continuous. 
		\item Equation~\eqref{S_and_Sprime:Sprime:eqn:distributions_of_functions} is the \emph{canonical way to interpret less nice functions as distributions}: we identify a suitable function $g : \R^d \longrightarrow \C$ with the distribution $L_g$. For instance, polynomially bounded smooth functions (think of $g(x) = x^2$) define continuous linear functionals in this manner since for any $g \in \Cont^{\infty}_{\mathrm{pol}}(\R^d)$, there exists $n \in \N_0$ such that $\sqrt{1 + x^2}^{\, -n} g(x)$ is bounded. Hence, for any $f \in \Schwartz(\R^d)$, Hölder's inequality yields 
		\begin{align*}
			\babs{\bigl ( g , f \bigr )} &= \abs{\int_{\R^d} \dd x \, g(x) \, f(x)} 
			= \abs{\int_{\R^d} \dd x \, \sqrt{1 + x^2}^{\, -n} \, g(x) \, \sqrt{1 + x^2}^{\, n} f(x)} 
			\\
			&
			\leq \bnorm{\sqrt{1 + x^2}^{\, -n} \, g(x)}_{L^{\infty}} \, \bnorm{\sqrt{1 + x^2}^{\, n} f(x)}_{L^1} 
			. 
		\end{align*}
		Later on, we will see that this point of view, interpreting “not so nice” functions as distributions, helps us extend operations from test functions to much broader classes of functions. 
	\end{enumerate}
\end{example}
Similar to the case of normed spaces, we see that continuity implies “boundedness”. 
\begin{proposition}
	A linear functional $L : \Schwartz(\R^d) \longrightarrow \C$ is a tempered distribution (\ie continuous) if and only if there exist constants $C > 0$ and $k,n \in \N_0$ such that 
	\begin{align*}
		\babs{L(f)} \leq C \sum_{\substack{\abs{a} \leq k \\
		\abs{\alpha} \leq n}} \norm{f}_{a \alpha} 
	\end{align*}
	for all $f \in \Schwartz(\R^d)$. 
\end{proposition}
Even though we will not give a proof, let us at least sketch its idea: because one has no control over the growth or decay of the seminorms $\norm{f}_{a \alpha}$, maxima or sums of seminorms are finite if and only if only finitely many of them enter. 
\medskip

\noindent
As mentioned before, we can interpret suitable functions $g$ as tempered distributions. In particular, every Schwartz function $g \in \Schwartz(\R^d)$ defines a tempered distribution so that 
\begin{align*}
	\bigl ( \partial_{x_k} g , f \bigr ) &= \int_{\R^d} \dd x \, \partial_{x_k} g(x) \, f(x) = - \int_{\R^d} \dd x \, g(x) \, \partial_{x_k} f(x) = \bigr ( g , - \partial_{x_k} f \bigr ) 
\end{align*}
holds for any $f \in \Schwartz(\R^d)$. We can use the right-hand side to \emph{define} derivatives of distributions: 
\begin{definition}[Weak derivative]
	For $\alpha \in \N_0^d$ and $L \in \Schwartz'(\R^d)$, we define the weak or distributional derivative of $L$ as 
	\begin{align*}
		\bigl ( \partial_x^{\alpha} L , f \bigr ) := \bigl ( L , (-1)^{\abs{\alpha}} \partial_x^{\alpha} f \bigr ) 
		, 
		&&
		\forall f \in \Schwartz(\R^d) 
		. 
	\end{align*}
\end{definition}
\begin{example}
	\begin{enumerate}[(i)]
		\item The weak derivative of $\delta$ is 
		\begin{align*}
			\bigl ( \partial_{x_k} \delta , f \bigr ) &= \bigl ( \delta , - \partial_{x_k} f \bigr ) 
			= - \partial_{x_k} f(0) 
			. 
		\end{align*}
		\item Let $g \in \Cont^{\infty}_{\mathrm{pol}}(\R^d)$. Then the weak derivative coincides with the usual derivative, by partial integration, we get 
		\begin{align*}
			\bigl ( \partial_{x_k} g , f \bigr ) &= - \bigl ( g , \partial_{x_k} f \bigr ) 
			= - \int_{\R^d} \dd x \, g(x) \, \partial_{x_k} f(x) 
			\\
			&
			= + \int_{\R^d} \dd x \, \partial_{x_k} g(x) \, f(x) 
			. 
		\end{align*}
	\end{enumerate}
\end{example}
Similarly, the Fourier transform can be extended to a bijection $\Schwartz'(\R^d) \longrightarrow \Schwartz'(\R^d)$. Theorem~\ref{S_and_Sprime:schwartz_functions:thm:unitarity_Fourier_on_S} tells us that if $g , f \in \Schwartz(\R^d)$, then 
\begin{align*}
	\bigl ( \Fourier g , f \bigr ) = \bigl ( g , \Fourier f \bigr ) 
\end{align*}
holds. If we replace $g$ with an arbitrary tempered distribution, the right-hand side again serves as \emph{definition} of the left-hand side: 
\begin{definition}[Fourier transform on $\mathcal{S}'(\R^d)$]
	For any tempered distribution $L \in \Schwartz'(\R^d)$, we define its Fourier transform to be 
	\begin{align*}
		\bigl ( \Fourier L , f \bigr ) := \bigl ( L , \Fourier f \bigr ) 
		&&
		\forall f \in \Schwartz(\R^d) 
		. 
	\end{align*}
\end{definition}
\begin{example}
	\begin{enumerate}[(i)]
		\item The Fourier transform of $\delta$ is the constant function $(2\pi)^{-\nicefrac{d}{2}}$, 
		\begin{align*}
			\bigl ( \Fourier \delta , f \bigr ) &= \bigl ( \delta , \Fourier f \bigr ) 
			= \Fourier f (0) 
			= \frac{1}{(2\pi)^{\nicefrac{d}{2}}}\int_{\R^d} \dd x \, f(x) 
			\\
			&= \bigl ( (2\pi)^{-\nicefrac{d}{2}} , f \bigr ) 
			. 
		\end{align*}
		\item The Fourier transform of $x^2$ makes sense as a tempered distribution on $\R$: $x^2$ is a polynomially bounded function and thus defines a tempered distribution via equation~\eqref{S_and_Sprime:Sprime:eqn:distributions_of_functions}: 
		\begin{align*}
			\bigl ( \Fourier x^2 , f \bigr ) &= \bigl ( x^2 , \Fourier f \bigr ) 
			= \int_{\R} \dd x \, x^2 \, \frac{1}{(2\pi)^{\nicefrac{1}{2}}} \int_{\R} \dd \xi \, \e^{- \ii x \cdot \xi} \, f(\xi) 
			% \\
			% &
			% = \frac{1}{(2\pi)^{\nicefrac{1}{2}}} \int_{\R} \dd \xi \int_{\R} \dd x \, x^2 \, \e^{- \ii x \cdot \xi} \, f(\xi) 
			\\
			&= \frac{1}{(2\pi)^{\nicefrac{1}{2}}} \int_{\R} \dd \xi \int_{\R} \dd x \, (+\ii)^2 \partial_{\xi}^2 \e^{- \ii x \cdot \xi} \, f(\xi) 
			\\
			&= (-1)^2 \cdot (-1) \, \int_{\R} \dd \xi \left ( \frac{1}{(2\pi)^{\nicefrac{1}{2}}} \int_{\R} \dd x \, \e^{- \ii x \cdot \xi} \right ) \, \partial_{\xi}^2 f(\xi) 
			\\
			&= - \int_{\R} \dd \xi \, (2\pi)^{\nicefrac{1}{2}} \, \delta(\xi) \, \partial_{\xi}^2 f(\xi) 
			\\
			&= - (2\pi)^{\nicefrac{1}{2}} \, \partial_{\xi}^2 f(0) 
			= \bigl ( (2\pi)^{\nicefrac{1}{2}} \delta , - \partial_{\xi}^2 f \bigr ) 
			= \bigl ( - (2\pi)^{\nicefrac{1}{2}} \delta'' , f \bigr ) 
		\end{align*}
		This is consistent with what we have shown earlier in Theorem~\ref{S_and_Sprime:schwartz_functions:thm:Fourier_is_bijection}, namely 
		\begin{align*}
			\Fourier \bigl ( x^2 f \bigr ) = (+ \ii \partial_{\xi})^2 \Fourier f = - \partial_{\xi}^2 \Fourier f 
			. 
		\end{align*}
	\end{enumerate}
\end{example}
We have just computed Fourier transforms of functions that do not have Fourier transforms in the usual sense. We can apply the idea we have used to define the derivative and Fourier transform on $\Schwartz'(\R^d)$ to other operators initially defined on $\Schwartz(\R^d)$. Before we do that though, we need to introduce the appropriate notion of continuity on $\Schwartz'(\R^d)$. 
\begin{definition}[Weak-$\ast$ convergence]
	Let $\Schwartz$ be a metric space with dual $\Schwartz'$. A sequence $(L_n)$ in $\Schwartz'$ is said to converge to $L \in \Schwartz'$ in the weak-$\ast$ sense if 
	\begin{align*}
		L_n(f) \xrightarrow{n \to \infty} L(f) 
	\end{align*}
	holds for all $f \in \Schwartz$. We will write $\wastlim_{n \to \infty} L_n = L$. 
\end{definition}
This notion of convergence implies a notion of continuity and is crucial for the next theorem. 
\begin{theorem}\label{S_and_Sprime:Sprime:thm:weak-star_continuity}
	Let $A : \Schwartz(\R^d) \longrightarrow \Schwartz(\R^d)$ be a linear continuous map. Then for all $L \in \Schwartz'(\R^d)$, the map $A' : \Schwartz'(\R^d) \longrightarrow \Schwartz'(\R^d)$
	\begin{align}
		\bigl ( A' L , f \bigr ) := \bigl ( L , A f \bigr ) 
		, 
		&& 
		f \in \Schwartz(\R^d) 
		, 
	\end{align}
	defines a weak-$\ast$ continuous linear map. \marginpar{2014.01.16}
\end{theorem}
Put in the terms of Chapter~\ref{operators:adjoint}, $A'$ is the \emph{adjoint} of $A$. 
\begin{proof}
	First of all, $A'$ is linear and well-defined, $A' L$ maps $f \in \Schwartz(\R^d)$ onto $\C$. To show continuity, let $(L_n)$ be a sequence of tempered distributions which converges in the weak-$\ast$ sense to $L \in \Schwartz'(\R^d)$. Then 
	\begin{align*}
		\bigl ( A' L_n , f \bigr ) &= \bigl ( L_n , A f \bigr ) 
		\xrightarrow{n \to \infty} \bigl ( L , A f \bigr ) = \bigl ( A' L , f \bigr ) 
	\end{align*}
	holds for all $f \in \Schwartz(\R^d)$ and $A'$ is weak-$\ast$ continuous. 
\end{proof}
As a last consequence, we can extend the convolution from $\ast : \Schwartz(\R^d) \times \Schwartz(\R^d) \longrightarrow \Schwartz(\R^d)$ to 
\begin{align*}
	\ast &: \Schwartz'(\R^d) \times \Schwartz(\R^d) \longrightarrow \Schwartz'(\R^d) 
	\\ 
	\ast &: \Schwartz(\R^d) \times \Schwartz'(\R^d) \longrightarrow \Schwartz'(\R^d) 
	. 
\end{align*}
For any $f , g , h \in \Schwartz(\R^d)$, we can push the convolution from one argument of the duality bracket to the other, 
\begin{align*}
	\bigl ( f \ast g , h \bigr ) &= \bigl ( g \ast f , h \bigr ) 
	= \int_{\R^d} \dd y \, (f \ast g)(y) \, h(y) 
	= \int_{\R^d} \dd y \int_{\R^d} \dd x \, f(x) \, g(y-x) \, h(y) 
	\\
	&
	= \int_{\R^d} \dd x \, f(x) \, \bigl ( g(- \; \cdot) \ast h \bigr )(x) 
	% \\
	% &
	= \bigl ( f , g(- \; \cdot) \ast h \bigr ) 
	. 
\end{align*}
Thus, we define 
\begin{definition}[Convolution on $\Schwartz'(\R^d)$]
	Let $L \in \Schwartz'(\R^d)$ and $f \in \Schwartz(\R^d)$. Then the convolution of $L$ and $f$ is defined as 
	\begin{align}
		\bigl ( L \ast f , g \bigr ) := \bigl ( L , f(- \; \cdot) \ast g \bigr ) 
		&&
		\forall g \in \Schwartz(\R^d) 
		. 
		\label{S_and_Sprime:eqn:convolution_Sprime}
	\end{align}
\end{definition}
By Theorem~\ref{S_and_Sprime:Sprime:thm:weak-star_continuity}, this extension of the convolution is weak-$\ast$ continuous. Moreover, the convolution has a neutral element in $\Schwartz'(\R^d)$, the delta distribution $\delta = \delta_0$: for all $f , g \in \Schwartz(\R^d)$ 
\begin{align*}
	\bigl ( \delta \ast f , g \bigr ) &= \bigl ( \delta , f(- \; \cdot) \ast g \bigr ) 
	= \bigl ( f(- \; \cdot) \ast g \bigr )(0) 
	\\
	&= \int_{\R^d} \dd y \, f \bigl ( - (0 -y) \bigr ) \, g(y) 
	= (f , g)
\end{align*}
holds, and thus we can succinctly write 
\begin{align}
	\delta \ast f = f 
	. 
	\label{Fourier:R:eqn:delta_unit_convolution}
\end{align}
In view of this, we can better understand what Dirac sequences are (\cf Definitions~\ref{Fourier:T:defn:approximate_id} and \ref{Fourier:R:defn:approximate_id}): since integrable functions define tempered distributions (\cf equation~\eqref{S_and_Sprime:Sprime:eqn:distributions_of_functions}), any Dirac sequence $\bigl ( \delta_{\eps} \bigr )_{\eps \in (0,\eps_0)}$ can be seen as a sequence of tempered distributions. 
Moreover, the inclusion $\imath : L^1(\R^d) \longrightarrow \Schwartz'(\R^d)$ is continuous, the fact that $\delta_{\eps} \ast f$ converges to $f$ in $L^1(\R^d)$, this sequence also converges in the distributional sense. Hence, $\delta_{\eps} \rightarrow \delta$ holds in the distributional sense as $\eps \to 0$. 
% section tempered_distributions (end)

\section{Partial differential equations on $\Schwartz'(\R^d)$} % (fold)
\label{S_and_Sprime:PDEs_on_Sprime}
We have extended the most common operations, taking Fourier transform, derivatives and the convolution, from Schwartz functions to tempered distributions. Hence, we have managed to ascribe meaning to the partial differential equation 
\begin{align*}
	L U := \sum_{\abs{\alpha} \leq N} c(\alpha) \, \partial_x^{\alpha} U = F 
\end{align*}
even if $U , F \in \Schwartz'(\R^d)$ are tempered distributions, and we can ask whether $L U = F$ has a solution. More precisely, the above equation means that 
\begin{align*}
	\bigl ( L U , \varphi \bigr ) &= \sum_{\abs{\alpha} \leq N} (-1)^{\abs{\alpha}} \, c(\alpha) \, \bigl ( U , \partial_x^{\alpha} \varphi \bigr ) 
	= \bigl ( F , \varphi \bigr ) 
\end{align*}
holds for all test functions $\varphi \in \Schwartz(\R^d)$. This point of view is commonly used when considering differential equations and apply them to functions for which derivatives in the ordinary sense do not exist. 
\medskip

\noindent
To give an explicit example, let us reconsider the heat equation 
\begin{align*}
	\partial_t u(t) = D \Delta_x u(t) 
	, 
	&&
	u(0) = u_0 
	. 
\end{align*}
In the context of $L^1(\R^d)$, the unique solution $u(t) = G(t) \ast u_0$ to the initial value problem  (Theorem~\ref{Fourier:R:thm:uniqueness_solution_heat_equation_L1}) involves the fundamental solution 
\begin{align*}
	G(t,x) &= \frac{1}{(4 \pi D t)^{\nicefrac{d}{2}}} \, \e^{- \frac{x^2}{4 D t}}
	. 
\end{align*}
Seeing as $G(t)$ is a Gaußian for $t > 0$, it is also an element of $\Schwartz(\R^d)$, and thus convolving it with a \emph{bona fide} tempered distribution makes sense. It stands to reason that 
\begin{align*}
	U(t) := G(t) \ast U_0 
\end{align*}
solves the heat equation for the initial condition $U(0) = U_0 \in \Schwartz'(\R^d)$: first of all, $U(t)$ satisfies the initial condition $U(0) = U_0$, because for all $\varphi \in \Schwartz$ 
\begin{align*}
	\lim_{t \searrow 0} \bigl ( U(t) , \varphi \bigr ) &= \lim_{t \searrow 0} \bigl ( G(t) \ast U_0 , \varphi ) 
	\\
	&
	= \lim_{t \searrow 0} \bigl ( U_0 , G(t) \ast \varphi \bigr ) 
\end{align*}
holds. Going from the first to the second line involves the definition of the convolution on $\Schwartz'(\R^d)$, equation~\eqref{S_and_Sprime:eqn:convolution_Sprime}, as well as $G(t,-x) = G(t,x)$. Given that $G(t)$ is a Dirac sequence, the limit $\lim_{t \searrow 0} G(t) \ast \varphi = \varphi$ exists in $L^1(\R^d)$; more specifically, the limit converges to the \emph{Schwartz function} $\varphi$, and we deduce
\begin{align*}
	\lim_{t \searrow 0} \bigl ( U(t) , \varphi \bigr ) &= \bigl ( U_0 , \varphi \bigr ) 
	. 
\end{align*}
Moreover, $U(t)$ solves the heat equation on $\Schwartz'(\R^d)$: 
\begin{align*}
	\Bigl ( \tfrac{\dd}{\dd t} U(t) , \varphi \Bigr ) &= \frac{\dd}{\dd t} \bigl ( U_0 , G(t) \ast \varphi \bigr )
	= \bigl ( U_0 , D \, \Delta G(t) \ast \varphi \bigr )
	\\
	&= \bigl ( D \, \Delta_x G(t) \ast U_0 , \varphi \bigr ) 
\end{align*}
Hence, $U(t)$ is a solution to the heat equation with initial condition $U_0$. Note that just like in the case of integrable functions, showing uniqueness involves additional conditions on $\partial_t U(t)$ and the initial condition $U_0$. \marginpar{2014.01.21}
% section partial_differential_equations_on_schartz_r_d_ (end)

\section{Other common spaces of distributions} % (fold)
\label{S_and_Sprime:others}
The ideas outlined in the last two sections can be applied to other spaces of test functions: one starts with a space of “nice” functions and the distributions are then comprised of the linear continuous functionals on that space. Operations such as derivatives are extended to distributions via the adjoint operator. 

Instead of Schwartz functions, often $\Cont^{\infty}_{\mathrm{c}}(\R^d)$ is used. However, working with this space is a little bit more unwieldy as it is not stable under the Fourier transform -- which also implies that $\Fourier$ does \emph{not} extend to a map $\Fourier : \Cont^{\infty}_{\mathrm{c}}(\R^d)' \longrightarrow \Cont^{\infty}_{\mathrm{c}}(\R^d)'$. Moreover, one often works on bounded domains, \ie sufficiently regular bounded subsets $\Omega \subset \R^d$. Here, the distributions are the dual of $\Cont^{\infty}(\Omega)$. Smoothness is also optional, for instance, the Dirac distribution is defined also on bounded, continuous functions, $\Cont_{\mathrm{b}}(\R^d)$, as $\delta_{x_0}(f) = f(x_0)$. 
% section other_common_spaces_of_distributions (end)
% chapter schwartz_functions_and_tempered_distributionqs (end)
\chapter{Green's functions} % (fold)
\label{Greens_functions}
The basis of this chapter is equation~\eqref{Fourier:R:eqn:delta_unit_convolution}: assume we are interested in solutions of the \emph{inhomogeneous} equation 
\begin{align}
	L u = f
	\label{Greens_functions:eqn:fundamental_PDE}
\end{align}
where we may take the differential operator $L$ to be of the form 
\begin{align*}
	L = \sum_{\abs{\alpha} \leq N} c(\alpha) \, \partial_x^{\alpha}
	, 
\end{align*}
for instance. In addition, we may impose boundary conditions if this equation is defined on a subset of $\R^d$ with boundary. Formally, the solution to \eqref{Greens_functions:eqn:fundamental_PDE} can be written as $u = L^{-1} f$ in case $L$ is invertible, but clearly, this is not very helpful as is. A more fruitful approach starts with the observation that also in the distributional sense
\begin{align*}
	\partial_x^{\alpha} \bigl ( G \ast f \bigr ) &= (\partial_x^{\alpha} G) \ast f 
	= G \ast (\partial_x^{\alpha} f) 
\end{align*}
holds true for any $G \in \Schwartz'(\R^d)$ and $f \in \Schwartz(\R^d)$. Hence, \emph{if} we can write the solution $u = G \ast f$ as the convolution of the inhomogeneity $f$ with some tempered distribution $G$, then $G$ necessarily satisfies 
\begin{align*}
	(L u)(x) &= L \bigl ( G \ast f \bigr )(x)
	% \\
	% &
	= \bigl ( L G \ast f \bigr )(x)
	\\
	&\overset{!}{=} f(x) = \bigl ( \delta_x , f \bigr ) 
	, 
\end{align*}
and we immediately obtain an equation for $G$, the \emph{Green's function} or \emph{fundamental solution}, that is \emph{independent} of $f$: 
\begin{align}
	L G(x) = \delta_x
	\label{Greens_functions:eqn:Greens_ansatz}
\end{align}
Once we solve this equation for $G$, we obtain a solution of \eqref{Greens_functions:eqn:fundamental_PDE} by setting 
\begin{align*}
	u(x) &= (G \ast f)(x) 
	= \int_{\R^d} \dd y \, G(x-y) \, f(y) 
	. 
\end{align*}
The partial differential operator $L$ as given above has \emph{translational symmetry}. In a more general setting, where translational symmetry is absent, the Green's function depends on \emph{two} variables $G(x,y)$ and the solution here is related to the inhomogeneity via 
\begin{align*}
	u(x) = \int_{\R^d} \dd y \, G(x,y) \, f(y) 
	. 
\end{align*}
The purpose of this chapter is to compute Green's functions for specific cases and explore some of the caveats. Even though \emph{a priori} it is not clear that Green's functions are actually defined in terms of a \emph{function} (as opposed to a \emph{bona fide} distribution), in many cases it turns out they are.

\section{Matrix equations as a prototypical example} % (fold)
\label{Greens_functions:matrix_eqns}
Another way to understand Green's functions is to appeal to the theory of matrices: assume $A \in \mathrm{Mat}_{\C}(n)$ is invertible and we are looking for solutions $x \in \C^n$ of the equation 
\begin{align*}
	A x = y
\end{align*}
for some fixed $y \in \C^n$. We can expand $y = \sum_{j = 1}^n y_j \, e_j$ in terms of the canonical basis vectors $e_j = ( 0 , \ldots , 0 , 1 , 0 , \ldots)$, and if we solve 
\begin{align*}
	A g_j = e_j 
\end{align*}
for all $j = 1 , \ldots , n$, then we obtain $x = \sum_{j = 1} y_j \, g_j$ as the solution of $A x = y$. 
Moreover, the matrix $G := \bigl ( g_1 \vert \cdots \vert g_n \bigr )$ whose columns are comprised of the vectors $g_j$ satisfies 
\begin{align*}
	x = G y
	. 
\end{align*}
Put another way, here $G$ is just the \emph{matrix inverse} of $A$. This already points to one fundamental obstacle for the existence of Green's functions, namely it hinges on the invertibility of $A$. 

The story is more complicated if $A$ is not invertible. For instance, the equation $A x = y$ also makes sense in case $A \in \mathrm{Mat}_{\C}(n,m)$ is a rectangular matrix, and the vectors $x \in \C^m$ and $y \in \C^n$ are from vector spaces of different dimension. Here, it is not clear whether a unique \emph{left}-inverse $G \in \mathrm{Mat}_{\C}(m,n)$ exists: there may be cases when one can find no such $G$ ($A x = y$ has no solution) or when there is a family of left-inverses ($A x = y$ does not have a unique solution). 
\medskip

\noindent
In the same vein, the Green's \emph{function} $G(x,y)$ gives the response at $x$ to a unit impulse at $y$. The solution $u(x) = G(x) \ast f$ at $x$ can be seen as the “infinite superposition” of $G(x,y) \, f(y)$ where the unit impulse at $y$ is scaled by the inhomogeneity $f(y)$. 
% section matrix_equations_as_a_prototypical_example (end)

\section{Simple examples} % (fold)
\label{Greens_functions:simple_examples}
To exemplify the general method, we solve \eqref{Greens_functions:eqn:Greens_ansatz} for a particularly simple case, the \emph{one-dimensional Poisson equation}
\begin{align}
	- \partial_x^2 u = f 
	. 
	\label{Greens_functions:eqn:1d_Poisson}
\end{align}
According to our discussion, we first need to solve 
\begin{align}
	- \partial_x^2 G(x,y) = \delta(x-y)
	\label{Greens_functions:eqn:1d_Poisson_fundamental}
\end{align}
for $G$. Put another way, $G$ is the second anti-derivative of the Dirac distribution $\delta$ which can be found “by hand”, namely 
\begin{align*}
	G(x,y) &= - \rho(x-y) + a x + b
\end{align*}
where 
\begin{align*}
	\rho(x) &= 
	\begin{cases}
		0 & x \in (-\infty,0) \\
		x & x \in [0,+\infty) \\
	\end{cases}
	. 
\end{align*}
It is instructive to verify \eqref{Greens_functions:eqn:1d_Poisson_fundamental}: clearly, the term $a x + b$ is harmless, because for smooth functions weak derivatives (meaning derivatives on $\Schwartz'(\R^d)$) and ordinary derivatives (of $\Cont^k(\R^d)$ functions) coincide. Thus, $- \partial_x^2 G(x,y) = - \partial_x^2 \rho(x-y)$ holds and \eqref{Greens_functions:eqn:1d_Poisson_fundamental} follows from the fact that the derivative of the \emph{Heavyside function} 
\begin{align*}
	\theta(x) &= 
	\begin{cases}
		0 & x \in (-\infty,0) \\
		1 & x \in [0,+\infty) \\
	\end{cases}
\end{align*}
is the Dirac distribution. Note that it does not matter how we define $\rho(0)$ and $\theta(0)$, this is just a modification on a set of measure $0$. For instance, it is the size of the jump 
\begin{align*}
	\lim_{x \nearrow 0} \theta(x) - \lim_{x \searrow 0} \theta(x) = 1 
\end{align*}
which matters, and its size is independent of how we define $\theta(0)$ and $\rho(0)$. 

So far we have derived the Green's function $G$ for the one-dimensional Poisson equation on all of $\R$. $G$ depends on the two parameters $a , b \in \R$, \ie $G$ is \emph{not} unique. Additional conditions, \eg \emph{boundary conditions}, are needed to nail down a \emph{unique} solution. For instance, we could restrict \eqref{Greens_functions:eqn:1d_Poisson} to the interval $[0,1]$ and use Dirichlet boundary conditions $u(0) = 0 = u(1)$. Note that we need \emph{two} condition to fix the values of the \emph{two} parameters $a$ and $b$. To determine the values of $a$ and $b$ (which depend on the second variable $y$), we solve $G(0,y) = b = 0$ and 
\begin{align*}
	 G(1,y) &= - (1 - y) + a = 0 
\end{align*}
for $a$ and $b$, and obtain 
\begin{align*}
	G(x,y) &= - \rho(x-y) + (1-y) \, x 
	= 
	\begin{cases}
		(1-y) \, x & x \in [0,y] \\
		y \, (1 - x) & x \in (y,1] \\
	\end{cases}
	. 
\end{align*}
The solution $u(x) = G(x) \ast f$ satisfies the boundary conditions: from $G(0,y) = 0$ we immediately deduce
\begin{align*}
	u(0) = \int_0^1 \dd y \, G(0,y) \, f(y) = 0 
\end{align*}
and similarly $u(1) = 0$ follows from $G(1,y) = 0$. 
\medskip

\noindent
Another example is the PDE from homework problem 42, 
\begin{align*}
	L u := \bigl ( \partial_{x_1}^2 + 2 \partial_{x_2}^2 + 3 \partial_{x_1} - 4 \bigr ) u = f 
	, 
\end{align*}
where the Fourier representation of $L$ is the multiplication operator associated to the polynomial $P(\xi) = - \xi_1^2 - 2 \xi_2^2 + \ii \, 3 \xi_1 - 4$. The inverse of this polynomial enters the solution 
\begin{align*}
	u(x) = \frac{1}{2 \pi} \, \bigl ( \Fourier^{-1} (\nicefrac{1}{P}) \bigr ) \ast f (x) 
	, 
\end{align*}
one can then read off the Green's function as 
\begin{align*}
	G(x,y) = \frac{1}{2 \pi} \, \bigl ( \Fourier^{-1} (\nicefrac{1}{P}) \bigr )(x-y) 
	. 
\end{align*}
%
% section simple_examples (end)

\section{Green's functions on $\R^d$} % (fold)
\label{Greens_functions:Rd}
The second example gives a strategy to \emph{compute} Green's functions on $\R^d$: the crucial ingredient here is the translational symmetry of the differential operator $L$, \ie $L$ commutes with the translation operator $(T_y u)(x) := u(x-y)$, $y \in \R^d$. Explicitly, we will discuss the Poisson equation 
\begin{align}
	- \Delta_x u = f
	\label{Greens_functions:eqn:Poisson_eqn}
\end{align}
in dimension $2$ and $3$ as well as the three-dimensional wave equation
\begin{align}
	\bigl ( \tfrac{1}{c^2} \partial_t^2 - \Delta_x \bigr ) u = f 
	. 
	\label{Greens_functions:eqn:wave_eqn}
\end{align}
Since the strategy to solve these equations is identical, we will outline how to solve $L u = f$ for a differential operator of the form 
\begin{align*}
	L = \sum_{\abs{\alpha} \leq N} c(\alpha) \, \partial_x^{\alpha} 
	. 
\end{align*}
Then in Fourier representation, 
\begin{align*}
	\Fourier \bigl ( L u \bigr ) &= \Fourier \, L \, \Fourier^{-1} \, \Fourier u 
	= P \hat{u}
	\overset{!}{=} \hat{f}
	, 
\end{align*}
the differential operator $L$ transforms to multiplication by the polynomial 
\begin{align*}
	P(\xi) = \sum_{\abs{\alpha} \leq N} \ii^{\abs{\alpha}} \, c(\alpha) \, \xi^{\alpha} 
	,  
\end{align*}
and $u(x) = G(x) \ast f$ can be expressed as the convolution of the Green's function 
\begin{align*}
	G(x,y) := (2\pi)^{\nicefrac{d}{2}} \, \bigl ( \Fourier^{-1} (\nicefrac{1}{P}) \bigr )(x-y) 
\end{align*}
with the inhomogeneity $f$. 

In other words, the problem of \emph{finding the Green's function reduces to computing the inverse Fourier transform of the rational function $\tfrac{1}{P}$}. For instance, one can interpret $\Fourier^{-1} (\nicefrac{1}{P})$ as an integral in the complex plane, use the method of partial fractions and employ Cauchy's integral formula. 
\medskip

\noindent
For the special case of the Poisson equation $- \Delta_x u = f$, another way to obtain the Green's function relies on Green's formula 
\begin{align}
	\int_{V} \dd x \, \bigl ( \Delta_x u(x) \; v(x) - u(x) \; \Delta_x v(x) \bigr ) = \int_{\partial V} \dd S \cdot \bigl ( \partial_n u(x) \; v(x) - u(x) \; \partial_n v(x) \bigr )
	. 
	\label{Greens_functions:eqn:Greens_formula}
\end{align}
Here, $V$ is a subset of $\R^n$ with boundary $\partial V$ and $u$ or $v$ has compact support. 
\begin{theorem}[Green's function for the Poisson equation]\label{Greens_functions:thm:Poisson_Green}
	The Green's function for the Poisson equation is 
	\begin{align*}
		G(x,y) = 
		\begin{cases}
			\tfrac{1}{C_d} \, \abs{x-y}^{2-d} & d \neq 2 \\
			- \frac{1}{2\pi} \, \ln \abs{x-y} & d = 2 \\
		\end{cases}
	\end{align*}
	where $C_d = (d-2) \, \mathrm{Area}(\mathbb{S}^{d-1})$ and $\mathrm{Area}(\mathbb{S}^{d-1})$ is the surface area of the $d-1$-dimensional sphere. 
\end{theorem}
\begin{proof}
	Let us prove the case $d \neq 2$, the arguments for two-dimensional case are virtually identical, one only needs to replace $\abs{x}^{d-2}$ with $\ln \abs{x}$. First of all, $\abs{x}^{2-d}$ and $\ln \abs{x}$ define tempered distributions (use polar coordinates to show and arguments analogous to Lemma~\ref{S_and_Sprime:schwartz_functions:lem:Lp_estimate} to prove continuity). Moreover, away from $x = 0$, $\Delta_x \abs{x}^{2-d} = 0$ and $\Delta_x \ln \abs{x} = 0$ hold in the ordinary sense. 
	
	Now let $\varphi \in \Schwartz(\R^d)$ be arbitrary. Since smooth function with compact support are dense in $\Schwartz(\R^d)$, we may assume $\varphi \in \Cont^{\infty}_{\mathrm{c}}(\R^d)$. Define $V_{\eps} := \R^d \setminus B_{\eps}$ where $B_{\eps} := \bigl \{ x \in \R^d \; \vert \; \abs{x} < \eps \bigr \}$. On $V_{\eps}$ where we have cut out a small hole around the origin, $\Delta_x \abs{x}^{2-d} = 0$ holds and we can write \marginpar{2014.01.23}
	\begin{align*}
		\Bigl ( - \Delta_x \abs{x}^{2-d} \, , \, \varphi \Bigr ) &= - \Bigl ( \abs{x}^{2-d} \, , \, \Delta_x \varphi \Bigr ) 
		= - \lim_{\eps \searrow 0} \int_{V_{\eps}} \dd x \, \abs{x}^{2-d} \, \Delta_x \varphi(x) 
		\\
		&= \lim_{\eps \searrow 0} \int_{V_{\eps}} \dd x \, \Bigl ( \Delta_x \abs{x}^{2-d} \, \varphi(x) - \abs{x}^{2-d} \, \Delta_x \varphi(x) \Bigr )
		. 
	\end{align*}
	Now Green's formula applies, 
	\begin{align*}	
		\ldots
		&= \lim_{\eps \searrow 0} \int_{\partial V_{\eps}} \dd S \cdot \Bigl ( \partial_r \abs{x}^{2-d} \, \varphi(x) - \abs{x}^{2-d} \, \partial_r \varphi(x) \Bigr )
		\\
		&= - \lim_{\eps \searrow 0} \int_{\partial B_{\eps}} \dd S \cdot \Bigl ( (2-d) \, \eps^{1-d} \, \varphi(x) - \eps^{2-d} \, \partial_r \varphi(x) \Bigr )
		, 
	\end{align*}
	and we obtain an integral with respect to $\dd S$, the surface measure of the sphere of radius $\eps$. Note that the minus sign is due to the difference in orientation (the outward normal on $\partial V_{\eps}$ points towards the origin while the surface normal of $\partial B_{\eps}$ points away from it). Since the surface area of $B_{\eps}$ scales like $\eps^{d-1}$, the second term vanishes while the first term converges to 
	\begin{align*}
		\Bigl ( - \Delta_x \abs{x}^{2-d} \, , \, \varphi \Bigr ) &= (d-2) \, \mathrm{Area}(\mathbb{S}^{d-1}) \, \varphi(0) 
		\\
		&= (d-2) \, \mathrm{Area}(\mathbb{S}^{d-1}) \, (\delta , \varphi) 
		. 
	\end{align*}
	%
	% This concludes the proof. 
\end{proof}
%
% section green_s_functions_on_r_d_ (end)

\section{Green's functions on domains: implementing boundary conditions} % (fold)
\label{Greens_functions:domains}
Predictably, the presence of boundaries complicates things. To exemplify some of the hurdles, let us discuss the Poisson equation 
\begin{align}
	- \Delta_x u = f 
	, 
	\qquad \qquad 
	u \vert_{\partial \Omega} = h \in \Cont^1(\partial \Omega)
	, 
	\label{Greens_function:eqn:Poisson_boundary}
\end{align}
on a bounded subset $\Omega$ of $\R^2$ with Dirichlet boundary conditions, \ie we prescribe the value of $u$ on the boundary (rather than that of its normal derivative). 

Imposing boundary conditions is necessary, because $- \Delta_x$ is not injective; The kernel of $- \Delta_x$ is made up of \emph{harmonic functions}, \ie functions which satisfy 
\begin{align*}
	- \Delta_x h = 0 
	. 
\end{align*}
For instance, the function $h(x) = \e^{x_1} \sin x_2$ is harmonic. Consequently, if $u$ is \emph{a} solution to \eqref{Greens_function:eqn:Poisson_boundary}, then $u + h$ is another. However, it turns out that fixing $u$ on the boundary $\partial \Omega$ singles out a \emph{unique} solution. 

The first step in this direction is the following representation of a sufficiently regular function on $\Omega$: 
\begin{proposition}\label{Greens_functions:prop:rep_solution_Poisson_domain}
	Suppose $\Omega \subset \R^2$ be a bounded domain with piecewise $\Cont^1$ boundary $\partial \Omega$ and pick $u \in \Cont^2(\Omega) \cap \Cont^1(\bar{\Omega})$. Then for any $x \in \Omega$, we can write 
	\begin{align*}
		u(x) &= - \int_{\Omega} \dd y \, G(x,y) \, \Delta_y u(y) \, + 
		\\
		&\qquad + 
		\int_{\partial \Omega} \dd S_y \cdot \Bigl ( G(x,y) \, \partial_{n_y} u(y) - \partial_{n_y} G(x,y) \, u(y) \Bigr ) 
	\end{align*}
	where the index $y$ in the surface measure $\dd S_y$ and the surface normal $n_y$ indicates that they are associated to the variable $y$. 
\end{proposition}
The idea here is to exploit $- \Delta_x G(x,y) = \delta(x-y)$ as well as Green's formula~\eqref{Greens_functions:eqn:Greens_formula}, and \emph{formally}, we immediately obtain 
\begin{align*}
	\int_{\Omega} \dd x \, &\Bigl ( \Delta_y G(x,y) \; u(y) - G(x,y) \; \Delta_y u(y) \Bigr ) = 
	\\
	&\qquad \qquad \qquad 
	= - u(x) - \int_{\Omega} \dd x \, G(x,y) \; \Delta_y u(y) 
	\\
	&\qquad \qquad \qquad 
	= \int_{\partial \Omega} \dd S \cdot \Bigl (\partial_n G(x,y) \; u(y) - G(x,y) \; \partial_n u(y) \Bigr )
\end{align*}
To make this rigorous, one has to cut a small hole around $x$ and adapt the strategy from the proof of Theorem~\ref{Greens_functions:thm:Poisson_Green}. However, we will skip the proof. 
\begin{theorem}\label{Greens_functions:thm:solution_Poisson_domain}
	The Green's function for the Dirichlet boundary value problem for the Poisson equation on a bounded domain $\Omega \subset \R^2$ with $\Cont^1$ boundary $\partial \Omega$ has the form 
	\begin{align*}
		G_{\Omega}(x,y) = G(x,y) + b(x,y)
	\end{align*}
	where $G(x,y) = - \frac{1}{2\pi} \, \ln \abs{x-y}$ is the Green's function of the free problem and $b(x,y)$ is the solution to the boundary value problem 
	\begin{align}
		\Delta_x b = 0 
		, 
		\qquad \qquad 
		b(x,y) = - G(x,y) 
		\; \forall x \in \partial \Omega \mbox{ or } \forall y \in \partial \Omega 
		. 
		\label{Greens_function:eqn:Greens_function_harmonic_term}
	\end{align}
	Then the solution $u$ to the inhomogeneous Dirichlet problem~\eqref{Greens_function:eqn:Poisson_boundary} is given by 
	\begin{align}
		u(x) &= \int_{\Omega} \dd y \, G_{\Omega}(x,y) \, f(y) - \int_{\partial \Omega} \dd S_y \cdot \partial_{n_y} G_{\Omega}(x,y) \, h(y) 
		. 
		\label{Greens_functions:eqn:solution_Poisson_domain}
	\end{align}
\end{theorem}
A priori it is not clear whether \eqref{Greens_function:eqn:Greens_function_harmonic_term} has a solution; we postpone a more in-depth discussion of harmonic functions and proceed under the assumption it does. 
\begin{proof}
	Clearly, seeing as $G_{\Omega}$ is the sum of the free fundamental solution $G$ and a harmonic function, $\Delta G_{\Omega}(x,y) = \delta(x-y)$ still holds. Moreover, the Green's function implements the boundary conditions, namely 
	\begin{align*}
		G_{\Omega}(x,y) = G(x,y) + b(x,y) = 0 
	\end{align*}
	is satisfied by construction on the boundary. 
	
	Then Green's identity~\eqref{Greens_functions:eqn:Greens_formula} and $\Delta_x b = 0$ imply 
	\begin{align*}
		0 &= - \int_{\Omega} \dd y \, b(x,y) \, \Delta_y u(y) 
		+ \int_{\partial \Omega} \dd S_y \cdot \Bigl ( b(x,y) \; \partial_{n_y} u(y) - \partial_{n_y} b(x,y) \; u(y) \Bigr ) 
		% \\
		% &= \int_{\Omega} \dd y \, b(x,y) \, f(y) 
		% + \int_{\partial \Omega} \dd S_y \cdot \Bigl ( -G(x,y) \; \partial_{n_y} u(y) - \partial_{n_y} b(x,y) \; h(y) \Bigr )
		. 
	\end{align*}
	The function $u$ solves the inhomogeneous Poisson equation $- \Delta_x u = f$; On the boundary, $b(x,y)$ coincides with $- G(x,y)$ and $u(y) = h(y)$ holds, and we deduce 
	\begin{align*}
		\int_{\partial \Omega} \dd S_y \cdot G(x,y) \; \partial_{n_y} u(y) &= - \int_{\Omega} \dd y \, b(x,y) \, f(y) 
		- \int_{\partial \Omega} \dd S_y \cdot \partial_{n_y} b(x,y) \; h(y) 
		. 
	\end{align*}
	This term cancels one of the boundary terms in the integral representation of  Proposition~\ref{Greens_functions:prop:rep_solution_Poisson_domain}, and we recover equation~\eqref{Greens_functions:eqn:solution_Poisson_domain}, 
	\begin{align*}
		u(x) &= - \int_{\Omega} \dd y \, G(x,y) \, \Delta_y u(y) \, + 
		\\
		&\qquad + 
		\int_{\partial \Omega} \dd S_y \cdot \Bigl ( G(x,y) \, \partial_{n_y} u(y) - \partial_{n_y} G(x,y) \, u(y) \Bigr ) 
		\\
		&= \int_{\Omega} \dd y \, G_{\Omega}(x,y) \, f(y) - \int_{\partial \Omega} \dd S_y \cdot \partial_{n_y} G_{\Omega}(x,y) \, h(y) 
		. 
	\end{align*}
\end{proof}
%
% section green_s_functions_on_domains_implementing_boundary_conditions (end)
% section Green's functions (end)
% chapter schwartz_functions_and_tempered_distributionqs (end)
\chapter{Quantum mechanics} % (fold)
\label{quantum}
The explanation of the photoelectric effect through light \emph{quanta} is the name sake for quantum mechanics. Quantization here refers to the idea that energy stored in light comes in “chunks” known as \emph{photons}, and that the energy per photon depends only on the frequency. This is quite a departure from the classical theory of light through Maxwell's equations (\cf Chapter~\ref{operators:Maxwell}). 
% CHANGED add book recommendations 

The reader can only get a glimpse of quantum theory in this chapter. 
A good standard physics textbook on the subject is \cite{Sakurai:quantum_mechanics:1994} while the mathematics of quantum mechanics is covered in more depth in \cite{Teschl:quantum_mechanics:2009,Gustafson_Sigal:quantum_mechanics:2011}.

\section{Paradigms} % (fold)
\label{quantum:paradigms}
% 
% TODO add picture Stern-Gerlach
The simplest \emph{bona fide} quantum system is that of a quantum spin, and it can be used to give an effective description of the \emph{Stern-Gerlach experiment} where a beam of neutral atoms with magnetic moment $g$ is sent through a magnet with inhomogeneous magnetic field $B = (B_1 , B_2 , B_3)$. It was observed experimentally that the beam splits in two rather than fan out with continuous distribution. Hence, the system behaves as if only two spin configurations, spin-up $\uparrow$ and spin-down $\downarrow$, are realized. A simplified (effective) model neglects the translational degree of freedom and focusses only on the internal spin degree of freedom. 
% FIXME mention 2013 paper
Then the energy observable, the \emph{hamiltonian}, is the matrix 
\begin{align*}
	H = g B \cdot S
\end{align*}
which involves the spin operator $S_j := \frac{\hbar}{2} \sigma_j$ defined in terms of Planck's constant $\hbar$ and the three Pauli matrices
\begin{align*}
	\sigma_1 = \left (
	\begin{matrix}
		0 & 1 \\
		1 & 0 \\
	\end{matrix}
	\right )
	, 
	\qquad \qquad 
	\sigma_2 = \left (
	\begin{matrix}
		0 & - \ii \\
		+ \ii & 0 \\
	\end{matrix}
	\right )
	, 
	\qquad \qquad 
	\sigma_3 = \left (
	\begin{matrix}
		+1 & 0 \\
		0 & -1 \\
	\end{matrix}
	\right ) 
	, 
\end{align*}
and the magnetic moment $g$ and the magnetic field $B$. The prefactor of the Pauli matrices are real, and thus $H = H^*$ is a hermitian matrix. 

For instance, assume $B = (0,0,b)$ points in the $x_3$-direction. Then spin-up and spin-down (seen from the $x_3$-direction) are the \emph{eigenvectors} of 
\begin{align*}
	H = 
	\left (
	\begin{matrix}
		+ \frac{\hbar g b}{2} & 0 \\
		0 & - \frac{\hbar g b}{2} \\
	\end{matrix}
	\right )
	, 
\end{align*}
\ie $\psi_{\uparrow} = (1,0)$ and $\psi_{\downarrow} = (0,1)$. The dynamical equation is the \emph{Schrödinger equation} 
\begin{align}
	\ii \, \hbar \, \frac{\partial}{\partial t} \psi(t) &= H \psi(t) 
	, 
	&&
	\psi(0) = \psi_0 \in \Hil 
	. 
	\label{quantum:eqn:Schroedinger_eqn}
\end{align}
The vector space $\Hil = \C^2$ becomes a Hilbert space if we equip it with the scalar product 
\begin{align*}
	\scpro{\psi}{\varphi}_{\C^2} := \sum_{j = 1 , 2} \overline{\psi_j} \, \varphi_j 
	. 
\end{align*}
Moreover, the hermitian matrix $H$ can always be diagonalized (\cf exercise~35--36), and the eigenvectors to distinct eigenvalues are orthogonal. The complex-valued \emph{wave function} $\psi$ encapsulates probabilities: for any $\psi \in \C^2$ normalized to $1 = \norm{\psi}_{\C^2}$, the probability to find the particle in the spin-up configuration is 
\begin{align*}
	\mathbb{P}(\mathsf{S}_3 = \uparrow) = \sabs{\psi_1}^2 
	= \babs{\sscpro{\psi_{\uparrow}}{\psi}}^2 
\end{align*}
since $\psi_{\uparrow} = (1,0)$. The above notation comes from probability theory and means “the probability of finding the random observable spin $\mathsf{S}_3$ in the spin-$\uparrow$ configuration $+ \frac{\hbar}{2}$”. \marginpar{2014.02.04}
\medskip

\noindent
The second exemplary quantum system describes a non-relativistic particle of mass $m$ subjected to an electric field generated by the potential $V$. The classical Hamilton \emph{function} $h(q,p) = \tfrac{1}{2m} p^2 + V(q)$ is then “quantized” to 
\begin{align*}
	h \bigl ( \hat{x},- \ii \hbar \nabla_x \bigr ) = H 
	= \frac{1}{2 m} \, \bigl ( - \ii \hbar \nabla_x \bigr )^2 + V(\hat{x}) 
\end{align*}
by replacing momentum $p$ by the momentum \emph{operator} $\mathsf{P} = - \ii \hbar \nabla_x$ and position $q$ by the multiplication operator $\mathsf{Q} = \hat{x}$.\footnote{To find a consistent quantization procedure is highly non-trivial. One possibility is to use Weyl quantization \cite{Weyl:qm_gruppentheorie:1927,Wigner:Wigner_transform:1932,Moyal:Weyl_calculus:1949,Folland:harmonic_analysis_hase_space:1989,Lein:quantization_semiclassics:2010}. Such a quantization procedure also yields a formulation of a semiclassical limit, and the names for various operators (\eg position, momentum and angular momentum) are then justified via a semiclassical limit. For instance, the momentum operator is $- \ii \hbar \nabla_x$, because in the semiclassical limit it plays the role of the classical momentum observable $p$ (\cf \eg \cite[Theorem~1.0.1]{Lein:quantization_semiclassics:2010} and \cite[Theorem~7.0.1]{Lein:quantization_semiclassics:2010}).} The hamiltonian is now an operator on the Hilbert space $L^2(\R^d)$ whose action on suitable vectors $\psi$ is 
\begin{align*}
	(H \psi)(x) &= - \frac{\hbar^2}{2 m} (\Delta_x \psi)(x) + V(x) \, \psi(x) 
	. 
\end{align*}
\begin{figure}
	\hfil\includegraphics[height=4cm]{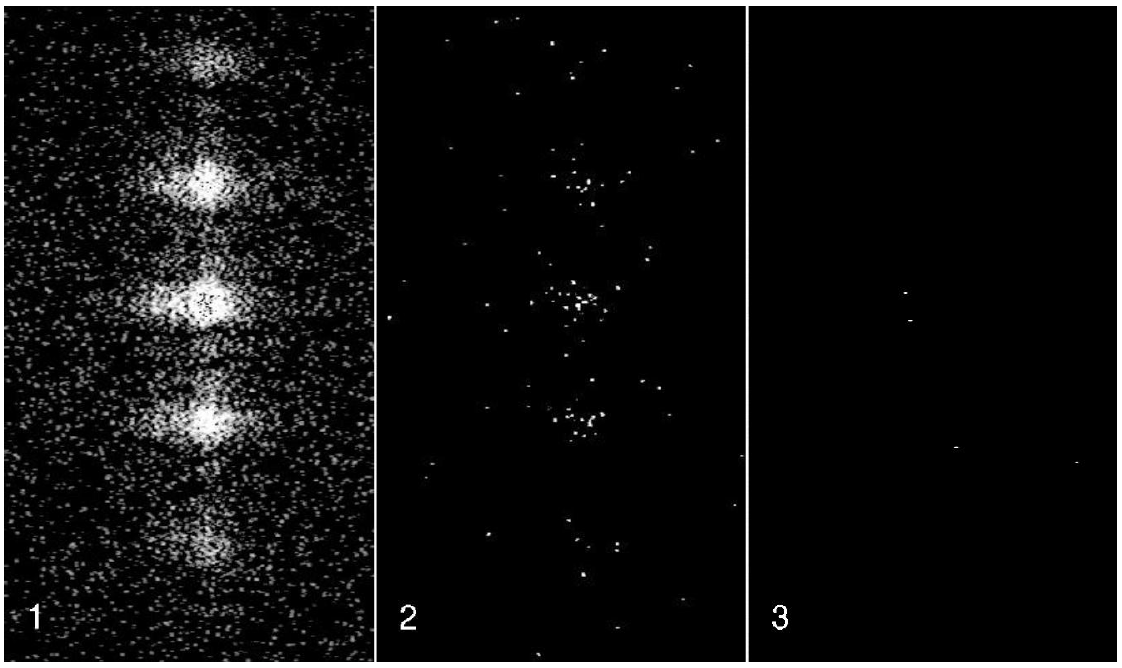}\hfil
	\caption{Images of a low-intensity triple slit experiment with photons (taken from \cite{Crowell:modern_revolution_physics:2008}). }
	\label{quantum:fig:low_intensity}
\end{figure}
Quantum particles simultaneously have wave and particle character: the Schrödinger equation~\eqref{quantum:eqn:Schroedinger_eqn} is structurally very similar to a wave equation. The physical constant $\hbar$ relates the energy of a particle with the associated wave length and has units $[\mathrm{energy} \cdot \mathrm{time}]$. The particle aspects comes when one measures outcomes of experiments: consider a version of the Stern-Gerlach experiment where the intensity of the atomic beam is so low that single atoms pass through the magnet. If the modulus square of the wave function $\sabs{\psi(t,x)}^2$ were to describe the intensity of a \emph{matter} wave, then one expects that the two peaks build up slowly, but \emph{simultaneously}. In actuality, one registers single impacts of atoms and only if one waits long enough, two peaks emerge (similar to what one sees in a low-intensity triple slit experiment in Figure~\ref{quantum:fig:low_intensity}). This is akin to tossing a coin: one cannot see the probabilistic nature in a few coin tosses, let alone a single one. Probabilities emerge only after repeating the experiment often enough. These experiments show that $\sabs{\psi(t,x)}^2$ is to be interpreted as a probability distribution, but more on that below. 

Pure states are described by wave functions, \ie complex-valued, square integrable functions. Put more precisely, we are considering $L^2(\R^d)$ made up of equivalence classes of functions with scalar product 
\begin{align*}
	\scpro{\varphi}{\psi} = \int_{\R^d} \dd x \, \overline{\varphi(x)} \, \psi(x) 
\end{align*}
and norm $\norm{\psi} := \sqrt{\scpro{\psi}{\psi}}$. In physics text books, one usually encounters the the \emph{bra-ket} notation: here $\ket{\psi}$ is a state and $\ipro{x}{\psi}$ is $\psi(x)$. The scalar product of $\phi,\psi \in L^2(\R^d)$ is denoted by $\ipro{\phi}{\psi}$ and corresponds to $\scpro{\phi}{\psi}$. Although bra-ket notation can be ambiguous, it is sometimes useful and in fact used in mathematics every once in a while. 

The fact that $L^2(\R^d)$ consists of \emph{equivalence classes} of functions is only natural from a physical perspective: if $\psi_1 \sim \psi_2$ are in the same equivalence class (\ie they differ on a set of measure $0$), then the arguments in Chapter~\ref{spaces:Banach:prototypical} state that the associated probabilities coincide: Physically, $\abs{\psi(x,t)}^2$ is interpreted as the \emph{probability to measure a particle at time $t$ in (an infinitesimally small box located in) location $x$}. If we are interested in the probability that we can measure a particle in a region $\Lambda \subseteq \R^d$, we have to integrate $\abs{\psi(x,t)}^2$ over $\Lambda$, 
\begin{align}
	\mathbb{P}(X(t) \in \Lambda) = \int_{\Lambda} \dd x \, \abs{\psi(x,t)}^2 
	. 
\end{align}
If we want to interpret $\abs{\psi}^2$ as \emph{probability} density, the wave function has to be \emph{normalized}, \ie 
\begin{align*}
	\norm{\psi}^2 = \int_{\R^d} \dd x \, \abs{\psi(x)}^2 = 1 
	. 
\end{align*}
This point of view is called \emph{Born rule}: $\abs{\psi}^2$ could either be a \emph{mass} or \emph{charge density} -- or a \emph{probability density}. To settle this, physicists have performed the double slit experiment with an electron source of low flux. If $\abs{\psi}^2$ were a density, one would see the whole interference pattern building up slowly. Instead, one \emph{measures} “single impacts” of electrons and the result is similar to the data obtained from experiments in statistics (\eg the Dalton board). Hence, we speak of \emph{particles}. 
% section paradigms (end)

\section{The mathematical framework} % (fold)
\label{quantum:framework}
To identify the common structures, let us study quantum mechanics in the abstract. Just like in the case of classical mechanics, we have to identify \emph{states}, \emph{observables} and \emph{dynamical equations} in \emph{Schrödinger and Heisenberg picture}.

\subsection{Quantum observables} % (fold)
\label{quantum:frameworks:observables}
Quantities that can be measured are represented by selfadjoint (hermitian in physics parlance) operators $F$ on the Hilbert space $\Hil$ (typically $L^2(\R^d)$), \ie special linear maps 
\begin{align*}
	F : \mathcal{D}(F) \subseteq \Hil \longrightarrow \Hil 
	. 
\end{align*}
Here, $\mathcal{D}(F)$ is the \emph{domain} of the operator since \emph{typical} observables are not defined for all $\psi \in \Hil$. \emph{This is not a mathematical subtlety with no physical content}, quite the contrary: consider the observable energy, typically given by 
\begin{align*}
	H = \frac{1}{2m} (- \ii \hbar \nabla_x)^2 + V(\hat{x}) 
	, 
\end{align*}
then states in the domain 
\begin{align*}
	\mathcal{D}(H) := \Bigl \{ \psi \in L^2(\R^d) \; \big \vert \; H \psi \in L^2(\R^d) \Bigr \} \subseteq L^2(\R^d)
\end{align*}
are those of \emph{finite energy}. For all $\psi$ in the domain of the hamiltonian $\mathcal{D}(H) \subseteq L^2(\R^d)$, the expectation value 
\begin{align*}
	\bscpro{\psi}{H \psi} < \infty
\end{align*}
is bounded. Well-defined observables have domains that are \emph{dense} in $\Hil$. Similarly, states in the domain $\mathcal{D}(\hat{x}_l)$ of the $l$th component of the position operator are those that are “localized in a finite region” in the sense of expectation values. Boundary conditions may also enter the definition of the domain: as seen in the example of the momentum operator on $[0,1]$, different boundary conditions yield different momentum operators (see Chapter~\ref{operators:unitary} for details). 

The energy observable is just a specific example, but it contains all the ingredients which enter the definition of a quantum observable: 
\begin{definition}[Observable]
	A quantum observable $F$ is a densely defined, selfadjoint operator on a Hilbert space. The spectrum $\sigma(F)$ (\cf Definition~\ref{operators:defn:spectrum}) is the set of outcomes of measurements. 
\end{definition}
Physically, results of measurements are real which is reflected in the selfadjointness of operators (\cf Chapter~\ref{operators:selfadjoint_operators}), $H^* = H$. (A symmetric operator is selfadjoint if $\mathcal{D}(H^{\ast}) = \mathcal{D}(H)$.) The set of possible outcomes of measurements is the spectrum $\sigma(H) \subseteq \R$ (the spectrum is defined as the set of complex numbers so that $H - z$ is not invertible, \cf Chapter~\ref{operators:defn:spectrum}). Spectra of selfadjoint operators are necessarily subsets of the reals (\cf Theorem~\ref{quantum:thm:spectrum_selfadjoint_real}). Typically one “guesses” quantum observables from classical observables: in $d = 3$, the angular momentum operator is given by 
\begin{align*}
	L = \hat{x} \wedge  (- \ii \hbar \nabla_x) 
	. 
\end{align*}
In the simplest case, one uses Dirac's recipe (replace $x$ by $\hat{x}$ and $p$ by $\hat{p} = - \ii \hbar \nabla_x$) on the classical observable angular momentum $L(x,p) = x \wedge p$. In other words, \emph{many quantum observables are obtained as quantizations of classical observables}: examples are position, momentum and energy. Moreover, the \emph{interpretation} of, say, the angular momentum operator as angular momentum is taken from classical mechanics. 

In the definition of the domain, we have already used the definition of expectation value: the expectation value of an observable $F$ with respect to a state $\psi$ (which we assume to be normalized, $\norm{\psi} = 1$) is given by 
\begin{align}
	\mathbb{E}_{\psi}(F) := \bscpro{\psi}{F \psi} 
	. 
\end{align}
The expectation value is finite if the state $\psi$ is in the domain $\mathcal{D}(F)$. The Born rule of quantum mechanics tells us that if we repeat an experiment measuring the observable $F$ many times for a particle that is prepared in the state $\psi$ each time, the statistical average calculated according to the relative frequencies converges to the expectation value $\mathbb{E}_{\psi}(F)$. 

Hence, quantum observables, selfadjoint operators on Hilbert spaces, are \emph{bookkeeping devices} that have two components: 
\begin{enumerate}[(i)]
	\item a \emph{set of possible outcomes} of measurements, the spectrum $\sigma(F)$, and 
	\item \emph{statistics}, \ie how often a possible outcome occurs. 
\end{enumerate}
%
% subsection quantum_observables (end)

\subsection{Quantum states} % (fold)
\label{quantum:frameworks:states}
Pure states are wave functions $\psi \in \Hil$, or rather, wave functions up to a total phase: just like one can measure only energy \emph{differences}, only phase \emph{shifts} are accessible to measurements. Hence, one can think of pure states as orthogonal \emph{projections} 
\begin{align*}
	P_{\psi} := \sopro{\psi}{\psi} = \sscpro{\psi}{\cdot} \, \psi 
	. 
\end{align*}
if $\psi$ is normalized to $1$, $\norm{\psi} = 1$. 
Here, one can see the elegance of bra-ket notation vs. the notation that is “mathematically proper”. A generalization of this concept are \emph{density operators} $\rho$ (often called density matrices): density matrices are defined via the trace. If $\rho$ is a suitable linear operator and $\{ \varphi_n \}_{n \in \N}$ and orthonormal basis of $\Hil$, then we define 
\begin{align*}
	\trace \, \rho := \sum_{n \in \N} \sscpro{\varphi_n}{\rho \varphi_n} 
	. 
\end{align*}
One can easily check that this definition is independent of the choice of basis (see homework problem~28). Clearly, $P_{\psi}$ has trace $1$ and it is also positive in the sense that 
\begin{align*}
	\bscpro{\varphi}{P_{\psi} \varphi} \geq 0 
\end{align*}
for all $\varphi \in \Hil$. This is also the good definition for quantum states: 
\begin{definition}[Quantum state]
	A quantum state (or density operator/matrix) $\rho = \rho^*$ is a non-negative operator of trace $1$, \ie 
	\begin{align*}
		\bscpro{\psi}{\rho \psi} &\geq 0 
		, 
		&& \forall \psi \in \Hil 
		, \\
		\trace \, \rho &= 1 
		. 
	\end{align*}
	If $\rho$ is also an orthogonal projection, \ie $\rho^2 = \rho$, it is a pure state.\footnote{Note that the condition $\trace \, \rho = 1$ implies that $\rho$ is a \emph{bounded} operator while the positivity implies the selfadjointness. Hence, if $\rho$ is a projection, \ie ${\rho}^2 = \rho$, it is automatically also an orthogonal projection. } Otherwise $\rho$ is a mixed state. 
\end{definition}
%
% CHANGED $\rho^2 = \rho$ if and only if $\rho = \sopro{\psi}{\psi}$ 
Density operators are projections if and only if they are rank-$1$ projections, \ie $\rho = \sopro{\psi}{\psi}$ for some $\psi \in \Hil$ of norm $1$ (see problem~28). 
\begin{example}
	Let $\psi_j \in \Hil$ be two wave functions normalized to $1$. Then for any $0 < \alpha < 1$
	\begin{align*}
		\rho = \alpha P_{\psi_1} + (1 - \alpha) P_{\psi_2} = \alpha \sopro{\psi_1}{\psi_1} + (1 - \alpha) \sopro{\psi_2}{\psi_2}
	\end{align*}
	is a mixed state as 
	\begin{align*}
		\rho^2 
		% &= 
		% \alpha^2 \sopro{\psi_1}{\psi_1} \sopro{\psi_1}{\psi_1} + \alpha (1 - \alpha) \bigl ( \sopro{\psi_1}{\psi_1} \sopro{\psi_2}{\psi_2} + \sopro{\psi_2}{\psi_2} \sopro{\psi_1}{\psi_1} \bigr ) + (1 - \alpha)^2 \sopro{\psi_2}{\psi_2} \sopro{\psi_2}{\psi_2} 
		% \\
		&= \alpha^2 \sopro{\psi_1}{\psi_1} + (1 - \alpha)^2 \sopro{\psi_2}{\psi_2} 
		+ \\
		&\qquad \qquad 
		+ \alpha (1 - \alpha) \bigl ( \sopro{\psi_1}{\psi_1} \sopro{\psi_2}{\psi_2} + \sopro{\psi_2}{\psi_2} \sopro{\psi_1}{\psi_1} \bigr )
		\\
		&\neq \rho
		. 
	\end{align*}
	Even if $\psi_1$ and $\psi_2$ are orthogonal to each other, since $\alpha^2 \neq \alpha$ and similarly $(1 - \alpha)^2 \neq (1 - \alpha)$, $\rho$ cannot be a projection. Nevertheless, it is a state since $\trace \, \rho = \alpha + (1 - \alpha) = 1$. Keep in mind that $\rho$ does not project on $\alpha \psi_1 + (1 - \alpha) \psi_2$! 
\end{example}
Also the expectation value of an observable $F$ with respect to a state $\rho$ is defined in terms of the trace, 
\begin{align*}
	\mathbb{E}_{\rho}(F) := \trace (\rho \, F)
	, 
\end{align*}
which for pure states $\rho = \sopro{\psi}{\psi}$ reduces to $\bscpro{\psi}{F \psi}$. 
% subsection quantum_states (end)

\subsection{Time evolution} % (fold)
\label{quantum:frameworks:time_evolution}
The time evolution is determined through the \emph{Schrödinger equation}, 
\begin{align}
	\ii \hbar \frac{\partial}{\partial t} \psi(t) = H \psi(t)
	, 
	&& \psi(t) \in \Hil, \; \psi(0) = \psi_0, \; \norm{\psi_0} = 1 
	. 
\end{align}
Alternatively, one can write $\psi(t) = U(t) \psi_0$ with $U(0) = \id_{\Hil}$. Then, we have 
\begin{align*}
	\ii \hbar \frac{\partial}{\partial t} U(t) = H U(t) 
	, 
	&& U(0) = \id_{\Hil} 
	. 
\end{align*}
If $H$ \emph{were} a number, one would immediately use the ansatz 
\begin{align}
	U(t) = \e^{- \ii \frac{t}{\hbar} H} 
\end{align}
as solution to the Schrödinger equation. If $H$ is a selfadjoint operator, this is \emph{still true}, but takes a lot of work to justify rigorously if the domain of $H$ is not all of $\Hil$ (the case of unbounded operators, the \emph{generic} case). 

As has already been mentioned, we can evolve either states or observables in time and one speaks of the Schrödinger or Heisenberg picture, respectively. In the Schrödinger picture, states evolve according to 
\begin{align*}
	\psi(t) = U(t) \psi_0 
\end{align*}
while observables remain fixed. Conversely, in the Heisenberg picture, states are kept fixed in time and observables evolve according to 
\begin{align}
	F(t) := U(t)^* \, F \, U(t) = \e^{+ \ii \frac{t}{\hbar} H} F \, \e^{- \ii \frac{t}{\hbar} H} 
	. 
\end{align}
Heisenberg observables satisfy \emph{Heisenberg's equation of motion}, 
\begin{align}
	\frac{\dd}{\dd t} F(t) = \frac{\ii}{\hbar} \bigl [ H , F(t) \bigr ] 
	, 
	&&
	F(0) = F 
	, 
\end{align}
which can be checked by plugging in the definition of $F(t)$ and elementary \emph{formal} manipulations. It is no coincidence that this equation looks structurally similar to equation~\eqref{classical_mechanics:eqn:eom_observables}! 

Just like in the classical case, density operators have to be evolved \emph{backwards} in time, meaning that $\rho(t) = U(t) \, \rho \, U(t)^*$ satisfies 
\begin{align*}
	\frac{\dd}{\dd t} \rho(t) = - \frac{\ii}{\hbar} \bigl [ H , \rho(t) \bigr ] 
	, 
	&&
	\rho(0) = \rho 
	. 
\end{align*}
The equivalence of Schrödinger and Heisenberg picture is seen by comparing expectation values just as in Chapter~\ref{classical_mechanics:equivalence_S_H}: the cyclicity of the trace, $\trace (AB) = \trace (BA)$, yields  \marginpar{2014.02.06}
\begin{align*}
	\mathbb{E}_{\rho(t)}(F) &= \trace \bigl ( \rho(t) \bigr ) 
	= \trace \Bigl ( U(t) \, \rho \, U(t)^* \, F \Bigr ) 
	\\
	&= \trace \Bigl ( \rho \, U(t)^* \, F \, U(t) \Bigr ) 
	= \trace \bigl ( \rho \, F(t) \bigr ) 
	= \mathbb{E}_{\rho} \bigl ( F(t) \bigr ) 
	. 
\end{align*}
As a last point, we mention the conservation of probability: if $\psi(t)$ solves the Schrödinger equation for some selfadjoint $H$, then we can check at least formally that the time evolution is unitary and thus preserves probability, 
\begin{align*}
	\frac{\dd}{\dd t} \bnorm{\psi(t)}^2 &= \frac{\dd}{\dd t} \bscpro{\psi(t)}{\psi(t)} 
	= \bscpro{\tfrac{1}{\ii \hbar} H \psi(t)}{\psi(t)} + \bscpro{\psi(t)}{\tfrac{1}{\ii \hbar} H \psi(t)} 
	\\
	&
	= \frac{\ii}{\hbar} \Bigl ( \bscpro{\psi(t)}{H^* \psi(t)} - \bscpro{\psi(t)}{H \psi(t)} \Bigr ) 
	\\
	&
	= \frac{\ii}{\hbar} \bscpro{\psi(t)}{(H^* - H) \psi(t)} 
	% \\
	% &
	= 0 
	. 
\end{align*}
Conservation of probability is reminiscent of Proposition~\ref{classical_mechanics:prop:states_stay_states}. We see that the condition $H^* = H$ is the key here: selfadjoint operators generate unitary evolution groups. As a matter of fact, there are cases when one \emph{wants} to violate conservation of proability: one has to introduce so-called \emph{optical potentials} which simulate particle creation and annihilation. 

The time evolution $\e^{- \ii \frac{t}{\hbar} H}$ is not the only unitary group of interest, other commonly used examples are \emph{translations} in position or momentum which are generated by the momentum and position operator, respectively (the order is reversed!), as well as rotations which are generated by the angular momentum operators. 
% subsection time_evolution (end)

\subsection{Comparison of the two frameworks} % (fold)
\label{quantum:frameworks:comparison}
Now that we have an understanding of the structures of classical and quantum mechanics, juxtaposed in Table~\ref{frameworks:comparison:table:overview_frameworks}, we can elaborate on the differences and similarities of both theories. 
% CHANGED table ugly, fix it! 
%
\begin{table}
	\begin{tabularx}{\textwidth}{>{\small\raggedright\hsize=3cm}X | >{\small\raggedright}X >{\small\raggedright}X >{\hsize=0cm}X} 
		\makebox[1cm]{}  & \textit{Classical} & \textit{Quantum} & \\ \hline 
		\textit{Observables} & $f \in \Cont^{\infty}(\Pspace,\R)$ & selfadjoint operators acting on Hilbert space $\Hil$ & \\ [0.5ex] 
		\textit{Building block observables} & position $x$ and momentum $p$ & position $\hat{x}$ and momentum $\hat{p}$ operators & \\ [0.5ex] 
		\textit{Possible results of measurements} & $\mathrm{im}(f)$ & $\sigma(F)$ & \\ [0.5ex] 
		\textit{States} & probability measures $\mu$ on phase space $\Pspace$ & density operators $\rho$ on $\Hil$ & \\ [0.5ex] 
		\textit{Pure states} & points in phase space $\Pspace$ & wave functions $\psi \in \Hil$ & \\ [0.5ex] 
		\textit{Generator of evolution} & hamiltonian function $H : \Pspace \longrightarrow \R$ & hamiltonian operator $H$ & \\ [0.5ex] 
		\textit{Infinitesimal time evolution equation} & $\frac{\dd}{\dd t} f(t) = \{ H , f(t) \}$ & $\frac{\dd }{\dd t} F(t) = \frac{\ii}{\hbar} [ H , F(t) ]$ & \\ [0.5ex]
		% CHANGED Die neue Inkarnation von LaTeX mag aus unerfindlichen Gründen \bigl und \bigr im tabularx environment nicht
		% \textit{Infinitesimal time evolution equation} & $\frac{\dd }{\dd t} f(t) = \bigl \{ f(t) , h \bigr \}$ & $\frac{\dd }{\dd t} A(t) = \frac{1}{\ii \hbar} \bigl [ A(t) , H \bigr ]$ & \\ [0.5ex] 
		\textit{Integrated time evolution} & hamiltonian flow $\phi_t$ & $\e^{+ \ii \frac{t}{\hbar} H} \, \Box \, \e^{- \ii \frac{t}{\hbar} H}$ & \\ 
	\end{tabularx}
	\caption{Comparison of classical and quantum framework}
	\label{frameworks:comparison:table:overview_frameworks}
\end{table}
For instance, observables form an \emph{algebra} (a vector space with multiplication): in classical mechanics, we use the \emph{pointwise product} of functions, 
\begin{align*}
	\cdot :& \, \Cont^{\infty}(\Pspace) \times \Cont^{\infty}(\Pspace) \longrightarrow \Cont^{\infty}(\Pspace) , \; (f , g) \mapsto f \cdot g 
	\\
	& (f \cdot g)(x,p) := f(x,p) \, g(x,p) 
	, 
\end{align*}
which is obviously commutative. We also admit \emph{complex}-valued functions and add \emph{complex conjugation} as involution (\ie $f^{\ast \ast} = f$). Lastly, we add the Poisson bracket to make $\Cont^{\infty}(\Pspace)$ into a so-called Poisson algebra. As we have seen, the notion of Poisson bracket gives rise to dynamics as soon as we choose an energy function (hamiltonian). 

On the quantum side, bounded operators (see~Chapter~\ref{operators:bounded}) form an algebra. This algebra is non-commutative, \ie 
\begin{align*}
	F \cdot G \neq G \cdot F 
	. 
\end{align*}
\emph{Exactly this is what makes quantum mechanics different.} Taking adjoints is the involution here and the commutator plays the role of the Poisson bracket. Again, once a hamiltonian (operator) is chosen, the dynamics of Heisenberg observables $F(t)$ is determined by the commutator of the $F(t)$ with the hamiltonian $H$. If an operator commutes with the hamiltonian, \emph{it is a constant of motion}. This is in analogy with Definition~\ref{classical:defn:conserved_quantity} where a classical observable is a constant of motion if and only if its Poisson bracket with the hamiltonian (function) vanishes.

\subsection{Representations} % (fold)
\label{quantum:framework:representations}
Linear algebra distinguishes abstract linear maps $H : \mathcal{X} \longrightarrow \mathcal{Y}$ and their representations as matrices using a basis in initial and target space: any pair of bases $\{ x_n \}_{n = 1}^N$ and $\{ y_k \}_{k = 1}^K$ of $\mathcal{X} \cong \C^N$ and $\mathcal{Y} \cong \C^K$ induces a matrix representation $h = (h_{nk}) \in \mathrm{Mat}_{\C}(N,K)$ of $H$ (called \emph{basis representation}) via 
\begin{align*}
	H x_n = \sum_{k = 1}^K h_{nk} \, y_k 
	. 
\end{align*}
The basis now identifies coordinates on the vector spaces: $x = \sum_{n = 1}^N \xi_n \, x_n \in \mathcal{X}$ has the coordinate $\xi = (\xi_1 , \ldots , \xi_n) \in \C^N$, and similarly $y = \sum_{k = 1}^K \eta_k \, y_k \in \mathcal{Y}$ is expressed in terms of the coordinate $\eta \in \C^K$. Using these coordinates, the equation $H x = y$ becomes the matrix equation $h \xi = \eta$. 

A change in basis can now be described in the same way: if $\{ x_n' \}_{j = 1}^N$ and $\{ y_k' \}_{k = 1}^K$ are two other orthonormal bases, then the coordinate representations of the maps 
\begin{align*}
	U_{xx'} &: x_n \mapsto x_n'
	\\
	U_{yy'} &: y_k \mapsto y_k'
\end{align*}
are unitary matrices $u_{xx'} \in \mathcal{U}(\C^N)$ and $u_{yy'} \in \mathcal{U}(\C^K)$, and these matrices connect the coordinate representations of $H$ with respect to $\{ x_n \}_{n = 1}^N$, $\{ y_k \}_{k = 1}^K$ and $\{ x_n' \}_{n = 1}^N$, $\{ y_k' \}_{k = 1}^K$, 
\begin{align*}
	h' = u_{yy'} \, h \, u_{xx'}^{-1} 
	. 
\end{align*}
$u_{xx'}^{-1}$ maps $\xi'$ onto $\xi$, $h$ maps $\xi$ onto $\eta$ and $u_{yy'}$ maps $\eta$ onto $\eta'$.

Similarly, we can represent operators on \emph{infinite}-dimensional Hilbert spaces such as $L^2(\R^d)$ in much the same way: for instance, consider the free Schrödinger operator $H = - \tfrac{1}{2} \Delta_x : \mathcal{D} \subset L^2(\R^d) \longrightarrow L^2(\R^d)$. Then the Fourier transform $\Fourier : L^2(\R^d) \longrightarrow L^2(\R^d)$ is such a unitary which changes from one “coordinate system” to another, and the free Schrödinger operator in this new representation becomes a simple multiplication operator 
\begin{align*}
	H^{\Fourier} := \Fourier \, H \, \Fourier^{-1} = \tfrac{1}{2} \hat{\xi}^2 
	. 
\end{align*}
Because initial and target space are one and the same, $\Fourier$ appears twice. 

Another unitary is a rescaling which can be seen as a change of units: for $\lambda > 0$ one defines 
\begin{align*}
	(U_{\lambda} \varphi)(x) := \lambda^{\nicefrac{d}{2}} \, \varphi(\lambda x) 
\end{align*}
where the scaling factor $\lambda$ relates the two scales. Similarly, other linear changes of the underlying configuration space $\R^d$ (\eg rotations) induce a unitary operator on $L^2(\R^d)$. \marginpar{2014.02.11}
\medskip

\noindent
One can exploit this freedom of representation to simplify a problem: Just like choosing spherical coordinates for a problem with spherical symmetry, we can work in a representation which simplifies the problem. For instance, the Fourier transform exploits the \emph{translational symmetry} of the free Schrödinger operator ($H$ commutes with translations). 

Another example would be to use an \emph{eigenbasis}: assume $H = H^* \geq 0$ as a set of eigenvectors $\{ \psi_n \}_{n \in \N}$ which span all of $\Hil$, \ie the $\psi_n$ are linearly independent and $H \psi_n = E_n \, \psi_n$ where $E_n \in \R$ is the eigenvalue. The eigenvalues are enumerated by magnitude and repeated according to their multiplicity, \ie $E_1 \leq E_2 \leq \ldots$. Just like in the case of hermitian matrices, the eigenvectors to distinct eigenvalues of selfadjoint operators are trivial, and hence, we can choose the $\{ \psi_n \}_{n \in \N}$ to be orthonormal. Then the suitable unitary is 
\begin{align*}
	U : \Hil \longrightarrow \ell^2(\N)
	, 
	\; \; 
	\psi = \sum_{n = 1}^{\infty} \widehat{\psi}(n) \, \psi_n \mapsto \widehat{\psi} \in \ell^2(\N)
\end{align*}
where $\widehat{\psi} = \bigl ( \widehat{\psi}(1) , \widehat{\psi}(2) , \ldots \bigr )$ is the sequence of coefficients and $\ell^2(\N)$ is the prototypical Hilbert space defined in Definition~\ref{spaces:Hilbert:defn:ell2}; moreover, the definition of orthonormal basis (Definition~\ref{spaces:Hilbert:defn:ONB}) implies that $\widehat{\psi}$ is necessarily square summable. 

In this representation, $H$ can be seen as an “infinite diagonal matrix”
\begin{align*}
	H = \sum_{n = 1}^{\infty} E_n \, P_{\psi_n} \mapsto H^U = U \, H \, U^{-1} 
	= \left (
	\begin{matrix}
		E_1 & 0 & \cdots & \cdots \\
		0 & E_2 & 0 & \cdots \\
		\vdots & & \ddots & \ddots \\
	\end{matrix}
	\right )
\end{align*}
where $P_{\psi} \varphi := \scpro{\psi}{\varphi} \, \psi$ are the rank-$1$ projections onto $\psi$. Put another way, $H^U$ acts on $\widehat{\psi} \in \ell^2(\N)$ as 
\begin{align*}
	H^U \widehat{\psi} = \bigl ( E_1 \, \widehat{\psi}(1) , E_2 \, \widehat{\psi}(2) , \ldots \bigr )
	. 
\end{align*}
The simple structure of this operator allows one to compute the unitary evolution group explicitly in terms of the projections $P_{\psi_n}$, 
\begin{align*}
	\e^{- \ii \frac{t}{\hbar} H} = \sum_{n = 1}^{\infty} \e^{- \ii \frac{t}{\hbar} E_n} \, P_{\psi_n}
	. 
\end{align*}
Sadly, most Schrödinger operators $H$ do not have a basis of eigenvectors. 
% subsection gauge_freedom_and_representations (end)
% section the_mathematical_framework (end)

\section{Spectral properties of hamiltonians} % (fold)
\label{quantum:spectrum}
The spectrum of an operator is the generalization of the set of eigenvalues for matrices. According to Definition~\ref{operators:defn:spectrum} the spectrum can be dived up into three parts, the \emph{point spectrum} 
\begin{align*}
	\sigma_{\mathrm{p}}(H) := \bigl \{ z \in \C \; \vert \; H - z \mbox{ is not injective} \bigr \} 
	, 
\end{align*}
the \emph{continuous spectrum}
\begin{align*}
	\sigma_{\mathrm{c}}(H) := \bigl \{ z \in \C \; \vert \; H - z \mbox{ is injective, } \im (H - z) \subseteq \Hil \mbox{ dense} \bigr \} 
	, 
\end{align*}
and the \emph{residual spectrum} 
\begin{align*}
	\sigma_{\mathrm{r}}(H) := \bigl \{ z \in \C \; \vert \; H - z \mbox{ is injective, } \im (H - z) \subseteq \Hil \mbox{ not dense} \bigr \} 
	. 
\end{align*}
Point spectrum is due to eigenvalues with eigenvector. Compared to matrices, the occurrence of continuous and residual spectra is new. The residual spectrum is not important for our discussion as it is empty for selfadjoint operators. 

The continuous spectrum can be attributed to cases where the eigenvectors are not elements of the Hilbert space. For instance, in case of the free Schrödinger operator $H = - \tfrac{1}{2} \Delta_x$ on $L^2(\R^d)$, the spectrum is $\sigma(H) = \sigma_{\mathrm{c}}(H) = [0,+\infty)$. Here, the eigenvectors are plane waves, $\e^{+ \ii \xi \cdot x}$ which are smooth, bounded functions; however, plane waves are not square integrable. Similarly, multiplication operators have Dirac distributions as eigen“functions”. 

Note that this distinction between the spectral components goes further than looking at the spectrum as a set: for instance, it is known that certain random Schrödinger operators have dense point spectrum which “looks” the same as continuous spectrum. The spectrum can be probed by means of approximate eigenfunctions (“Weyl's Criterion”, see Theorem~\ref{operators:thm:Weyl_criterion}). 
\medskip

\noindent
% TODO add physical interpretation of essential vs. discrete via RAGE theorem
There is also a second helpful classification of spectrum which cannot be made rigorous with the tools we have at hand, and that is the distinction between \emph{essential} spectrum $\sigma_{\mathrm{ess}}(H)$ and \emph{discrete} spectrum $\sigma_{\mathrm{disc}}(H)$. The essential spectrum is stable under local, short-range perturbations while the discrete spectrum may change. One has the following characterization for the essential spectrum: 
\begin{theorem}[Theorem VII.10 in \cite{Reed_Simon:M_cap_Phi_1:1972}]
	$\lambda \in \sigma_{\mathrm{ess}}(H)$ iff \emph{one} or more of the following holds: 
	\begin{enumerate}[(i)]
		\item $\lambda \in \sigma_{\mathrm{cont}}(H)$
		\item $\lambda$ is a limit point of $\sigma_{\mathrm{p}}(H)$. 
		\item $\lambda$ is an eigenvalue of infinite multiplicity. 
	\end{enumerate}
\end{theorem}
Similarly, the discrete spectrum has a similar characterization: 
\begin{theorem}[Theorem VII.11 in \cite{Reed_Simon:M_cap_Phi_1:1972}]
	$\lambda \in \sigma_{\mathrm{disc}}(H)$ if and only if \emph{both} of the following hold: 
	\begin{enumerate}[(i)]
		\item $\lambda$ is an isolated point of $\sigma(H)$, \ie for some $\eps > 0$ we have $\bigl ( \lambda - \eps , \lambda + \eps \bigr ) \cap \sigma(H) = \{ \lambda \}$. 
		\item $\lambda$ is an eigenvalue of \emph{finite} multiplicity. 
	\end{enumerate}
\end{theorem}

\subsection{Spectra of common selfadjoint operators} % (fold)
\label{quantum:framework:common}
Quite generally, the spectrum of selfadjoint operators is purely real. But before we prove that, let us discuss some examples from physics:

\paragraph{Multiplication operators} % (fold)
The spectrum of the multiplication operator 
\begin{align*}
	\bigl ( f(\hat{x}) \psi \bigr )(x) := f(x) \, \psi(x)
\end{align*}
is given by the range, $\sigma \bigl ( f(\hat{x}) \bigr ) = \overline{\ran f}$, where $f : \R^d \longrightarrow \R$ is a piecewise-continuous function.\footnote{This condition can be relaxed and is chosen just for ease of use.} 

To see this claim, we rely on the Weyl criterion: in order to show $\sigma \bigl ( f(\hat{x}) \bigr ) \supseteq \overline{\ran f}$, pick any $\lambda \in \overline{\ran f}$. Then there exists a sequence $x_n$ such that $\abs{\lambda - f(x_n)} < \nicefrac{1}{n}$. Then by shifting an $L^2$-Dirac sequence by $x_n$ (\eg scaled Gaußians), we obtain a sequence of vectors $\psi_n$ with $\bnorm{\bigl ( f(\hat{x}) - \lambda \bigr ) \psi_n} \xrightarrow{n \to \infty} 0$. Hence, this reasoning shows $\overline{\ran f} \subseteq \sigma \bigl ( f(\hat{x}) \bigr )$. 

To show the converse inclusion, let $\lambda \in \sigma \bigl ( f(\hat{x}) \bigr )$. Then there exists a Weyl sequence $\{ \psi_n \}_{n \in \N}$ with $\bnorm{\bigl ( f(\hat{x}) - \lambda \bigr ) \psi_n} \rightarrow 0$ as $n \to \infty$. Assume $\inf_{x \in \R^d} \babs{f(x) - \lambda} = c > 0$, \ie $\lambda \not \in \overline{\ran f}$, then $\{ \psi_n \}$ cannot be a Weyl sequence to $\lambda$, 
\begin{align*}
	\bnorm{\bigl ( f(\hat{x}) - \lambda \bigr ) \psi_n} \geq \inf_{x \in \R^d} \babs{f(x) - \lambda} \, \snorm{\psi_n} \geq c > 0 
	, 
\end{align*}
which is absurd. 
\medskip

\noindent
Should $f$ be constant and equal to $\lambda_0$ on a set of positive measure, there are infinitely many eigenfunctions associated to the eigenvalue $\lambda_0$. Otherwise, $f$ has continuous spectrum. In any case, the spectrum of $f(\hat{x})$ is purely essential. 

Clearly, this takes care of any operator which is unitarily equivalent to a multiplication operator, \eg the free Laplacian on $\R^d$, $\T^d$ or the tight-binding hamiltonians from Chapter~\ref{Fourier:T:periodic_operators:tight_binding}. 
% paragraph multiplication_operators (end)

\paragraph{The hydrogen atom} % (fold)
One of the most early celebrated successes of quantum mechanics is the explanation of the spectral lines by Schrödinger \cite{Schroedinger:quantisierung_ew_problem_1:1926,Schroedinger:quantisierung_ew_problem_2:1926,Schroedinger:quantisierung_ew_problem_3:1926,Schroedinger:quantisierung_ew_problem_4:1926}. Here, the operator 
\begin{align*}
	H := - \frac{\hbar^2}{2 m} \Delta_x - \frac{e}{\sabs{\hat{x}}}
\end{align*}
acts on a dense subspace of $L^2(\R^3)$. A non-obvious fact is that this operator is bounded from below, \ie there exits a constant $c > 0$ such that $H \geq - c$. This is false for the corresponding classical system, because the function $h(q,p) = \tfrac{1}{2m} p^2 - \tfrac{e}{\abs{q}}$ is \emph{not} bounded from below. 

The reason for that is that states of low potential energy (\ie wave functions which are sharply peaked around $0$) must pay an ever larger price in kinetic energy (sharply peaked means large gradient). One heuristic way to see that is to compute the energy expectation value of $\psi_{\lambda} := \lambda^{\nicefrac{3}{2}} \, \psi(\lambda x)$ for $\lambda \gg 1$ where $\psi \in \Schwartz(\R^d)$: 
\begin{align*}
	\mathbb{E}_{\psi_{\lambda}}(H) &= \frac{\hbar^2}{2m} \, \bscpro{\psi_{\lambda}}{- \Delta_x \psi_{\lambda}} - e \, \bscpro{\psi_{\lambda}}{\sabs{\hat{x}}^{-1} \psi_{\lambda}} 
	\\
	&= \frac{\hbar^2}{2m} \, \lambda^2 \, \int_{\R^3} \dd x \, \lambda^3 \, \babs{\nabla_x \psi(\lambda x)}^2 - e \, \lambda \, \int_{\R^3} \dd x \, \lambda^3 \, \frac{\sabs{\psi(\lambda x)}}{\lambda \, \sabs{x}} 
	\\
	&= \lambda^2 \, \scpro{\psi}{- \tfrac{\hbar^2}{2m} \Delta_x \psi} - \lambda \, \bscpro{\psi}{e \, \sabs{\hat{x}}^{-1} \psi} 
\end{align*}
Clearly, if one replaces the Coulomb potential by $- \abs{x}^{-3}$, the kinetic energy wins and the quantum particle can “fall down the well”. 

The negative potential gives rise to a family of eigenvalues (the spectral lines) while $-\Delta_x$ contributes continuous spectrum $[0,+\infty)$, 
\begin{align*}
	\sigma(H) &= \{ E_n \}_{n \in \N} \cup [0,+\infty)
	, 
	\\
	\sigma_{\mathrm{cont}}(H) &= [0,+\infty) = \sigma_{\mathrm{ess}}(H)
	, 
	\\
	\sigma_{\mathrm{p}}(H) &= \{ E_n \}_{n \in \N} = \sigma_{\mathrm{disc}}(H) 
	. 
\end{align*}
%
% paragraph the_hydrogen_atom (end)
% subsection common_systems (end)

\subsection{The spectrum of selfadjoint operators is real} % (fold)
\label{quantum:framework:selfadjoint_operator_real}
As a side note, let us show that the spectrum of selfadjoint operators is purely real. 
\begin{theorem}\label{quantum:thm:spectrum_selfadjoint_real}
	Let $H = H^*$ be a selfadjoint operator on the Hilbert space $\Hil$. Then the following holds true: 
	\begin{enumerate}[(i)]
		\item $\sigma(H) \subseteq \R$
		\item $H \geq 0$ $\Rightarrow$ $\sigma(H) \subseteq [0,+\infty)$
	\end{enumerate}
\end{theorem}
To prove this, we use the following 
\begin{lemma}
	Let $\Hil_j$, $j = 1 , 2$, be Hilbert spaces. Then an operator $T \in \mathcal{B}(\Hil_1,\Hil_2)$ is invertible if and only if there exists a constant $C > 0$ such that $T^* \, T \geq C \, \id_{\Hil_1}$ and $T \, T^* \geq C \, \id_{\Hil_2}$ hold. 
\end{lemma}
\begin{proof}
	“$\Rightarrow$:” Assume $T$ is invertible. Then $T^* : \Hil_2 \longrightarrow \Hil_1$ is also invertible with inverse ${T^*}^{-1} = {T^{-1}}^*$. Set $C := \norm{T^{-1}}^{-2} = \norm{{T^*}^{-1}}^{-2}$. Then the inequality 
	\begin{align*}
		\norm{\psi} &= \norm{T^{-1} T \psi} \leq \norm{T^{-1}} \, \norm{T \psi} 
	\end{align*}
	proves $\norm{T \psi} \geq \norm{T^{-1}}^{-1}$, and thus also 
	\begin{align}
		\scpro{\psi}{T^* T \psi} &= \norm{T \psi}^2 \geq \norm{T^{-1}}^{-2} \, \norm{\psi}^2 = C \norm{\psi}^2
		\label{quantum:eqn:Tast_T_geq_C_id}
		, 
	\end{align}
	\ie we have shown $T^* \, T \geq C \, \id_{\Hil_1}$. The non-negativity of $T \, T^*$ is shown analogously. 
	
	“$\Leftarrow$:” Suppose there exists $C > 0$ such that $T^* \, T \geq C \, \id_{\Hil_2}$ and $T \, T^* \geq C \, \id_{\Hil_2}$. Then from \eqref{quantum:eqn:Tast_T_geq_C_id} we deduce $\norm{T \psi} \geq \sqrt{C} \, \norm{\psi}$ holds for all $\psi \in \Hil_1$. First of all, this proves that $T$ is injective, and secondly $T$ has closed range in $\Hil_2$ (one can see the latter by considering convergence of $T \psi_n$ for any Cauchy sequence $\{ \psi_n \}_{n \in \N}$). Moreover, one can easily see 
	\begin{align*}
		\ran T = \overline{\ran T} = \bigl ( \ker T^* \bigr )^{\perp}
		. 
	\end{align*}
	Since we can make the same arguments for $T^*$, we also know that $T^*$ is injective, and thus $\ker T^* = \{ 0 \}$. This shows that $T$ is surjective, \ie it is bijective, and hence, invertible. 
\end{proof}
With the proof of the Lemma complete, we can now prove the statement: 
\begin{proof}[Theorem~\ref{quantum:thm:spectrum_selfadjoint_real}]
	\begin{enumerate}[(i)]
		\item Let $H = H^*$ be selfadjoint and $z = \lambda + \ii \mu \in \C \setminus \R$ be a complex number with non-vanishing imaginary part $\mu$. We will show that $z \not\in \sigma(H)$, \ie that $H - \lambda$ is invertible: a quick computation shows 
		\begin{align*}
			\bigl ( H - z \bigr )^* \, \bigl ( H - z \bigr ) &= H^2 - 2 \, (\Re z) \, H + \sabs{z}^2 
			= H^2 - 2 \lambda \, H + (\lambda^2 + \mu^2)
			\\
			&= \mu^2 + \bigl ( H - \lambda \bigr )^2 
			. 
		\end{align*}
		The last term is non-negative, and thus, we have shown 
		\begin{align*}
			\bigl ( H - z \bigr )^* \, \bigl ( H - z \bigr ) \geq \mu^2 
			. 
		\end{align*}
		By the Lemma, this means $H - \lambda$ is necessarily invertible, and $z \not \in \sigma(H)$. 
		\item We have to show that for $\lambda \in (-\infty,0)$, the operator $H - \lambda$ is invertible. This follows as before from 
		\begin{align*}
			\bigl ( H - \lambda \bigr )^* \, \bigl ( H - \lambda \bigr ) &= H^2 - 2 \lambda \, H + \lambda^2 
			\geq \lambda^2 
			, 
		\end{align*}
		the non-negativity of $- 2 \lambda \, H = 2 \sabs{\lambda} \, H$ and the Lemma. 
	\end{enumerate}
\end{proof}
%
% subsection subsection_name (end)

\subsection{Eigenvalues and bound states} % (fold)
\label{quantum:spectrum:bound_states}
The hydrogen atom is a prototypical example of the type of problem we are interested in, namely Schrödinger operators on $L^2(\R^d)$ of the form 
\begin{align*}
	H = - \Delta_x + V
\end{align*}
where $V \leq 0$ is a non-positive potential decaying at infinity ($\lim_{\abs{x} \to \infty} V(x) = 0$). Under suitable technical conditions on the potential, $H$ defines a selfadjoint operator which is bounded from below, that is $H \geq c$ holds for some $c \in \R$, and we have 
\begin{align*}
	\sigma_{\mathrm{ess}}(H) &= \sigma(-\Delta_x) = [0,+\infty) 
	. 
\end{align*}
Now the question is whether $\sigma_{\mathrm{p}}(H) = \emptyset$ or 
\begin{align*}
	\sigma_{\mathrm{p}}(H) &= \{ E_n \}_{n = 0}^{N} \subset (-\infty,0)
\end{align*}
for some $N \in \N_0 \cup \{ \infty \}$. We shall always assume that the eigenvalues are ordered by magnitude, 
\begin{align*}
	E_0 \leq E_1 \leq \ldots 
\end{align*}
The \emph{ground state} $\psi_0$ is the eigenfunction to the lowest eigenvalue $E_0$. Eigenfunctions $\psi$ are localized: the weakest form of localization is $\psi \in L^2(\R^d)$, but usually one can expect \emph{exponential} localization. 

So there are two natural questions which we will answer in turn: 
\begin{enumerate}[(1)]
	\item Do eigenvalues below the essential spectrum \emph{exist}? 
	\item Can we give \emph{estimates} on their numerical values? 
\end{enumerate}

\subsubsection{The Birman-Schwinger principle} % (fold)
\label{quantum:spectrum:bound_states:Birman_Schwinger}
We begin with the Birman-Schwinger principle which gives a criterion for the existence and absence of eigenvalues at a specific energy level. It is \emph{the} standard tool for showing the existence or absence of eigenvalues. Assume $\varphi$ is an eigenvector of $H$ to the eigenvalue $-E < 0$. Then the eigenvalue equation is equivalent to 
\begin{align*}
	\bigl ( - \Delta_x + E \bigr ) \varphi &= - V \varphi 
	= \abs{V} \varphi
	. 
\end{align*}
If we define the vector $\psi := \abs{V}^{\nicefrac{1}{2}} \varphi$ and use that $-E \not \in \sigma(-\Delta_x) = [0,+\infty)$, we obtain 
\begin{align*}
	\abs{V}^{\nicefrac{1}{2}} \bigl ( - \Delta_x + E \bigr )^{-1} \abs{V}^{\nicefrac{1}{2}} \psi &= \psi
	. 
\end{align*}
In other words, we have just shown the 
\begin{theorem}[Birman-Schwinger principle]
	The function $\varphi \in L^2(\R^d)$ is an eigenvector of $H = - \Delta_x + V$ to the eigenvalue $-E < 0$ if and only if $\psi = \abs{V}^{\nicefrac{1}{2}} \varphi$ is an eigenvector of the Birman-Schwinger operator 
	\begin{align}
		K_E := \abs{V}^{\nicefrac{1}{2}} \bigl ( - \Delta_x + E \bigr )^{-1} \, \abs{V}^{\nicefrac{1}{2}} 
		\label{quantum:eqn:Birman_Schwinger_operator}
	\end{align}
	to the eigenvalue $1$. \marginpar{2014.0213}
\end{theorem}
The only assumption we have glossed over is the boundedness of $K_E$. One may think that solving $K_E \psi = \psi$ is just as difficult as $H \varphi = - E \varphi$, but it is not. For instance, we immediately obtain the following 
\begin{corollary}
	Assume the Birman-Schwinger operator $K_E \in \mathcal{B} \bigl ( L^2(\R^d) \bigr )$ is bounded. Then for $\lambda_0$ small enough, $H_{\lambda} = - \Delta_x + \lambda V$ has no eigenvalue at $-E$ for all $0 \leq \lambda < \lambda_0$. 
\end{corollary}
\begin{proof}
	% CHANGED add citation
	Replacing $V$ with $\lambda V$ in equation~\eqref{quantum:eqn:Birman_Schwinger_operator} yields that the Birman-Schwinger operator for $H_{\lambda}$ is $\lambda K_E$. Thus, for $\lambda$ small enough, we can make $\lambda \norm{K_E} < 1$ arbitrarily small and since $\sup \abs{\sigma \bigl ( K_E \bigr )} \leq \norm{K_E}$,\footnote{This is a general fact: if $T \in \mathcal{B}(\mathcal{X})$ is an operator on a Banach space, then $\sup \abs{\sigma(T)} \leq \norm{T}$ holds \cite[Chapter~VIII.2, Theorems~3 and 4]{Yoshida:functional_analysis:1980}.} this means $1$ cannot be an eigenvalue. Hence, by the Birman-Schwinger principle there cannot exist an eigenvalue at $-E$. 
\end{proof}
Another advantage is that we have an explicit expression for the \emph{operator kernel} of $K_E$, the \emph{Birman-Schwinger kernel}, which allows us to make explicit estimates. In general, an operator kernel $K_T$ for an operator $T$ is a distribution on $\R^d \times \R^d$ so that 
\begin{align*}
	(T \psi)(x) = \int_{\R^d} \dd y \, K_T(x,y) \, \psi(y) 
	. 
\end{align*}
For the sake of brevity, we will also write $T(x,y)$ for $K_T(x,y)$. We have dedicated Chapter~\ref{Greens_functions} to one specific example: assume the operator $L$ is invertible and $L u = f$, then 
\begin{align*}
	u(x) = \int_{\R^d} \dd y \, G(x,y) \, f(y) 
	= \bigl ( L^{-1} f \bigr )(x)
\end{align*}
holds. In other words, the Green's function $G$ is the operator kernel of $L^{-1}$. 

Seeing as $K_E$ is the product of the multiplication operator $\abs{V}^{\nicefrac{1}{2}}$ and $\bigl ( - \Delta_x + E \bigr )^{-1}$, the dimension-dependent, explicit expression of Birman-Schwinger kernel involves only the Green's function of $- \Delta_x + E$ in that particular dimension, 
\begin{align*}
	K_E(x,y) = \abs{V(x)}^{\nicefrac{1}{2}} \bigl ( - \Delta_x + E \bigr )^{-1}(x,y) \abs{V(y)}^{\nicefrac{1}{2}} 
	. 
\end{align*}
In odd dimension, there exist closed expressions for $\bigl ( - \Delta_x + E \bigr )^{-1}(x,y)$ while for even $d$, no neat formulas for it exist. Nevertheless, its behavior can be characterized. 
\medskip

\noindent
Let us return to the original question: Can we show the \emph{existence} of eigenvalues as well via the Birman-Schwinger principle? The answer is yes, and we will treat a particular case: 
\begin{theorem}[\cite{Simon:bound_states_low_d_Schroedinger:1976}]\label{quantum:thm:existence_bound_state_1d_Schroedinger}
	Consider the Schrödinger operator $H_{\lambda} = - \partial_x^2 + \lambda V$ on $L^2(\R)$ where $\lambda > 0$ and the potential satisfies $V \in L^1(\R)$, $V \neq 0$, $V \leq 0$, and 
	\begin{align*}
		\int_{\R} \dd x \, \bigl ( 1 + x^2 \bigr ) \, \abs{V(x)} < \infty 
		. 
	\end{align*}
	Then there exists $\lambda_0 > 0$ small enough so that $H_{\lambda}$ has exactly one eigenvalue 
	\begin{align}
		E_{\lambda} &= - \frac{\lambda^2}{4} \, \left ( \int_{\R} \dd x \, \abs{V(x)} \right )^2 + \order(\lambda^4)
		\label{quantum:eqn:BS_neg_eigenvalue}
	\end{align}
	for all $\lambda \in (0,\lambda_0)$. 
\end{theorem}
The eigenvalue gives an intuition on the shape of the eigenfunction: it has few oscillations to minimize kinetic energy and is approximately constant in the region where $V$ is appreciably different from $0$ (this region is not too large because of the decay assumption $\int_{\R} \dd x \, x^2 \, \abs{V(x)} < \infty$). Hence, the eigenfunction sees only the average value of the potential. 

This intuition neither explains why other eigenvalues may appear nor that for $d \geq 3$, the theorem is false. 
\begin{proof}
	The arguments in \cite[Section~2]{Simon:bound_states_low_d_Schroedinger:1976} ensure the boundedness of the Birman-Schwinger operator. Moreover, in one dimension the Green's function for $-\partial_x^2 + E$ exists ($-E \not\in \sigma(-\partial_x^2)$) and can be computed explicitly, namely 
	\begin{align*}
		\bigl ( -\partial_x^2 + E \bigr )^{-1}(x,y) = \sqrt{2\pi} \, \bigl ( \Fourier (\xi^2 + E)^{-1} \bigr )(x-y)
		= \frac{\e^{- \sqrt{E} \abs{x-y}}}{2 \sqrt{E}}
		. 
	\end{align*}
	To simplify notation, let us define $\mu := \sqrt{E}$. Thus, the Birman-Schwinger kernel is the function 
	\begin{align*}
		K_{\mu^2}(x,y) &= \frac{1}{2 \mu} \, \abs{V(x)}^{\nicefrac{1}{2}} \, \e^{- \mu \abs{x-y}} \, \abs{V(y)}^{\nicefrac{1}{2}} 
		. 
	\end{align*}
	In addition, define the operators 
	\begin{align*}
		L_{\mu} := \frac{1}{2 \mu} \, \bopro{\abs{V}^{\nicefrac{1}{2}}}{\abs{V}^{\nicefrac{1}{2}}}
	\end{align*}
	and $M_{\mu} := K_{\mu^2} - L_{\mu}$. Clearly, given that $V \in L^1(\R)$, its square root is $L^2$ and $L_E$ is a bounded rank-$1$ operator. Moreover, the operator kernel 
	\begin{align*}
		M_{\mu}(x,y) = \abs{V(x)}^{\nicefrac{1}{2}} \, \frac{\e^{- \mu \abs{x-y}} - 1}{2 \mu} \, \abs{V(y)}^{\nicefrac{1}{2}} 
	\end{align*}
	is well-defined in the limit $\mu \rightarrow 0$ and analytic for $\mu \in \C$ with $\Re \mu > 0$. 
	
	The Birman-Schwinger principle tells us that $H_{\lambda}$ has an eigenvalue at $-\mu^2$ if and only if $1 \in \sigma_{\mathrm{p}}(K_{\mu^2})$: for $\lambda \ll 1$ small enough we have $\bnorm{\lambda \, M_{\mu}} < 1$ which means the Neumann series\footnote{In this context, the geometric series is usually referred to as Neumann series.} 
	\begin{align}
		\bigl ( 1 - \lambda \, M_{\mu} \bigr )^{-1} &= \sum_{n = 0}^{\infty} \lambda^n \, M_{\mu}^n 
		= 1 + \lambda \, M_{\mu} + \order(\lambda^2)
		\label{quantum:eqn:Neuman_series_M_mu}
	\end{align}
	exists in $\mathcal{B} \bigl ( L^2(\R) \bigr )$. Hence, the invertibility of 
	\begin{align*}
		1 - \lambda \, K_{\mu^2} &= 1 - \lambda \, M_{\mu} - \lambda \, L_{\mu} 
		\\
		&= \bigl ( 1 - \lambda \, M_{\mu} \bigr ) \, \Bigl ( 1 - \lambda \, \bigl ( 1 - \lambda \, M_{\mu} \bigr )^{-1} \, L_{\mu} \Bigr ) 
	\end{align*}
	hinges on whether $1$ is an eigenvalue of \marginpar{2014.02.25}
	\begin{align*}
		\lambda \, \bigl ( 1 - \lambda \, M_{\mu} \bigr )^{-1} \, L_{\mu} &= \opro{\tfrac{\lambda}{2 \mu} \bigl ( 1 - \lambda \, M_{\mu} \bigr )^{-1} \, \abs{V}^{\nicefrac{1}{2}}}{\abs{V}^{\nicefrac{1}{2}}}
		. 
	\end{align*}
	This is again a rank-$1$ operator, and thus, we can read off the eigenvector 
	\begin{align*}
		\psi_{\lambda,\mu} &= \frac{\lambda}{2 \mu} \bigl ( 1 - \lambda \, M_{\mu} \bigr )^{-1} \, \abs{V}^{\nicefrac{1}{2}} \in L^2(\R) 
	\end{align*}
	to its only non-zero eigenvalue. Moreover, we can compute this eigenvalue, 
	\begin{align*}
		\scpro{\abs{V}^{\nicefrac{1}{2}}}{ \; \tfrac{\lambda}{2 \mu} \bigl ( 1 - \lambda \, M_{\mu} \bigr )^{-1} \, \abs{V}^{\nicefrac{1}{2}}}
		, 
	\end{align*}
	and this is equal to $1$ if and only if $\mu$ satisfies the self-consistent equation 
	\begin{align*}
		\mu &= G(\mu) := \frac{\lambda}{2} \, \scpro{\abs{V}^{\nicefrac{1}{2}}}{ \; \bigl ( 1 - \lambda \, M_{\mu} \bigr )^{-1} \, \abs{V}^{\nicefrac{1}{2}}}
		. 
	\end{align*}
	Given that $\norm{\lambda \, M_{\mu}} < 1$ for $\lambda \ll 1$ small enough, we can express $\bigl ( 1 - \lambda \, M_{\mu} \bigr )^{-1}$ in terms of \eqref{quantum:eqn:Neuman_series_M_mu}. Keeping only the first term of the expansion~\eqref{quantum:eqn:Neuman_series_M_mu}, we approximate $G$ by the average of the potential 
	\begin{align}
		G(\mu) &= \tfrac{\lambda}{2} \, \bscpro{\abs{V}^{\nicefrac{1}{2}}}{\abs{V}^{\nicefrac{1}{2}}} + \order(\lambda^2) 
		= \frac{\lambda}{2} \, \int_{\R} \dd x \, \abs{V(x)} + \order(\lambda^2)
		. 
		\label{quantum:eqn:approximate_eigenvalue_sqrt}
	\end{align}
	Hence, $G(\mu) = \mu$ has \emph{a} solution $\mu_{\ast}$ provided $\lambda$ is small enough; additionally \emph{any} solution to this equation satisfies $\mu^{-1} \leq C_1 \, \lambda^{-1}$ for some constant $C_1 > 0$ and $\lambda$ small. 
	\medskip
	
	\noindent
	Now that we know that \emph{a} solution exists, we need to show uniqueness: Suppose we have found two solutions $\mu_1 \leq \mu_2$. Then they both solve the self-consistent equation $G(\mu_j) = \mu_j$, and assuming for a moment that $G$ is continuously differentiable in $\mu$, we use the fundamental theorem of calculus to obtain 
	\begin{align*}
		\babs{\mu_2 - \mu_1} &= \babs{G(\mu_2) - G(\mu_1)} 
		= \abs{\int_{\mu_1}^{\mu_2} \dd \mu \, \partial_{\mu} G(\mu)} 
		\\
		&
		\leq \sup_{\mu \in [\mu_1,\mu_2]} \babs{\partial_{\mu} G(\mu)} \, \babs{\mu_2 - \mu_1}
		. 
	\end{align*}
	\emph{If} we can show $G$ is continuously differentiable \emph{and} its derivative is bounded by $\nicefrac{1}{2}$ for $\lambda$ small enough, then the above inequality reads $\babs{\mu_2 - \mu_1} \leq \tfrac{1}{2} \babs{\mu_2 - \mu_1}$. This is only possible if $\mu_1 = \mu_2$, and the solution is unique. 
	
	To show the last bit, we note that $M_{\mu}$ and $(1-z)^{-1}$ are real-analytic in $\mu$ so that their composition $\bigl ( 1 - \lambda \, M_{\mu} \bigr )^{-1}$ is also real-analytic. The analyticity of $M_{\mu}$ for $\mu \in \C$, $\Re \mu > 0$, also yields the bound 
	\begin{align}
		\bnorm{\partial_{\mu} M_{\mu}} \leq C_2 \, \mu^{-1}
		\label{quantum:eqn:BS_proof_a_priori_estimate_G_mu}
	\end{align}
	via the Cauchy integral formula, because the maximal radius of the circular contour is less than $\mu$. 
	
	The derivative of the resolvent can be related to $\partial_{\mu} M_{\mu}$ via the useful trick
	\begin{align*}
		0 = \partial_{\mu} (\id) 
		&= \partial_{\mu} \Bigl ( \bigl ( 1 - \lambda \, M_{\mu} \bigr )^{-1} \; \bigl ( 1 - \lambda \, M_{\mu} \bigr ) \Bigr ) 
		\\
		&= \partial_{\mu} \bigl ( 1 - \lambda \, M_{\mu} \bigr )^{-1} \; \bigl ( 1 - \lambda \, M_{\mu} \bigr ) + \lambda \, \bigl ( 1 - \lambda \, M_{\mu} \bigr )^{-1} \; \partial_{\mu} M_{\mu}
	\end{align*}
	which yields 
	\begin{align*}
		\babs{\partial_{\mu} G(\mu)} &= \abs{\frac{\lambda^2}{2} \scpro{\abs{V}^{\nicefrac{1}{2}}}{ \; \bigl ( 1 - \lambda \, M_{\mu} \bigr )^{-1} \, \partial_{\mu} M_{\mu} \, \bigl ( 1 - \lambda \, M_{\mu} \bigr )^{-1} \, \abs{V}^{\nicefrac{1}{2}}}}
		. 
	\end{align*}
	The right-hand side can be estimated with the help of the Cauchy-Schwarz inequality 
	\begin{align*}
		\ldots &\leq \lambda^2 \, \bnorm{\abs{V}^{\nicefrac{1}{2}}}_{L^2(\R)}^2 \, \norm{\bigl ( 1 - \lambda \, M_{\mu} \bigr )^{-1}}^2 \, \bnorm{\partial_{\mu} M_{\mu}}
		=: C_3 \, \lambda^2 \, \bnorm{\partial_{\mu} M_{\mu}} 
		. 
	\end{align*}
	Combining \eqref{quantum:eqn:BS_proof_a_priori_estimate_G_mu} with $\mu^{-1} \leq C_1 \, \lambda^{-1}$ (which we obtained from $\mu = G(\mu)$), we find 
	\begin{align*}
		C_3 \, \lambda^2 \, \bnorm{\partial_{\mu} M_{\mu}} \leq C_3 \, \lambda^2 \, C_2 \, \mu^{-1} 
		\leq C_1 C_2 C_3 \, \lambda 
		. 
	\end{align*}
	Put another way, we have deduced the bound $\babs{\partial_{\mu} G(\mu)} \leq C \, \lambda$ which means that for $\lambda$ small enough, we can ensure that the derivative is less than $\nicefrac{1}{2}$. Thus, the eigenvalue is unique and we have shown the theorem. 
\end{proof}
%
% subsubsection the_birman_schwinger_principle (end)

\subsubsection{The min-max principle} % (fold)
\label{quantum:spectrum:bound_states:the_min_max_principle}
Now that we have established criteria for the \emph{existence} of bound states below the continuous spectrum for operators of the form $H = - \Delta_x + V$, we proceed to find other ways to give estimates of their numerical values. Crucially, we shall always assume $H \geq c$ for some $c \in \R$. Most of the methods of this chapter do not depend on the particular form of the hamiltonian. 
\medskip

\noindent
So let us assume we have established the existence of a \emph{ground state} $\psi_0$, \ie there exists an eigenvalue $E_0 = \inf \sigma(H) < 0 = \inf \sigma_{\mathrm{ess}}(H)$ at the bottom of the spectrum, the \emph{ground state energy}, whose eigenfunction is $\psi_0$. Then simplest estimate is obtained by minimizing the \emph{Rayleigh quotient}
\begin{align*}
	\frac{\mathbb{E}_{\psi}(H)}{\norm{\psi}^2} &= \frac{\scpro{\psi}{H \psi}}{\norm{\psi}^2}
\end{align*}
for a family of trial wave functions (see also homework problem~54). Clearly, the Rayleigh quotient is bounded from below by $E_0$ for otherwise, $E_0$ is not the infimum of the spectrum. 
\begin{proposition}[The Rayleigh-Ritz principle]
	Let $H$ with a densely defined, selfadjoint operator which is bounded from below, \ie there exists $c \in \R$ such that $H \geq c$. Then 
	\begin{align}
		\inf \sigma(H) \leq \frac{\scpro{\psi}{H \psi}}{\norm{\psi}^2}
		\label{quantum:eqn:Rayleigh_Ritz}
	\end{align}
	holds for all $\psi \in \Hil \setminus \{ 0 \}$. 
\end{proposition}
A rigorous proof of this innocent-looking fact (see \eg \cite[Theorem~XIII.1]{Reed_Simon:M_cap_Phi_4:1978}) requires machinery that is not yet available to us. 

A non-obvious fact is that we can also give a \emph{lower} bound on the ground state energy: 
\begin{theorem}[Temple's inequality, Theorem~XIII.5 in \cite{Reed_Simon:M_cap_Phi_4:1978}]
	Let $H$ be a selfadjoint operator that is bounded from below with ground state $E_0 \in \sigma_{\mathrm{p}}(H)$, $E_0 < 0$. Suppose in addition $E_0 < E_1$ where $E_1$ is either the second eigenvalue (in case more eigenvalues exist) or the bottom of the essential spectrum. 
	Then for $\mu \in (E_0,E_1)$ and $\psi$ with $\norm{\psi} = 1$ and $\scpro{\psi}{H \psi} < \mu$, Temple's inequality holds: 
	\begin{align*}
		E_0 \geq \scpro{\psi}{H \psi} - \frac{\scpro{\psi}{H^2 \psi} - \scpro{\psi}{H \psi}^2}{\mu - \scpro{\psi}{H \psi}} 
		= \scpro{\psi}{H \psi} - \frac{\mathrm{Var}_{\psi}(H)}{\mu - \scpro{\psi}{H \psi}} 
	\end{align*}
\end{theorem}
Temple's inequality gives an energy window for the ground state energy: if $\psi$ is close to the ground state wave function, then the right-hand side is also close to $E_0$. On the other hand, one needs to know a lower bound on the \emph{second} eigenvalue $E_1$. 
\begin{proof}
	By assumption, $E_0$ is an isolated eigenvalue of finite multiplicity (otherwise $E_0 = E_1 = E_n$ for all $n \in \N$), and thus the operator $(H - E_0) (H - \mu) \geq 0$ is non-negative: the product is $= 0$ if applied to the ground state and $> 0$ otherwise because $\mu < E_1$. Consequently, 
	\begin{align}
		\bscpro{\psi}{(H - \tilde{E_1}) H \psi} \geq E_0 \, \bscpro{\psi}{(H - \mu) \psi} 
	\end{align}
	holds which, combined with the hypothesis $\bscpro{\psi}{(H - \mu) \psi} < 0$, yields 
	\begin{align*}
		E_0 \geq \frac{\mu \, \scpro{\psi}{H \psi} - \scpro{\psi}{H^2 \psi}}{\mu - \scpro{\psi}{H \psi}} 
		. 
	\end{align*}
\end{proof}
What about other bound states below the essential spectrum (the ionization threshold)? Usually, we do not know \emph{whether} and \emph{how many} eigenvalues exist. Nevertheless, we can define a sequence of non-decreasing real numbers that coincides with the eigenvalues if they exist: the Rayleigh quotient suggests to use 
\begin{align*}
	E_0 := \inf_{\substack{\varphi \in \mathcal{D}(H), \norm{\varphi} = 1}} \scpro{\varphi}{H \varphi}
\end{align*}
as the definition of the ground state energy. Note that even if $H$ does not have eigenvalues, $E_0$ is still well-defined and yields $\inf \sigma(H)$ (use a Weyl sequence). A priori, we do \emph{not} know whether a $E_0$ is an eigenvalue, so we do not know whether an eigenvector exists. However, \emph{if} $E_0$ is an eigenvalue, then the eigenvector $\psi_1$ to the next eigenvalue $E_1$ (if it exists) would necessarily have to be orthogonal to $\psi_0$. Then the next eigenvalue satisfies 
\begin{align*}
	E_1 = \sup_{\varphi_0 \in \domain(H) \setminus \{ 0 \}} \inf_{\substack{\varphi \in \mathcal{D}(H), \norm{\varphi} = 1 \\ \varphi \in \{ \varphi_0 \}^{\perp}}} \scpro{\varphi}{H \varphi}
	. 
\end{align*}
It turns out that this is the good definition even if $E_0 \not \in \sigma_{\mathrm{disc}}(H)$ is not an eigenvalue of finite multiplicity, because then $E_0 = E_1$. Quite generally, the candidate for the $n$th eigenvalue is 
\begin{align*}
	E_n := \sup_{\substack{\varphi_1 , \ldots , \varphi_n \in \mathcal{D}(H) \\ \scpro{\varphi_j}{\varphi_k} = \delta_{jk}}} \inf_{\substack{\varphi \in \mathcal{D}(H), \norm{\varphi} = 1 \\ \varphi \in \{ \varphi_1 , \ldots , \varphi_n \}^{\perp}}} \scpro{\varphi}{H \varphi}
	. 
\end{align*}
Thus, we obtain a sequence of non-decreasing real numbers 
\begin{align*}
	E_0 \leq E_1 \leq E_2 \leq \ldots 
\end{align*}
which -- if they exist -- are the eigenvalues repeated according to their multiplicities. One can show rigorously that if $E_n = E_{n+1} = E_{n + 2} = \ldots$, then $E_n = \inf \sigma_{\mathrm{ess}}(H)$ is the bottom of the essential spectrum. Otherwise, the $E_n < \inf \sigma_{\mathrm{ess}}(H)$ are \emph{eigenvalues of finite multiplicity}. In that case, there are at most $n$ eigenvalues below the essential spectrum. 

One may object that quite generally, it is impossible to evaluate $E_n$. Here is where the \emph{min-max principle} comes into play: assume we have chosen $n$ trial wave functions. Then this family of trial wave functions is a good candidate for the first few eigenfunctions if the eigenvalues $\lambda_j$ of the matrix $h := \Bigl ( \bscpro{\varphi_j}{H \varphi_k} \Bigr )_{0 \leq j , k \leq n-1} $ (ordered by size) are close to the $E_j$. 
\begin{theorem}[The min-max principle]
	Suppose $H$ is a selfadjoint operator on the Hilbert space $\Hil$ with domain $\mathcal{D}(H)$. Moreover, assume $H$ is bounded from below. Let $\bigl \{ \varphi_0 , \ldots , \varphi_{n-1} \bigr \} \subset \mathcal{D}(H)$ be an orthonormal system of $n$ functions and consider the $n \times n$ matrix 
	\begin{align*}
		h := \Bigl ( \bscpro{\varphi_j}{H \varphi_k} \Bigr )_{0 \leq j , k \leq n-1} 
	\end{align*}
	with eigenvalues $\lambda_0 \leq \lambda_1 \leq \ldots \leq \lambda_{n-1}$. Then we have that \marginpar{2014.02.27}
	\begin{align*}
		E_j \leq \lambda_j 
		\qquad \qquad \forall j = 0 , \ldots , n-1
		. 
	\end{align*}
\end{theorem}
\begin{proof}
	We proceed by induction over $k$ (which enumerates the eigenvalues of $h$): denote the normalized eigenvector to the lowest eigenvalue $\lambda_0$ with $v_0 = \bigl ( v_{0,0} , \ldots v_{0,n-1} \bigr )$. Then the normalized vector $\chi_0 := \sum_{j = 0}^{n-1} v_{0,j} \, \varphi_j$ satisfies 
	\begin{align*}
		\lambda_0 = \bscpro{v_0}{h v_0}_{\C^n} = \bscpro{\chi_0}{H \chi_0} 
		\geq E_0
	\end{align*}
	by the Rayleigh-Ritz principle. 
	
	Now assume we have shown that $E_l \leq \lambda_l$ holds for all $l = 0 , \ldots , k \leq n-2$. Clearly, the eigenvectors $v_0 , \ldots , v_k$ to $h$, and the space spanned by the corresponding normalized $\chi_l = \sum_{j = 0}^{n-1} v_{l,j} \, \varphi_j$ is $k+1$-dimensional. Hence, for any 
	\begin{align*}
		\chi = \sum_{j = 0}^{n-1} w_j \, \chi_j \in \span \bigl \{ \chi_0 , \ldots , \chi_k \bigr \}^{\perp}
	\end{align*}
	with coefficients $w \in \{ v_0 , \ldots , v_k \}^{\perp}$ we obtain 
	\begin{align*}
		\scpro{w}{h w} = \scpro{\chi}{H \chi} \geq E_{k+1}
	\end{align*}
	because $\chi$ is orthogonal to a $k+1$-dimensional subspace of $\mathcal{D}(H)$. The left-hand side can be minimized by setting $w = v_{k+1}$, the eigenvector to $\lambda_{k+1}$, and thus, $E_{k+1} \leq \lambda_{k+1}$. This concludes the proof. 
\end{proof}
One can use the min-max principle to make the following intuition rigorous: Assume one is given an operator $H(V) = -\Delta_x + V$ whose potential vanishes sufficiently rapidly at $\infty$, and one knows that $H(V)$ has a certain number of eigenvalues $\{ E_j(V) \}_{j \in \mathcal{I}}$, $\mathcal{I} \subseteq \N_0$. The decay conditions on $V$ ensure $\sigma_{\mathrm{ess}} \bigl ( H(V) \bigr ) = [0,+\infty)$. Then if $W \leq V$ is a second potential of the same type, the min-max principle implies 
\begin{align*}
	E_j(W) \leq E_j(V) 
	. 
\end{align*}
In particular, $H(W)$ has \emph{at least} as many eigenvalues as $H(V)$. This fact combined with Theorem~\ref{quantum:thm:existence_bound_state_1d_Schroedinger} immediately yields 
\begin{corollary}
	Suppose we are in the setting of Theorem~\ref{quantum:thm:existence_bound_state_1d_Schroedinger}. Then for all $\lambda > 0$ the Schrödinger operator $H = - \partial_x^2 + \lambda V$ has at least one eigenvalue $E_0 < 0$. 
\end{corollary}
%
% section spectral_properties_of_hamiltonians (end)

\section{Magnetic fields} % (fold)
\label{quantum:magnetic}
Classically, there are \emph{two} ways to include magnetic fields (\cf Chapter~\ref{classical_mechanics:magnetic}): either by minimal substitution $p \mapsto p - A(x)$ which involves the magnetic vector \emph{potential} $A$ or one modifies the symplectic form to include the magnetic \emph{field} $B = \nabla_x \times A$. Note that the physical observable is the magnetic \emph{field} rather than the vector \emph{potential}, because there are many vector potentials which represent the same magnetic field. For instance, if $A$ is a vector potential to the magnetic field $B = \nabla_x \times A$, then also $A' = A + \nabla_x \phi$ is another vector potential to $B$, because $\nabla_x \times \nabla_x \phi = 0$. The scalar function $\phi$ generates a \emph{gauge transformation}. 

In contrast, one always needs to choose a vector potential in quantum mechanics, and the hamiltonian for a non-relativistic particle subjected to an electromagnetic field $(E,B) = \bigl ( - \nabla_x V , \nabla_x \times A \bigr )$ is obtained by minimal substitution as well, 
\begin{align}
	H^A = \bigl ( - \ii \nabla_x - A \bigr )^2 + V 
	. 
\end{align}
What happens if we choose an \emph{equivalent} gauge $A' = A + \nabla_x \phi$? It turns out that $H^A$ and $H^{A + \nabla_x \phi}$ are \emph{unitarily equivalent} operators, and the unitary which connects the two is $\e^{- \ii \phi}$, 
\begin{align*}
	\e^{+ \ii \phi} \, H^A \, \e^{- \ii \phi} = H^{A + \nabla_x \phi}
\end{align*}
Using the lingo of Chapter~\ref{quantum:framework:representations}, $\e^{- \ii \phi}$ is a unitary that connects two different representations. This has several very important ramifications. The spectrum $\sigma(H^A)$, for instance, only depends on the magnetic field $B = \nabla_x \times A$ because unitarily equivalent operators necessarily have the same spectrum. Moreover, the gauge freedom is essential to solving problems, because \emph{some gauges are nicer to work with than others}. One such condition is $\nabla_x \cdot A = 0$, known as Coulomb gauge. \marginpar{2014.03.04}

The natural domain of these operators are 
\begin{definition}[Magnetic Sobolev spaces $H^m_A(\R^d)$]\label{quantum:defn:mag_Sobolev}
	Suppose $A \in \Cont^1(\R^d,\R^d)$. Then we define the magnetic Sobolev space of order $m$ to be 
	\begin{align}
		H_A^m(\R^d) := \bigl \{ \psi \in L^2(\R^d) \; \; \vert \; \; \norm{\psi}_{H^m_A} < \infty \bigr \} 
	\end{align}
	where the $m$th magnetic Sobolev norm is 
	\begin{align}
		\norm{\psi}_{H^m_A}^2 := \sum_{\abs{\gamma} \leq m} \norm{\bigl ( - \ii \nabla_x - A \bigr )^{\gamma} \psi}_{L^2}^2 
		. 
	\end{align}
	For $A = 0$, we abbreviate the (ordinary) Sobolev space with $H^m(\R^d) := H^m_{A = 0}(\R^d)$. 
\end{definition}
The definition just means we are looking at those $\psi \in L^2(\R^d)$ whose weak derivatives of up to $m$th order are all in $L^2(\R^d)$.\footnote{The weak derivative is well-defined, because we can view $L^2(\R^d)$ as a subspace of the tempered distributions $\Schwartz'(\R^d)$.} One can see that Sobolev spaces are complete and can be equipped with a scalar product (see \eg \cite[Theorem~7.3]{Lieb_Loss:analysis:2001} for the case $m = 1$ and $A = 0$). 

Magnetic fields have the property that they induce oscillations, and these induced oscillations, in turn, \emph{increase the kinetic energy}. The diamagnetic inequality makes this intuition rigorous: 
\begin{theorem}[Diamagnetic inequality]\label{quantum:thm:diamag_inequality}
	Let $A : \R^d \longrightarrow \R^d$ be in $\Cont^1(\R^d,\R^d)$ and $\psi$ be in $H^1_A(\R^d)$. Then $\abs{\psi}$, the absolute value of $\psi$, is in $H^1(\R^d)$ and the \emph{diamagnetic inequality}, 
	\begin{align}
		\babs{\nabla_x \abs{\psi}(x)} \leq \babs{\bigl (- \ii \nabla_x - A(x) \bigr ) \psi(x)}
		, 
		\label{quantum:eqn:diamag_inequality}
	\end{align}
	holds pointwise for almost all $x \in \R^d$. 
\end{theorem}
\begin{proof}
	Since $\psi \in L^2(\R^d)$ and each component of $A$ is in $\Cont^1(\R^d,\R^d) \subset L^2_{\mathrm{loc}}(\R^d,\R^d)$, the distributional gradient of $\psi$ is in $L^1_{\mathrm{loc}}(\R^d)$. The distributional derivative of the absolute value can be computed explicitly, 
	\begin{align}
		\partial_{x_j} \negthinspace \abs{\psi}(x) = 
		\begin{cases}
			\Re \Bigl ( \tfrac{\overline{\psi(x)}}{\abs{\psi(x)}} \, \partial_{x_j} \psi(x) \Bigr ) & \psi(x) \neq 0 \\
			0 & \psi(x) = 0 \\
		\end{cases}
		, 
		\label{quantum:eqn:diamag_ineq_derivative_gradient}
	\end{align}
	and the right-hand side is again a function in $L^1_{\mathrm{loc}}(\R^d)$. Given that $A$ and $\abs{\psi}$ are real, 
	\begin{align*}
		\Re \left ( \frac{\overline{\psi(x)}}{\abs{\psi(x)}} \, \ii A_j(x) \, \psi(x) \right ) &= \Re \bigl ( \ii A_j(x) \, \abs{\psi(x)} \bigr ) = 0 
		, 
	\end{align*}
	and we can add this term to equation~\eqref{quantum:eqn:diamag_ineq_derivative_gradient} free of charge to obtain 
	\begin{align*}
		\partial_{x_j} \abs{\psi}(x) = 
		\begin{cases}
			\Re \Bigl ( \tfrac{\overline{\psi(x)}}{\abs{\psi(x)}} \, \bigl ( \partial_{x_j} + \ii A_j(x) \bigr ) \, \psi(x) \Bigr ) & \psi(x) \neq 0 \\
			0 & \psi(x) = 0 \\
		\end{cases}
		.
	\end{align*}
	The diamagnetic inequality now follows from $\babs{\Re z} \leq \abs{z}$, $z \in \C$. The left-hand side of \eqref{quantum:eqn:diamag_inequality} is in $L^2(\R^d)$ since the right-hand side is by assumption on $\psi$. 
\end{proof}
The simplest example of a magnetic Hamilton operator $H^A = (- \ii \nabla_x - A)^2$ is the so-called \emph{Landau hamiltonian} where $d = 2$, $B$ is constant and $V = 0$. For instance, one can choose the \emph{symmetric gauge}
\begin{align}
	A(x) = \frac{B}{2} 
	\left (
	\begin{matrix}
		- x_2 \\
		+ x_1 \\
	\end{matrix}
	\right ) 
\end{align}
or the \emph{Landau gauge}
\begin{align}
	A(x) = B 
	\left (
	\begin{matrix}
		- x_2 \\
		0 \\
	\end{matrix}
	\right ) 
	. 
\end{align}
The spectrum of $\sigma(H^A) = \{ 2 n + 1 \; \vert \; n \in \N_0 \} = \sigma_{\mathrm{ess}}(H^A)$ are the \emph{Landau levels}, a collection of \emph{infinitely degenerate} eigenvalues accumulating at $+ \infty$. Physically, the origin for this massive degeneracy is translation-covariance: if I have a bound state $\psi_0$, then $\psi_0(\cdot - x_0)$ is an eigenvector to a possibly gauge-transformed hamiltonian $H^{A + \nabla_x \phi}$. From a classical perspective, the existence of bound states as well as translational symmetry are also clear: a constant magnetic field traps a particle in a circular orbit, and the analog of this \emph{classical} bound state is a \emph{quantum} bound state, an eigenvector. 
% section magnetic_fields (end)

\section{Bosons vs{.} fermions} % (fold)
\label{quantum:bosons_fermions}
The extension of \emph{single}-particle quantum mechanics to \emph{multi}-particle quantum mechanics is highly non-trivial. To clarify the presentation, let us focus on \emph{two} identical particles moving in $\R^d$. Two options are arise: either the compound wave function $\Psi$ is a function on $\R^d$, \ie it acts like a \emph{density}, or it is a function of $\R^d \times \R^d$ where each set of coordinates $x = (x_1,x_2)$ is associated to one particle. It turns out that wave functions depend on $\R^{N d}$ where $N$ is the number of particles. 

However, that is not all, there is an added complication: classically, we can label identical particles by tracking their trajectory. This is impossible in the quantum framework, because the uncertainty principle forbids any such tracking procedure. Given that $\babs{\Psi(x_1,x_2)}^2$ is a physical observable, the inability to distinguish particles implies 
\begin{align*}
	\babs{\Psi(x_1,x_2)}^2 &= \babs{\Psi(x_2,x_1)}^2
	, 
\end{align*}
and hence, $\Psi(x_1,x_2) = \e^{+ \ii \theta} \, \Psi(x_2,x_1)$. However, given that exchanging variables twice must give the same wave function, the only two admissible phase factors are $\e^{+ \ii \theta} = \pm 1$. 

Particles for which $\Psi(x_1,x_2) = \Psi(x_2,x_1)$ holds are \emph{boson} (integer spin) while those for which $\Psi(x_1,x_2) = - \Psi(x_2,x_1)$ are \emph{fermions} (half-integer spin). Examples are bosonic photons and fermionic electrons. This innocent looking fact has very, very strong consequences on the physical and mathematical properties of quantum systems. The most immediate implication is \emph{Pauli's exclusion principle} for fermions, 
\begin{align*}
	\Psi(x,x) = 0
	, 
\end{align*}
a fact that is colloquially summarized by saying that bosons are social (because they like to bunch together) while sociophobic fermions tend to avoid one another. 

To make this more rigorous, let us consider the splitting 
\begin{align*}
	L^2(\R^d \times \R^d) \cong L^2_{\mathrm{s}}(\R^d \times \R^d) \oplus L^2_{\mathrm{as}}(\R^d \times \R^d)
\end{align*}
into symmetric and antisymmetric part induced via $f = f_{\mathrm{s}} + f_{\mathrm{as}}$ where 
\begin{align*}
	f_{\mathrm{s}}(x_1,x_2) :& \negmedspace= \tfrac{1}{2} \bigl ( f(x_1,x_2) + f(x_2,x_1) \bigr ) 
	,
	\\
	f_{\mathrm{as}}(x_1,x_2) :& \negmedspace= \tfrac{1}{2} \bigl ( f(x_1,x_2) - f(x_2,x_1) \bigr ) 
	. 
\end{align*}
Then one can proceed and restrict the two-particle Schrödinger operator 
\begin{align*}
	H = \sum_{j = 1 , 2} \bigl ( - \Delta_{x_j} + V(x_j) \bigr ) 
\end{align*}
to either the bosonic space $L^2_{\mathrm{s}}(\R^d \times \R^d)$. The kinetic energy $- \sum_{j = 1 , 2} \Delta_{x_j}$ preserves the (anti-)symmetry, \eg in the antisymmetric (fermionic case) it defines a bounded linear map 
\begin{align*}
	H : L^2_{\mathrm{as}}(\R^d \times \R^d) \cap H^2(\R^d \times \R^d) \longrightarrow L^2_{\mathrm{as}}(\R^d \times \R^d)
	. 
\end{align*}
% 
% section Bosons vs{.} fermions (end)

\section{Perturbation theory} % (fold)
\label{quantum:perturbations}
One last, but important remark concerns perturbation theory. Almost none of the systems one encounters in “real life” has a closed-form solution, so it is immediate to study perturbations of known systems first. The physics literature usually contents itself studying approximations of \emph{eigenvalues} of the hamiltonian, but the more fundamental question is what happens to the \emph{dynamics}? In other words, does $H_1 \approx H_2$ imply $\e^{- \ii t H_1} \approx \e^{- \ii t H_2}$. The answer is \emph{yes} and uses a very, very nifty trick, the \emph{Duhamel formula}. The idea is to write the difference 
\begin{align}
	\e^{- \ii t H_1} - \e^{- \ii t H_2} &= \int_0^t \dd s \, \frac{\dd}{\dd s} \Bigl ( \e^{- \ii s H_1} \, \e^{- \ii (t-s) H_2} \Bigr ) 
	\notag
	\\
	&= - \ii \, \int_0^t \dd s \, \e^{- \ii s H_1} \, \bigl ( H_1 - H_2 \bigr ) \, \e^{- \ii (t-s) H_2} 
\end{align}
as the integral of a total derivative. So if we assume $H_2 = H_1 + \eps \, W$, then one has to estimate 
\begin{align*}
	\e^{- \ii s H_1} \, \bigl ( H_1 - H_2 \bigr ) \, \e^{- \ii (t-s) H_2} &= \eps \, \e^{- \ii s H_1} \, W \, \e^{- \ii (t-s) H_2} 
	= \order(\eps) 
	. 
\end{align*}
Note that this holds for \emph{all} times, because quantum mechanics is a \emph{linear} theory. Otherwise, we would have to use the Grönwall Lemma~\ref{odes:lem:Groenwall} that places restrictions on the time scale for which $\psi_1(t) = \e^{- \ii t H_1} \psi$ and $\psi_2(t) = \e^{- \ii t H_2} \psi$ remain close. 
% section perturbation_theory (end)

% \section{Semiclassics} % (fold)
% \label{quantum:semiclassics}
% % 
% %
% \begin{itemize}
% 	\item semiclassical states (Gaußian, WKB)
% 	\item Ehrenfest theorem
% 	\item intuition semiclassics
% 	\item semiclassics $\equiv$ ray optics
% 	\item Wigner transform 
% \end{itemize}
% %
% 
% 
% % section semiclassics (end)
% chapter Quantum mechanics (end)
\chapter{Variational calculus} % (fold)
\label{variation}
Functionals $\mathcal{E} : \Omega \subseteq \mathcal{X} \longrightarrow \C$ are maps from a subset $\Omega$ of a Banach space $\mathcal{X}$ over the field $\C$ (or $\R$) to $\C$ (or $\R$). In case $\mathcal{X}$ is finite-dimensional, a functional is just a function $\C^n \longrightarrow \C$, and so the cases we are really interested in are when $\mathcal{X}$ is \emph{infinite}-dimensional. 

Functionals arise very often in physics as a way to formulate certain fundamental principles (\eg energy, action and the like); their analysis often produces linear and non-linear PDEs which are interesting in their own right. For instance, the energy functional 
\begin{align}
	\mathcal{E}(\psi) := \left ( \int_{\R^d} \dd x \, \babs{\psi(x)}^2 \right )^{-1} \; \int_{\R^d} \dd x \, \Bigl ( \babs{\nabla_x \psi(x)}^2 + V(x) \, \babs{\psi(x)}^2 \Bigr ) 
	\label{variation:eqn:energy_functional}
\end{align}
associated to the Schrödinger operator $H = - \Delta_x + V$ can be seen as a map $H^1(\R^d) \ni \psi \mapsto \mathcal{E}(\psi) \in \R$. Here, $H^1(\R^d)$ is the first Sobolev space, \cf Definition~\ref{quantum:defn:mag_Sobolev}. Let us assume $H$ is selfadjoint, bounded from below and has a ground state $\psi_0$. Then if we minimize $\mathcal{E}$ under the constraint $\norm{\psi} = 1$, the functional has a \emph{global minimum} at $\psi_0$. Alternatively, we can view it as the minimizer of the Rayleigh-Ritz quotient. 

So let us perturb the Rayleigh-Ritz quotient in the vicinity of the minimizer $\psi_0$, \ie we define 
\begin{align*}
	F(s) := \mathcal{E}(\psi_0 + s \varphi) \geq F(0) = \mathbb{E}(\varphi_0)
\end{align*}
where $\varphi \in H^1(\R^d)$ is arbitrary, but fixed. One can express denominator and numerator explicitly as quadratic polynomials in $s$, and one finds 
\begin{align*}
	\frac{\dd}{\dd s} F(0) &= 2 \, \snorm{\psi_0}^{-2} \, \Re \scpro{\varphi}{\bigl ( - \Delta_x + V - E_0 \bigr ) \psi_0} = 0
	, 
\end{align*}
\ie $F$ has a global minimum at $0$ independently of the choice of function $\varphi$. Put more succinctly: \emph{if it exists,} the minimizer $\psi_0$ of the Rayleigh-Ritz quotient is the eigenfunction of the Schrödinger operator $H = -\Delta_x + V$ at $E_0 = \inf \sigma(H)$. 

The energy functional only serves as an \emph{amuse gueule}. Among other things, it suggests that one can ask the same questions for functionals that one can asks for (ordinary) functions: 
\begin{enumerate}[(1)]
	\item Existence of local and global \emph{extrema}, \emph{convexity} properties and the existence of extrema under \emph{constraints}. 
	\item One can study ODEs where the vector field is connected to derivatives of a functional; in the simplest case, we want to look at gradient flows. The same fundamental questions arise: Where are the \emph{fixed points}? Are these fixed points \emph{stable or unstable}? 
\end{enumerate}

\section{Extremals of functionals} % (fold)
\label{variation:extremals}
Given that functionals are just functions on infinite-dimensional spaces, it is not surprising that the same type of questions are raised as with ordinary functions: continuity, differentiability, existence of local and global extrema. As one can guess, a rigorous treatment of functionals is a lot more technical.

\subsection{The Gâteaux derivative} % (fold)
\label{variation:extremals:derivatives}
Apart from continuity, the most fundamental property a function has is that of differentiability. On functionals, the starting point is the \emph{directional derivative} which then gives rise to the \emph{Gâteaux derivative}. Similar to functions on $\R^n$ vs{.} functions on $\C^n$, we start with the real case first and postpone the discussion of complex derivatives to Chapter~\ref{variation:extremals:complex_derivatives}. 
\begin{definition}[Gâteaux derivative]\label{variation:defn:Gateaux_derivative}
	Let $\mathcal{E} : \Omega \subset \mathcal{X} \longrightarrow \R$ a continuous linear functional defined on an \emph{open} subset $\Omega$ of the real Banach space $\mathcal{X}$. Then the Gâteaux derivative $\dd \mathcal{E}(\psi)$ at $\psi \in \Omega$ is defined as the linear functional on $\mathcal{X}$ for which 
	\begin{align}
		\dd \mathcal{E}(\psi) \varphi := \left . \frac{\dd }{\dd s} \mathcal{E} \bigl ( \psi + s \varphi \bigr ) \right \vert_{s = 0}
		\label{variation:eqn:Gateaux_derivative}
	\end{align}
	holds for all $\varphi \in \mathcal{X}$. If the Gâteaux derivative exists for all $\psi \in \Omega$, we say $\mathcal{E}$ is $\Cont^1$. 
\end{definition}
Higher derivatives are multilinear forms which are defined iteratively. 
% subsection taking_derivatives (end)

\subsection{Extremal points and the Euler-Lagrange equations} % (fold)
\label{variation:extremals:Euler_Lagrange}
\emph{Critical points} are those $\psi_{\ast} \in \mathcal{X}$ for which the Gâteaux derivative $\dd \mathcal{E}(\psi_{\ast}) = 0$ vanishes. To illustrate the connection between PDEs and critical points, let us consider the functional 
\begin{align*}
	\mathcal{E}(u) := \int_{\R^d} \dd x \, \Bigl ( \tfrac{1}{2} \babs{\nabla_x u}^2 + u \, f \Bigr ) 
\end{align*}
where $f : \R^d \longrightarrow \R$ is some fixed function and $u \in \Cont^{\infty}_{\mathrm{c}}(\R^d) \subset H^1(\R^d,\R)$. A quick computation yields the Gâteaux derivative, and if we set the right-hand side of 
\begin{align}
	\dd \mathcal{E}(u) v &= \int_{\R^d} \dd x \, \bigl ( \nabla_x u \cdot \nabla_x v + f \, v \bigr ) 
	\notag \\
	&
	= \int_{\R^d} \dd x \, \bigl ( - \Delta_x u + f \bigr ) \, v 
	\overset{!}{=} 0
	\label{variation:eqn:Euler_Lagrange_integral}
\end{align}
to zero for all $v \in H^1(\R^d,\R)$, we obtain the condition for a local extremum: $u$ is a critical point if and only if $u$ satisfies the Poisson equation 
\begin{align}
	- \Delta_x u + f = 0 
	. 
	\label{variation:eqn:Euler_Lagrange_differential}
\end{align}
Depending on the context, we call either \eqref{variation:eqn:Euler_Lagrange_integral} or \eqref{variation:eqn:Euler_Lagrange_differential} the \emph{Euler-Lagrange equation} to $\mathcal{E}$. So if $\mathcal{X}$ is comprised of functions, then \emph{the search for critical points is equivalent to solving a linear or non-linear PDE.} \marginpar{2014.03.06}
% subsection extremal_points_and_the_euler_lagrange_equations (end)

\subsection{Functionals on submanifolds} % (fold)
\label{variation:extremals:tangent}
Very often the functional is defined on a subset $\Omega \subset \mathcal{X}$ which lacks a linear structure (\eg Lagrangian mechanics below) so that $\mathcal{E}(\psi + s \varphi)$ need not make sense, because $\psi + s \varphi \not\in \Omega$. 

Instead, one has to replace the simple linear combination $\psi + s \varphi$ with differentiable \emph{paths} $(-\delta,+\delta) \ni s \mapsto \psi_s \in \Omega$. Then a tangent vector $\xi$ at $\psi$ is an equivalence class of paths so that $\psi_0 = \psi$ and $\frac{\dd}{\dd s} \psi_s \big \vert_{s = 0} = \xi$; the tangent space $T_{\psi} \Omega$ is then the vector space spanned by these tangent vectors. Hence, we proceed as in Definition~\ref{variation:defn:Gateaux_derivative} and set 
\begin{align*}
	\dd \mathcal{E}(\psi) \xi := \left . \frac{\dd}{\dd s} \mathcal{E}(\psi_s) \right \vert_{s = 0}
	. 
\end{align*}
Clearly, $\dd \mathcal{E}(\psi) \in \bigl ( T_{\psi} \Omega \bigr )'$ is an element of the dual space since it maps a tangent vector onto a scalar (\cf Definition~\ref{spaces:defn:dual_space} of the dual space). In general, this is just an abstract vector space, but here we can identify $T_{\psi} \Omega$ with a subvector space of $\mathcal{X}$. For a detailed description of the mathematical structures (manifolds and tangent bundles), we refer to \cite[Chapter~4.1]{Marsden_Ratiu:intro_mechanics_symmetry:1999}. 
% subsection functionals_on_submanifolds (end)

\subsection{Lagrangian mechanics} % (fold)
\label{variation:extremals:lagrangian_mechanics}
One extremely important example where the variation takes place over a \emph{non-linear} space is that of classical mechanics. Here we start with the space of paths 
\begin{align}
	\mathcal{D}(x_0,x_1) := \Bigl \{ q : [0,T] \longrightarrow \mathcal{X} \; \; \big \vert \; \; q \in \Cont^2 \bigl ( [0,T] , \mathcal{X} \bigr ) 
	, \; \; 
	q(0) = x_0 , \; \; 
	q(T) = x_1 \Bigr \} 
	\label{variation:eqn:path_space}
\end{align}
which join $x_0$ and $x_1$ in the Banach space $\mathcal{X}$; initially, let us concentrate on the case $\mathcal{X} = \R^d$ but more general choices are also admissible. As we will see later on, the choices of $x_0$, $x_1$ and $T$ are completely irrelevant, the Euler-Lagrange equations will be independent of them. 

The idea is to \emph{derive classical equations of motion from the Euler-Langrange equations of the action functional}
\begin{align}
	S(q) := \int_0^T \dd t \, L \bigl ( q(t) , \dot{q}(t) \bigr )
\end{align}
where $L \in \Cont^2 \bigl ( \R^d \times \R^d \bigr )$ is the \emph{Lagrange function}; $L(x,v)$ depends on position $x$ and velocity $v$ (as opposed to momentum). Physicists call this \emph{principle of stationary action}. 

We exploit the linear structure of $\mathcal{X} = \R^d$ and propose that we can canonically identify the tangent space $T_q \mathcal{D}(x_0,x_1)$ with $\mathcal{D}(0,0)$: if $h \in \mathcal{D}(0,0)$, then by definition $q + s h \in \mathcal{D}(x_0,x_1)$ is a path which joins $x_0$ and $x_1$ for all $s \in \R$ with tangent vector $h$. The Euler-Lagrange equations can now be derived easily: by definition $\dd S(q) = 0$ means $\dd S(q) h = 0$ holds for all $h \in \mathcal{D}(0,0)$, and we compute 
\begin{align}
	\dd S(q) h &= \left . \frac{\dd }{\dd s} S(q + s h) \right \vert_{s = 0} 
	= \int_0^T \dd t \, \left . \frac{\dd }{\dd s} L \Bigl ( q(t) + s h(t) \, , \, \dot{q}(t) + s \dot{h}(t) \Bigr ) \right \vert_{s = 0} 
	\notag \\
	&= \int_0^T \dd t \, \Bigl ( \nabla_x L \bigl ( q(t) , \dot{q}(t) \bigr ) \cdot h(t) + \nabla_v L \bigl ( q(t) , \dot{q}(t) \bigr ) \cdot \dot{h}(t) \Bigr )
	\notag \\
	&= \int_0^T \dd t \, \left ( \nabla_x L \bigl ( q(t) , \dot{q}(t) \bigr ) - \frac{\dd}{\dd t} \nabla_v L \bigl ( q(t) , \dot{q}(t) \bigr ) \right ) \cdot h(t)
	. 
	\label{variation:eqn:Euler_Lagrange_with_L_integral}
\end{align}
Note that the boundary terms vanish, because $h(0) = 0 = h(T)$. Clearly, the Euler-Lagrange equations
\begin{align}
	\nabla_x L \bigl ( q(t) , \dot{q}(t) \bigr ) - \frac{\dd}{\dd t} \nabla_v L \bigl ( q(t) , \dot{q}(t) \bigr ) &= 0 
	\label{variation:eqn:Euler_Lagrange_Lagrangian_mechanics}
\end{align}
are independent of $x_0$, $x_1$ and $T$ -- as promised. Moreover, the linear nature of $\mathcal{X} = \R^d$ is not crucial in the derivation, but a rigorous definition is more intricate in the case where $\mathcal{X}$ is a manifold (\cf \cite[Chapter~7]{Marsden_Ratiu:intro_mechanics_symmetry:1999}).

\subsubsection{Classical mechanics on $\R^d$} % (fold)
\label{variation:extremals:lagrangian_mechanics:classical_mechanics_Rd}
\emph{The} standard example for a Lagrangian is $L(x,v) = \frac{m}{2} v^2 - U(x)$ where $U$ is the potential (we use $U$ rather than $V$ to avoid the ambiguity between the velocity $v$ and the potential $V$). A simple computation yields 
\begin{align*}
	0 &= \nabla_x L \bigl ( q(t) , \dot{q}(t) \bigr ) - \frac{\dd}{\dd t} \nabla_v L \bigl ( q(t) , \dot{q}(t) \bigr ) 
	\\
	&= - \nabla_x U \bigl ( q(t) \bigr ) - m \ddot{q}(t) 
\end{align*}
or $m \ddot{q} = - \nabla_x U$. This second-order equation can be reduced to a first-order ODE by setting $\dot{q} = v$, and one obtains 
\begin{align*}
	\frac{\dd}{\dd t} \left (
	\begin{matrix}
		q \\
		v \\
	\end{matrix}
	\right ) &= \left (
	\begin{matrix}
		v \\
		- m^{-1} \, \nabla_x U \\
	\end{matrix}
	\right )
	. 
\end{align*}
Glancing back at the beginning of Chapter~\ref{classical_mechanics}, we guess the simple change of variables $p := m v$ and recover Hamilton's equations of motion~\eqref{classical_mechanics:eqn:hamiltons_eom}. In fact, this innocent change of variables is an instance of a much deeper fact, namely that momentum can be \emph{defined} as 
\begin{align*}
	p := \nabla_v L
	. 
\end{align*}
%
% subsubsection classical_mechanics_on_r_d_ (end)

\subsubsection{Derivation of Maxwell's equations} % (fold)
\label{variation:extremals:lagrangian_mechanics:Maxwell}
The idea to \emph{derive the dynamical equations of a physical theory as a critical point from an action functional} is extremely successful; almost any physical theory (\eg general relativity, quantum electrodynamics and fluid dynamics) can be derived in this formalism, and hence, a better understanding of functionals gives one access to a richly stocked toolbox. Moreover, they yield equations of motion in situations where one wants to couple degrees of freedom of a different nature (\eg fluid dynamics and electrodynamics). 

To illustrate this, we will derive the vacuum Maxwell equations (\cf Chapter~\ref{operators:Maxwell} for $\eps = 1 = \mu$). It all starts with a clever choice of Lagrange function, in this case 
\begin{align*}
	L \bigl ( t , A , \phi , \alpha , \varphi \bigr ) &= \int_{\R^3} \dd x \, \mathbb{L} \bigl ( t \, , \, A(x) \, , \, \phi(x) \, , \, \alpha(x) \, , \, \varphi(x) \bigr ) 
	\\
	&= \int_{\R^3} \dd x \, \Bigl ( \tfrac{1}{2} \babs{- \alpha(x) - \nabla_x \phi(x)}^2 - \tfrac{1}{2} \babs{\nabla_x \times A(x)}^2 
	\, + \Bigr . \\
	&\qquad \qquad \quad \Bigl . 
	+ \, j(t,x) \cdot A(x) - \rho(t,x) \, \phi(x) \Bigr )
\end{align*}
where $A$ is the vector potential, $\phi$ the scalar potential, $j$ the current density and $\rho$ the charge density. Because of \emph{charge conservation} 
\begin{align}
	\nabla_x \cdot j + \partial_t \rho = 0 
	,
	\label{variation:eqn:charge_conservation}
\end{align}
current and charge density are linked. The potentials are linked to the electromagnetic field via 
\begin{align*}
	\mathbf{E} &= - \partial_t A - \nabla_x \phi 
	, 
	\\
	\mathbf{B} &= \nabla_x \times A 
	. 
\end{align*}
Given that $L$ is defined in terms of a quadratic polynomial in the fields, we can easily deduce the equations of motion from the action functional 
\begin{align}
	S(A,\phi) &= \int_0^T \dd t \, L \bigl ( A(t) , \phi(t) , \partial_t A(t) , \partial_t \phi(t) \bigr )
	% . 
	\label{variation:eqn:Maxwell_action}
\end{align}
where $(A,\phi)$ is a path in the space of potentials. Physicists usually use $\delta$ to denote (what they call) “functional differentiation”, and the Euler-Langrange equations~\eqref{variation:eqn:Euler_Lagrange_Lagrangian_mechanics} are expressed as 
\begin{subequations}\label{variation:eqn:Maxwell_Euler_Lagrange}
	\begin{align}
		\frac{\dd}{\dd t} \frac{\delta \mathbb{L}}{\delta \dot{\phi}} - \frac{\delta \mathbb{L}}{\delta \phi} &= 0 
		, 
		\label{variation:eqn:Maxwell_Euler_Lagrange_phi}
		\\
		\frac{\dd}{\dd t} \frac{\delta \mathbb{L}}{\delta \dot{A}} - \frac{\delta \mathbb{L}}{\delta A} &= 0 
		. 
		\label{variation:eqn:Maxwell_Euler_Lagrange_A}
	\end{align}
\end{subequations}
These can be computed by pretending that the integrand is just a polynomial in $A$, $\partial_t A$, $\phi$ and $\partial_t \phi$. We postpone a proper derivation until after the discussion of these two equations.\marginpar{2014.03.12}

The \emph{Lagrange density $\mathbb{L}$} is independent of $\dot{\phi} = \partial_t \phi$ so that partial integration yields \emph{Gauß's law} (charges are the sources of electric fields), 
\begin{align}
	\nabla_x \cdot \bigl ( - \partial_t A - \nabla_x \phi \bigr ) &= \nabla_x \cdot \mathbf{E} 
	= \rho 
	. 
	\label{variation:eqn:edyn_constraint_equation}
\end{align}
This equation acts as a constraint and is not a dynamical equation of motion. Note that since $\mathbf{B} = \nabla_x \times A$ is the curl of a vector field, it is automatically divergence-free, $\nabla_x \cdot \mathbf{B} = 0$. This takes care of the two constraints, equations~\eqref{operators:eqn:source_Maxwell_eqns} for $\eps = 1 = \mu$. 

Equation~\eqref{variation:eqn:Maxwell_Euler_Lagrange_A} yields both of the dynamical Maxwell equations, starting with  
\begin{align}
	- \partial_t^2 A - \nabla_x \partial_t \phi &= \partial_t \mathbf{E} 
	= \nabla_x \times \nabla_x A - j
	= \nabla_x \times \mathbf{B} - j 
	. 
	\label{variation:eqn:edyn_dynamical_equation}
\end{align}
To obtain the other dynamical Maxwell equation, we derive $\mathbf{B} = \nabla_x \times A$ with respect to time and use $\nabla_x \times \nabla_x \phi = 0$: 
\begin{align*}
	\partial_t \mathbf{B} &= \nabla_x \times \partial_t A 
	= - \nabla_x \times \bigl ( - \partial_t A - \nabla_x \phi \bigr )
	\\
	&= - \nabla_x \times \mathbf{E} 
\end{align*}
Hence, we obtain the usual Maxwell equations after introducing a pair of new variables, the electric field $\mathbf{E}$ and the magnetic field $\mathbf{B}$. These fields are \emph{independent} of the choice of gauge, but more on that below. 

Solutions to the Maxwell equations are stationary points of the action functional~\eqref{variation:eqn:Maxwell_action}, and a proper derivation involves computing the functional derivative: 
\begin{align*}
	\bigl ( S(A,\phi) \bigr )(a,\varphi) &= \left . \frac{\dd}{\dd s} S \bigl ( A + s a \, , \, \phi + s \varphi \bigr ) \right \vert_{s = 0}
	\\
	&= \int_0^T \dd t \int_{\R^3} \dd x \, \frac{\dd}{\dd s} \left ( 
	\frac{1}{2} \Babs{- \partial_t A - \nabla_x \phi - s \, \partial_t a - s \, \nabla_x \varphi}^2 
	+ \right . \\
	&\qquad \qquad \qquad \qquad \qquad \left . 
	- \frac{1}{2} \Babs{\nabla_x \times A + s \, \nabla_x \times a}^2 
	+ \right . \\
	&\qquad \qquad \qquad \qquad \qquad + \biggl . 
	j \cdot A - \rho \, \phi + s \, j \cdot a - s \, \rho \, \varphi 
	\biggr ) \bigg \rvert_{s = 0}
	\\
	&= \int_0^T \dd t \int_{\R^3} \dd x \, \Bigl ( 
	\bigl ( - \partial_t A - \nabla_x \phi \bigr ) \cdot \bigl ( - \partial_t a - \nabla_x \varphi \bigr ) 
	\, + \Bigr . 
	\\
	&\qquad \qquad \qquad \qquad \Bigl . 
	- \bigl ( \nabla_x \times A \bigr ) \cdot \bigl ( \nabla_x \times a \bigr ) + j \cdot a - \rho \, \varphi 
	\Bigr ) 
	\\
	&= \int_0^T \dd t \int_{\R^3} \dd x \, \biggl ( 
	\Bigl ( - \partial_t^2 A - \nabla_x \partial_t \phi - \nabla_x \times \nabla_x \times A + j \Bigr ) \cdot a 
	+ \biggr . \\
	&\qquad \qquad \qquad \qquad \biggl . + 
	\Bigl ( \nabla_x \cdot \bigl ( - \partial_t A - \nabla_x \phi \bigr ) - \rho \Bigr ) \, \varphi 
	\biggr ) 
\end{align*}
Setting $\dd S(A,\phi) = 0$ yields equations \eqref{variation:eqn:edyn_dynamical_equation} and \eqref{variation:eqn:edyn_constraint_equation}.

\paragraph{Eliminating the constraints} % (fold)
The presence of the constraint equations 
\begin{align*}
	\nabla_x \cdot \mathbf{E} &= \rho 
	\\
	\nabla_x \cdot \mathbf{B} &= 0 
\end{align*}
means we can in fact eliminate some variables. The idea is to decompose the electromagnetic fields 
\begin{align*}
	\mathbf{E} &= \mathbf{E}_{\parallel} + \mathbf{E}_{\perp} 
	, 
	\\
	\mathbf{B} &= \mathbf{B}_{\parallel} + \mathbf{B}_{\perp} 
\end{align*}
into longitudinal ($\parallel$) and transversal ($\perp$) component. Transversal fields are those for which $\nabla_x \cdot \mathbf{E}_{\perp} = 0$ holds while the longitudinal component is simply the remainder $\mathbf{E}_{\parallel} = \mathbf{E} - \mathbf{E}_{\perp} = \nabla_x \chi$ which can always be written as a gradient of some function; on $\R^3$ this \emph{Helmholtz decomposition} of fields is unique. 
\begin{theorem}[Helmholtz-Hodge-Weyl-Leeray decomposition]
	The Hilbert space 
	\begin{align*}
		L^2(\R^3,\C^3) = J \oplus G 
	\end{align*}
	decomposes into the orthogonal subspaces 
	\begin{align}
		J := \ker \mathrm{div} = \ran \mathrm{curl}
		\label{variation:eqn:definition_J_Helmholtz}
	\end{align}
	and 
	\begin{align}
		G := \ran \mathrm{grad} = \ker \mathrm{curl} 
		.  
		\label{variation:eqn:definition_G_Helmholtz}
	\end{align}
	In other words, any vector field $\mathbf{E} = \mathbf{E}_{\parallel} + \mathbf{E}_{\perp}$ can be uniquely decomposed into $\mathbf{E}_{\parallel} \in G$ and $\mathbf{E}_{\perp} \in J$. 
\end{theorem}
\begin{proof}[Sketch]
	We will disregard most technical questions and content ourselves showing orthogonality of the subspaces and $J \cap G = \{ 0 \}$. Note that since $\Cont^{\infty}_{\mathrm{c}}(\R^3,\C^3)$ is dense in $L^2(\R^3,\C^3)$, it suffices to work with vectors from that dense subspace. The proof of 
	equation~\eqref{variation:eqn:definition_G_Helmholtz} can be found in \cite[Chapter~I, Theorem~1.4, equation~(1.34)]{Temam:theory_Navier_Stokes:2001}; equation~\eqref{variation:eqn:definition_G_Helmholtz} is shown in \cite[Chapter~I, Theorem~1.4, equation~(1.33) and Remark~1.5]{Temam:theory_Navier_Stokes:2001} and \cite[Theorem~1.1]{Picard:selfadjointness_curl:1998}. 
	
	Let us start with orthogonality: Pick $\phi = \nabla_x \varphi \in G \cap \Cont^{\infty}_{\mathrm{c}}(\R^3,\C^3)$ where $\varphi \in \Cont^{\infty}_{\mathrm{c}}(\R^3)$ and $\psi \in J = \ker \mathrm{div}$. Then partial integration yields 
	\begin{align*}
		\bscpro{\psi}{\phi} = \bscpro{\psi}{\nabla_x \varphi} = \bscpro{\nabla_x \cdot \psi}{\varphi} = 0 
		, 
	\end{align*}
	meaning that the vectors are necessarily orthogonal. 
	
	It is crucial here that the space of \emph{harmonic vector fields}
	\begin{align*}
		\mathrm{Har}(\R^3,\R^3) := \ker \mathrm{div} \cap \ker \mathrm{curl} 
		= J \cap G 
		= \{ 0 \} 
	\end{align*}
	is trivial, because by 
	\begin{align*}
		\nabla_x \times \nabla_x \times \phi = \nabla_x \bigl ( \nabla_x \cdot \phi \bigr ) - \Delta_x \phi 
	\end{align*}
	the subvector space consists of functions which component-wise satisfy $\Delta_x \phi = 0$. But on $\Cont^{\infty}_{\mathrm{c}}(\R^3,\R^3)$ and $L^2(\R^3,\R^3)$, this equation has no non-trivial solution. 
\end{proof}
On bounded subsets of $\R^3$, there \emph{are} harmonic vector fields because then at least the constant function is square integrable. Then the Helmholtz splitting is more subtle. 

Clearly, by definition $\mathbf{B}_{\parallel} = 0$ and $\nabla_x \cdot \mathbf{E} = \nabla_x \cdot \mathbf{E}_{\parallel} = \rho$. Moreover, $\nabla_x \times \mathbf{E} = \nabla_x \times \mathbf{E}_{\perp}$ holds so that we obtain 
\begin{align*}
	\mathbf{B}_{\parallel}(t) &= 0 
	, 
	\\
	\nabla_x \cdot \mathbf{E}_{\parallel}(t) &= \rho(t) 
	. 
\end{align*}
Note that the last equation can be solved with the help of the Fourier transform. 

The two dynamical contributions now only involve the \emph{transversal} components of the fields, 
\begin{align*}
	\partial_t \mathbf{E}_{\perp} &= \nabla_x \times \mathbf{B}_{\perp} - j
	, 
	\\
	\partial_t \mathbf{B}_{\perp} &= - \nabla_x \times \mathbf{E}_{\perp} 
	. 
\end{align*}
%
% paragraph eliminating_the_constraints (end)

\paragraph{Emergence of $\mathbf{E}$ and $\mathbf{B}$} % (fold)
One may wonder what motivates one to set $\mathbf{E} = - \partial_t A - \nabla_x \phi$ and $\mathbf{B} = \nabla_x \times A$. Let us re-examine the Euler-Lagrange equation for the second variable, 
\begin{align*}
	- \partial_t^2 A - \nabla_x \partial_t \phi &= \nabla_x \times \nabla_x A - j
	. 
\end{align*}
The left-hand side involves second-order time-derivatives, and if we want to write it as first-order equation, we can introduce the new variable $\mathbf{E} = - \partial_t A - \nabla_x \phi$ to obtain 
\begin{align*}
	\frac{\dd}{\dd t} \left (
	\begin{matrix}
		A \\
		\mathbf{E} \\
	\end{matrix}
	\right ) &= \left (
	\begin{matrix}
		- \mathbf{E} - \nabla_x \phi \\
		\nabla_x \times A \\
	\end{matrix}
	\right )
	. 
\end{align*}
%
% paragraph emergence_of_e_and_b_ (end)

\paragraph{Gauge symmetry} % (fold)
% 
% CHANGED add citation to Noether's theorem
The constraint equation \eqref{variation:eqn:edyn_constraint_equation} is related to a continuous symmetry and leads to a \emph{conserved quantity}. The relation between the two, a continuous symmetry and a conserved quantity, is made precise by \emph{Noether's theorem} (\cf \cite[Theorem~11.4.1]{Marsden_Ratiu:intro_mechanics_symmetry:1999}). Here, the \emph{gauge symmetry} of the action functional leads to \eqref{variation:eqn:edyn_constraint_equation}: if $\chi : \R \times \R^3 \longrightarrow \R$ is a smooth function depending on time and space, then a quick computation shows 
\begin{align*}
	S \bigl ( A , \phi \bigr ) = S \bigl ( A + \nabla_x \chi \, , \, \phi - \partial_t \chi \bigr )
\end{align*}
holds because the extra terms in the first two terms cancel exactly while those in the last two cancel because of \emph{charge conservation}~\eqref{variation:eqn:charge_conservation}. 
% paragraph gauge_symmetry (end)
% subsubsection derivation_of_maxwell_s_equations (end)
% subsection lagrangian_mechanics (end)

\subsection{Constraints} % (fold)
\label{variation:extremals:constraints}
Two of the four Maxwell equations describe \emph{constraints}, and we have seen how to factor out the constraints in the dynamical equations by splitting $\mathbf{E}$ and $\mathbf{B}$ into longitudinal and transversal components. This represents one way to deal with constraints, one introduces a \emph{suitable parametrization} (coordinates) to factor them out. Unfortunately, this is often non-obvious and practically impossible. But fortunately, one can apply a well-known technique if one wants to find extrema for functions on subsets of $\R^d$ under constraints, \emph{Lagrange multipliers}. 

For simplicity, let us reconsider the case of functions on $\R^d$. There, the idea of Lagrange multipliers is very simple: assume one wants to find the extrema of $f : \R^d \longrightarrow \R$ under the constraints 
\begin{align*}
	g(x) = \bigl ( g_1(x) , \ldots , g_n(x) \bigr ) = 0 \in \R^n
	. 
\end{align*}
Then this problem is equivalent to finding the (ordinary) extrema of the function 
\begin{align*}
	F : \R^d \times \R^n \longrightarrow \R 
	, \; \; 
	F(x,\lambda) := f(x) + \lambda \cdot g(x) 
\end{align*}
and the condition for the extrema, 
\begin{align*}
	\nabla_x f(x) + \lambda \, \nabla_x g(x) &= 0 
	, 
	\\
	g(x) &= 0 
	, 
\end{align*}
shows how the recipe of using Lagrange multipliers works behind the scenes. There is also a simple visualization in case $n = 1$: setting $\nabla_x f(x) + \lambda \, \nabla_x g(x) = 0$ means $\nabla_x f(x)$ and $\nabla_x g(x)$ are parallel to one another. Assume, for instance, that $x_0$ is a local maximum, but that $\nabla_x f(x_0)$ and $\nabla_x g(x_0)$ are not parallel. First of, $\nabla_x g(x)$ is always \emph{normal} to the surface defined by $g(x) = 0$, and we can split $\nabla_x f(x_0) = v_{\parallel} + v_{\perp}$ where $v_{\parallel} \parallel \nabla_x g(x_0)$ and $v_{\perp}$ is orthogonal. Then $v_{\perp} \neq 0$ means we can increase the value of $f(x)$ along $g(x) = 0$ by going in the direction of $v_{\perp}$, because 
\begin{align*}
	v_{\perp} \cdot \partial_{v_{\perp}} f(x_0) = v_{\perp}^2 > 0 
	. 
\end{align*}
The argument for a local minimum is analogous, one simply has to walk in the direction opposite of $v_{\perp}$ to lower $f(x)$. Put another way, tangential to the surface, all components of $\nabla_x f(x)$ need to vanish, but along the surface normal $\nabla_x f(x)$ need not be zero. But it is prohibited to travel along this direction, because one would leave the surface $\{ g(x) = 0 \}$. 

To go back to the realm of functionals, we merely need to translate all these ideas properly: consider a functional $\mathcal{E} : \Omega \subseteq \mathcal{X} \longrightarrow \R$ restricted to 
\begin{align*}
	U := \bigl \{ x \in \Omega \; \; \vert \; \; \mathcal{J}(x) = 0 \bigr \} 
\end{align*}
described by the constraint functional $\mathcal{J}$. 
\begin{theorem}
	Let $\mathcal{E}$ and $\mathcal{J}$ both be $\Cont^1$ functionals and assume $x_0 \in U$ is a critical point of $\mathcal{E} \vert_U$. Then there exists $\lambda \in \R$ such that 
	\begin{align*}
		\dd \mathcal{E}(x_0) + \lambda \, \dd \mathcal{J}(x_0) = 0 
		. 
	\end{align*}
\end{theorem}
\begin{proof}
	Let $x_s$ be a differentiable path in $U$ for the tangent vector $\xi = \partial_t x_s \vert_{s = 0}$, \ie it is a differentiable path in $\Omega$ such that $J(x_s) = 0$. Then $J(x_s) = 0$ implies automatically $\dd J(x_0) \xi = 0$ for all such paths, \ie $\xi \in \ker \dd J(x_0)$ (“$\dd J(x)$ is normal to $J(x) = 0$”). 
	
	The fact that $x_0$ is a critical points of $\mathcal{E}$ in $U$ means 
	\begin{align*}
		\dd \mathcal{E}(x_0) \xi &= \left . \frac{\dd}{\dd s} \mathcal{E}(x_s) \right \vert_{s = 0} = 0 
		, 
	\end{align*}
	meaning $\xi \in \ker J(x_0)$ at critical points implies $\xi \in \ker \dd \mathcal{E}(x_0)$ (“$\dd \mathcal{E}(x_0)$ and $\dd \mathcal{J}(x_0)$ are parallel”). 
\end{proof}
The proof once again states that in the tangential direction to $U$, one needs to have $\dd \mathcal{E}(x_0) = 0$ while in the normal direction, $\dd \mathcal{E}(x_0) \neq 0$ is permissible because it is not possible to travel in this direction as one would leave the “surface” $U = \{ J(x) = 0 \}$.\marginpar{2014.03.13} 
\begin{example}
	Assume we would like to find the equations of motions of a particle of mass $1$ moving on the surface of the sphere $\mathbb{S}^2 \subset \R^3$ with radius $1$. Then the Lagrange function describing this motion is just the free Lagrange function in $\R^3$, namely $L(x,v) = \frac{1}{2} v^2$. The constraint functional is described in terms of the function $J(x,v) = \tfrac{1}{4} \bigl ( x^2 - 1 \bigr )^2$, 
	\begin{align*}
		\mathcal{J}(q) := \int_0^T \dd t \, J \bigl (q(t) , \dot{q}(t) \bigr ) 
		. 
	\end{align*}
	Paths on the sphere satisfy $\mathcal{J}(q) = 0$. Equation~\eqref{variation:eqn:Euler_Lagrange_with_L_integral} now yields a very efficient way to compute $\dd S + \lambda \, \dd \mathcal{J}$, namely for any path $q$ on the surface of the sphere, we obtain 
	\begin{align*}
		\bigl ( \dd S + \lambda \, \dd \mathcal{J} \bigr ) h &= \int_0^T \dd t \, \left ( \nabla_x \bigl ( L + \lambda \, \mathcal{J} \bigr ) \bigl ( q(t) , \dot{q}(t) \bigr ) - \frac{\dd}{\dd t} \nabla_v \bigl ( L + \lambda \, \mathcal{J} \bigr )  \bigl ( q(t) , \dot{q}(t) \bigr ) \right ) \cdot h(t)
		\\
		&= \int_0^T \dd t \, \bigl ( - \ddot{q}(t) + \lambda \, (q(t)^2 - 1) \, q(t) \bigr ) \cdot h(t)
		. 
	\end{align*}
	By construction, $h(t)$ is tangent to the sphere at $q(t)$, \ie $h(t)$ needs to be perpendicular to $q(t)$, and thus the second term vanishes (the term also vanishes because $q(t)^2 = 1$, but the other argument still holds true if we change the constraint function). In fact, $q(t)$ is always normal to the tangent plane, and thus saying that $\ddot{q}(t)$ is perpendicular to the plane really means 
	\begin{align}
		\ddot{q}(t) = \lambda(t) \, q(t) 
		. 
		\label{variation:eqn:ddotq_free_particle_sphere}
	\end{align}
	Writing this equation as a first-order equation, we introduce the variable $v(t) := \dot{q}(t)$ which by definition is tangent to $\mathbb{S}^2$ at $q(t)$. Because we are in three dimensions, there exists a unique vector $\omega(t)$ with $\omega(t) \cdot q(t) = 0$ ($\omega(t)$ must lie in the tangent plane), $\omega(t) \cdot \dot{q}(t) = 0$ and 
	\begin{align*}
		\dot{q}(t) = \omega(t) \times q(t) 
		, 
	\end{align*}
	because $\omega(t)$ is the vector which completes $\{ q(t) , \dot{q}(t) \}$ to an orthogonal basis of $\R^3$ and is proportional to $q(t) \times \dot{q}(t)$. In fact, the above equation implies $\omega(t) = q(t) \times \dot{q}(t)$ since 
	\begin{align*}
		q(t) \times \dot{q}(t) &= q(t) \times \bigl ( \omega(t) \times q(t) \bigr ) 
		= \underbrace{\abs{q(t)}^2}_{= 1} \, \omega(t) - \underbrace{q(t) \cdot \omega(t)}_{= 0} \, q(t)
		= \omega(t)
		. 
	\end{align*}
	Deriving the left-hand side with respect to time yields an equation of motion for $\omega$, 
	\begin{align*}
		\frac{\dd}{\dd t} \omega(t) &= \frac{\dd}{\dd t} \bigl ( q(t) \times \dot{q}(t) \bigr ) 
		= \dot{q}(t) \times \dot{q}(t) + q(t) \times \ddot{q}(t)
		= q(t) \times \ddot{q}(t)
		, 
	\end{align*}
	and with the help equation \eqref{variation:eqn:ddotq_free_particle_sphere}, we deduce $\frac{\dd}{\dd t} \omega(t) = 0$. That means $\omega = \omega(0) = \omega(t) = q(0) \times \dot{q}(0)$ is constant in time and the equations of motion reduce to 
	\begin{align*}
		\dot{q}(t) &= \omega \times q(t) 
		. 
	\end{align*}
	We have solved these equation numerous times in the course, the solutions are rotations around the axis along $\omega$ with angular velocity $\abs{\omega}$, just as expected. 
\end{example}
%
% subsection constraints (end)

\subsection{Complex derivatives} % (fold)
\label{variation:extremals:complex_derivatives}
To frame the discussion, we will quickly recap why complex derivatives of functions are fundamentally different from real derivatives. In what follows let us denote complex numbers with $z = z_{\Re} + \ii z_{\Im} \in \C$ where $z_{\Re} , z_{\Im} \in \R$ are real and imaginary part. We also split any function $f : \C \longrightarrow \C$ into real and imaginary part as $f = f_{\Re} + \ii f_{\Im}$ where now $f_{\Re} , f_{\Im} : \C \longrightarrow \R$. The identification of $\C \cong \R^2$ via $z \mapsto \vec{z} :=  \bigl ( z_{\Re} , z_{\Im} \bigr )$ allows us to think of 
\begin{align}
	\partial_z f(z_0) = \lim_{z \to z_0} \frac{f(z) - f(z_0)}{z - z_0}
	\label{variation:eqn:complex_derivative}
\end{align}
as a limit in $\R^2$. As we are in more than one (real) dimension, the above equation implicitly assumes that the limit is \emph{independent of the path} taken as $z \to z_0$. To simplify the notation a little, we will assume without loss of generality that $z_0 = 0$ and $f(0) = 0$. Then the existence of the limit~\eqref{variation:eqn:complex_derivative} implies that 
\begin{align*}
	\partial_z f(0) &= \lim_{z_{\Re} \to 0} \frac{f_{\Re}(z_{\Re}) + \ii f_{\Im}(z_{\Re})}{z_{\Re}} 
	= \partial_{z_{\Re}} f(0) 
	\\
	&= \lim_{z_{\Im} \to 0} \frac{f_{\Im}(\ii z_{\Im}) + \ii f_{\Im}(\ii z_{\Im})}{\ii z_{\Im}} 
	= - \ii \, \partial_{z_{\Im}} f(0)
	\\
	&= \tfrac{1}{2} \bigl ( \partial_{z_{\Re}} - \ii \partial_{z_{\Im}} \bigr ) f(0) 
\end{align*}
are in fact equal. Thus, equating real and imaginary part, we immediately get the Cauchy-Riemann equations, 
\begin{subequations}\label{variation:eqn:Cauchy_Riemann}
	\begin{align}
		\partial_{z_{\Re}} f_{\Re}(0) &= + \partial_{z_{\Im}} f_{\Im}(0) 
		, 
		\\
		\partial_{z_{\Im}} f_{\Re}(0) &= - \partial_{z_{\Re}} f_{\Im}(0) 
		. 
	\end{align}
\end{subequations}
This reasoning shows that complex differentiability of $f$ (\ie the limit in \eqref{variation:eqn:complex_derivative} exists) implies \eqref{variation:eqn:Cauchy_Riemann}. In fact, these two are equivalent, a function is complex differentiable or \emph{holomorphic} if and only if 
\begin{align*}
	\partial_{\bar{z}} f(z) = \tfrac{1}{2} \bigl ( \partial_{z_{\Re}} + \ii \partial_{z_{\Im}} \bigr ) f(z) 
	= 0 
	. 
\end{align*}
So let us turn our attention back to functionals. The idea here is the same: we can identify each complex Banach space $\mathcal{X} \cong \mathcal{X}_{\R} \oplus \ii \mathcal{X}_{\R}$ as a Banach space over $\R$ whose dimension is twice as large via the identification $x \mapsto \vec{x} = \bigl ( \Re x , \Im x \bigr ) = \bigl ( x_{\Re} , x_{\Im} \bigr )$. Hence, we can associate to any functional $\mathcal{E} : \Omega \subseteq \mathcal{X} \longrightarrow \C$ another functional on $\mathcal{X}_{\R} \oplus \ii \mathcal{X}_{\R}$ via 
\begin{align*}
	\mathcal{E}_{\R}(\vec{x}) = \mathcal{E} \bigl ( x_{\Re} + \ii x_{\Im} \bigr ) 
	. 
\end{align*}
For this functional, we can take (partial) derivatives with respect to $x_{\Re}$ and $x_{\Im}$ as before (we are in the setting of functionals over a real Banach space again) as well as define the gradient $\partial_{\vec{x}}$. Then analogously to derivatives on $\C$, we define the complex partial  
\begin{subequations}
	\begin{align}
		\partial_x \mathcal{E}(x) &:= \partial_{x_{\Re}} \mathcal{E}(x) - \ii \partial_{x_{\Im}} \mathcal{E}(x) 
		, 
		\\
		\partial_{\bar{x}} \mathcal{E}(x) &:= \partial_{x_{\Re}} \mathcal{E}(x) + \ii \partial_{x_{\Im}} \mathcal{E}(x) 
		. 
	\end{align}
\end{subequations}
and the complex Gâteaux derivatives 
\begin{subequations}
	\begin{align}
		\dd \mathcal{E}(x) \equiv \dd_x \mathcal{E}(x) &:= \dd_{x_{\Re}} \mathcal{E}(x) - \ii \dd_{x_{\Im}} \mathcal{E}(x) 
		, 
		\\
		\bar{\dd} \mathcal{E}(x) \equiv \dd_{\bar{x}} \mathcal{E}(x) &:= \dd_{x_{\Re}} \mathcal{E}(x) + \ii \dd_{x_{\Im}} \mathcal{E}(x) 
		. 
	\end{align}
\end{subequations}
Now that derivatives have been extended, let us define the notion of 
\begin{definition}[Critical point]
	Let $\mathcal{E} : \Omega \subseteq \mathcal{X} \longrightarrow \C$ be a $\Cont^1$ functional. Then $x_0 \in \Omega$ is a critical point if and only if $\bar{\dd} \mathcal{E}(x_0) = 0$. 
\end{definition}
One way to compute $\bar{\dd} \mathcal{E}(x)$ is to treat $x$ and $\bar{x}$ as independent functions and compute the corresponding partial Gâteaux derivative. Note that in case of polynomials, 

The reason why we have chosen to take the derivative $\bar{\dd}$ instead of $\dd$ is just a matter of convention and convenience: in the context of Hilbert spaces, we can write 
\begin{align*}
	\bar{\dd} \mathcal{E}(\psi) \varphi = \bscpro{\varphi}{\mathcal{E}'(\psi)}
\end{align*}
as a scalar product; the vector $\mathcal{E}'(\psi)$ is defined via the Riesz representation theorem~\ref{hilbert_spaces:dual_space:thm:Riesz_Lemma}. If we had used $\dd$ instead, then $\mathcal{E}'(\psi)$ would appear in the first, the \emph{anti-linear} argument of the scalar product.

\paragraph{Derivation of the Ginzburg-Landau equations} % (fold)
As a simple example, let us derive the Euler-Lagrange equation for the Ginzburg-Landau energy functional, 
\begin{align}
	\mathcal{E}_{\Omega}(\psi,A) := \int_{\Omega} \dd x \, \Bigl ( \babs{\bigl ( - \ii \nabla_x - A(x) \bigr ) \psi(x)}^2 + \tfrac{\kappa^2}{2} \bigl ( \sabs{\psi(x)}^2 - 1 \bigr )^2 + \bigl ( \nabla_x \times A(x) \bigr )^2 \Bigr ) 
	. 
	\label{variation:eqn:Ginzburg_Landau_functional}
\end{align}
It describes the difference in Helmholtz free energy between the ordinary and superconducting phase of an ordinary superconductor in the bounded region $\Omega \subseteq \R^3$ which is subjected to a magnetic field $B = \nabla_x \times A$. Here, $\psi$ is an \emph{order parameter} which describes whether the material is in an ordinary conductor $\psi = 0$ or the electrons have formed Cooper pairs which carry a superconducting current ($\psi \neq 0$). 

Admissible states are by definition critical points of the Ginzburg-Landau functional, 
\begin{align}
	\bigl ( \dd \mathcal{E}_{\Omega}(\psi,A) \bigr )(\varphi,a) &= \left . \frac{\dd}{\dd s} \mathcal{E}_{\Omega} \bigl ( \psi + s \varphi \, , \, A + s a \bigr ) \right \vert_{s = 0} 
	\notag \\
	&= 2 \Re \int_{\Omega} \dd x \, \Bigl ( 
	  a \cdot \Bigl ( - \overline{\psi} \, \bigl ( - \ii \nabla_x - A \bigr ) \psi + \nabla_x \times \nabla_x \times A \Bigr ) 
	  + \Bigr . \notag \\
	  &\qquad \qquad \qquad \; \Bigl . + \, 
	  \overline{\varphi} \, \Bigl ( \bigl ( - \ii \nabla_x - A \bigr )^2 \psi - \kappa^2 \bigl ( 1 - \sabs{\psi}^2 \bigr ) \, \psi 
	\Bigr ) 
	\overset{!}{=} 0 
	, 
	\label{variation:eqn:Ginzburg_Landau_critical_point}
\end{align}
and leads to the \emph{Ginzburg-Landau equations}, 
\begin{subequations}
	\begin{align}
		\nabla_x \times \nabla_x \times A &= \Re \bigl ( \overline{\psi} \, \bigl ( - \ii \nabla_x - A \bigr ) \psi \bigr ) =: j(\psi,A)
		,
		\\
		0 &= \bigl ( - \ii \nabla_x - A \bigr )^2 \psi - \kappa^2 \bigl ( 1 - \sabs{\psi}^2 \bigr ) \, \psi
		. 
	\end{align}
\end{subequations}
We could have also obtained these equations by deriving the integrand with respect to $\overline{\psi}$. 

Clearly, for all $A$ which describe \emph{constant} magnetic fields $B = \nabla_x \times A$ the normal conducting phase $\psi = 0$ as well as the perfect superconductor $\psi = 1$ are solutions, and the question is for which values of $\kappa$ and $B$ there are other solutions (mixed phase where normal and superconducting states coexist). We will continue this example later on in the chapter. 
% paragraph derivation_of_the_ginzburg_landau_equations (end)
% subsection complex_derivatives (end)

\subsection{Second derivatives} % (fold)
\label{variation:extremals:second_derivatives}
When we introduced the Gâteaux derivative for functionals, we have claimed that one can compute second- and higher-order derivatives in the same fashion. We choose to illustrate the general principle with the Ginzburg-Landau equations: the Euler-Lagrange equations can either be understood as the integral expression~\eqref{variation:eqn:Ginzburg_Landau_critical_point} or the \emph{$L^2$ gradient}
\begin{align}
	\mathcal{E}_{\Omega}'(\psi,A) &= \left (
	\begin{matrix}
		\bigl ( - \ii \nabla_x - A \bigr )^2 \psi - \kappa^2 \bigl ( 1 - \sabs{\psi}^2 \bigr ) \, \psi \\
		\nabla_x \times \nabla_x \times A - \Re \bigl ( \overline{\psi} \, \bigl ( - \ii \nabla_x - A \bigr ) \psi \bigr ) \\
	\end{matrix}
	\right )
\end{align}
defined by $\bscpro{(\varphi,a)}{\mathcal{E}_{\Omega}'(\psi,A)} := \bigl ( \dd \mathcal{E}_{\Omega}(\psi,A) \bigr )(a,\varphi)$. The \emph{Hessian} $\mathcal{E}_{\Omega}''$ of $\mathcal{E}_{\Omega}$ now is just the \emph{linearization} of $\mathcal{E}_{\Omega}'(\psi)$, 
\begin{align}
	\bigl ( \mathcal{E}_{\Omega}''(\psi,A) \bigr )(\varphi,a) := \left . \frac{\dd}{\dd s}  \mathcal{E}_{\Omega}' \bigl ( \psi + s \varphi \, , \, A + s a \bigr ) \right \vert_{s = 0}
	. 
\end{align}
Equivalently, we could have used the quadratic form 
\begin{align*}
	\bscpro{(\eta,\alpha)}{\mathcal{E}_{\Omega}''(\psi,A) (\varphi,a)} &= \frac{\partial^2}{\partial s \, \partial t} \mathcal{E}_{\Omega} \bigl ( \psi + s \varphi + t \eta \, , \, A + s a + t \alpha \bigr ) \bigg \vert_{s = t = 0}
\end{align*}
as a definition. A somewhat lengthier computation yields an explicit expression, 
\begin{align*}
	&\bigl ( \mathcal{E}_{\Omega}''(\psi,A) \bigr ) (\varphi,a) = 
	\\
	&\, = \left (
	\begin{matrix}
		\bigl ( - \ii \nabla_x - A \bigr )^2 \varphi + \kappa^2 \, \bigl ( 2 \, \sabs{\psi}^2 - 1 \bigr ) \, \varphi + \kappa^2 \, \psi^2 \, \overline{\varphi} - 2 \, \bigl ( - \ii \nabla_x - A \bigr ) \psi \cdot a + \ii \, \psi \, \nabla_x \cdot a \\
		\bigl ( \nabla_x \times \nabla_x + \sabs{\psi}^2 \bigr ) a - \Re \bigl ( \overline{\varphi} \, \bigl ( - \ii \nabla_x - A \bigr ) \psi + \overline{\psi} \, \bigl ( - \ii \nabla_x - A \bigr ) \varphi \bigr ) \\
	\end{matrix}
	\right )
	, 
\end{align*}
and the take-away message here is that $\mathcal{E}_{\Omega}''(\psi,A)$ is an $\R$-linear map, and thus, we can use the tools of functional analysis to probe the properties of the Hessian. 
\medskip

\noindent
Quite generally, the Hessian of a functional $\mathcal{E}$ is just the second-order Gâteaux derivative, 
\begin{align*}
	\bscpro{\beta}{\mathcal{E}''(\psi) \alpha} &= \frac{\partial^2}{\partial s \, \partial t} \mathcal{E} \bigl ( \psi + s \alpha + t \beta \bigr ) \bigg \vert_{s = t = 0}
	, 
\end{align*}
which implicitly defines a \emph{linear} partial differential operator. The properties of the Hessian characterize the behavior of the fixed point: is it a local maximum, a local minimum or a saddle point? 

In short, we have the following hierarchy: We would like find the critical points of a given a functional $\mathcal{E}$ and characterize them. The critical point equation is then equivalent to a \emph{non-linear PDE} whose solutions are the fixed points. The linearization of this PDE at a critical point yields a \emph{linear PDE} which serves as the starting point for the \emph{stability analysis} of that fixed point. This is exactly what we have done for ODEs in Chapter~\ref{odes:stability}: the non-linear vector field determines the fixed points while its linearization can be used to classify the fixed point. Hence, functionals generate many interesting linear and non-linear PDEs. 
% subsection convexity (end)
% section extremals_of_functionals (end)

\section{Key points for a rigorous analysis} % (fold)
\label{variation:rigorous}
% 
%
% \begin{itemize}
% 	\item existence of Gâteaux derivative does not imply continuity (ditto for functions, drag out standard example)
% 	\item Michael pp.~22
% 	\item alternate strategy: $\mathcal{E}$ bounded from below, convex
% \end{itemize}
%
Up to now, we have not been very rigorous. A careful reader will have noticed that even though we have defined a notion of differentiability, we have \emph{not} defined \emph{continuity} which in real analysis precedes differentiability. That was a conscious choice, because the question of continuity for functionals turns out to be much more involved. But continuity is not necessary to define the Gâteaux derivative:

So let us sketch what is involved in a rigorous analysis of functionals. The purpose here is not to completely cover and overcome all of the difficulties, but to merely point them out. So let us consider a functional $\mathcal{E}$ on $\Omega \subseteq \mathcal{X}$, and assume we would like to find a \emph{minimizer}, \ie an element $x_0 \in \Omega$ for which 
\begin{align*}
	\mathcal{E}(x_0) = \inf_{x \in \Omega} \mathcal{E}(x)
	, 
\end{align*}
and whether this minimizer is \emph{unique}. Clearly, we have to assume that the functional is \emph{bounded from below}, 
\begin{align*}
	E_0 := \inf_{x \in \Omega} \mathcal{E}(x)
	, 
\end{align*}
otherwise such a minimizer cannot exist. 

Now let us consider the case where $\Omega$ is a closed subset of $\R^d$, \ie we are looking at functions (in the ordinary sense) from $\R^d$ to $\R$. We split this minimization procedure in 3 steps; this splitting is chosen to 

\begin{enumerate}[(1)]
	\item Pick a \emph{minimizing sequence $\{ x_n \}_{n \in \N}$} for which $\lim_{n \to \infty} \mathcal{E}(x_n) = E_0$. Such a sequence exists, because $\mathcal{E}$ is bounded from below. 
	\item Investigate whether $\{ x_n \}_{n \in \N}$ or at least a subsequence converges. Just imagine if $\mathcal{E}$ has two minima $x_0$ and $x_0'$. Then the alternating sequence would certainly be a minimizing sequence which does not converge. Another situation may occur, namely that the minimizer is located at “$\infty$”: take $\mathcal{E}(x) = \e^{- x^2}$ on $\R^d$, then no minimizer exists. 
	
	So how do we show that $\{ x_n \}_{n \in \N}$ has a convergent subsequence? Let us first assume in addition that $\mathcal{E}$ is \emph{coercive}, \ie $\mathcal{E}(x) \rightarrow \infty$ if and only if $\snorm{x} \rightarrow \infty$. In that case, we may assume that all the elements in the minimizing sequence satisfy $\mathcal{E}(x_n) \leq E_0 + 1$ (just discard all the others). Then since $\mathcal{E}$ is coercive, also $\snorm{x_n} \leq C$ holds for some $C > 0$, and the existence of a convergent subsequence follows from Bolzano-Weierstrass (every bounded sequence in $\R^d$ has a convergent subsequence). For simplicity, we denote this convergent subsequence by $\{ x_n \}_{n \in \N}$. 
	\item Then the limit point of this convergent subsequence, $x_0 = \lim_{n \to \infty} x_n$ is a candidate for a minimizer. \emph{If} $\mathcal{E}$ is continuous, then by 
	\begin{align*}
		E_0 = \lim_{n \to \infty} \mathcal{E}(x_n) = \mathcal{E} \Bigl ( \lim_{n \to \infty} x_n \Bigr ) 
		= \mathcal{E}(x_0)
	\end{align*}
	$x_0$ is a minimizer of $\mathcal{E}$. 
\end{enumerate}
Not surprisingly, things are more complicated if $\mathcal{X}$ is infinite-dimensional. 
\begin{enumerate}[(i)]
	\item First of all, showing that a functional is bounded from below is not as immediate as in the case of functions. For instance, consider the energy functional~\eqref{variation:eqn:energy_functional} associated to $H = - \Delta_x + V$ where $V$ is not bounded from below (think of something like the Coulomb potential). Then it is a priori \emph{not} clear whether $H$ -- and thus $\mathcal{E}$ -- is bounded from below. 
	\item In point (2) the Bolzano-Weiserstrass theorem was crucial to extract a convergent subsequence. Things are not as easy when going to infinite dimensions, because the unit ball $\bigl \{ \snorm{x} \leq 1 \; \vert \; x \in \mathcal{X} \bigr \}$ is \emph{no longer compact}. However, if the bidual $\mathcal{X}''$ is isometrically isomorphic to $\mathcal{X}$ (\ie $\mathcal{X}$ is reflexive), then by the Banach-Alaoglu theorem \cite[Theorem~IV.21]{Reed_Simon:M_cap_Phi_1:1972} every bounded sequence has a \emph{weakly} convergent subsequence (\cf Definition~\ref{spaces:defn:weak_convergence}). For instance, if $\mathcal{X}$ is a Hilbert space, it is reflexive. 
	\item The essential ingredient in (3) was the continuity of the functional, but this is usually either very hard to prove or even wrong. However, for the purpose of finding minimizers, \emph{weak lower semi-continuity} (w-lsc) suffices, \ie 
	\begin{align*}
		x_n \rightharpoonup x_0 
		\; \; \Longrightarrow \; \; 
		\liminf_{n \to \infty} \mathcal{E}(x_n) \geq \mathcal{E}(x_0)
		. 
	\end{align*}
	Because if $x_n$ is initially a minimizing sequence, then by definition 
	\begin{align}
		\lim_{n \to \infty} \mathcal{E}(x_n) = E_0 \geq \mathcal{E}(x_0)
		\geq E_0 
		\label{variation:eqn:wlsc_existence_minimizer}
	\end{align}
	minimizes the value of the functional. \marginpar{2014.03.25}
\end{enumerate}
\begin{theorem}\label{variation:thm:key_theorem_existence_minimizer}
	Assume that
	\begin{enumerate}[(i)]
		\item $\Omega \subseteq \mathcal{X}$ is closed under weak limits (weakly convergent sequences converge in $\Omega$), 
		\item $\mathcal{E}$ is weakly lower semi-continuous, and 
		\item $\mathcal{E}$ is coercive. 
	\end{enumerate}
	Then $\mathcal{E}$ is bounded from below and attains its minimum in $\Omega$, \ie there exists a possibly non-unique minimizer in $\Omega$. 
\end{theorem}
\begin{proof}
	Set $E_0 := \inf_{x \in \Omega} \mathcal{E}(x)$ and pick a minimizing sequence $\{ x_n \}_{n \in \N}$, \ie a sequence with $E_0 = \lim_{n \to \infty} \mathcal{E}(x_n)$. At this point, we do not know whether $E_0$ is finite or $-\infty$. In case $E_0$ is finite, we may assume without loss of generality that $\mathcal{E}(x_n) \leq E_0 + 1$ (simply discard all elements for which this does not hold). By the Banach-Alaoglu theorem, $\{ x_n \}_{n \in \N}$ contains a weakly convergent subsequence $\{ x_{n_k} \}_{k \in \N}$ for which $x_{n_k} \rightharpoonup x_0$. The point $x_0$ is the candidate for the minimizer. Given that $\mathcal{E}$ is w-lsc, we conclude $\lim_{k \to \infty} \mathcal{E}(x_{n_k}) \geq \mathcal{E}(x_0)$. However, in view of \eqref{variation:eqn:wlsc_existence_minimizer} we already know $\mathcal{E}(x_0) = E_0$ which not only shows the existence of \emph{a} minimizer, but also that $E_0 > -\infty$. 
\end{proof}
One \emph{application} of this theorem is to show the existence of solutions to a PDE, and the strategy is as follows: Find a functional so that its Euler-Lagrange equation are the PDE in question. Then the existence of a minimizer implies the existence of a solution to the PDE, because the PDE characterizes the critical points of the functional. 

To get uniqueness of solutions, we have to impose additional assumptions. One of the standard ones is \emph{convexity} which is defined analogously to the case of functions. 
\begin{definition}[Convex functional]
	Let $\Omega \subset \mathcal{X}$ be a convex subset and $\mathcal{E} : \Omega \longrightarrow \R$ a functional. 
	\begin{enumerate}[(i)]
		\item $\mathcal{E}$ is \emph{convex} iff $\mathcal{E} \bigl ( s x + (1 - s) y \bigr ) \leq s \, \mathcal{E}(x) + (1-s) \, \mathcal{E}(y)$ for all $x , y \in \Omega$ and $s \in [0,1]$. 
		\item $\mathcal{E}$ is \emph{strictly convex} iff $\mathcal{E} \bigl ( s x + (1 - s) y \bigr ) < s \, \mathcal{E}(x) + (1-s) \, \mathcal{E}(y)$ for all $x \neq y$ and $s \in (0,1)$. 
	\end{enumerate}
\end{definition}
Convexity has strong implications on minimizers just like in the case of functions. 
\begin{proposition}\label{variation:prop:minimizer_convexity}
	\begin{enumerate}[(i)]
		\item Every local minimum of a convex functional is a global minimum. 
		\item A convex functional has at most one minimizer. 
	\end{enumerate}
\end{proposition}
\begin{proof}
	\begin{enumerate}[(i)]
		\item Assume $x_0$ is just a local, but not a global minimum. Then there exists a neighborhood $V$ of $x_0$ such that $\mathcal{E}(x_0) \leq \mathcal{E}(x)$ for all $x \in V$ (local minimum) as well as a point $y \not\in V$ with $\mathcal{E}(y) < \mathcal{E}(x_0)$ (local but not global minimum). Connecting $x_0$ and $y$ by a line segment, convexity implies in conjunction with $\mathcal{E}(y) < \mathcal{E}(x_0)$ 
		\begin{align*}
			\mathcal{E} \bigl ( s y + (1 - s) x_0 \bigr ) \leq s \, \mathcal{E}(y) + (1-s) \, \mathcal{E}(x_0) < \mathcal{E}(x_0)
			. 
		\end{align*}
		Hence, for $s$ small enough, $x' = s y + (1-s) x_0 \in V$ and we have found points in $V$ for which $\mathcal{E}(x') < \mathcal{E}(x_0)$ -- in contradiction to the assumption that $x_0$ is a local minimum. Hence, every local minimum is also a global minimum. 
		\item Assume there exist two distinct minimizers $x_0$ and $x_0'$ with $\mathcal{E}(x_0) = E_0 = \mathcal{E}(x_0')$. Then strict convexity 
		\begin{align*}
			E_0 \leq \mathcal{E} \bigl ( s x_0 + (1-s) x_0' \bigr ) < s \, \mathcal{E}(x_0) + (1-s) \, \mathcal{E}(x_0') = E_0 
		\end{align*}
		leads to a contradiction, and the minimum -- if it exists -- is unique. 
	\end{enumerate}
\end{proof}
Note that the definition of convexity does not require $\mathcal{E}$ to be once or twice Gâteaux-differentiable. However, if we assume in addition that the function is once or twice differentiable, we obtain additional characterizations of convexity. 
\begin{proposition}\label{variation:prop:convexity_differential_characterization}
	Assume $\mathcal{E} : \Omega \longrightarrow \R$ is Gâteaux differentiable. 
	\begin{enumerate}[(i)]
		\item $\mathcal{E}$ is convex iff $\mathcal{E}(x) \geq \mathcal{E}(y) + \bigl ( \dd \mathcal{E}(y) \bigr )(x-y)$ for all $x,y \in \Omega$
		\item $\mathcal{E}$ is convex iff $\bigl ( \dd \mathcal{E}(x) - \dd \mathcal{E}(y) \bigr )(x-y) \geq 0$ for all $x,y \in \Omega$
		\item $\mathcal{E}$ is convex and twice Gâteaux differentiable iff $\bscpro{y}{\dd^2 \mathcal{E}(x) y} \geq 0$ for all $x,y \in \Omega$
	\end{enumerate}
\end{proposition}
\begin{proof}
	\begin{enumerate}[(i)]
		\item “$\Rightarrow$:” The convexity of $\mathcal{E}$ can be expressed by regrouping the terms in the definition, 
		\begin{align*}
			\mathcal{E} \bigl ( y + s (x - y) \bigr ) \leq \mathcal{E}(y) + s \bigl ( \mathcal{E}(x) - \mathcal{E}(y) \bigr ) 
			, 
		\end{align*}
		which is equivalent to 
		\begin{align*}
			\mathcal{E}(x) - \mathcal{E}(y) \geq \frac{\mathcal{E} \bigl ( y + s (x-y) \bigr ) - \mathcal{E}(y)}{s} 
			. 
		\end{align*}
		Taking the limit $s \searrow 0$ yields 
		\begin{align*}
			\mathcal{E}(x) - \mathcal{E}(y) \geq \bigl ( \dd \mathcal{E}(y) \bigr )(x-y)
		\end{align*}
		which implies $\mathcal{E}(x) \geq \mathcal{E}(y) + \bigl ( \dd \mathcal{E}(y) \bigr )(x-y)$. 
		
		“$\Leftarrow$:” Upon substituting $x \mapsto x$ and $y \mapsto y + s (x-y)$ as well as $y \mapsto x$ and $y \mapsto y + s (x-y)$ yields 
		\begin{align*}
			\mathcal{E}(x) &\geq \mathcal{E} \bigl ( y + s (x-y) \bigr ) + (1-s) \, \Bigl ( \mathcal{E} \bigl ( y + s (x-y) \bigr ) \Bigr )(x-y)
			, 
			\\
			\mathcal{E}(y) &\geq \mathcal{E} \bigl ( y + s (x-y) \bigr ) - s \, \Bigl ( \mathcal{E} \bigl ( y + s (x-y) \bigr ) \Bigr )(x-y)
			, 
		\end{align*}
		and if we multiply the first inequality by $s$, the second with $1-s$ and add the two, we obtain $\mathcal{E} \bigl ( s x + (1-s) y \bigr ) \leq s \, \mathcal{E}(x) + (1-s) \, \mathcal{E}(y)$. 
		\item “$\Rightarrow$:” From (i) we deduce $\mathcal{E}(x) \geq \mathcal{E}(y) + \bigl ( \dd \mathcal{E}(y) \bigr )(x-y)$ and $\mathcal{E}(y) \geq \mathcal{E}(x) + \bigl ( \mathcal{E}(x) \bigr )(y-x)$; adding the two yields the right-hand side. 
		
		“$\Leftarrow$:” Set $f(s) := \mathcal{E} \bigl ( y + s (x-y) \bigr )$. Then $f'(s) = \Bigl ( \dd \mathcal{E} \bigl ( y + s(x-y) \bigr ) \Bigr ) (x-y)$, and by assumption 
		\begin{align*}
			f'(s) - f'(0) &= \Bigl ( \dd \mathcal{E} \bigl ( y + s(x-y) \bigr ) - \mathcal{E}(y) \Bigr ) (x-y) \geq 0 
		\end{align*}
		holds. Integrating this equation with respect to $s$ over $[0,1]$, we obtain 
		\begin{align*}
			f(1) - f(0) - f'(0) = \mathcal{E}(x) - \mathcal{E}(y) - \bigl ( \dd \mathcal{E}(y) \bigr )(x-y) \geq 0 
			, 
		\end{align*}
		which by virtue of (i) is equivalent to convexity. 
		\item “$\Rightarrow$:” The convexity of $\mathcal{E}$ means 
		\begin{align*}
			\Bigl ( \dd \mathcal{E} \bigl ( x + s y \bigr ) - \mathcal{E}(x) \Bigr )(sy) \geq 0 
			, 
		\end{align*}
		and consequently, 
		\begin{align*}
			\bscpro{y}{\bigl ( \dd^2 \mathcal{E}(x) \bigr )(y)} &= \lim_{s \searrow 0} \frac{\Bigl ( \dd \mathcal{E} \bigl ( x + s y \bigr ) - \mathcal{E}(x) \Bigr )(sy)}{s} 
			\\
			&= \lim_{s \searrow 0} \Bigl ( \dd \mathcal{E} \bigl ( x + s y \bigr ) - \mathcal{E}(x) \Bigr ) y
			. 
		\end{align*}
		“$\Leftarrow$:” The convexity follows from a Taylor expansion of $\mathcal{E}$. 
	\end{enumerate}
\end{proof}
There is an analogous version of the proposition characterizing \emph{strict} convexity; since the proofs are virtually identical, we leave it to the reader to modify them appropriately. 
\begin{corollary}
	Assume $\mathcal{E} : \Omega \longrightarrow \R$ is Gâteaux differentiable. 
	\begin{enumerate}[(i)]
		\item $\mathcal{E}$ is strictly convex iff $\mathcal{E}(x) > \mathcal{E}(y) + \bigl ( \dd \mathcal{E}(y) \bigr )(x-y)$ for all $x \neq y$
		\item $\mathcal{E}$ is strictly convex iff $\bigl ( \dd \mathcal{E}(x) - \dd \mathcal{E}(y) \bigr )(x-y) > 0$ for all $x \neq y$
		\item $\mathcal{E}$ is strictly convex and twice Gâteaux differentiable iff $\bscpro{y}{\dd^2 \mathcal{E}(x) y} > 0$ for all $x \neq y$
	\end{enumerate}
\end{corollary}
Now we go back to the problem at hand, existence of minimizers. The first step is the following 
\begin{lemma}\label{variation:lem:convex_functional_w-lsc}
	Suppose $\mathcal{E}$ is a convex functional defined on a closed subset $\Omega \subseteq \mathcal{X}$ of a reflexive Banach space $\mathcal{X}$; moreover, we assume its Gâteaux derivative exists for all $x \in \Omega$. Then $\mathcal{E}$ is w-lsc. 
\end{lemma}
\begin{proof}
	Pick an arbitrary $x \in \mathcal{X}$ which we leave fixed. Moreover, let $(x_n)_{n \in \N}$ be a sequence which converges weakly to $x$. Then by the characterization of convexity in Proposition~\ref{variation:prop:convexity_differential_characterization}~(i) yields 
	\begin{align*}
		\mathcal{E}(x_n) \geq \mathcal{E}(x) + \bigl ( \dd \mathcal{E}(x) \bigr ) (x_n - x)
		. 
	\end{align*}
	Upon taking the limit $n \to \infty$, the last term on the right-hand side vanishes since $x_n \rightharpoonup x$, and we deduce $\mathcal{E}$ is weakly lower semicontinuous, 
	\begin{align*}
		\lim_{n \to \infty} \mathcal{E}(x_n) \geq \mathcal{E}(x) 
		. 
	\end{align*}
	%
	% 
	% 
	% 
	% 
	% Let $x \in \mathcal{X}$ be arbitrary and fixed. The convexity of $\mathcal{E}$ implies the convexity of $f(s) := \mathcal{E}(x + s y)$ for all $y \in \mathcal{X}$. 
	% % CHANGED Why is that true? 
	% Hence, the function $\partial_s f(s) = \partial_s \mathcal{E}(x + s y)$ is monotonically non-decreasing. 
	% In what follows, we set $y = x_n - x$ where eventually $x_n \rightharpoonup x$. Then the difference 
	% %
	% \begin{align*}
	% 	\mathcal{E}(x_n) - \mathcal{E}(x) &= \int_0^1 \dd s \, \partial_s \mathcal{E}(x + s y)
	% 	= \int_0^1 \dd s \, \bigl ( \partial_s f(s) - \partial_s f(0) \bigr ) + \partial_s f(0) 
	% \end{align*}
	% %
	% can be expressed in terms of $f$, and the monotonicity of $f$ means
	% %
	% \begin{align*}
	% 	\mathcal{E}(x_n) \geq \mathcal{E}(x) + \partial_s f(0) 
	% 	= \mathcal{E}(x) + \bigl ( \dd \mathcal{E}(x) \bigr ) (x_n - x)
	% 	. 
	% \end{align*}
	% %
	% Now if $\{ x_n \}_{n \in \N}$ is a sequence which converges weakly to $x$, the second term $\bigl ( \dd \mathcal{E}(x) \bigr ) (x_n - x) \rightarrow 0$ vanishes as $x_n \rightharpoonup x$, and hence, taking $\liminf_{n \to \infty}$ on both sides yields 
	% %
	% \begin{align*}
	% 	\liminf_{n \to \infty} \mathcal{E}(x_n) \geq \mathcal{E}(x) + \liminf_{n \to \infty} \bigl ( \dd \mathcal{E}(x) \bigr ) (x_n - x) = \mathcal{E}(x) 
	% 	, 
	% \end{align*}
	% %
	% \ie $\mathcal{E}$ is w-lsc. 
\end{proof}
\begin{theorem}\label{variation:thm:existence_uniquness_convex_functional}
	Assume 
	\begin{enumerate}[(i)]
		\item $\Omega \subseteq \mathcal{X}$ is closed under weak limits (weakly convergent sequences converge in $\Omega$), 
		\item $\mathcal{E}$ is strictly convex and 
		\item $\mathcal{E}$ is coercive. 
	\end{enumerate}
	Then there exists a \emph{unique} minimizer. 
\end{theorem}
\begin{proof}
	The convexity of $\mathcal{E}$ implies weakly lower semicontinuity (Lemma~\ref{variation:lem:convex_functional_w-lsc}). Hence, Theorem~\ref{variation:thm:key_theorem_existence_minimizer} applies and we know there exists \emph{a} minimizer, and in view of Proposition~\ref{variation:prop:minimizer_convexity}~(ii) this minimizer is necessarily unique. 
\end{proof}
As mentioned earlier, variational calculus can be used to show existence of solutions if the PDE in question coincide with the Euler-Lagrange equations of a functional. If the functional is in addition strictly convex, then our arguments here show that the solution is \emph{unique}.  
% section key_points_for_a_rigorous_analysis (end)

% \section{Symmetries and Noether's theorem} % (fold)
% \label{variation:symmetries}
% % 
% %
% \begin{itemize}
% 	\item homework: translational symmetry $\Rightarrow$ conservation of momentum
% 	\item take material from the bible of classical mechanics
% \end{itemize}
% %
% 
% 
% % section symmetries_and_noether_s_theorem (end)
% 
% 
\section{Bifurcations} % (fold)
\label{variation:bifurcation}
The idea to understand certain physical phenomena as incarnations of a “bifurcation” where certain properties of a system change abruptly. Before we show how this meshes with the preceding content of the chapter, let us quickly explore one incarnation from physics in order to introduce some of the necessary terminology. 

Phase transitions, for instance, are points where some \emph{order parameter} changes abruptly when an \emph{external parameter} is changed. For instance, in a ferromagnet the order parameter is the macroscopic \emph{magnetization} while the external parameter is temperature. Below a critical temperature, the Curie temperature, the microscopic magnets can align to produce a non-zero macroscopic magnetization. For instance, another magnet can be used to align these microscopic magnets. This magnetization persists even when the magnet is heated -- up to a certain specific \emph{critical temperature} when the magnetization suddenly drops to $0$. In summary, below the critical temperature, there are two states, the unmagnetized state where the macroscopic magnetization $M(T) = 0$ vanishes and the magnetized state where $M(T) \neq 0$. Above the critical temperature, only the $M(T) = 0$ state persists. Other physical effects such as superconductivity can be explained along the exact same lines. 

% TODO add a picture for bifurcation
To tie this section to the theme of the chapter, let us consider a \emph{parameter-dependent functional} $\mathcal{E} : \R \times \mathcal{D} \longrightarrow \R$ where $\mathcal{D}$ is dense in some Hilbert space $\Hil$; we will denote the external parameter with $\mu$ and the Banach space variable with $x$, \ie $\mathcal{E}(\mu,x)$. The bifurcation analysis starts with 
\begin{align*}
	F(\mu,x) := \partial_x \mathcal{E}(\mu,x)
\end{align*}
which enters the Euler-Lagrange equations and determines the stationary points of the functional. Here, the partial derivative is defined in terms of the scalar product as 
\begin{align*}
	\bigl ( \dd_x \mathcal{E}(\mu,x) \bigr )(y) = \bscpro{y}{\partial_x \mathcal{E}(\mu,x)} 
	. 
\end{align*}
In what follows, we assume that the “normal solution” $x = 0$ solves $F(\mu,0) = 0$ for all $\mu \in \R$ and we want to know whether there is a bifurcation solution $x(\mu) \neq 0$. Then in the present context, we define the notion of 
\begin{definition}[Bifurcation point]
	$(\mu_0,0)$ is a bifurcation point if there exists $x(\mu)$ on a neighborhood $[\mu_0,\mu_0 + \delta)$ so that $x(\mu) \neq 0$ on $(\mu_0,\mu_0 + \delta)$ and $F \bigl ( \mu , x(\mu) \bigr ) = 0$. 
\end{definition}
A consequence of the Implicit Function Theorem \cite[Theorem~1.6]{Teschl:nonlinear_functional_analysis:2010} is the following 
\begin{proposition}
	If $(\mu_0,0)$ is a bifurcation point, then $\dd_x F(\mu_0,0)$ does not have a bounded inverse. 
\end{proposition}
\begin{proof}[Sketch]
	Because if $\dd_x F(\mu_0,0)$ were invertible, then by the Implicit Function Theorem \cite[Theorem~1.6]{Teschl:nonlinear_functional_analysis:2010} we could extend the trivial solution $x(\mu)$ in a vicinity of $\mu_0$. But given the multivaluedness, this is clearly false. 
\end{proof}
%
% %
% \begin{example}[Branching in the linear case]
% 	% TODO add example
% 	Add example. 
% \end{example}
% %
% %
% \begin{example}[No branching despite $\dd_x F(\mu_0,0)$ not being invertible]
% 	% TODO add example
% 	Add example. 
% \end{example}
% %
The last example has shown us that $\dd_x F(\mu_0,0)$ not being invertible is just \emph{necessary} but \emph{not sufficient}. Hence, we need to impose additional conditions, and one of the standard results in this direction is a Theorem due to Krasnoselski \cite{Krasnoselski:non_linear_functional_analysis:1984}
\begin{theorem}[Krasnoselski]
	Assume $\mathcal{D} \subseteq \Hil$ is a dense subset of a Hilbert space and the $\Cont^1$ map $F : \R \times \mathcal{D} \longrightarrow \Hil$ is such that 
	\begin{enumerate}[(i)]
		\item $\dd_x F(\mu,x)$ is a $\Cont^1$ map at $(\mu_0,0)$, 
		\item $\dd_x F(\mu_0,0)$ is selfadjoint, 
		\item $0$ is an isolated eigenvalue of \emph{odd} and finite multiplicity, and 
		\item there exists $v \in \ker \dd_x F(\mu_0,0)$ such that $\bscpro{v}{\partial_{\mu} \dd_x F(\mu_0,0) v} \neq 0$. 
	\end{enumerate}
	Then $(\mu_0,0)$ is a bifurcation point. \marginpar{2014.03.27}
\end{theorem}
The main ingredient is a procedure called \emph{Lyapunov-Schmidt reduction}: define the linear operator $L(\mu) := \dd_x F(\mu,0)$ and the orthogonal projection $P$ onto $\ker L(\mu_0)$. The projection and its orthogonal complement $P^{\perp} = 1 - P$ induce a splitting of the Hilbert space 
\begin{align*}
	\Hil = \ran P \oplus \ran P^{\perp} 
	. 
\end{align*}
Hence, any $x = x_{\parallel} + x_{\perp}$ can be uniquely decomposed into $x_{\parallel}$ from the finite-dimensional space $\ran P$ and $x_{\perp} \in \ran P^{\perp}$. Consequently, also the equation $F(x,\mu) = 0 \in \Hil$ can equivalently be seen as 
\begin{subequations}
	\begin{align}
		F_{\parallel}(\mu,x_{\parallel},x_{\perp}) := P F \bigl ( \mu , x_{\parallel} + x_{\perp} \bigr ) &= 0 \in \ran P
		,
		\label{variation:eqn:defn_F_parallel}
		\\
		F_{\perp}(\mu,x_{\parallel},x_{\perp}) := P^{\perp} F \bigl ( \mu , x_{\parallel} + x_{\perp} \bigr ) &= 0 \in \ran P^{\perp}
		.
		\label{variation:eqn:defn_F_perp}
	\end{align}
\end{subequations}
The idea of the Liapunov-Schmidt reduction is to solve these equations iteratively: Given $(\mu,x_{\parallel})$ we first solve the \emph{branching equation}
\begin{align*}
	F_{\perp}(\mu,x_{\parallel},x_{\perp}) = 0 
\end{align*}
in the vicinity of the branching point $(\mu_0,0)$. By assumption $\dd_{x_{\perp}} F(\mu_0,0,0)$ has a bounded inverse, and thus, the Implicit Function Theorem \cite[Theorem~1.6]{Teschl:nonlinear_functional_analysis:2010} yields a solution $x_{\perp}(\mu,x_{\parallel})$. We then proceed by inserting $x_{\perp}(\mu,x_{\parallel})$ into $F_{\parallel}$, thereby obtaining a function 
\begin{align*}
	f(\mu,x_{\parallel}) := F_{\parallel} \bigl ( \mu , x_{\parallel} , x_{\perp}(\mu,x_{\parallel}) \bigr ) 
\end{align*}
that depends only on finitely many variables ($n = \dim \ker L(\mu_0)$ is finite by assumption), \ie $f : \R \times \R^n \longrightarrow \R^n$. The remaining equation $f(\mu,x_{\parallel})$ then needs to be solved by other means. 
\begin{proof}
	We will use the notation introduced in the preceding paragraphs. First of all the map 
	\begin{align*}
		F_{\perp} : \bigl ( \R \times P \mathcal{D} \bigr ) \times P^{\perp} \mathcal{D} \longrightarrow \ran P^{\perp}
	\end{align*}
	satisfies the assumptions of the Implicit Function Theorem \cite[Theorem~1.6]{Teschl:nonlinear_functional_analysis:2010}: evidently, it inherits the $\Cont^1$ property from $F$, and $F_{\perp}(\mu,0,0) = 0$ holds for all $\mu \in \R$. Moreover, $\dd_{x_{\perp}} F_{\perp}(\mu_0,0,0)$ has a bounded inverse because $0 \in \sigma \bigl ( \dd_x F(\mu_0,0) \bigr )$ is an isolated eigenvalue. Consequently, there exists a neighborhood of $(\mu_0,0) \in \R \times P \mathcal{D}$ and a function $x_{\perp}(\mu,x_{\parallel})$ on that neighborhood that is uniquely determined by 
	\begin{align*}
		F_{\perp} \bigl ( \mu , x_{\parallel} , x_{\perp}(\mu,x_{\parallel}) \bigr ) = 0
		. 
	\end{align*}
	Later on, we will crucially need the technical estimate 
	\begin{align}
		x_{\perp} , \partial_{\mu} x_{\perp} = \order \bigl ( \snorm{x_{\parallel}} \, \sabs{\mu - \mu_0} \bigr ) + o(\snorm{x_{\parallel}})
		% . 
		\label{variation:eqn:asymptotics_x_perp}
	\end{align}
	as $\mu \rightarrow \mu_0$ and $x_{\parallel} \rightarrow 0$, but we postpone the proof of \eqref{variation:eqn:asymptotics_x_perp} until the end. 
	
	For simplicity, let us only prove Krasnoselski's Theorem for $\dim \ker L(\mu_0) = 1$. 
	% we will sketch the proof for the case $n > 1$ in Remark~\ref{variation:rem:n_greater_1}
	Then the kernel is spanned by a single normalized vector $v \in \ker L(\mu_0)$, and we can write $x_{\parallel} = s v$ for some $s \in \R$. Combining \eqref{variation:eqn:defn_F_parallel} with $x_{\perp}(\mu,sv)$ yields a \emph{scalar} function $f : \R \times \R \longrightarrow \R$, and we are looking for solutions to 
	\begin{align}
		f(\mu,s) := \frac{1}{s} \, \Bscpro{v}{F \bigl ( \mu, sv + x_{\perp}(\mu,sv) \bigr )} = 0 
		. 
		\label{variation:eqn:defn_f_mu_s}
	\end{align}
	We have added the prefactor $\nicefrac{1}{s}$ to make $\lim_{s \searrow 0} \partial_{\mu} f(\mu_0,s) \neq 0$, but more on that later. As in the case of functions, we can Taylor expand 
	\begin{align}
		F(\mu,x) &= F(\mu,0) + \bigl ( \dd F(\mu,0) \bigr ) x + R(\mu,x) 
		% \notag \\
		% &
		= L(\mu) x + R(\mu,x)
		,
		\label{variation:eqn:Taylor_expansion_F}
	\end{align}
	to first order; here, the remainder $R(\mu,x) = o(\snorm{x})$ vanishes as $x \to 0$. Plugging \eqref{variation:eqn:Taylor_expansion_F} into \eqref{variation:eqn:defn_f_mu_s} yields 
	\begin{align*}
		f(\mu,s) &= \bscpro{v}{L(\mu) v} + \bscpro{v}{L(\mu) s^{-1} x_{\perp}(\mu,sv)} 
		+ \Bscpro{v}{s^{-1} \, R \bigl ( \mu , sv + x_{\perp}(\mu,sv) \bigr )}
		. 
	\end{align*}
	We obtain the branching solution through the Implicit Function Theorem: once we have proven $\partial_{\mu} f(\mu_0,0) \neq 0$, we obtain a function $s(\mu)$ in the vicinity of $\mu_0$ so that $f \bigl ( \mu , s(\mu) \bigr ) = 0$. And this function $s(\mu)$ also defines the branching solution 
	\begin{align}
		x(\mu) := s(\mu) \, v + x_{\perp} \bigl ( \mu , s(\mu) v \bigr )
		\neq 0 
		\label{variation:eqn:branching_solution}
	\end{align}
	satisfying $F \bigl ( \mu , x(\mu) \bigr ) = 0$ by construction. 
	
	Hence, we compute the derivative 
	\begin{align}
		\partial_{\mu} f(\mu,s) &= \bscpro{v}{\partial_{\mu} L(\mu) v} + \bscpro{v}{\partial_{\mu} L(\mu) s^{-1} x_{\perp}(\mu,sv)} + \bscpro{v}{L(\mu) s^{-1} \, \partial_{\mu} x_{\perp}(\mu,sv)} 
		+ \notag \\
		&\quad 
		+ \Bscpro{v}{s^{-1} \, \partial_{\mu} R \bigl ( \mu , sv + x_{\perp}(\mu,sv) \bigr )} 
		+ \notag \\
		&\quad 
		+ \Bscpro{v}{\bigl ( \dd_x R \bigl ( \mu , sv + x_{\perp}(\mu,sv) \bigr ) \bigr ) \, s^{-1} \, \partial_{\mu} x_{\perp}(\mu,sv)}
		, 
		\label{variation:eqn:partial_mu_f}
	\end{align}
	and use \eqref{variation:eqn:asymptotics_x_perp} in conjunction with assumption (iv), 
	\begin{align*}
		\bscpro{v}{\partial_{\mu} L(\mu_0) v} = \bscpro{v}{\partial_{\mu} \dd_x F(\mu_0,0) v} \neq 0 
		, 
	\end{align*}
	to deduce that all but the first term vanish in the limit $\mu \to \mu_0$ and $s \searrow 0$. First of all, \eqref{variation:eqn:asymptotics_x_perp} tells us that 
	\begin{align*}
		s^{-1} x_{\perp}(\mu,sv) , s^{-1} \, \partial_{\mu} x_{\perp}(\mu,sv) = \order(\sabs{\mu - \mu_0}) + o(1) 
	\end{align*}
	goes to $0$ in the limit $\mu \to \mu_0$ and $s \searrow 0$. Moreover, the assumption that $\partial_{\mu} F$ exists as a Gâteaux derivative means we can Taylor expand this function to first order, and the terms are just the $\mu$-derivatives of \eqref{variation:eqn:Taylor_expansion_F}. Consequently, $\partial_{\mu} L(\mu_0)$ is bounded and $\partial_{\mu} R(\mu,x) = o(\snorm{x})$, and we deduce that the second and third term in \eqref{variation:eqn:partial_mu_f} vanish. \marginpar{2014.04.01}
	
	Moreover, the terms involving $R$ also vanish: by assumption also $\partial_{\mu} F$ has a Taylor expansion at $(\mu,0)$, so that 
	\begin{align*}
		s^{-1} \, &\norm{\partial_{\mu} R \bigl ( \mu , sv + x_{\perp}(\mu,sv) \bigr )} 
		= o \bigl ( \bnorm{sv + x_{\perp}(\mu,sv)} \bigr ) 
	\end{align*}
	is necessarily small, and if we combine that with \eqref{variation:eqn:asymptotics_x_perp}, we deduce that this term vanishes as $\mu \to \mu_0$ and $s \searrow 0$. 
	
	The last term can be treated along the same lines, $\dd_x F(\mu,x)$ is $\Cont^1$ in $(\mu_0,0)$ by assumption so that 
	\begin{align*}
		\dd_x R \bigl ( \mu,sv + x_{\perp}(\mu,sv) \bigr ) &= \dd_x F \bigl ( \mu,sv + x_{\perp}(\mu,sv) \bigr ) - \dd_x L(\mu) \bigl ( sv + x_{\perp}(\mu,sv) \bigr ) 
		\\
		&= o(\snorm{x}) \xrightarrow{\mu \to \mu_0 , \; s \searrow 0} 0 
		. 
	\end{align*}
	In conclusion, we have just shown $\partial_{\mu} f(\mu_0,0) \neq 0$, the Implicit Function Theorem applies, and we obtain the branching solution~\eqref{variation:eqn:branching_solution}. 
	
	All that remains are proofs of the estimates \eqref{variation:eqn:asymptotics_x_perp}. We use the Taylor expansion~\eqref{variation:eqn:Taylor_expansion_F} to rewrite $F_{\perp}(\mu,x_{\parallel},x_{\perp}) = 0$ as 
	\begin{align*}
		P^{\perp} L(\mu) P^{\perp} x_{\perp}(\mu,sv) + P^{\perp} L(\mu) sv + P^{\perp} R \bigl ( \mu , sv + x_{\perp}(\mu,sv) \bigr ) = 0 
		. 
	\end{align*}
	$L^{\perp}(\mu) := P^{\perp} L(\mu) P^{\perp}$ is invertible on $\ran P^{\perp}$ in a neighborhood of $\mu_0$, because $\mu \mapsto L(\mu)$ is continuous and thus, the spectral gap around $0$ does not suddenly collapse \cite[Chapter~VII.3]{Kato:perturbation_theory:1995}. Now we can solve for $x_{\perp}$ and insert $L(\mu_0) sv = 0$ free of charge, 
	\begin{align}
		x_{\perp}(\mu,sv) &= - L^{\perp}(\mu)^{-1} \, P^{\perp} \Bigl ( L(\mu) sv + R \bigl ( \mu , sv + x_{\perp}(\mu,sv) \bigr ) \Bigr ) 
		\notag \\
		&
		= - L^{\perp}(\mu)^{-1} \, P^{\perp} \Bigl ( \bigl ( L(\mu) - L(\mu_0) \bigr ) sv + R \bigl ( \mu , sv + x_{\perp}(\mu,sv) \bigr ) \Bigr ) 
		\label{variation:eqn:selfconsistent_x_perp} 
		\\
		&= \order \bigl ( \snorm{sv} \, \sabs{\mu - \mu_0} \bigr ) + o \bigl ( \bnorm{sv + x_{\perp}(\mu,sv)} \bigr )
		.
		\notag  
	\end{align}
	The first term is clearly $\order \bigl ( s \, \sabs{\mu - \mu_0} \bigr )$. Initially, we merely obtain $o \bigl ( \bnorm{sv + x_{\perp}(\mu,sv)} \bigr )$ for the second term, but if $x_{\perp}(\mu,sv) \rightarrow \tilde{x} \neq 0$ as $\mu \to \mu_0$ and $s \searrow 0$, we would obtain 
	\begin{align*}
		\tilde{x} = o(\snorm{\tilde{x}})
	\end{align*}
	which is absurd. Thus, the asymptotic behavior of $x_{\perp}(\mu,sv)$ is described by \eqref{variation:eqn:asymptotics_x_perp}. The proof for $\partial_{\mu} x_{\perp}(\mu,sv)$ is analogous, one starts from equation~\eqref{variation:eqn:selfconsistent_x_perp} and uses 
	\begin{align*}
		\partial_{\mu} L^{\perp}(\mu)^{-1} = - L^{\perp}(\mu)^{-1} \, \partial_{\mu} L^{\perp}(\mu) \, L^{\perp}(\mu)^{-1}%
	\end{align*}
	as well as the assumption that $\mu \mapsto \dd_x F(\mu,x)$ is $\Cont^1$ in $(\mu_0,0)$. This concludes the proof.\marginpar{2014.04.03}
\end{proof}
%
% %
% \begin{remark}\label{variation:rem:n_greater_1}
% 	We will sketch the modifications for the proof in case $n > 1$: according to assumption (iv) there exists $v_n \in \ker L(\mu_0)$, $\snorm{v_n} = 1$, such that $\bscpro{v_n}{\partial_{\mu} \dd_x F(\mu_0,0) v_n} \neq 0$. We can complete $\{ v_1 , \ldots , v_n \}$ to an orthonormal basis of $\ker L(\mu_0)$, and consequently express each $v = v(s) = \sum_{j = 1}^n s_j \, v_j \in \ker L(\mu_0)$ in terms of a coordinate vector $s = (s_1 , \ldots , s_n)$. Then we need to consider the function  
% 	%
% 	\begin{align*}
% 		f(\mu,s) := F \bigl ( \mu , v(s) + x_{\perp}(\mu,v(s)) \bigr ) 
% 		= f(\tilde{\mu},s_n) 
% 	\end{align*}
% 	%
% 	where in the last step we have introduced $\tilde{\mu} = \bigl ( \mu , s_1 , \ldots , s_{n-1} \bigr )$; the $s_1 , \ldots , s_{n-1}$ are free parameters . Assuming the derivative $\nabla_{\mu} f \bigl ( \mu_0,s_1 , \ldots , s_{n-1} , 0 \bigr )$ is an invertible $n \times n$ matrix, the Implicit Function Theorem now yields $s_n(\tilde{\mu})$
% 	
% 	
% \end{remark}
% %
% section bifurcation (end)
% chapter Variational calculus (end)
% \include{./chapter_11}
% \include{./chapter_12}
% \include{./chapter_13}
% \include{./chapter_14}
% \include{./chapter_15}
% \include{./chapter_16}

\backmatter

\printbibliography

\end{document}